\documentclass[preprint,12pt]{elsarticle}

\usepackage{amsmath, amssymb, amsthm}
\usepackage[margin=2.5cm]{geometry}
\usepackage{setspace}
\usepackage{booktabs}
\usepackage{graphicx}
\graphicspath{{pdf/}}
\usepackage{subcaption}
\usepackage{placeins}
\usepackage{tikz}
\usetikzlibrary{positioning}
\usepackage{hyperref}
\hypersetup{colorlinks=true, linkcolor=blue, citecolor=blue}

\newtheorem{proposition}{Proposition}
\newtheorem{theorem}{Theorem}
\newtheorem{lemma}{Lemma}
\newtheorem{remark}{Remark}

\journal{Journal of Computational Physics}

\begin{document}

\begin{frontmatter}
\title{Wave-Particle Decomposition for Kinetic Equations I: \\ Theory and Numerics}
\author[iapcm]{Chang Liu\corref{cor1}}
\cortext[cor1] {Corresponding author.}
\ead{liuchang@iapcm.ac.cn}
\author[hkustmath,hkustmae,hkustsri]{Kun Xu}
\ead{makxu@ust.hk}

\address[iapcm]{Institute of Applied Physics and Computational Mathematics, Beijing, P.R. China}
\address[hkustmath]{Department of Mathematics, Hong Kong University of Science and Technology, Clear Water Bay, Kowloon, Hong Kong}
\address[hkustmae]{Department of Mechanical and Aerospace Engineering, Hong Kong University of Science and Technology, Clear Water Bay, Kowloon, Hong Kong}
\address[hkustsri]{Shenzhen Research Institute, Hong Kong University of Science and Technology, Shenzhen, China}

\begin{abstract}

This paper presents a wave-particle decomposition (WPD) for kinetic relaxation equations, formulated around a local evolution timescale and its associated kinetic horizon. By leveraging the characteristic integral solution, we decompose the distribution function into an analytically accumulated wave component and a purely kinetic particle component. The latter is defined by the collisionless transport that survives beyond a prescribed local domain of influence, termed the horizon. This continuous formulation yields a unified wave-particle system valid across the entire Knudsen spectrum, comprising a source-free total conservation law, a wave equation, and a particle equation. The wave operator admits a Chapman--Enskog expansion, whose moments yield Euler and Navier--Stokes fluxes with horizon-dependent coefficients, while the particle equation governs the remaining non-equilibrium kinetic transport.

At the algorithmic level, this system is discretized by a conservative macro-micro method. The total conservative variables are advanced by a finite-volume update using the sum of a Navier--Stokes gas-kinetic wave flux and a particle flux computed by either a deterministic discrete-ordinate $(S_N)$ method or a Monte Carlo representation. Unlike the global time-step splitting in the unified gas-kinetic wave-particle (UGKWP) method, the present partition is defined at the PDE level and governed by the local ratio of evolution timescale to relaxation time. The particle component is therefore a fractional kinetic population generated by the collisionless factor of the integral solution. Formal analysis establishes the asymptotic-preserving continuum limit, rarefied-regime consistency, and regime-adaptive scaling of active kinetic degrees of freedom. Numerical tests in one, two, and three dimensions validate the accuracy, multiscale capability, and efficiency of the framework.

\end{abstract}

\end{frontmatter}

\section{Introduction}\label{sec:intro}

The Boltzmann equation provides a fundamental kinetic description of dilute gas dynamics, resolving molecular transport and collision processes at the scale of the mean free path and collision time~\cite{boltzmann1872,cercignani1988}. Its asymptotic analysis forms a central link between microscopic mechanics and macroscopic fluid dynamics: in the small-Knudsen-number limit, the Chapman--Enskog and Hilbert expansions formally recover the Euler and Navier--Stokes equations~\cite{chapman1970}, thereby connecting kinetic theory with the program underlying Hilbert's sixth problem. This theoretical connection, however, does not immediately translate into an efficient numerical method. A direct kinetic discretization must still resolve microscopic scales even when the solution is close to the Navier--Stokes limit, while a purely macroscopic solver loses validity when the local flow departs from equilibrium. Consequently, a practical multiscale method must retain the kinetic description where non-equilibrium effects are essential, while seamlessly reducing, at the discrete level, to the appropriate hydrodynamic model in the continuum regime.

The numerical treatment of multiscale kinetic equations has led to several major methodological paradigms. The direct simulation Monte Carlo (DSMC) method remains the standard particle approach for highly rarefied gas dynamics, but its statistical noise and stringent cell-size and time-step restrictions become computationally prohibitive near the continuum limit~\cite{bird1994}. Deterministic moment methods reduce velocity dependence through finite moment closures, providing compact macroscopic descriptions of non-equilibrium effects~\cite{grad1949}. To address stiffness, asymptotic-preserving (AP) implicit-explicit and penalized schemes treat the stiff relaxation term implicitly while preserving the limiting fluid behavior at the discrete level~\cite{pareschi2005,filbet2010}. Macro-micro decompositions further couple a macroscopic conservative system with a kinetic residual, offering a systematic route to AP discretizations and exact moment conservation~\cite{liu2004,gamba2019}.
Based on the direct modeling methodology for computational fluid dynamics, the unified gas-kinetic scheme (UGKS) constructs multiscale fluxes from the integral solution of the kinetic model, thereby coupling transport and collision over the mesh-scale time step~\cite{xu2014direct,xu2010,huang2012ugks}; its applications have since extended from continuum-rarefied gas dynamics to multi-component plasma transport and dilute gas-particle multiphase systems~\cite{liu2017plasma,liu2019multiphase}. The discrete unified gas-kinetic scheme (DUGKS) provides a related finite-volume discrete-velocity framework for flows across all Knudsen numbers~\cite{guo2013}. More recently, the unified gas-kinetic wave-particle (UGKWP) method was introduced as a wave-particle realization in which the collision probability determines the sampled particle population, while the remaining contribution is evolved analytically~\cite{liu2020ugkwp,zhu2019ugkwp}. Beyond gas dynamics, UGKWP-type wave-particle modeling has been successfully applied to multiscale photon transport and non-equilibrium turbulent-flow modeling~\cite{li2020ugkwp,liu2021,yang2026turbulence}, as well as the analysis of generalized fluid dynamics equations~\cite{guo2026}. For steady rarefied gas simulations, the general synthetic iterative scheme (GSIS) accelerates convergence by solving kinetic equations in tandem with synthetic macroscopic equations~\cite{su2020,luo2024}.
The present work synthesizes two key structural ideas from these developments (Fig.~\ref{fig:wpd_lineage}): integral-solution wave-particle modeling and the macro-micro coupling between a conservation law and a kinetic component. The novel contribution of this work is the definition of the wave-particle split at the governing equation level through a local kinetic horizon. Under this formulation, the particle component represents a physically consistent, collisionless fractional population rather than a signed kinetic residual, offering a robust and regime-adaptive framework for multiscale flow simulations.

An effective multiscale kinetic method must simultaneously preserve two essential properties. The first is the asymptotic-preserving (AP) property: as the Knudsen number approaches zero, the numerical scheme must recover the correct hydrodynamic limit on a mesh and time-step that do not resolve the microscopic mean free path or collision time. This requirement guarantees both accuracy and consistency across disparate flow regimes. The second is a regime-adaptive kinetic representation---a concept central to the UGKWP methodology---wherein the active kinetic degrees of freedom dynamically adjust to the local state of relaxation~\cite{liu2020ugkwp,zhu2019ugkwp,li2020ugkwp}. Guided by the collisionless factor inherent in the integral solution, only the non-equilibrated fraction of the distribution function is represented by particles, while the equilibrated contribution is resolved analytically by the wave component. Consequently, asymptotic preservation and regime-adaptive kinetic representation are not competing objectives but rather complementary requirements: the former guarantees the mathematical correctness of the limiting macroscopic solution, while the latter ensures that the computational representation and its associated cost scale dynamically with the local physical regime.

The present formulation places the wave-particle decomposition directly at the level of the governing kinetic equation. In the traditional UGKWP method, the analytical-versus-sampled partition is determined during the numerical update by the time-step collision probability, $\exp{(-\Delta t/\tau)}$. Consequently, on a non-uniform mesh, a single highly refined cell can dictate the global time step and artificially bias the particle fraction across the entire domain. In contrast, the proposed formulation defines this split using a local kinetic horizon $\mathcal{T}$, whereby the partition is governed by the local relaxation-to-evolution ratio $\mathcal{T}/\tau$. The resulting decomposition is entirely local with respect to the mesh-scale relaxation process. Furthermore, because the particle component is generated dynamically via the collisionless factor in the integral solution, it represents a physically consistent, fractional kinetic population rather than a signed mathematical residual obtained by subtracting an equilibrium state.

This construction yields both a rigorous continuous formulation and an efficient numerical realization. At the continuous level, the integral solution establishes a unified, closed wave-particle system of equations valid across the entire Knudsen spectrum, comprising a wave equation, a particle equation, and a source-free total conservation law. At the algorithmic level, this continuous system naturally translates into a conservative macro-micro numerical method. Within this framework, the total variables are updated in a finite-volume form, while the cell-local horizon ensures a regime-adaptive kinetic representation that remains completely decoupled from the smallest time step in the mesh.

\begin{figure}[!htbp]
\centering
\includegraphics[width=\linewidth]{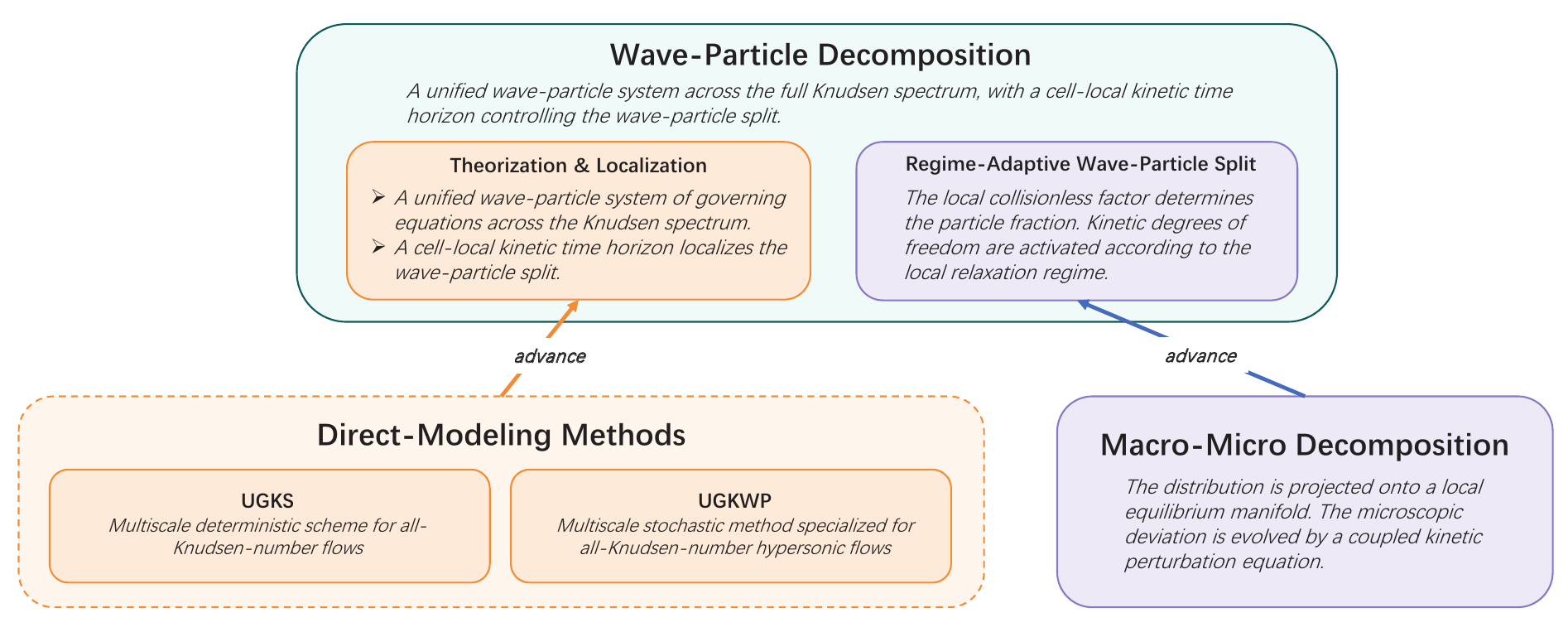}
\caption{Positioning of the present wave-particle decomposition relative to the unified gas-kinetic (UGKS, UGKWP) and macro-micro decomposition lines.}
\label{fig:wpd_lineage}
\end{figure}

The remainder of this paper is organized to reflect this theory-algorithm structure. Section~\ref{sec:wp_decomposition} develops the continuous wave-particle decomposition from the characteristic integral solution, defines the local kinetic horizon, and derives the total conservation law alongside the wave and particle equations. Section~\ref{sec:coupled_scheme} constructs the conservative macro-micro algorithm, in which the wave flux is evaluated via a gas-kinetic macroscopic solver, while the particle equation is resolved using either Monte Carlo sampling or a deterministic discrete-ordinate
($S_N$) discretization. Section~\ref{sec:analysis} provides a formal analysis of the asymptotic-preserving continuum limit, consistency in the rarefied regime, and the regime-adaptive kinetic representation. Section~\ref{sec:results} presents one-, two-, and three-dimensional numerical benchmarks to assess the accuracy, multiscale behavior, and computational efficiency of the proposed method. Finally, Section~\ref{sec:conclusion} summarizes the main findings and outlines future research directions.

\section{Wave-Particle Decomposition}\label{sec:wp_decomposition}

This section derives the wave-particle decomposition from the characteristic integral solution. A local kinetic horizon separates the relaxation history into an accumulated wave contribution and a collisionless particle contribution, yielding the wave equation, the particle equation, and the total conservation law. The final subsection assembles the conservative system and identifies the horizon-weighted Euler and Navier--Stokes structure of the wave flux.

\subsection{Integral Solution and Wave-Particle Operators}\label{sec:kinetic_integral_kernel}

We consider a kinetic model equation where the molecular velocity distribution function, \(f(\boldsymbol{x},\boldsymbol{\xi},t)\), relaxes toward a target distribution, \(\mathcal{G}[f]\), over a local relaxation time, \(\tau\):
\begin{equation}\label{eq:kinetic_relaxation}
    \partial_t f + \boldsymbol{\xi} \cdot \nabla f = \frac{\mathcal{G}-f}{\tau}.
\end{equation}
Here, \(\boldsymbol{x}\) denotes the spatial coordinate and \(\boldsymbol{\xi}\) represents the molecular velocity. For notational convenience in subsequent sections, we define the material derivative along a molecular trajectory as \(D_t = \partial_t + \boldsymbol{\xi} \cdot \nabla\). The target distribution can be specified as a Maxwellian (BGK), Shakhov, or ellipsoidal statistical (ES-BGK) target, provided it preserves the collision invariants:
\begin{equation}\label{eq:target_conservation_general}
    \left\langle \boldsymbol{\psi}\mathcal{G} \right\rangle = \left\langle \boldsymbol{\psi}f \right\rangle.
\end{equation}
By construction, the target distribution shares the same conservative moments as \(f\). Depending on the chosen relaxation model, it may also incorporate prescribed non-conserved moments---such as the heat flux in the Shakhov model---to accurately recover the desired macroscopic transport coefficients. The conservative variables are defined as
\begin{equation}
    \mathbf{U} = \left\langle \boldsymbol{\psi}f \right\rangle,\qquad
    \boldsymbol{\psi} = \left(1, \boldsymbol{\xi}, \frac{1}{2}|\boldsymbol{\xi}|^2\right)^T,
\end{equation}
where \(\mathbf{U} = (\rho, \rho\boldsymbol{u}, \rho E)^T\), with \(\rho\) representing the density, \(\boldsymbol{u}\) the macroscopic velocity, and \(\rho E\) the total energy density.

For the numerical examples presented in this work, we employ the Shakhov target~\cite{shakhov1968}:
\begin{equation}\label{eq:shakhov_target}
    \mathcal{G} = \mathcal{M} \left[ 1 + (1-\Pr)\frac{\boldsymbol{c} \cdot \boldsymbol{q}}{5pRT} \left(\frac{|\boldsymbol{c}|^2}{RT}-5\right) \right].
\end{equation}
Here, \(\boldsymbol{c} = \boldsymbol{\xi} - \boldsymbol{u}\) is the peculiar velocity, \(p = \rho R T\) is the pressure, \(T\) is the temperature, \(R\) is the specific gas constant, \(\Pr\) is the prescribed Prandtl number, and
\begin{equation}\label{eq:maxwellian}
    \mathcal{M} = \frac{\rho}{(2\pi RT)^{3/2}} \exp\left(-\frac{|\boldsymbol{\xi}-\boldsymbol{u}|^2}{2RT}\right)
\end{equation}
denotes the local Maxwellian distribution. The heat flux incorporated into the Shakhov correction is given by
\begin{equation}\label{eq:heat_flux_def}
    \boldsymbol{q} = \left\langle \frac{1}{2}|\boldsymbol{c}|^2\boldsymbol{c} f \right\rangle.
\end{equation}
This target inherently satisfies
\begin{equation}\label{eq:shakhov_conservation}
    \left\langle \boldsymbol{\psi}\mathcal{G} \right\rangle = \left\langle \boldsymbol{\psi}f \right\rangle = \mathbf{U},
\end{equation}
ensuring that the relaxation process strictly conserves mass, momentum, and total energy. Ultimately, the Shakhov correction adjusts the heat-flux moment to yield the correct Prandtl number while maintaining the conservative relaxation structure required for our formulation. The WPD framework derived below relies exclusively on the conservative relaxation form \eqref{eq:kinetic_relaxation} and the moment conservation property \eqref{eq:target_conservation_general}. The specific choice of the target distribution only influences the resulting transport coefficients and higher-order non-equilibrium moments.

Along the backward molecular trajectory
\begin{equation}
    \boldsymbol{x}_\sigma=\boldsymbol{x}-\boldsymbol{\xi}\sigma,\qquad t_\sigma=t-\sigma,
\end{equation}
we introduce the cumulative collision frequency
\begin{equation}\label{eq:Phi_def}
    \Phi(\sigma) = \int_0^{\sigma} \frac{1}{\tau_\nu} \mathrm{d}\nu.
\end{equation}
Here $\sigma$ is the backward flight time measured along the molecular path, and a subscript $\nu$ or $\sigma$ denotes evaluation at the retarded point $(\boldsymbol{x}-\boldsymbol{\xi}\nu,t-\nu)$ or $(\boldsymbol{x}-\boldsymbol{\xi}\sigma,t-\sigma)$. Hence $\partial_\sigma\Phi(\sigma)=1/\tau_\sigma$. The collisionless factor $e^{-\Phi(\sigma)}$ is the probability that a molecule travels over the interval $\sigma$ without a collision.

The corresponding integral solution can be written in terms of the relaxation kernel
\begin{equation}\label{eq:kernel}
    K(\sigma;\boldsymbol{x},\boldsymbol{\xi},t)
    =
    \frac{\mathcal{G}(\boldsymbol{x}-\boldsymbol{\xi}\sigma,\boldsymbol{\xi},t-\sigma)}
    {\tau(\boldsymbol{x}-\boldsymbol{\xi}\sigma,t-\sigma)}
    \exp\left[
    -\int_0^\sigma
    \frac{\mathrm{d}\nu}{\tau(\boldsymbol{x}-\boldsymbol{\xi}\nu,t-\nu)}
    \right].
\end{equation}
For any admissible kinetic history length $\mathcal{T}\ge0$, the solution at the current observation point satisfies
\begin{equation}\label{eq:integral_solution_T}
    f(\boldsymbol{x},\boldsymbol{\xi},t)
    =
    \int_0^\mathcal{T} K(\sigma)\,\mathrm{d}\sigma
    +
    e^{-\Phi(\mathcal{T})}f(\boldsymbol{x}-\boldsymbol{\xi} \mathcal{T},\boldsymbol{\xi},t-\mathcal{T}).
\end{equation}
Here $\mathcal{T}$ has the dimension of time and specifies how far the molecular history is represented analytically. This identity is the starting point for the wave-particle decomposition: the first term represents the relaxation contribution accumulated over the local kinetic horizon, while the second term represents the free-transport memory that survives beyond that horizon.

The integral solution defines a family of local operators parameterized by the admissible kinetic horizon $\mathcal{T}$. The analytical wave operator is
\begin{equation}\label{eq:wave_def}
    \begin{aligned}
    \mathcal{W}_\mathcal{T}(\boldsymbol{x},\boldsymbol{\xi},t)
    &=
    \int_0^{\mathcal{T}(\boldsymbol{x},t,\boldsymbol{\xi})}
    \frac{\mathcal{G}_\sigma}{\tau_\sigma}e^{-\Phi(\sigma)}
    \,\mathrm{d}\sigma .
    \end{aligned}
\end{equation}
Figure~\ref{fig:wave_operator} illustrates the structure of this local wave operator. The upper integration limit is the local kinetic horizon, the kernel contains the local relaxation target, and the exponential factor measures the collisionless attenuation along the backward molecular path.

\begin{figure}[!htbp]
\centering
\includegraphics[width=0.56\linewidth]{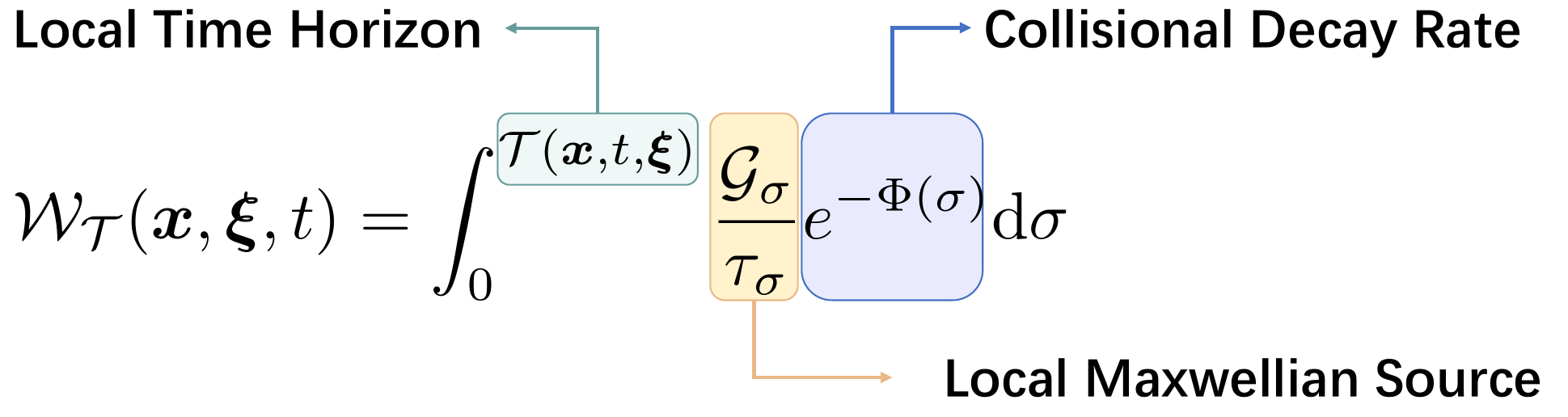}
\caption{Schematic of the local wave operator.}
\label{fig:wave_operator}
\end{figure}

The surviving particle free-transport operator is
\begin{equation}\label{eq:P_def}
    \mathcal{P}_\mathcal{T}(\boldsymbol{x},\boldsymbol{\xi},t)
    =
    e^{-\Phi(\mathcal{T})}f(\boldsymbol{x}-\boldsymbol{\xi} \mathcal{T},\boldsymbol{\xi},t-\mathcal{T}).
\end{equation}
Thus the distribution can be represented as the sum of a locally accumulated relaxation contribution and a surviving kinetic memory over the same molecular history interval. For compact notation the subscript $\mathcal{T}$ will be omitted when no ambiguity is possible. The derivative of the horizon along the molecular trajectory is
\begin{equation}\label{eq:DT_def}
    D_t \mathcal{T}=\partial_t \mathcal{T}+\boldsymbol{\xi}\cdot\nabla \mathcal{T} ,
\end{equation}
with the same convention if $\mathcal{T}$ also depends on $\boldsymbol{\xi}$ as a parameter.

\begin{lemma}[Exact evolution of the relaxation kernel]\label{lem:kernel}
The integral kernel $K(\sigma)$ satisfies the fundamental differential identity:
\begin{equation}\label{eq:kernel_identity}
    D_t K(\sigma) + \frac{\partial K(\sigma)}{\partial\sigma} = -\frac{K(\sigma)}{\tau_0},
\end{equation}
where $\tau_0=\tau(\boldsymbol{x},t)$. The identity follows only from the retarded characteristic structure of the relaxation kernel and does not require frozen coefficients.
\end{lemma}

\begin{proof}
Let $K(\sigma)=g(\sigma)e^{-\Phi(\sigma)}$ with $g(\sigma)=\mathcal{G}_\sigma/\tau_\sigma$. Since $g(\sigma)$ is evaluated at the retarded point, the derivative along the molecular trajectory satisfies $D_t g(\sigma)=-\partial_\sigma g(\sigma)$. Differentiation with respect to the backward flight time gives
\begin{equation}\label{eq:dsigma_K}
    \partial_\sigma K = \left[ \partial_\sigma g(\sigma) - g(\sigma)\partial_\sigma \Phi(\sigma) \right] e^{-\Phi}
    = \left[ \partial_\sigma g(\sigma) - \frac{g(\sigma)}{\tau_\sigma} \right] e^{-\Phi}.
\end{equation}
Moreover,
\begin{equation}
    D_t\Phi(\sigma) = \int_0^\sigma D_t \left( \frac{1}{\tau_\nu} \right) \mathrm{d}\nu = \int_0^\sigma -\partial_\nu \left( \frac{1}{\tau_\nu} \right) \mathrm{d}\nu = \frac{1}{\tau_0} - \frac{1}{\tau_\sigma}.
\end{equation}
Therefore
\begin{equation}\label{eq:Dt_K}
    D_t K = \left[ -\partial_\sigma g(\sigma) - g(\sigma)\left(\frac{1}{\tau_0} - \frac{1}{\tau_\sigma}\right) \right] e^{-\Phi}.
\end{equation}
Adding Eq.~\eqref{eq:dsigma_K} and Eq.~\eqref{eq:Dt_K} cancels the $\partial_\sigma g$ terms and the retarded relaxation terms $g/\tau_\sigma$, yielding Eq.~\eqref{eq:kernel_identity}.
\end{proof}

\subsection{The Wave Equation and Its Moment Balance}\label{sec:wave}

The wave component obeys an exact evolution equation. Its conservative moments form a macroscopic balance law that is not closed by itself: the right-hand side carries the relaxation transfer from the total target state together with the local-horizon exchange with the particle component. The hydrodynamic content of the wave flux---its Chapman--Enskog reduction to the Euler and Navier--Stokes levels with horizon-dependent coefficients---is developed in Section~\ref{sec:hydro_structure}, after the coupled conservative system has been assembled, so that the expansion acts on the source-free total conservation law rather than on the sourced wave balance.

\begin{proposition}[Wave evolution equation]\label{prop:wave_evol}
For any admissible horizon $\mathcal{T}$, the wave component satisfies
\begin{equation}\label{eq:W_evol}
    D_t\mathcal{W}+\frac{\mathcal{W}}{\tau_0}
    =
    \frac{\mathcal{G}_0}{\tau_0}-\mathcal B_\mathcal{T} ,
\end{equation}
where $\mathcal{G}_0=\mathcal{G}(\boldsymbol{x},\boldsymbol{\xi},t)$ and
\begin{equation}\label{eq:BL_def}
    \mathcal B_\mathcal{T}=K(\mathcal{T})\left(1-D_t \mathcal{T}\right)
    =
    \frac{\mathcal{G}_\mathcal{T}}{\tau_\mathcal{T}}e^{-\Phi(\mathcal{T})}\left(1-D_t \mathcal{T}\right)
\end{equation}
is the kinetic contribution crossing the boundary of the local relaxation horizon in flight-time space.
\end{proposition}

\begin{proof}
Applying the Leibniz rule to Eq.~\eqref{eq:wave_def} gives
\begin{equation}
    D_t\mathcal{W}
    =
    \int_0^\mathcal{T} D_t K(\sigma)\,\mathrm{d}\sigma
    +K(\mathcal{T})D_t \mathcal{T} .
\end{equation}
Using the kernel identity in Lemma~\ref{lem:kernel},
\begin{equation}
    D_t\mathcal{W}
    =
    \int_0^\mathcal{T} \left(-\partial_\sigma K-\frac{K}{\tau_0}\right)\mathrm{d}\sigma
    +K(\mathcal{T})D_t \mathcal{T}
\end{equation}
and $\Phi(0)=0$, $K(0)=\mathcal{G}_0/\tau_0$, therefore
\begin{equation}
    D_t\mathcal{W}
    =
    \frac{\mathcal{G}_0}{\tau_0}
    -K(\mathcal{T})
    -\frac{\mathcal{W}}{\tau_0}
    +K(\mathcal{T})D_t \mathcal{T} ,
\end{equation}
which is Eq.~\eqref{eq:W_evol}.
\end{proof}

Taking conservative moments of Eq.~\eqref{eq:W_evol} gives the macroscopic balance law carried by the wave component. Let
\begin{equation}
    \mathbf U_{\mathcal{W}}=\left\langle \boldsymbol{\psi}\mathcal{W} \right\rangle,
    \qquad
    \mathbf F_{\mathcal{W}}=\left\langle \boldsymbol{\xi}\boldsymbol{\psi}\mathcal{W} \right\rangle.
\end{equation}
Since $\left\langle \boldsymbol{\psi}\mathcal{G}_0 \right\rangle=\mathbf U$, the wave moment equation is
\begin{equation}\label{eq:W_moment_equation}
    \partial_t\mathbf U_{\mathcal{W}}+\nabla\cdot\mathbf F_{\mathcal{W}}
    =
    \frac{\mathbf U-\mathbf U_{\mathcal{W}}}{\tau_0}
    -\left\langle \boldsymbol{\psi}\mathcal B_\mathcal{T} \right\rangle.
\end{equation}
In component form,
\begin{equation}\label{eq:W_moment_component}
\left\{
\begin{aligned}
&\partial_t\rho_{\mathcal{W}}+\nabla\cdot\left\langle \boldsymbol{\xi}\mathcal{W} \right\rangle
    =
    \frac{\rho-\rho_{\mathcal{W}}}{\tau_0}
    -\left\langle \mathcal B_\mathcal{T} \right\rangle,\\
&\partial_t(\rho\boldsymbol{u})_{\mathcal{W}}+\nabla\cdot\left\langle \boldsymbol{\xi}\otimes\boldsymbol{\xi}\,\mathcal{W} \right\rangle
    =
    \frac{\rho\boldsymbol{u}-(\rho\boldsymbol{u})_{\mathcal{W}}}{\tau_0}
    -\left\langle \boldsymbol{\xi}\mathcal B_\mathcal{T} \right\rangle,\\
&\partial_t(\rho E)_{\mathcal{W}}
    +\nabla\cdot\left\langle \frac{1}{2}|\boldsymbol{\xi}|^2\boldsymbol{\xi}\,\mathcal{W} \right\rangle
    =
    \frac{\rho E-(\rho E)_{\mathcal{W}}}{\tau_0}
    -\left\langle \frac{1}{2}|\boldsymbol{\xi}|^2\mathcal B_\mathcal{T} \right\rangle.
\end{aligned}
\right.
\end{equation}
The wave balance represents the relaxed fraction of the total kinetic state. Its right-hand side contains relaxation from the total target state and the local-horizon exchange with the particle component. Since $\mathbf U_{\mathcal{W}}=\mathcal A_0\mathbf U+\mathcal O(\tau_0D_t\mathcal{G}_0)$ in the hydrodynamic expansion (Section~\ref{sec:hydro_structure}), closure of the conservative dynamics is obtained by coupling this wave balance with the particle moment balance and the source-free total conservation law derived below.

\subsection{The Particle Equation and Its Kinetic Closure}\label{sec:particle_operator}

The particle component represents the part of the characteristic history not included in the local relaxation integral. Its flux is a kinetic moment of $\mathcal{P}$ and is not closed by a hydrodynamic constitutive relation. Since $f=\mathcal{W}+\mathcal{P}$ at the level of the integral solution, the particle equation follows by subtracting the wave equation from the original kinetic equation. More precisely, subtracting Eq.~\eqref{eq:W_evol} from
\begin{equation}
    D_t f+\frac{f}{\tau_0}=\frac{\mathcal{G}_0}{\tau_0}
\end{equation}
gives the particle equation with the local-horizon source written explicitly,
\begin{equation}\label{eq:P_evol}
    D_t\mathcal{P}+\frac{\mathcal{P}}{\tau_0}
    =
    \frac{\mathcal{G}_\mathcal{T}}{\tau_\mathcal{T}}\,e^{-\Phi(\mathcal{T})}\left(1-D_t \mathcal{T}\right).
\end{equation}
For compact notation in the moment equations below, this local-horizon source is denoted by
\begin{equation}\label{eq:source_term}
    \mathcal B_\mathcal{T}
    =
    \frac{\mathcal{G}_\mathcal{T}}{\tau_\mathcal{T}}\,e^{-\Phi(\mathcal{T})}\left(1-D_t \mathcal{T}\right).
\end{equation}
This term is the boundary contribution at the local horizon in flight-time space. The factor $\mathcal{G}_\mathcal{T}/\tau_\mathcal{T}$ is the retarded relaxation production evaluated at the horizon, $e^{-\Phi(\mathcal{T})}$ is the collisionless probability from the horizon to the observation point, and $1-D_t \mathcal{T}$ is the signed relative speed between the molecular flight coordinate and the moving horizon.

To interpret the moving-boundary factor, introduce the forward molecular flight-time coordinate $s$ measured from the present point along a characteristic. The analytical wave part occupies the local interval $0\le s\le \mathcal{T}(\boldsymbol{x}+\boldsymbol{\xi} s,t+s,\boldsymbol{\xi})$, while the particle component lies beyond this moving kinetic horizon. Along the same characteristic, the signed relative speed between the molecular flight coordinate and the moving horizon is
\begin{equation}\label{eq:horizon_relative_speed}
    \frac{\mathrm{d}}{\mathrm{d} s}\left[
    s-\mathcal{T}(\boldsymbol{x}+\boldsymbol{\xi} s,t+s,\boldsymbol{\xi})
    \right]
    =
    1-D_t \mathcal{T} .
\end{equation}
Thus $\mathcal B_\mathcal{T}$ is the boundary flux through the wave-particle interface in flight-time space. Positive values of $1-D_t \mathcal{T}$ correspond to injection from the relaxation layer into the particle component, whereas negative values correspond to transfer in the opposite direction. In the particle implementation below, only the positive injection part is sampled, and the remaining exchange is retained in the wave reconstruction.

Taking moments of Eq.~\eqref{eq:P_evol} gives the particle macroscopic balance law. With
\begin{equation}
    \mathbf U_{\mathcal{P}}=\left\langle \boldsymbol{\psi}\mathcal{P} \right\rangle,
    \qquad
    \mathbf F_{\mathcal{P}}=\left\langle \boldsymbol{\xi}\boldsymbol{\psi}\mathcal{P} \right\rangle,
\end{equation}
one obtains
\begin{equation}\label{eq:P_moment_equation}
    \partial_t\mathbf U_{\mathcal{P}}+\nabla\cdot\mathbf F_{\mathcal{P}}
    =
    -\frac{\mathbf U_{\mathcal{P}}}{\tau_0}
    +\left\langle \boldsymbol{\psi}\mathcal B_\mathcal{T} \right\rangle.
\end{equation}
Equivalently,
\begin{equation}\label{eq:P_moment_component}
\left\{
\begin{aligned}
&\partial_t\rho_{\mathcal{P}}+\nabla\cdot\left\langle \boldsymbol{\xi}\mathcal{P} \right\rangle
    =
    -\frac{\rho_{\mathcal{P}}}{\tau_0}
    +\left\langle \mathcal B_\mathcal{T} \right\rangle,\\
&\partial_t(\rho\boldsymbol{u})_{\mathcal{P}}+\nabla\cdot\left\langle \boldsymbol{\xi}\otimes\boldsymbol{\xi}\,\mathcal{P} \right\rangle
    =
    -\frac{(\rho\boldsymbol{u})_{\mathcal{P}}}{\tau_0}
    +\left\langle \boldsymbol{\xi}\mathcal B_\mathcal{T} \right\rangle,\\
&\partial_t(\rho E)_{\mathcal{P}}
    +\nabla\cdot\left\langle \frac{1}{2}|\boldsymbol{\xi}|^2\boldsymbol{\xi}\,\mathcal{P} \right\rangle
    =
    -\frac{(\rho E)_{\mathcal{P}}}{\tau_0}
    +\left\langle \frac{1}{2}|\boldsymbol{\xi}|^2\mathcal B_\mathcal{T} \right\rangle.
\end{aligned}
\right.
\end{equation}
The particle fluxes in Eq.~\eqref{eq:P_moment_component} are kinetic moments of $\mathcal{P}$. They are not expressible, in general, by Euler or Navier--Stokes constitutive laws, and are therefore evaluated by a deterministic discrete-ordinate discretization or by a Monte Carlo particle representation in Section~\ref{sec:coupled_scheme}.

\subsection{The Coupled Conservative System and Its Hydrodynamic Structure}\label{sec:wpd}\label{sec:hydro_structure}

The preceding construction gives the wave-particle decomposition
\begin{equation}\label{eq:wp_decomposition}
    f=\mathcal{W}+\mathcal{P},
\end{equation}
where $\mathcal{W}$ is the local relaxation history in Eq.~\eqref{eq:wave_def} and $\mathcal{P}$ is the collisionless contribution in Eq.~\eqref{eq:P_def}. At the kinetic level, the two components satisfy the coupled microscopic wave-particle system
\begin{equation}\label{eq:wp_conservative_system}
\left\{
\begin{aligned}
&D_t\mathcal{W}+\frac{\mathcal{W}}{\tau_0}
    =
    \frac{\mathcal{G}_0}{\tau_0}
    -\frac{\mathcal{G}_\mathcal{T}}{\tau_\mathcal{T}}e^{-\Phi(\mathcal{T})}\left(1-D_t \mathcal{T}\right),
    \\
&D_t\mathcal{P}+\frac{\mathcal{P}}{\tau_0}
    =
    \frac{\mathcal{G}_\mathcal{T}}{\tau_\mathcal{T}}e^{-\Phi(\mathcal{T})}\left(1-D_t \mathcal{T}\right).
\end{aligned}
\right.
\end{equation}
These two equations contain equal and opposite local-horizon exchange terms, together with relaxation transfer between $\mathcal{P}$ and the target state. These terms are internal to the decomposition. They cancel in the conservative moment system and hence do not appear as external sources in the total equation.

Taking moments of the original kinetic equation~\eqref{eq:kinetic_relaxation} gives the total conservative macroscopic equations. Since the target distribution satisfies Eq.~\eqref{eq:target_conservation_general},
\begin{equation}\label{eq:macro_conservation}
\left\{
\begin{aligned}
&\partial_t\rho+\nabla\cdot\left(\mathbf F_{\mathcal{W}}^{\rho}+\mathbf F_{\mathcal{P}}^{\rho}\right)=0,\\
&\partial_t(\rho\boldsymbol{u})+\nabla\cdot\left(\mathbf F_{\mathcal{W}}^{\rho\boldsymbol{u}}+\mathbf F_{\mathcal{P}}^{\rho\boldsymbol{u}}\right)=0,\\
&\partial_t(\rho E)+\nabla\cdot\left(\mathbf F_{\mathcal{W}}^{\rho E}+\mathbf F_{\mathcal{P}}^{\rho E}\right)=0,
\end{aligned}
\right.
\end{equation}
with
\begin{equation}\label{eq:total_flux_split}
    \mathbf U=\left\langle \boldsymbol{\psi}(\mathcal{W}+\mathcal{P}) \right\rangle,
    \qquad
    \mathbf F
    =
    \left\langle \boldsymbol{\xi}\boldsymbol{\psi} f \right\rangle
    =
    \mathbf F_{\mathcal{W}}+\mathbf F_{\mathcal{P}},
    \qquad
    \mathbf F_{\mathcal{W}}=\left\langle \boldsymbol{\xi}\boldsymbol{\psi}\mathcal{W} \right\rangle,
    \qquad
    \mathbf F_{\mathcal{P}}=\left\langle \boldsymbol{\xi}\boldsymbol{\psi}\mathcal{P} \right\rangle.
\end{equation}
Here $\mathbf F^\rho$, $\mathbf F^{\rho\boldsymbol{u}}$, and $\mathbf F^{\rho E}$ denote the mass, momentum, and energy flux components of the corresponding macroscopic flux vector.

The same conservation law follows by adding the two component moment balances. The horizon-exchange terms $\mp\left\langle \boldsymbol{\psi}\mathcal B_\mathcal{T} \right\rangle$ cancel exactly, and the identity $\mathbf U=\mathbf U_{\mathcal{W}}+\mathbf U_{\mathcal{P}}$ eliminates the relaxation transfer:
\begin{equation}\label{eq:wp_moment_exchange}
    \frac{\mathbf U-\mathbf U_{\mathcal{W}}}{\tau_0}-\frac{\mathbf U_{\mathcal{P}}}{\tau_0}=0,
\end{equation}
recovering Eq.~\eqref{eq:macro_conservation}. The local-horizon source and the relaxation transfer are therefore internal exchanges between $\mathcal{W}$ and $\mathcal{P}$: they appear in the individual component balances but cancel in the total conservation law. This cancellation is the structural basis for the numerical update in Section~\ref{sec:coupled_scheme}, where the total conservative variables are advanced by the sum of the wave and particle fluxes.

The hydrodynamic reduction is applied to the wave flux in the source-free total conservation law~\eqref{eq:macro_conservation}, not to the sourced wave balance alone. The particle flux is retained as a kinetic closure for the remaining transport.

The asymptotic structure of the wave operator is obtained by expanding the relaxation history over a smooth local kinetic horizon. Let
\begin{equation}
    \eta=\frac{\mathcal{T}}{\tau_0}
\end{equation}
be the local horizon-to-relaxation ratio. Under the usual Chapman--Enskog ordering, the target distribution along the backward molecular path is expanded as
\begin{equation}
    \mathcal{G}_\sigma
    =
    \mathcal{G}_0-\sigma D_t\mathcal{G}_0
    +\frac{\sigma^2}{2}D_t^2\mathcal{G}_0
    +\mathcal O(\sigma^3),
\end{equation}
and the leading local relaxation factor is $e^{-\sigma/\tau_0}$. Whereas the wave evolution equation in Proposition~\ref{prop:wave_evol} is exact and requires no frozen coefficients, the asymptotic coefficients below are obtained under a local frozen-relaxation approximation, $\tau_\sigma\approx\tau_0$ and $\Phi(\sigma)\approx\sigma/\tau_0$. Thus Eq.~\eqref{eq:wave_expansion} is obtained directly by substituting the Taylor expansion of the retarded target into the wave integral~\eqref{eq:wave_def} and collecting the zeroth and first moments of the exponential relaxation kernel. This is the same local integral-solution expansion that underlies the gas-kinetic scheme for Euler and Navier--Stokes fluxes~\cite{xu2001gks}, here applied to the finite-horizon wave operator. Integrating term by term gives
\begin{equation}\label{eq:wave_expansion}
    \mathcal{W}
    =
    \mathcal A_0\mathcal{G}_0
    -\mathcal A_1\tau_0D_t\mathcal{G}_0
    +\mathcal O(\tau_0^2D_t^2\mathcal{G}_0),
\end{equation}
with
\begin{equation}\label{eq:A_coefficients}
    \mathcal A_0=1-e^{-\eta},\qquad
    \mathcal A_1=1-(1+\eta)e^{-\eta}.
\end{equation}
These coefficients depend only on the ratio $\eta=\mathcal{T}/\tau_0$ and carry the exponential decay factor $e^{-\eta}$ inherited from the integral solution: the analytical wave part grows toward the full hydrodynamic flux when the horizon is large compared with the relaxation time, and it decays toward zero when the dynamics approach free molecular transport. The expansion is applied here only to the wave component, which is the already relaxed part of the distribution by construction; the surviving kinetic memory is not expanded but is carried exactly by the particle flux $\mathbf F_{\mathcal{P}}$. Since $f=\mathcal{W}+\mathcal{P}$ is an exact algebraic identity, the expanded wave flux and the particle flux are complementary, and no contribution is double counted.

Taking moments of Eq.~\eqref{eq:wave_expansion} gives the macroscopic wave flux
\begin{equation}\label{eq:wave_flux_decomposition}
    \mathbf F_{\mathcal{W}}
    =
    \mathcal A_0\mathbf F_{\mathrm E}
    +\mathcal A_1\mathbf F_{\mathrm{NS}}
    +\mathcal R_{\mathcal{W}}^{(2)} ,
\end{equation}
where
\begin{equation}
    \mathbf F_{\mathrm E}=\left\langle \boldsymbol{\xi}\boldsymbol{\psi}\mathcal{M} \right\rangle,
\end{equation}
and $\mathbf F_{\mathrm{NS}}$ denotes the first-order kinetic correction generated by $-\tau_0D_t\mathcal{G}_0$. The remainder $\mathcal R_{\mathcal{W}}^{(2)}$ contains the second- and higher-order terms in the wave expansion. The first-order heat-flux coefficient is determined by the chosen relaxation target and gives the prescribed Prandtl number in the computations below. Because the relaxation target carries the total conservative moments, $\left\langle \boldsymbol{\psi}\mathcal{G}_0 \right\rangle=\mathbf U$, the fluxes $\mathbf F_{\mathrm E}$ and $\mathbf F_{\mathrm{NS}}$ are the Euler and Navier--Stokes constitutive fluxes of the total macroscopic state $\mathbf U$.

Substituting Eq.~\eqref{eq:wave_flux_decomposition} into the total conservation law~\eqref{eq:macro_conservation} gives the macroscopic equation associated with the coupled system,
\begin{equation}\label{eq:macro_wave_particle_expansion}
    \partial_t\mathbf U
    +\nabla\cdot\Big[
    \underbrace{\mathcal A_0\mathbf F_{\mathrm E}
    +\mathcal A_1\mathbf F_{\mathrm{NS}}
    +\mathcal R_{\mathcal{W}}^{(2)}}_{\text{wave Euler/NS flux}}
    +\underbrace{\mathbf F_{\mathcal{P}}}_{\text{particle flux}}
    \Big]=0 .
\end{equation}
The equation is source-free: the wave contributes the Euler and Navier--Stokes fluxes weighted by the horizon-dependent decay coefficients $\mathcal A_0$ and $\mathcal A_1$, while the particle flux carries the complementary kinetic flux. The particle flux is not purely non-equilibrium: in a uniform equilibrium state $f=\mathcal{M}$ the decomposition gives $\mathcal{W}=\mathcal A_0\mathcal{M}$ and $\mathcal{P}=(1-\mathcal A_0)\mathcal{M}$, so $\mathbf F_{\mathcal{P}}$ retains the fraction $1-\mathcal A_0$ of the local equilibrium flux together with the entire surviving non-equilibrium memory. This is the structural origin of the regime-adaptive behaviour. In an under-resolved cell, $\eta\to0$ gives $\mathcal A_0\sim\eta$ and $\mathcal A_1\sim\eta^2/2$, so the hydrodynamic wave fluxes are suppressed and the transport is carried by $\mathbf F_{\mathcal{P}}$; in a well-resolved cell, $\eta\to\infty$ gives $\mathcal A_0,\mathcal A_1\to1$ while the local-horizon source $\mathcal B_{\mathcal T}\propto e^{-\Phi(\mathcal{T})}\to0$ and $\mathbf F_{\mathcal{P}}\to0$.

Taking conservative moments of Eq.~\eqref{eq:wave_flux_decomposition} component by component, and reducing the time derivatives through the Euler-level relations, the wave-flux components are
\begin{equation}\label{eq:wave_flux_components}
\left\{
\begin{aligned}
\mathbf F_{\mathcal{W}}^{\rho}
    &=\mathcal A_0\,\rho\boldsymbol{u},\\
\mathbf F_{\mathcal{W}}^{\rho\boldsymbol{u}}
    &=\mathcal A_0\left(\rho\boldsymbol{u}\otimes\boldsymbol{u}+p\mathbf I\right)
    -\mathcal A_1\boldsymbol{\sigma}_{\mathrm{NS}}
    +\mathcal R_{\mathcal{W}}^{\rho\boldsymbol{u},(2)},\\
\mathbf F_{\mathcal{W}}^{\rho E}
    &=\mathcal A_0\left(\rho E+p\right)\boldsymbol{u}
    -\mathcal A_1\left(\boldsymbol{\sigma}_{\mathrm{NS}}\cdot\boldsymbol{u}-\boldsymbol{q}_{\mathrm{NS}}\right)
    +\mathcal R_{\mathcal{W}}^{\rho E,(2)},
\end{aligned}
\right.
\end{equation}
where $\mathbf I$ is the identity tensor, $E=e+\frac{1}{2}|\boldsymbol{u}|^2$ is the total specific energy, the viscous stress and heat flux are
\begin{equation}\label{eq:ns_constitutive}
    \boldsymbol{\sigma}_{\mathrm{NS}}
    =
    \mu\left[\nabla\boldsymbol{u}+(\nabla\boldsymbol{u})^T-\frac{2}{3}(\nabla\cdot\boldsymbol{u})\mathbf I\right],
    \qquad
    \boldsymbol{q}_{\mathrm{NS}}=-\kappa\nabla T,
\end{equation}
with $\mu=p\tau_0$ and $\kappa=5Rp\tau_0/(2\Pr)$ for a monatomic gas (the Shakhov model enters through the value of $\Pr$). The terms $\mathcal R_{\mathcal{W}}^{\rho\boldsymbol{u},(2)}$ and $\mathcal R_{\mathcal{W}}^{\rho E,(2)}$ denote second- and higher-order remainders of the wave expansion and are not used in the numerical flux below. The mass flux carries no first-order correction, consistent with the absence of diffusive mass transport for a single species.

Substituting Eq.~\eqref{eq:wave_flux_components} into the component conservation law~\eqref{eq:macro_conservation} gives the coupled macroscopic system at each truncation order of the wave flux, each written as an $\mathcal A_i$-weighted hydrodynamic wave flux plus the particle flux $\mathbf F_{\mathcal{P}}$. Retaining the Euler-level wave flux ($\mathcal A_0$) gives
\begin{equation}\label{eq:euler_system}
\left\{
\begin{aligned}
&\partial_t\rho
    +\nabla\cdot\left(\mathcal A_0\,\rho\boldsymbol{u}+\mathbf F_{\mathcal{P}}^{\rho}\right)=0,\\
&\partial_t(\rho\boldsymbol{u})
    +\nabla\cdot\left[\mathcal A_0\left(\rho\boldsymbol{u}\otimes\boldsymbol{u}+p\mathbf I\right)
    +\mathbf F_{\mathcal{P}}^{\rho\boldsymbol{u}}\right]=0,\\
&\partial_t(\rho E)
    +\nabla\cdot\left[\mathcal A_0\left(\rho E+p\right)\boldsymbol{u}
    +\mathbf F_{\mathcal{P}}^{\rho E}\right]=0;
\end{aligned}
\right.
\end{equation}
At this order the wave flux is purely advective: it transports the $\mathcal A_0$-weighted inviscid Euler flux, and all dissipative and high-order transport is left to the particle flux $\mathbf F_{\mathcal{P}}$. Including the first-order wave flux adds the viscous stress and heat conduction through the coefficient $\mathcal A_1$, giving the Navier--Stokes--level system
\begin{equation}\label{eq:ns_system}
\left\{
\begin{aligned}
&\partial_t\rho
    +\nabla\cdot\left(\mathcal A_0\,\rho\boldsymbol{u}+\mathbf F_{\mathcal{P}}^{\rho}\right)=0,\\
&\partial_t(\rho\boldsymbol{u})
    +\nabla\cdot\left[\mathcal A_0\left(\rho\boldsymbol{u}\otimes\boldsymbol{u}+p\mathbf I\right)
    -\mathcal A_1\boldsymbol{\sigma}_{\mathrm{NS}}
    +\mathbf F_{\mathcal{P}}^{\rho\boldsymbol{u}}\right]=0,\\
&\partial_t(\rho E)
    +\nabla\cdot\left[\mathcal A_0\left(\rho E+p\right)\boldsymbol{u}
    -\mathcal A_1\left(\boldsymbol{\sigma}_{\mathrm{NS}}\cdot\boldsymbol{u}-\boldsymbol{q}_{\mathrm{NS}}\right)
    +\mathbf F_{\mathcal{P}}^{\rho E}\right]=0;
\end{aligned}
\right.
\end{equation}
Here the wave part carries Newtonian viscous stress and Fourier heat conduction with horizon-scaled coefficients $\mathcal A_1\mu$ and $\mathcal A_1\kappa$, so the full transport coefficients are recovered only as $\mathcal A_1\to1$ at a well-resolved horizon; otherwise the remaining viscous and conductive transport is supplied by the particle flux. Thus the wave part contributes the $\mathcal A_i$-weighted Euler and Navier--Stokes fluxes, while the particle part $\mathbf F_{\mathcal{P}}$ carries the complementary kinetic flux. In the continuum limit $\eta\to\infty$, the coefficients satisfy $\mathcal A_0,\mathcal A_1\to1$ and the particle flux vanishes, $\mathbf F_{\mathcal{P}}\to0$, so the system reduces to the classical source-free Navier--Stokes equations.

Equivalently, one may introduce the local relaxation-to-horizon ratio $\varepsilon=\tau_0/\mathcal{T}=1/\eta$, a Knudsen-type small parameter measuring how far the local horizon resolves the collision scale. From Eq.~\eqref{eq:A_coefficients}, the continuum limit $\varepsilon\to0$ gives $\mathcal A_0,\mathcal A_1\to1$ with exponentially small corrections, so the coupled system reduces to the source-free Navier--Stokes equations at the retained order. In the rarefied limit $\varepsilon\to\infty$, the same coefficients vanish as $\mathcal A_0\sim1/\varepsilon$ and $\mathcal A_1\sim1/(2\varepsilon^2)$, switching off the hydrodynamic wave flux so that the transport is carried by the particle component.

The preceding derivation completes the equation-level wave-particle decomposition. The microscopic wave-particle equations describe the exchange between the relaxed wave and collisionless particle components, while the macroscopic conservative system provides the source-free closure for their total moments. The next section constructs a conservative numerical realization of this coupled system.

\section{Numerical Algorithm}\label{sec:coupled_scheme}

This section gives the algorithmic realization of the continuous wave-particle system derived in Section~\ref{sec:wp_decomposition}. The discretization starts from the conservative total equation, so that the macroscopic variables are updated by the sum of a wave flux and a particle flux. Following the wave-flux decomposition~\eqref{eq:wave_flux_decomposition}, the wave flux is evaluated by a gas-kinetic solver at the Navier--Stokes level, i.e. through the horizon-weighted Euler and Navier--Stokes contributions $\mathcal A_0\mathbf F_{\mathrm E}+\mathcal A_1\mathbf F_{\mathrm{NS}}$. The particle flux is evaluated either by a deterministic discrete-ordinate method or by Monte Carlo particles. The two parts are coupled through the relaxation target and the local-horizon source term.

\subsection{Conservative Update and Wave Flux}\label{sec:algorithm_wave_flux}

In the continuous decomposition, the horizon $\mathcal{T}$ is an admissible flight time. In the numerical scheme it is chosen from the cell-scale evolution time rather than from the global time step. A convenient choice used in the present computations is
\begin{equation}\label{eq:algorithm_local_horizon}
   \mathcal{T}_i^n=\mathrm{CFL}_l\,\Delta t_{i,\mathrm{loc}}^n,
    \qquad
    \Delta t_{i,\mathrm{loc}}^n
    =
    \left[
    \frac{1}{|\Omega_i|}
    \sum_{j\in N(i)}
    |\Gamma_{ij}|\left(
    |\boldsymbol{u}_i^n\cdot\boldsymbol{n}_{ij}|+a_i^n
    \right)
    \right]^{-1},
\end{equation}
where $a_i^n$ is the acoustic speed and $\mathrm{CFL}_l$ is the local-horizon CFL number. The actual time step for the conservative update is still restricted by the usual global CFL condition, for example $\Delta t^n=\mathrm{CFL}\min_i\Delta t_{i,\mathrm{loc}}^n$. The numerical tests reported below use $\mathrm{CFL}_l=1$, so the wave-particle split is controlled by the cell-scale evolution time.

Let $\Omega_i$ be a control volume, $|\Omega_i|$ its volume, $\Gamma_{ij}$ the interface shared with a neighboring cell $j$, $|\Gamma_{ij}|$ its area, and $\boldsymbol{n}_{ij}$ the outward unit normal from cell $i$ to cell $j$. The finite-volume update for the total conservative variables is
\begin{equation}\label{eq:fv_update_total}
    \mathbf U_i^{n+1}
    =
    \mathbf U_i^n
    -
    \frac{\Delta t}{|\Omega_i|}
    \sum_{j\in N(i)}
    \left(
    \mathbf F_{\mathcal{W},ij}
    +
    \mathbf F_{\mathcal{P},ij}
    \right)|\Gamma_{ij}|,
\end{equation}
where $\mathbf F_{\mathcal{W},ij}$ and $\mathbf F_{\mathcal{P},ij}$ are the time-averaged normal fluxes of the wave and particle components across $\Gamma_{ij}$.

The wave flux is computed at the Navier--Stokes level. Following the gas-kinetic scheme for viscous compressible flow~\cite{xu2001gks}, the interface flux is obtained from a local kinetic evolution rather than by adding a separate central viscous discretization. At the interface $\Gamma_{ij}$, a local coordinate is introduced with $x_n$ along the normal direction $\boldsymbol{n}_{ij}$. The primitive variables are reconstructed from the two neighboring cells to obtain left and right states $\mathbf W_{ij}^{L}$ and $\mathbf W_{ij}^{R}$ and their slopes. These states define the upwind initial distribution near the interface, while a compatible interface target state $\mathcal{G}_{ij}^{0}$ is constructed from the conservative compatibility condition.

The wave flux is built directly from the local kinetic integral solution, not by rescaling a completed Navier--Stokes flux. At the interface $\Gamma_{ij}$ the spatial and temporal derivatives of the target are written, following the gas-kinetic construction in Ref.~\cite{xu2001gks}, as
\begin{equation}\label{eq:gks_A_definition}
    \boldsymbol{\xi}\cdot\nabla\mathcal{G}_{ij}^{0}=a_{ij}\mathcal{G}_{ij}^{0},
    \qquad
    \partial_t\mathcal{G}_{ij}^{0}=a_{t,ij}\mathcal{G}_{ij}^{0},
\end{equation}
where $a_{ij}$ represents the directional spatial derivative along the molecular velocity. The derivative coefficients are obtained locally from the reconstructed macroscopic slopes. For each spatial direction $x_m$, $m=1,\ldots,d$, with $d$ the physical dimension,
\begin{equation}\label{eq:gks_derivative_basis}
    \partial_{x_m}\mathcal{G}_{ij}^{0}
    =
    a_{m,ij}\mathcal{G}_{ij}^{0},
    \qquad
    a_{m,ij}=\boldsymbol{\lambda}_{m,ij}\cdot\boldsymbol{\psi},
    \qquad
    a_{t,ij}=\boldsymbol{\lambda}_{t,ij}\cdot\boldsymbol{\psi}.
\end{equation}
The corresponding coefficient vectors are determined by moment matching and by the compatibility condition of the local kinetic evolution. With $\mathbf U_{ij}^{0}=\left\langle\boldsymbol{\psi}\mathcal{G}_{ij}^{0}\right\rangle$ and
\[
    \mathbf M_{ij}
    =
    \left\langle
    \boldsymbol{\psi}\boldsymbol{\psi}^{T}\mathcal{G}_{ij}^{0}
    \right\rangle ,
\]
the spatial derivatives satisfy
\begin{equation}\label{eq:gks_spatial_derivative_coefficients}
    \mathbf M_{ij}\boldsymbol{\lambda}_{m,ij}
    =
    \partial_{x_m}\mathbf U_{ij}^{0},
    \qquad m=1,\ldots,d,
\end{equation}
where $\partial_{x_m}\mathbf U_{ij}^{0}$ is obtained from the reconstructed conservative slopes. The time derivative coefficient is then fixed by
$\left\langle\boldsymbol{\psi}(\partial_t\mathcal{G}_{ij}^{0}+\boldsymbol{\xi}\cdot\nabla\mathcal{G}_{ij}^{0})\right\rangle=0$, namely
\begin{equation}\label{eq:gks_derivative_coefficients}
    \mathbf M_{ij}\boldsymbol{\lambda}_{t,ij}
    =
    -\left\langle
    \boldsymbol{\psi}
    \left(\sum_{m=1}^{d}\xi_m a_{m,ij}\right)
    \mathcal{G}_{ij}^{0}
    \right\rangle .
\end{equation}
Finally,
$a_{ij}=\sum_{m=1}^{d}\xi_m a_{m,ij}$ in Eq.~\eqref{eq:gks_A_definition}; in a one-dimensional normal reduction this becomes $a_{ij}=(\boldsymbol{\xi}\cdot\boldsymbol{n}_{ij})\alpha_{ij}$, where $\alpha_{ij}$ is the logarithmic derivative in the interface-normal direction. Thus the spatial derivative is supplied by reconstruction, while the time derivative is an algebraic consequence of conservation compatibility. The wave component at the interface is the relaxation integral truncated at the local horizon $\mathcal{T}_{ij}$. Carrying the integral only to $\sigma=\mathcal{T}_{ij}$ ties the horizon weights of the wave operator directly into the flux: with $\int_0^{\mathcal{T}_{ij}}\tfrac{1}{\tau_{ij}}e^{-\sigma/\tau_{ij}}\mathrm{d}\sigma=\mathcal A_{0,ij}$ and $\int_0^{\mathcal{T}_{ij}}\tfrac{1}{\tau_{ij}}e^{-\sigma/\tau_{ij}}\sigma\,\mathrm{d}\sigma=\tau_{ij}\mathcal A_{1,ij}$, together with the local expansion $\mathcal{G}_{ij}^{0}(t-\sigma)=\mathcal{G}_{ij}^{0}[1+a_{t,ij}t-\sigma(a_{ij}+a_{t,ij})+\mathcal O(t^2+\sigma^2)]$, the horizon integral yields the wave interface distribution in closed form,
\begin{equation}\label{eq:wave_interface_distribution}
    f_{\mathcal{W},ij}(t,\boldsymbol{\xi})
    =
    \mathcal{G}_{ij}^{0}\left[
    \mathcal A_{0,ij}
    -\mathcal A_{1,ij}\tau_{ij}a_{ij}
    +
    \left(\mathcal A_{0,ij}t-\mathcal A_{1,ij}\tau_{ij}\right)a_{t,ij}
    \right],
\end{equation}
where the horizon coefficients $\mathcal A_{0,ij}$ and $\mathcal A_{1,ij}$ are defined in Eq.~\eqref{eq:algorithm_A01} below. The decay coefficients are thus generated by the same integral solution that builds the flux, not applied to an assembled flux afterwards. The surviving free-transport memory $e^{-t/\tau_{ij}}f_{\mathrm{up},ij}$, with $f_{\mathrm{up},ij}$ the upwind reconstructed distribution selected by the sign of $\boldsymbol{\xi}\cdot\boldsymbol{n}_{ij}$, is excluded from $f_{\mathcal{W},ij}$ and is carried by the particle flux $\mathbf F_{\mathcal{P},ij}$ in Eq.~\eqref{eq:particle_flux_closure}.

The wave flux is the time average of the normal moment of Eq.~\eqref{eq:wave_interface_distribution} over $[0,\Delta t]$,
\begin{equation}\label{eq:wave_flux_algorithm}
    \mathbf F_{\mathcal{W},ij}
    =
    \frac{1}{\Delta t}\int_0^{\Delta t}
    \left\langle (\boldsymbol{\xi}\cdot\boldsymbol{n}_{ij})\,\boldsymbol{\psi}\,f_{\mathcal{W},ij}(t,\boldsymbol{\xi}) \right\rangle\,\mathrm{d} t
    =
    \left\langle (\boldsymbol{\xi}\cdot\boldsymbol{n}_{ij})\,\boldsymbol{\psi}\,
    \mathcal{G}_{ij}^{0}\left(
    b_{0,ij}+b_{a,ij}\,a_{ij}+b_{t,ij}\,a_{t,ij}
    \right) \right\rangle .
\end{equation}
The step time-integration coefficients are the first two time moments over the step,
\begin{equation}\label{eq:gks_time_coeffs}
    \theta_0=\frac{1}{\Delta t}\int_0^{\Delta t}\mathrm{d} t=1,
    \qquad
    \theta_1=\frac{1}{\Delta t}\int_0^{\Delta t}t\,\mathrm{d} t=\frac{\Delta t}{2},
\end{equation}
and combining them with the horizon coefficients $\mathcal A_{0,ij},\mathcal A_{1,ij}$ produces the single wave coefficient set
\begin{equation}\label{eq:wave_combined_coeffs}
\left\{
\begin{aligned}
    b_{0,ij}&=\theta_0\,\mathcal A_{0,ij}=\mathcal A_{0,ij},\\
    b_{a,ij}&=-\theta_0\,\mathcal A_{1,ij}\tau_{ij}=-\mathcal A_{1,ij}\tau_{ij},\\
    b_{t,ij}&=\theta_1\,\mathcal A_{0,ij}-\theta_0\,\mathcal A_{1,ij}\tau_{ij}
    =\tfrac12\Delta t\,\mathcal A_{0,ij}-\mathcal A_{1,ij}\tau_{ij}.
\end{aligned}
\right.
\end{equation}
The single coefficient set $(b_{0,ij},b_{a,ij},b_{t,ij})$ combines the gas-kinetic time-integration constants with the horizon decay coefficients, and its structure is transparent. The term $b_{0,ij}=\mathcal A_{0,ij}$ is the horizon-weighted Euler flux. The $\tfrac12\Delta t\,\mathcal A_{0,ij}$ part of $b_{t,ij}$ is the first-order temporal evolution of the equilibrium, so the Euler part carries its own time derivative; the $-\mathcal A_{1,ij}\tau_{ij}$ parts of $b_{a,ij}$ and $b_{t,ij}$ are the Chapman--Enskog viscous stress and heat conduction acting on the spatial and temporal target gradients. A single gas-kinetic moment evaluation thus returns the wave flux at the Navier--Stokes level with the horizon weighting built in, recovering the full Navier--Stokes gas-kinetic flux as $\mathcal A_{0,ij},\mathcal A_{1,ij}\to1$.

The local coefficients are evaluated from
\begin{equation}\label{eq:algorithm_A01}
    \eta_{ij}=\frac{\mathcal{T}_{ij}}{\tau_{ij}},\qquad
    \mathcal A_{0,ij}=1-e^{-\eta_{ij}},\qquad
    \mathcal A_{1,ij}=1-(1+\eta_{ij})e^{-\eta_{ij}} .
\end{equation}
Thus, when $\mathcal{T}_{ij}/\tau_{ij}$ is large, the wave flux approaches the full Navier--Stokes GKS flux; when $\mathcal{T}_{ij}/\tau_{ij}$ is small, the wave flux is reduced and the remaining transport is carried by $\mathbf F_{\mathcal{P},ij}$.

For the target used in the computations, the viscosity and heat conductivity entering the wave flux are
\begin{equation}\label{eq:algorithm_transport_coeff}
    \mu=p\tau,\qquad
    \kappa=\frac{5}{2}\frac{Rp\tau}{\Pr}.
\end{equation}
Equivalently, if a viscosity law is prescribed, for example
\begin{equation}\label{eq:viscosity_law}
    \mu=\mu_{\mathrm{ref}}\left(\frac{T}{T_{\mathrm{ref}}}\right)^\omega ,
\end{equation}
then the relaxation time is set by $\tau=\mu/p$ and the heat conductivity follows from the prescribed Prandtl number. In implementation, the factor $\mathcal A_{1,ij}$ enters through the combined coefficients $b_{a,ij}$ and $b_{t,ij}$ in Eq.~\eqref{eq:wave_combined_coeffs}; equivalently, it may be absorbed once into the transport coefficients used to construct the viscous flux. The corresponding effective wave transport coefficients are
\begin{equation}\label{eq:effective_wave_transport}
    \mu_{\mathcal{W},ij}=\mathcal A_{1,ij}\mu_{ij},\qquad
    \kappa_{\mathcal{W},ij}=\mathcal A_{1,ij}\kappa_{ij}.
\end{equation}

\begin{remark}[Single counting of viscosity]
The factor $\mathcal A_{1,ij}$ is an activation weight for the wave's Navier--Stokes part, not an additional collision time, so it must be applied once---either to $\mathbf F_{\mathrm{vis},ij}^{\mathrm{GKS}}$ in Eq.~\eqref{eq:wave_flux_algorithm} or to the transport coefficients in Eq.~\eqref{eq:effective_wave_transport}, but not both, which would give $\mathcal A_{1,ij}^2\mu_{ij}$. The wave thus carries a viscosity $\mathcal A_{1,ij}\mu_{ij}$, the remainder being supplied by the particle flux; as $\eta_{ij}\to\infty$, $\mathcal A_{1,ij}\to1$ and the full Navier--Stokes GKS viscosity is recovered (Theorem~\ref{thm:ap}).
\end{remark}

The macroscopic variables entering the GKS flux are reconstructed by a second-order van Leer limiter. For a scalar cell variable $\phi_i$ in one coordinate direction, the limited slope is
\begin{equation}\label{eq:vanleer_limiter}
    \delta\phi_i
    =
    \operatorname{VL}(\phi_i-\phi_{i-1},\phi_{i+1}-\phi_i),
    \qquad
    \operatorname{VL}(a,b)
    =
    \frac{\operatorname{sign}(a)+\operatorname{sign}(b)}{2}
    \frac{2|a||b|}{|a|+|b|+\epsilon_{\mathrm{lim}}},
\end{equation}
where $\epsilon_{\mathrm{lim}}$ is a small positive number used only to avoid division by zero. The left and right interface states are then obtained from the reconstructed cell profiles and are used to construct the interface target distribution, its Maxwellian part, and the spatial derivatives required by the GKS Navier--Stokes flux.

\subsection{Discretization of the Particle Equation}\label{sec:algorithm_particle}

The particle equation is solved in a conservative transport form. The particle flux provides the missing kinetic contribution in Eq.~\eqref{eq:fv_update_total},
\begin{equation}\label{eq:particle_flux_closure}
    \mathbf F_{\mathcal{P},ij}
    =
    \left\langle (\boldsymbol{\xi}\cdot\boldsymbol{n}_{ij})\boldsymbol{\psi}\mathcal{P} \right\rangle_{ij},
\end{equation}
and closes the total macroscopic finite-volume update together with the wave flux. The relaxation sink $-\mathcal{P}/\tau_0$ and the local-horizon source $\mathcal B_\mathcal{T}$ are internal exchanges between $\mathcal{W}$ and $\mathcal{P}$; they do not appear as external sources in the total conservative update.

For non-negative particle weights, only the positive injection part is sampled. At a velocity node $\boldsymbol{\xi}_k$, the particle source used in the present computations is
\begin{equation}\label{eq:positive_particle_source}
    \mathcal B_{\mathcal{T},i,k}^{+}
    =
    \frac{\mathcal{G}_{\mathcal{T},i,k}}{\tau_i}
    \exp\!\left[-\Phi_{i,k}(\mathcal{T}_i)\right]
    \alpha_{i,k},
    \qquad
    \alpha_{i,k}
    =
    \max\left(0,1-D_t \mathcal{T}_i(\boldsymbol{\xi}_k)\right),
\end{equation}
where $\mathcal{G}_{\mathcal{T},i,k}$ is evaluated from the reconstructed local state at the horizon, $\tau_i$ is the cell-local frozen relaxation time used in the source integration, and $\Phi_{i,k}(\mathcal{T}_i)$ is the cumulative collision frequency over that horizon. The part of the signed exchange that is not sampled by particles is carried by the residual wave reconstruction. Equivalently, the source state for particle sampling is
\begin{equation}\label{eq:particle_source_state}
    \mathcal{P}_{\mathrm{src},i,k}^{+}
    =
    \tau_i\mathcal B_{\mathcal{T},i,k}^{+}
    =
    \mathcal{G}_{\mathcal{T},i,k}
    \exp\!\left[-\Phi_{i,k}(\mathcal{T}_i)\right]
    \alpha_{i,k}.
\end{equation}
Under the same local frozen-time approximation, the factor $1/\tau_i$ in $\mathcal B_{\mathcal{T}}^{+}$ is canceled by the exponential source integration. Both the deterministic $S_N$ solver and the Monte Carlo solver used in this paper employ Eq.~\eqref{eq:particle_source_state}.

In the deterministic implementation, the velocity space is discretized by quadrature points and weights $\{(\boldsymbol{\xi}_k,\omega_k)\}_{k=1}^{N_v}$. Denote the cell average of the particle distribution by $\mathcal{P}_{i,k}$. The normal particle flux across $\Gamma_{ij}$ is evaluated by an upwind finite-volume formula,
\begin{equation}\label{eq:sn_particle_flux}
    \mathbf F_{\mathcal{P},ij}^{S_N}
    =
    \sum_{k=1}^{N_v}
    \omega_k
    (\boldsymbol{\xi}_k\cdot\boldsymbol{n}_{ij})\boldsymbol{\psi}_k
    \mathcal{P}_{ij,k}^{\mathrm{up}},
\end{equation}
where $\boldsymbol{\psi}_k=\left(1,\boldsymbol{\xi}_k,\frac{1}{2}|\boldsymbol{\xi}_k|^2\right)^T$. The upwind value is
\begin{equation}
    \mathcal{P}_{ij,k}^{\mathrm{up}}
    =
    \begin{cases}
    \mathcal{P}_{ij,k}^{L}, & \boldsymbol{\xi}_k\cdot\boldsymbol{n}_{ij}\ge0,\\
    \mathcal{P}_{ij,k}^{R}, & \boldsymbol{\xi}_k\cdot\boldsymbol{n}_{ij}<0,
    \end{cases}
\end{equation}
with left and right interface values reconstructed by the same van Leer limiter used for the macroscopic wave variables. In one-step notation, the transport update is
\begin{equation}\label{eq:sn_transport_update}
    \mathcal{P}_{i,k}^{*,n+1}
    =
    \mathcal{P}_{i,k}^{n}
    -
    \frac{\Delta t}{|\Omega_i|}
    \sum_{j\in N(i)}
    (\boldsymbol{\xi}_k\cdot\boldsymbol{n}_{ij})
    \mathcal{P}_{ij,k}^{\mathrm{up}} |\Gamma_{ij}|.
\end{equation}
The relaxation and particle source are then applied to the particle component by the exponential update of
$\partial_t\mathcal{P}+\mathcal{P}/\tau_i=\mathcal B_\mathcal{T}^+$ with frozen coefficients over one step. Let
$r_i=\exp(-\Delta t/\tau_i)$. Then
\begin{equation}\label{eq:sn_relax_source_update}
    \mathcal{P}_{i,k}^{n+1}
    =
    r_i\mathcal{P}_{i,k}^{*,n+1}
    +
    (1-r_i)\mathcal{P}_{\mathrm{src},i,k}^{+}.
\end{equation}
This is an exponential convex combination of the transported particle state and the source state for particle sampling.
The particle moments used in the macroscopic closure are obtained by quadrature,
\begin{equation}\label{eq:sn_particle_moments}
    \mathbf U_{\mathcal{P},i}^{n+1}
    =
    \sum_{k=1}^{N_v}\omega_k\boldsymbol{\psi}_k\mathcal{P}_{i,k}^{n+1}.
\end{equation}

In the stochastic implementation, $\mathcal{P}$ is represented by computational particles $\{(\boldsymbol{x}_p,\boldsymbol{\xi}_p,m_p)\}$, where $m_p$ is the particle weight. During one time step, particles move along molecular characteristics,
\begin{equation}\label{eq:mc_streaming}
    \boldsymbol{x}_p^{\,n+1}=\boldsymbol{x}_p^n+\boldsymbol{\xi}_p\Delta t .
\end{equation}
Face crossings during this motion are accumulated as the Monte Carlo estimate of $\mathbf F_{\mathcal{P},ij}$. The relaxation sink in Eq.~\eqref{eq:P_evol} is applied by survival sampling:
\begin{equation}\label{eq:mc_survival}
    P_{\mathrm{surv},i}=r_i=\exp(-\Delta t/\tau_i),
    \qquad
    P_{\mathrm{abs},i}=1-r_i .
\end{equation}
A particle in cell $i$ is retained with probability $P_{\mathrm{surv},i}$; otherwise its mass, momentum, and energy are no longer represented by $\mathcal{P}$ and are returned to the wave component by the residual moment reconstruction. New particles are sampled from the source state $\mathcal{P}_{\mathrm{src}}^{+}$ in Eq.~\eqref{eq:particle_source_state}. In practice, the expected source moments in a cell are
\begin{equation}\label{eq:mc_source_moment}
    \Delta \mathbf U_{\mathcal{P},i}^{+}
    =
    P_{\mathrm{abs},i}
    \left\langle \boldsymbol{\psi}\mathcal{P}_{\mathrm{src}}^{+} \right\rangle_i,
\end{equation}
and the sampled particle ensemble is corrected, when necessary, so that its low-order conservative moments match the prescribed source moments. This keeps the Monte Carlo update consistent with the same macroscopic exchange used by the deterministic formulation.

The factor $1-D_t \mathcal{T}$ can also be treated by a moving-boundary, or particle-chasing, construction. Consider a molecule emitted from the wave field and let $s$ be its forward flight time. The exact crossing time into the particle region is the first root, if it exists, of
\begin{equation}\label{eq:chasing_root}
    H(s)=s-\mathcal{T}(\boldsymbol{x}+\boldsymbol{\xi} s,t+s,\boldsymbol{\xi})=0,
    \qquad 0<s\le\Delta t .
\end{equation}
The corresponding survival factor is
\begin{equation}\label{eq:chasing_survival}
    P_{\mathrm{chase}}
    =
    \exp\!\left[
    -\int_0^{s_*}
    \frac{\mathrm{d}\nu}{\tau(\boldsymbol{x}+\boldsymbol{\xi}\nu,t+\nu)}
    \right],
\end{equation}
where $s_*$ is the first crossing time. If $\mathcal{T}$ and $\tau$ are frozen locally, this reduces to the approximate expression
\begin{equation}\label{eq:frozen_chasing_probability}
    s_*\approx\frac{\mathcal{T}}{1-D_t \mathcal{T}},
    \qquad
    P_{\mathrm{chase}}
    \approx
    \exp\!\left[-\frac{\mathcal{T}}{\tau(1-D_t \mathcal{T})}\right],
    \qquad 1-D_t \mathcal{T}>0 .
\end{equation}
This construction follows the continuous moving-horizon geometry more directly, but it requires a local first-crossing calculation and becomes less straightforward near rapidly varying horizons or boundaries. The numerical results in this paper therefore use the source state in Eq.~\eqref{eq:particle_source_state} for both $S_N$ and Monte Carlo (MC) particles.

\subsection{Macro-Micro Coupling Procedure}\label{sec:algorithm_coupling}

The wave and particle solvers are coupled through the conservative variables and the relaxation target. At the beginning of a time step, the cell-averaged total conservative variables are decomposed as
\begin{equation}\label{eq:algorithm_macro_micro_split}
    \mathbf U_i^n=\mathbf U_{\mathcal{W},i}^n+\mathbf U_{\mathcal{P},i}^n,
    \qquad
    \mathbf U_{\mathcal{P},i}^n=\left\langle \boldsymbol{\psi}\mathcal{P}_i^n \right\rangle.
\end{equation}
The total state $\mathbf U_i^n$ determines the primitive variables, the relaxation time $\tau_i^n$, the target distribution $\mathcal{G}_i^n$, and the local horizon $\mathcal{T}_i^n$. These quantities provide the wave flux, the particle source, and the absorption probability used by the two subsolvers. Thus the macroscopic field supplies the thermodynamic closure for the microscopic equation, while the microscopic field supplies the non-equilibrium flux and moments needed by the macroscopic conservation law.

The finite-volume update is performed in conservative form. First, the wave flux $\mathbf F_{\mathcal{W},ij}$ is computed from the GKS Navier--Stokes flux in Eq.~\eqref{eq:wave_flux_algorithm}. Second, the particle transport solver, either $S_N$ or MC, evaluates $\mathbf F_{\mathcal{P},ij}$ and advances the particle component according to Eqs.~\eqref{eq:sn_transport_update}--\eqref{eq:sn_relax_source_update} or Eqs.~\eqref{eq:mc_streaming}--\eqref{eq:mc_source_moment}. The total conservative variables are updated only through the sum of the wave and particle fluxes in Eq.~\eqref{eq:fv_update_total}. After the particle moments have been reconstructed from the updated particle distribution or particle ensemble, the wave conservative variables are obtained by the conservative residual
\begin{equation}\label{eq:wave_residual_reconstruction}
    \mathbf U_{\mathcal{W},i}^{n+1}
    =
    \mathbf U_i^{n+1}
    -
    \mathbf U_{\mathcal{P},i}^{n+1}.
\end{equation}
This step is the macro-micro closure of the algorithm: the particle part contributes its transported non-equilibrium moments, and the wave part carries the remaining conservative moments. As long as the total finite-volume flux is conservative and the particle moment reconstruction is consistent with the sampled or quadrature distribution, Eq.~\eqref{eq:wave_residual_reconstruction} preserves the cell-wise identity $\mathbf U=\mathbf U_{\mathcal{W}}+\mathbf U_{\mathcal{P}}$.

The local horizon source closes the exchange between the two components. The source moments sampled into $\mathcal{P}$ are removed from the wave component through Eq.~\eqref{eq:wave_residual_reconstruction}, while particles absorbed by the factor $e^{-\Delta t/\tau}$ are no longer included in $\mathbf U_{\mathcal{P}}^{n+1}$ and are therefore carried by the reconstructed wave state. In the deterministic $S_N$ implementation this exchange is evaluated by velocity quadrature; in the Monte Carlo implementation it is evaluated by particle moments with a low-order conservative correction when source particles are generated. Therefore the two microscopic discretizations share the same macroscopic conservation update and differ only in how $\mathcal{P}$ and $\mathbf F_{\mathcal{P}}$ are represented. The complete WPD algorithm is summarized in Fig.~\ref{fig:wpd_algorithm_flowchart}.

\begin{remark}[Temporal accuracy]
The discretization is second order in space but first order in time, owing to the forward-Euler update and the transport--relaxation--source splitting. This temporal error does not affect the reported results, since the smooth accuracy test uses $\Delta t=\mathcal O(\Delta x^2)$ and the remaining cases are either advanced to a steady state or compared against a reference solution computed with identical time stepping. A second-order extension, by a two-stage flux update with Strang splitting of the relaxation and source, leaves the wave-particle decomposition unchanged and is left to future work.
\end{remark}
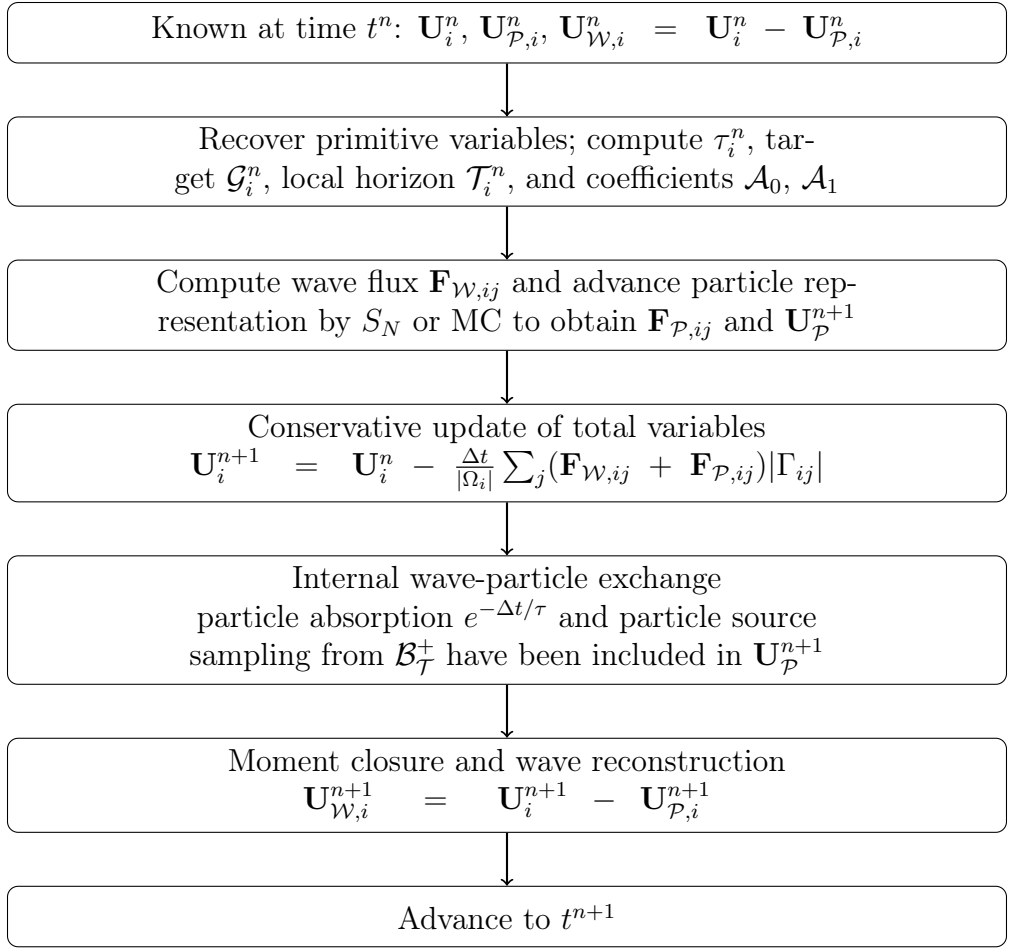
\begin{figure}[!htbp]
\centering
\begin{tikzpicture}[
    node distance=7mm,
    block/.style={rectangle,draw,rounded corners,align=center,text width=0.78\linewidth,minimum height=8mm},
    arrow/.style={->,thick}
]
\node[block] (start) {Known at time $t^n$: $\mathbf U_i^n$, $\mathbf U_{\mathcal{P},i}^n$, $\mathbf U_{\mathcal{W},i}^n=\mathbf U_i^n-\mathbf U_{\mathcal{P},i}^n$};
\node[block,below=of start] (macro) {Recover primitive variables; compute $\tau_i^n$, target $\mathcal{G}_i^n$, local horizon $\mathcal T_i^n$, and coefficients $\mathcal A_0$, $\mathcal A_1$};
\node[block,below=of macro] (flux) {Compute wave flux $\mathbf F_{\mathcal{W},ij}$ and advance particle representation by $S_N$ or MC to obtain $\mathbf F_{\mathcal{P},ij}$ and $\mathbf U_{\mathcal{P}}^{n+1}$};
\node[block,below=of flux] (total) {Conservative update of total variables\\$\mathbf U_i^{n+1}=\mathbf U_i^n-\frac{\Delta t}{|\Omega_i|}\sum_j(\mathbf F_{\mathcal{W},ij}+\mathbf F_{\mathcal{P},ij})|\Gamma_{ij}|$};
\node[block,below=of total] (exchange) {Internal wave-particle exchange\\particle absorption $e^{-\Delta t/\tau}$ and particle source sampling from $\mathcal B_{\mathcal T}^+$ have been included in $\mathbf U_{\mathcal{P}}^{n+1}$};
\node[block,below=of exchange] (recombine) {Moment closure and wave reconstruction\\$\mathbf U_{\mathcal{W},i}^{n+1}=\mathbf U_i^{n+1}-\mathbf U_{\mathcal{P},i}^{n+1}$};
\node[block,below=of recombine] (next) {Advance to $t^{n+1}$};

\draw[arrow] (start) -- (macro);
\draw[arrow] (macro) -- (flux);
\draw[arrow] (flux) -- (total);
\draw[arrow] (total) -- (exchange);
\draw[arrow] (exchange) -- (recombine);
\draw[arrow] (recombine) -- (next);
\end{tikzpicture}
\caption{Macro-micro coupling procedure of the WPD algorithm. The macroscopic conservative state determines the relaxation target, relaxation time, and local horizon; the wave and particle solvers provide complementary fluxes; the updated particle moments close the total conservative update by residual reconstruction of the wave moments.}
\label{fig:wpd_algorithm_flowchart}
\end{figure}

\section{Analysis}\label{sec:analysis}

This section records the multiscale consistency properties of the discretization. The results are not intended as a full nonlinear convergence theory for the complete kinetic scheme. The analysis is restricted to three structural properties: recovery of a Navier--Stokes finite-volume discretization in the continuum limit, consistency with kinetic transport in rarefied regimes, and the scaling of the active kinetic representation. No pointwise positivity theorem for the full distribution is asserted; in particular, the Shakhov correction used in the computations is not pointwise non-negative in general.

\subsection{Asymptotic-Preserving Continuum Limit}

Let
\begin{equation}\label{eq:chi_def}
    \chi_i=\frac{\mathcal{T}_i}{\tau_i},
    \qquad
    \beta_i=e^{-\chi_i}
\end{equation}
be the cellwise local horizon-to-relaxation ratio and the corresponding particle collisionless factor. At a cell interface the wave-flux coefficients use
\begin{equation}\label{eq:eta_interface_analysis}
    \eta_{ij}=\frac{\mathcal{T}_{ij}}{\tau_{ij}},
\end{equation}
as in Eq.~\eqref{eq:algorithm_A01}. The continuum limit on a fixed mesh corresponds to $\tau_i\to0$ with $\mathcal{T}_i$ determined by the local mesh-scale time in Eq.~\eqref{eq:algorithm_local_horizon}; hence $\chi_i\to\infty$ and $\eta_{ij}\to\infty$.

\begin{theorem}[Asymptotic-preserving continuum limit]\label{thm:ap}
Assume that the macroscopic variables remain smooth, that the interface GKS fluxes remain bounded, and that the time step is restricted by a mesh-scale CFL condition independent of the relaxation time. Let
\[
    \chi_i^\ast=\min\left(\chi_i,\min_{j\in N(i)}\eta_{ij}\right).
\]
In the limit $\chi_i^\ast\to\infty$, the WPD finite-volume update reduces to a Navier--Stokes GKS discretization with an exponentially small particle correction:
\begin{equation}\label{eq:ap_limit_update}
    \mathbf U_i^{n+1}
    =
    \mathbf U_i^n
    -
    \frac{\Delta t}{|\Omega_i|}
    \sum_{j\in N(i)}
    \left(
    \mathbf F_{\mathrm E,ij}^{\mathrm{GKS}}
    +
    \mathbf F_{\mathrm{vis},ij}^{\mathrm{GKS}}
    \right)|\Gamma_{ij}|
    +
    \mathcal O\!\left((1+\chi_i^\ast)e^{-\chi_i^\ast}\right) .
\end{equation}
\end{theorem}

\begin{proof}
From Eq.~\eqref{eq:algorithm_A01},
\begin{equation}
    \mathcal A_{0,ij}=1-e^{-\eta_{ij}}\to1,
    \qquad
    \mathcal A_{1,ij}=1-(1+\eta_{ij})e^{-\eta_{ij}}\to1
    \qquad
    \text{as }\eta_{ij}\to\infty .
\end{equation}
More precisely,
\[
    1-\mathcal A_{0,ij}=e^{-\eta_{ij}},
    \qquad
    1-\mathcal A_{1,ij}=(1+\eta_{ij})e^{-\eta_{ij}} .
\]
Therefore the wave flux in Eq.~\eqref{eq:wave_flux_algorithm} approaches the full Navier--Stokes GKS flux,
\begin{equation}
    \mathbf F_{\mathcal{W},ij}
    =
    \mathbf F_{\mathrm E,ij}^{\mathrm{GKS}}
    +
    \mathbf F_{\mathrm{vis},ij}^{\mathrm{GKS}}
    +
    \mathcal O\!\left((1+\eta_{ij})e^{-\eta_{ij}}\right) .
\end{equation}
The surviving particle memory is controlled by the factor $\beta_i=e^{-\chi_i}$ in the local-horizon decomposition and by the relaxation factor $e^{-\Delta t/\tau_i}$ during the particle update. For bounded particle moments and smooth reconstructed macroscopic states, its transported contribution is exponentially small in the continuum limit. Substituting these estimates into the conservative update \eqref{eq:fv_update_total} gives Eq.~\eqref{eq:ap_limit_update}. Since the time step is controlled by the acoustic CFL condition rather than by $\tau_i$, the scheme is asymptotic-preserving in the sense that it becomes a consistent Navier--Stokes finite-volume discretization without resolving the collision time.
\end{proof}

The theorem concerns the limiting discrete equation. It states that, on a fixed mesh, the kinetic update reduces to the Navier--Stokes GKS discretization; it does not imply that the finite-mesh solution is the exact solution of the continuous Navier--Stokes equations.

\subsection{Rarefied-Limit Consistency}

The opposite limit corresponds to a local relaxation time much larger than the horizon, $\chi_i=\mathcal{T}_i/\tau_i\to0$. In this regime, the wave operator should not impose an artificial hydrodynamic closure.

\begin{proposition}[Kinetic consistency in rarefied regimes]\label{prop:rarefied_consistency}
As $\eta_{ij}\to0$ at an interface and $\chi_i\to0$ in the adjacent cells, the wave flux vanishes to leading order and the total transport is carried by the particle solver:
\begin{equation}\label{eq:rarefied_limit_flux}
    \mathcal A_{0,ij}=\eta_{ij}+\mathcal O(\eta_{ij}^2),
    \qquad
    \mathcal A_{1,ij}=\frac{1}{2}\eta_{ij}^2+\mathcal O(\eta_{ij}^3),
    \qquad
    \mathbf F_{ij}
    =
    \mathbf F_{\mathcal{P},ij}
    +
    \mathcal O(\eta_{ij}) .
\end{equation}
\end{proposition}

\begin{proof}
The expansions follow directly from the Taylor expansion of $e^{-\eta_{ij}}$:
\begin{equation}
    1-e^{-\eta_{ij}}
    =
    \eta_{ij}+\mathcal O(\eta_{ij}^2),
    \qquad
    1-(1+\eta_{ij})e^{-\eta_{ij}}
    =
    \frac{1}{2}\eta_{ij}^2+\mathcal O(\eta_{ij}^3).
\end{equation}
Thus the Euler part of the wave flux is first order in $\eta_{ij}$, while the viscous part is second order in $\eta_{ij}$, provided the corresponding GKS fluxes remain bounded. The particle collisionless factor satisfies $\beta_i=e^{-\chi_i}=1+\mathcal O(\chi_i)$, so the kinetic particle representation remains active. The total flux therefore approaches the $S_N$ or MC kinetic transport flux used for $\mathcal{P}$.
\end{proof}

This proposition gives the complementary rarefied-regime consistency of the same formulation: when local relaxation is weak over the mesh-scale horizon, the scheme does not force the solution into a fluid closure, and the kinetic solver remains the dominant transport mechanism.

\subsection{Regime-Adaptive Kinetic Representation}

The local-horizon decomposition also provides a quantitative measure of the active kinetic degrees of freedom. The particle representation is controlled by the same collisionless factor that appears in the integral solution,
\begin{equation}\label{eq:beta_particle_fraction}
    \beta_i=e^{-\mathcal{T}_i/\tau_i}.
\end{equation}
This factor is close to one in rarefied cells and exponentially small in continuum cells.

\begin{theorem}[Scaling of active kinetic degrees of freedom]\label{thm:ra}
For fixed local conservative variables and fixed cell volume, the particle collisionless factor $\beta_i=e^{-\mathcal{T}_i/\tau_i}$ is a strictly decreasing function of the local horizon-to-relaxation ratio:
\begin{equation}\label{eq:beta_monotonicity}
    \frac{\partial \beta_i}{\partial \chi_i}
    =
    -e^{-\chi_i}<0,
    \qquad
    \chi_i=\frac{\mathcal{T}_i}{\tau_i}.
\end{equation}
In a Monte Carlo realization with reference particle weight $m_p$, the equilibrium scaling of the expected number of transported particles is
\begin{equation}\label{eq:expected_particle_number}
    \mathbb E[N_i]
    \simeq
    \frac{\rho_i|\Omega_i|}{m_p}
    e^{-\mathcal{T}_i/\tau_i}.
\end{equation}
\end{theorem}

\begin{proof}
Differentiating the collisionless factor gives the stated monotonicity. Equation~\eqref{eq:expected_particle_number} follows by multiplying the cell mass $\rho_i|\Omega_i|$ by the surviving particle fraction and dividing by the prescribed particle weight. Hence the expected transported particle population decreases exponentially with the local horizon-to-relaxation ratio.
\end{proof}

For the deterministic $S_N$ implementation, Eq.~\eqref{eq:beta_particle_fraction} should be interpreted as the amplitude of the active kinetic component rather than as an automatic wall-clock reduction if a dense velocity grid is swept everywhere. Computational savings in deterministic implementations require exploiting this decay through sparse, thresholded, or adaptive velocity-space treatment. For the Monte Carlo implementation, the same factor directly controls the expected number of transported particles.

\begin{proposition}[Local-horizon reduction on nonuniform meshes]\label{prop:local_horizon_reduction}
Let $\Delta t_{\min}=\min_i\Delta t_{i,\mathrm{loc}}$ be the global time step scale imposed by the smallest cell. With the same local-horizon CFL number $\mathrm{CFL}_l$ as in Eq.~\eqref{eq:algorithm_local_horizon}, a UGKWP-type global time-step split would use
\begin{equation}
    \beta_i^{g}=\exp\left(-\frac{\mathrm{CFL}_l\Delta t_{\min}}{\tau_i}\right),
\end{equation}
whereas the local-horizon split uses
\begin{equation}
    \beta_i^{l}=\exp\left(-\frac{\mathrm{CFL}_l\Delta t_{i,\mathrm{loc}}}{\tau_i}\right).
\end{equation}
Since $\Delta t_{i,\mathrm{loc}}\ge\Delta t_{\min}$ for every cell,
\begin{equation}\label{eq:local_vs_global_beta}
    \beta_i^{l}
    \le
    \beta_i^{g},
    \qquad
    \frac{\beta_i^{l}}{\beta_i^{g}}
    =
    \exp\left[
    -\frac{\mathrm{CFL}_l(\Delta t_{i,\mathrm{loc}}-\Delta t_{\min})}{\tau_i}
    \right]
    \le1 .
\end{equation}
\end{proposition}

This inequality quantifies the effect of the local horizon on nonuniform meshes. A small cell may determine the global time step, but it does not prescribe the same particle fraction in cells with larger local evolution times. The particle representation is controlled by local relaxation over the local mesh-scale horizon.

\section{Numerical Results}\label{sec:results}

The numerical results are organized by physical configuration. Unless otherwise stated, the wave-particle split in all WPD computations uses the cell-local horizon \(\mathcal{T}_i=\mathrm{CFL}_l\Delta t_{i,\mathrm{loc}}\) with \(\mathrm{CFL}_l=1\). The one-dimensional tests examine the kinetic wave-particle decomposition in canonical non-equilibrium problems. The two-dimensional cavity and cylinder cases then assess the method in wall-bounded, continuum, and external rarefied flows. The final three-dimensional cases illustrate the direct extension of the same formulation to realistic multidimensional geometries.

\subsection{One-Dimensional Benchmarks}

\subsubsection{Smooth-wave accuracy test}

The first one-dimensional test verifies spatial accuracy in a smooth periodic problem. A small-amplitude sinusoidal perturbation is imposed on an otherwise uniform equilibrium state,
\[
\rho(x,0)=1+10^{-2}\sin(2\pi x),\qquad u(x,0)=0,\qquad T(x,0)=1,
\]
on \(x\in[0,1]\). The pure deterministic \(S_N\) limit is used, with velocity interval \([-4.5,4.5]\), \(128\) uniform velocity points, \(\Delta t=0.5\Delta x^2\), and final time \(t=0.05\). The three Knudsen numbers, \(\mathrm{Kn}=10^{-1}\), \(10^{-2}\), and \(10^{-4}\), cover rarefied, transition, and near-continuum regimes. Errors are measured against the finest available solution on the same velocity grid; Table~\ref{tab:sine_accuracy_linf} reports \(L_\infty\) errors and observed orders for \(N_x=20,50,100,150,200\).

\begin{table}[!htbp]
\centering
\caption{\(L_\infty\) errors and observed orders for the smooth-wave accuracy test.}
\label{tab:sine_accuracy_linf}
\resizebox{\textwidth}{!}{%
\begin{tabular}{cccccccc}
\toprule
\(\mathrm{Kn}\) & \(N_x\) & \(E_\rho\) & order & \(E_u\) & order & \(E_T\) & order \\
\midrule
\(10^{-1}\) & 20  & \(9.3750{\times}10^{-6}\) & --   & \(1.1548{\times}10^{-5}\) & --   & \(5.9603{\times}10^{-6}\) & -- \\
            & 50  & \(9.9811{\times}10^{-7}\) & 2.44 & \(2.0147{\times}10^{-6}\) & 1.91 & \(5.7042{\times}10^{-7}\) & 2.56 \\
            & 100 & \(1.8778{\times}10^{-7}\) & 2.41 & \(4.9354{\times}10^{-7}\) & 2.03 & \(9.1700{\times}10^{-8}\) & 2.64 \\
            & 150 & \(7.3846{\times}10^{-8}\) & 2.30 & \(2.1577{\times}10^{-7}\) & 2.04 & \(2.8857{\times}10^{-8}\) & 2.85 \\
            & 200 & \(2.6337{\times}10^{-8}\) & 3.58 & \(1.2131{\times}10^{-7}\) & 2.00 & \(1.0405{\times}10^{-8}\) & 3.55 \\
\midrule
\(10^{-2}\) & 20  & \(9.4192{\times}10^{-6}\) & --   & \(1.1801{\times}10^{-5}\) & --   & \(6.2584{\times}10^{-6}\) & -- \\
            & 50  & \(1.0018{\times}10^{-6}\) & 2.45 & \(2.0915{\times}10^{-6}\) & 1.89 & \(5.7382{\times}10^{-7}\) & 2.61 \\
            & 100 & \(1.8853{\times}10^{-7}\) & 2.41 & \(5.0977{\times}10^{-7}\) & 2.04 & \(8.9523{\times}10^{-8}\) & 2.68 \\
            & 150 & \(7.4153{\times}10^{-8}\) & 2.30 & \(2.2140{\times}10^{-7}\) & 2.06 & \(2.7967{\times}10^{-8}\) & 2.87 \\
            & 200 & \(2.6502{\times}10^{-8}\) & 3.58 & \(1.2422{\times}10^{-7}\) & 2.01 & \(1.0329{\times}10^{-8}\) & 3.46 \\
\midrule
\(10^{-4}\) & 20  & \(1.2357{\times}10^{-5}\) & --   & \(1.2097{\times}10^{-5}\) & --   & \(1.6175{\times}10^{-5}\) & -- \\
            & 50  & \(2.2697{\times}10^{-6}\) & 1.85 & \(2.3031{\times}10^{-6}\) & 1.81 & \(4.9780{\times}10^{-6}\) & 1.29 \\
            & 100 & \(5.9613{\times}10^{-7}\) & 1.93 & \(5.7624{\times}10^{-7}\) & 2.00 & \(1.4796{\times}10^{-6}\) & 1.75 \\
            & 150 & \(2.6011{\times}10^{-7}\) & 2.05 & \(2.5611{\times}10^{-7}\) & 2.00 & \(6.7617{\times}10^{-7}\) & 1.93 \\
            & 200 & \(1.4051{\times}10^{-7}\) & 2.14 & \(1.4310{\times}10^{-7}\) & 2.02 & \(3.7328{\times}10^{-7}\) & 2.07 \\
\bottomrule
\end{tabular}%
}
\end{table}

Figure~\ref{fig:sine_accuracy_linf} shows the corresponding \(L_\infty\) convergence curves. The velocity error follows the second-order reference slope over all three Knudsen numbers. The density and temperature errors also decrease at approximately second order in the smooth range, with mild pre-asymptotic variation caused by the finite velocity quadrature and the use of a numerical reference solution. This test indicates that the deterministic transport discretization used in the WPD-\(S_N\) solver retains the expected second-order spatial accuracy for smooth solutions.

\begin{figure}[!htbp]
    \centering
    \begin{subfigure}{0.32\textwidth}
        \centering
        \includegraphics[width=\linewidth]{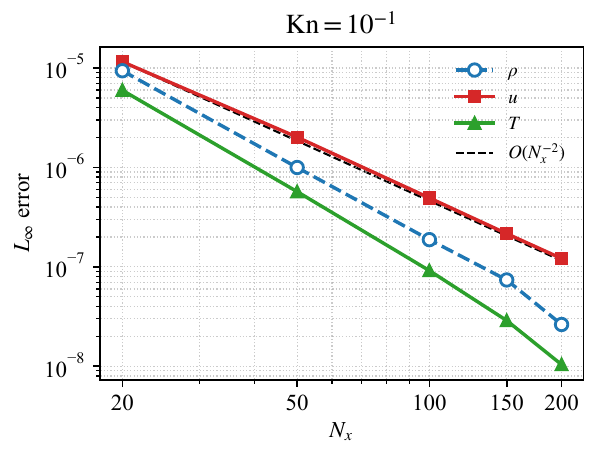}
        \caption{\(\mathrm{Kn}=10^{-1}\)}
    \end{subfigure}
    \begin{subfigure}{0.32\textwidth}
        \centering
        \includegraphics[width=\linewidth]{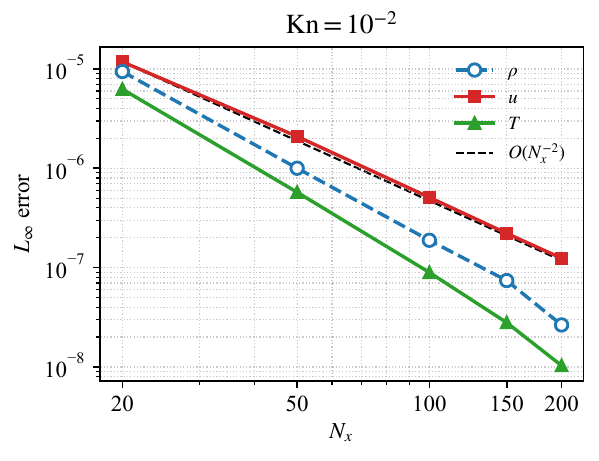}
        \caption{\(\mathrm{Kn}=10^{-2}\)}
    \end{subfigure}
    \begin{subfigure}{0.32\textwidth}
        \centering
        \includegraphics[width=\linewidth]{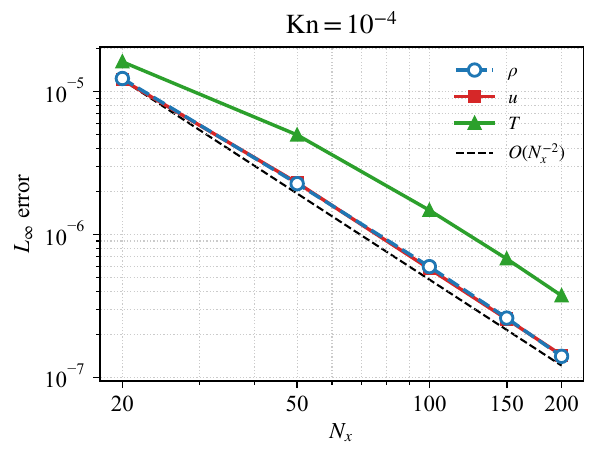}
        \caption{\(\mathrm{Kn}=10^{-4}\)}
    \end{subfigure}
    \caption{\(L_\infty\) convergence curves for the smooth-wave accuracy test.}
    \label{fig:sine_accuracy_linf}
\end{figure}

\subsubsection{Sod shock tube}

The Sod shock tube tests the deterministic WPD-\(S_N\) formulation on an unsteady kinetic Riemann problem, with a UGKS solution as the reference. The initial discontinuity is placed at the center of the computational domain,
\[
(\rho,u,p)_L=(1,0,1),\qquad
(\rho,u,p)_R=(0.125,0,0.1).
\]
The computational domain is \(x\in[0,2]\), with the initial interface at \(x=1\). A uniform mesh with \(N_x=800\) cells is used. The velocity space is discretized by \(N_\xi=513\) uniform points over \([-8,8]\), and the final time is \(t=0.2\). The same physical mesh, velocity mesh, gas model, and time step are used for UGKS and WPD-\(S_N\); in all WPD-\(S_N\) results reported here the local wave-particle horizon is set by \(\mathrm{CFL}_l=1\).

Figures~\ref{fig:sod_kn1}, \ref{fig:sod_kn1em2}, and \ref{fig:sod_kn1em4} compare density, velocity, and temperature at representative Knudsen numbers. In the rarefied regime, the solution contains broad kinetic transition layers and the WPD-\(S_N\) result follows the UGKS reference without introducing additional numerical diffusion. As the Knudsen number decreases, the profiles sharpen toward the continuum limit, and the agreement remains consistent across all three macroscopic variables.

\begin{table}[!htbp]
\centering
\caption{Discrete \(L_2\) errors of WPD-\(S_N\) relative to UGKS for the Sod shock tube.}
\label{tab:sod_l2_error}
\begin{tabular}{cccc}
\toprule
\(\mathrm{Kn}\) & \(\rho\) & \(u\) & \(T\) \\
\midrule
\(1\) & \(1.3316\times10^{-5}\) & \(5.4762\times10^{-6}\) & \(4.4671\times10^{-6}\) \\
\(10^{-2}\) & \(1.3262\times10^{-4}\) & \(1.1496\times10^{-4}\) & \(1.4756\times10^{-4}\) \\
\(10^{-4}\) & \(3.9922\times10^{-4}\) & \(9.7397\times10^{-4}\) & \(1.9609\times10^{-3}\) \\
\bottomrule
\end{tabular}
\end{table}

\begin{figure}[!htbp]
    \centering
    \begin{subfigure}{0.32\textwidth}
        \centering
        \includegraphics[width=\linewidth]{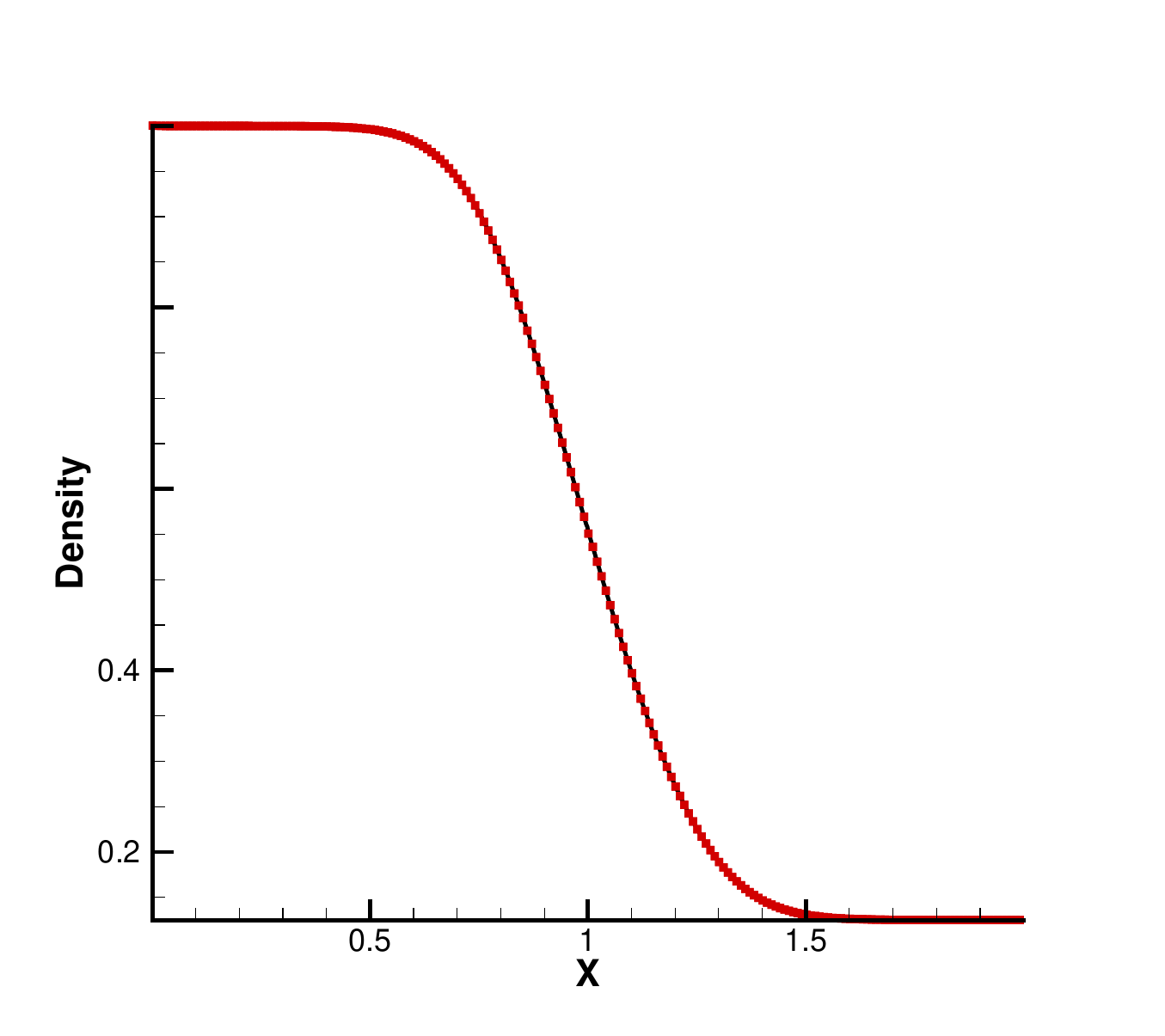}
        \caption{Density}
    \end{subfigure}
    \begin{subfigure}{0.32\textwidth}
        \centering
        \includegraphics[width=\linewidth]{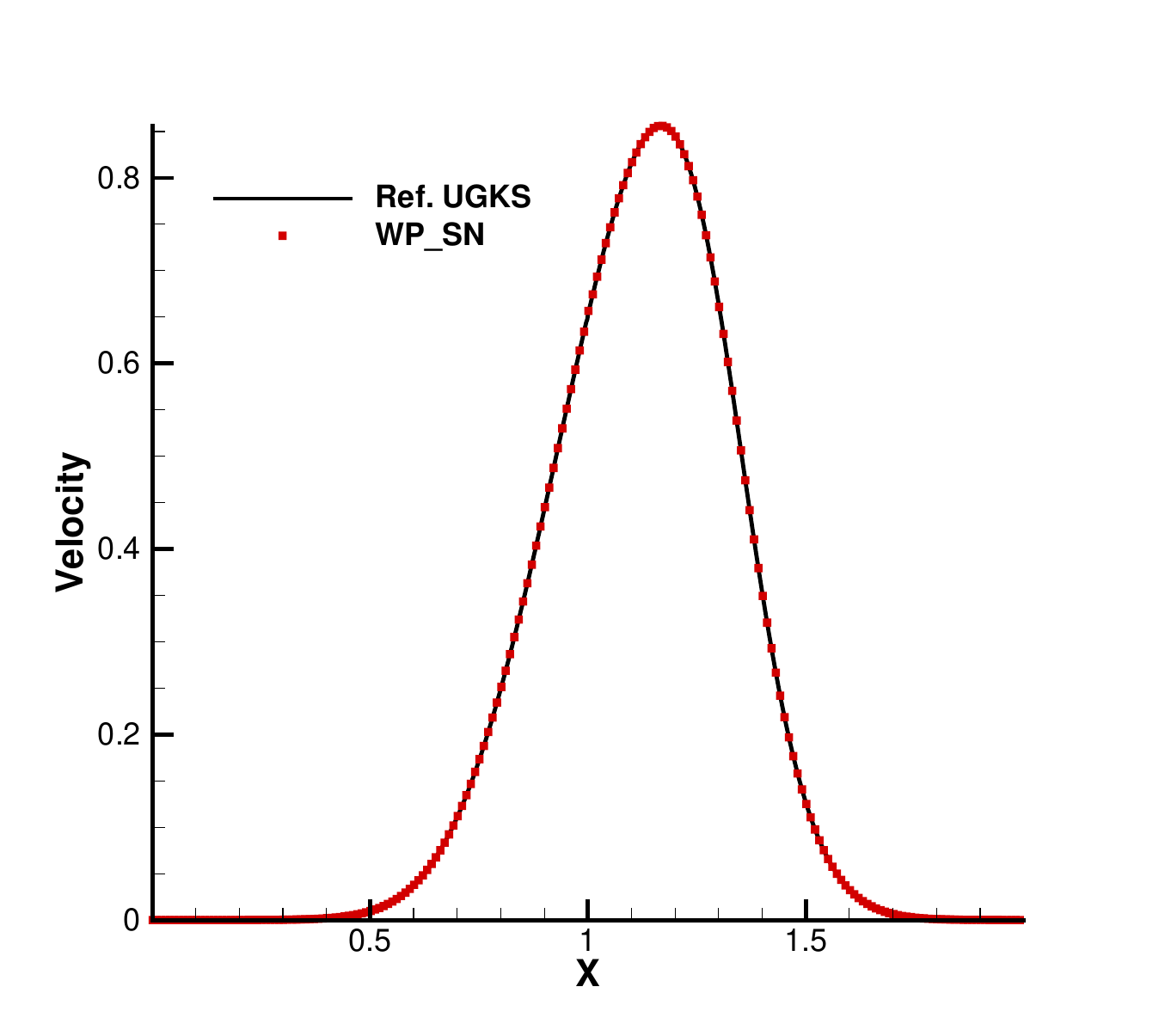}
        \caption{Velocity}
    \end{subfigure}
    \begin{subfigure}{0.32\textwidth}
        \centering
        \includegraphics[width=\linewidth]{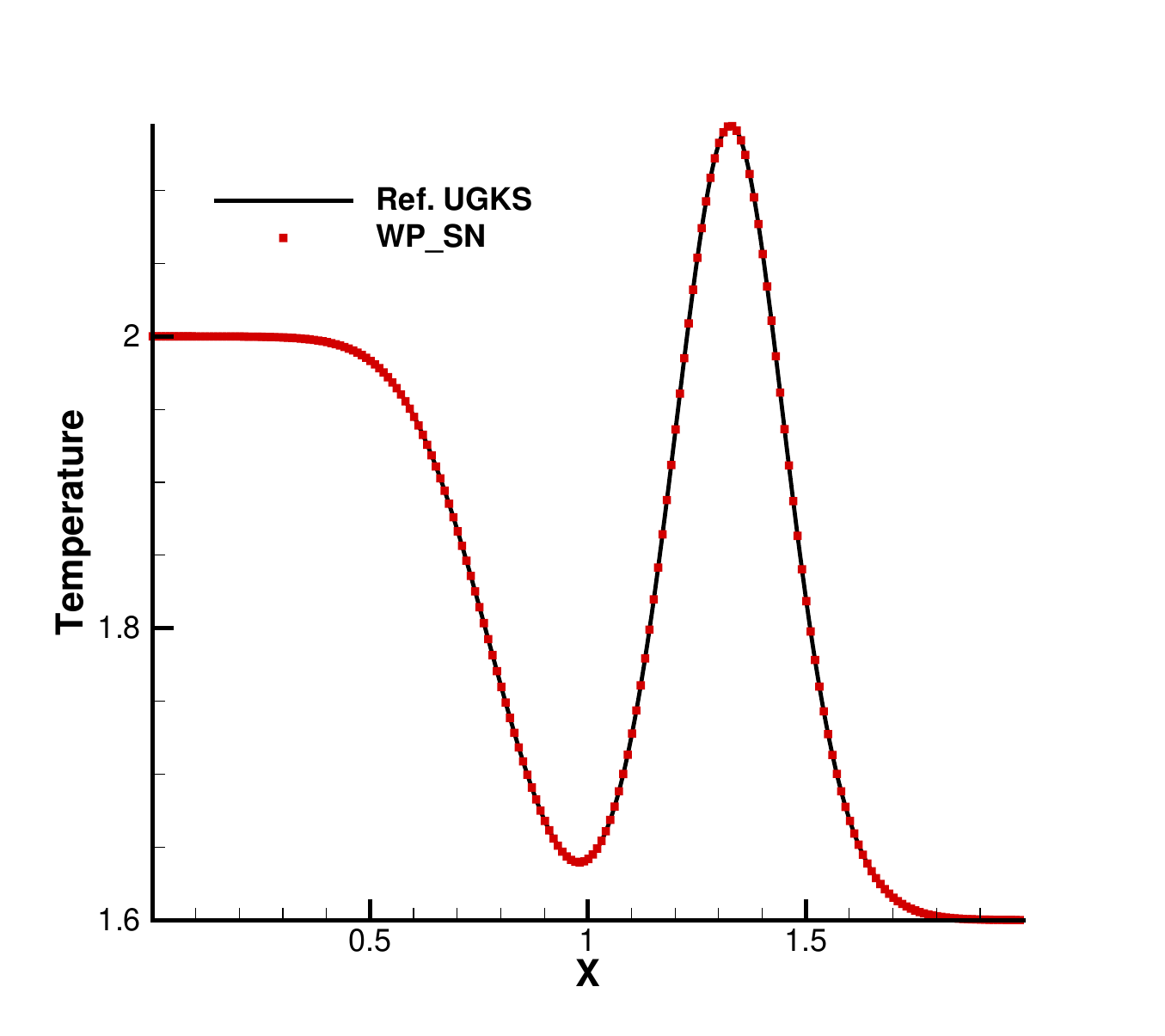}
        \caption{Temperature}
    \end{subfigure}
    \caption{Sod shock tube at $\mathrm{Kn}=1$.}
    \label{fig:sod_kn1}
\end{figure}

\begin{figure}[!htbp]
    \centering
    \begin{subfigure}{0.32\textwidth}
        \centering
        \includegraphics[width=\linewidth]{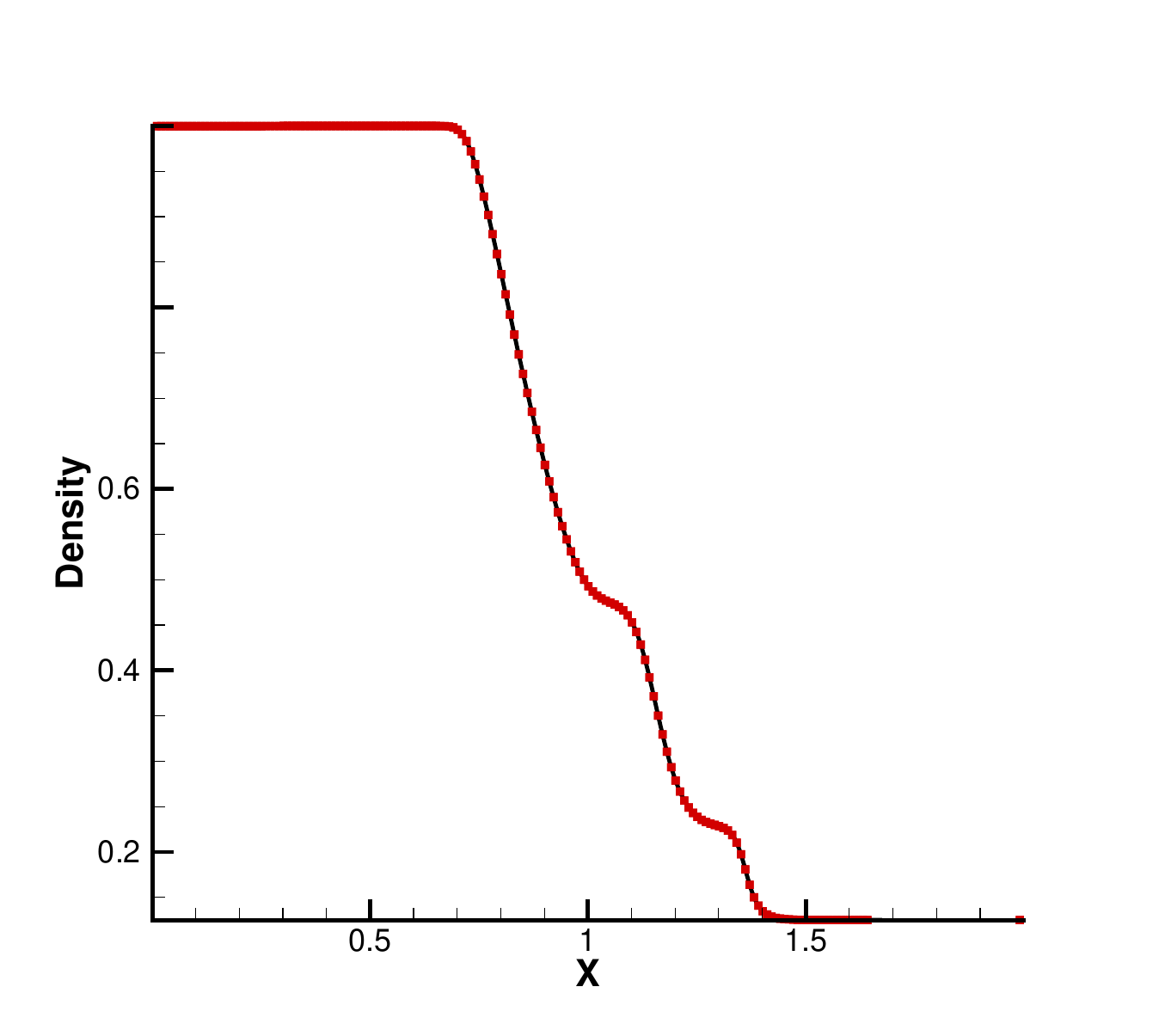}
        \caption{Density}
    \end{subfigure}
    \begin{subfigure}{0.32\textwidth}
        \centering
        \includegraphics[width=\linewidth]{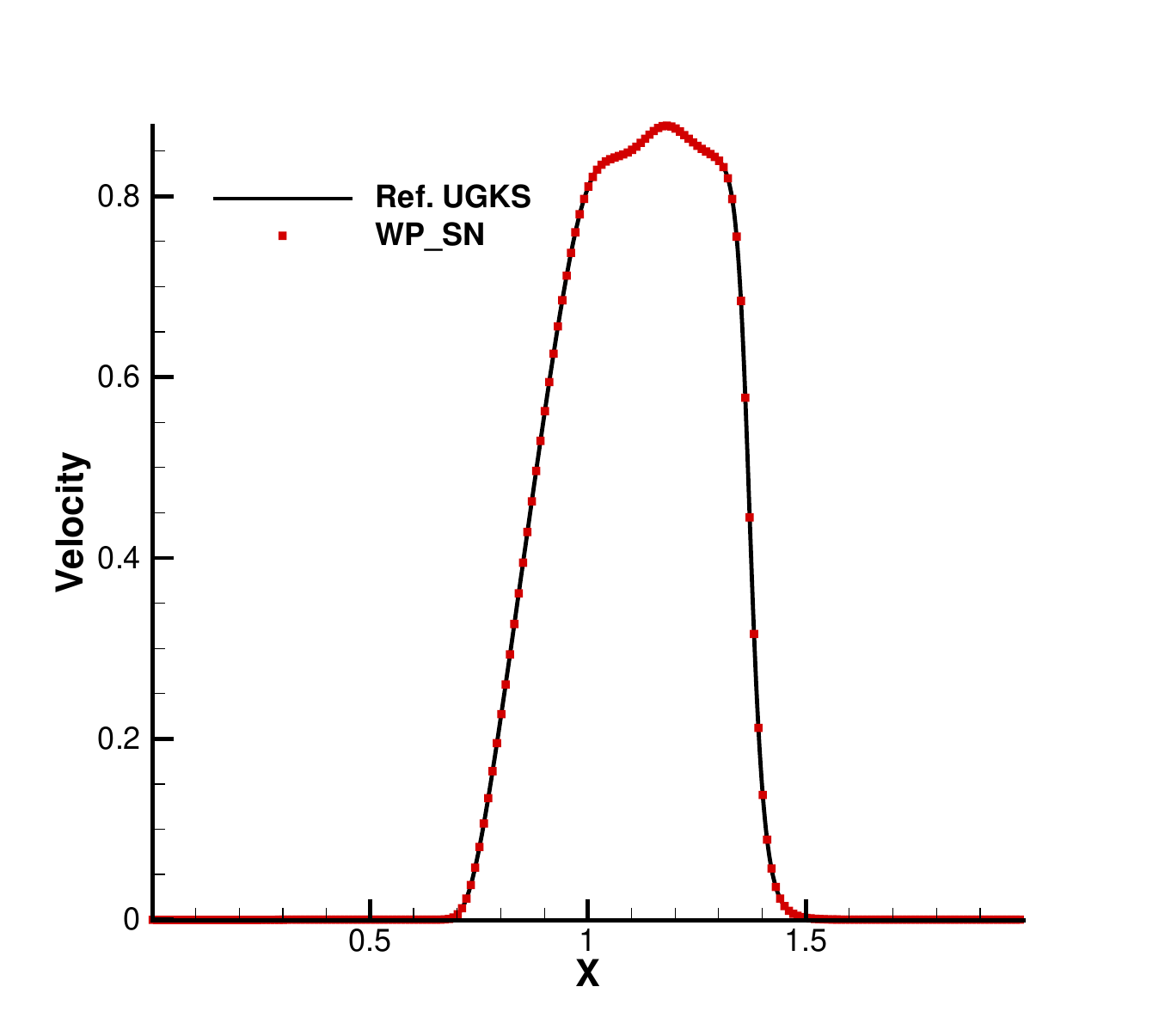}
        \caption{Velocity}
    \end{subfigure}
    \begin{subfigure}{0.32\textwidth}
        \centering
        \includegraphics[width=\linewidth]{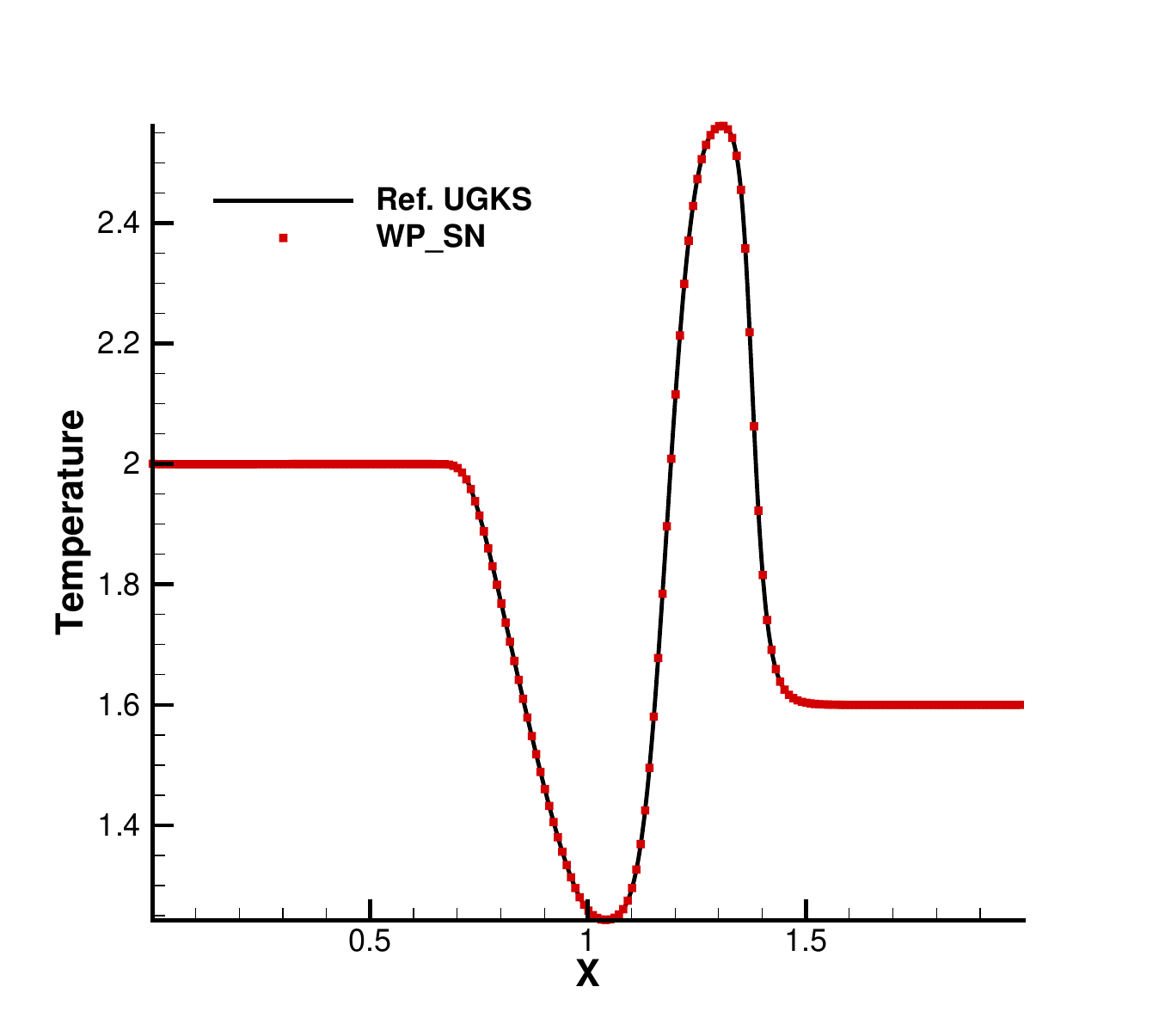}
        \caption{Temperature}
    \end{subfigure}
    \caption{Sod shock tube at $\mathrm{Kn}=10^{-2}$.}
    \label{fig:sod_kn1em2}
\end{figure}

\begin{figure}[!htbp]
    \centering
    \begin{subfigure}{0.32\textwidth}
        \centering
        \includegraphics[width=\linewidth]{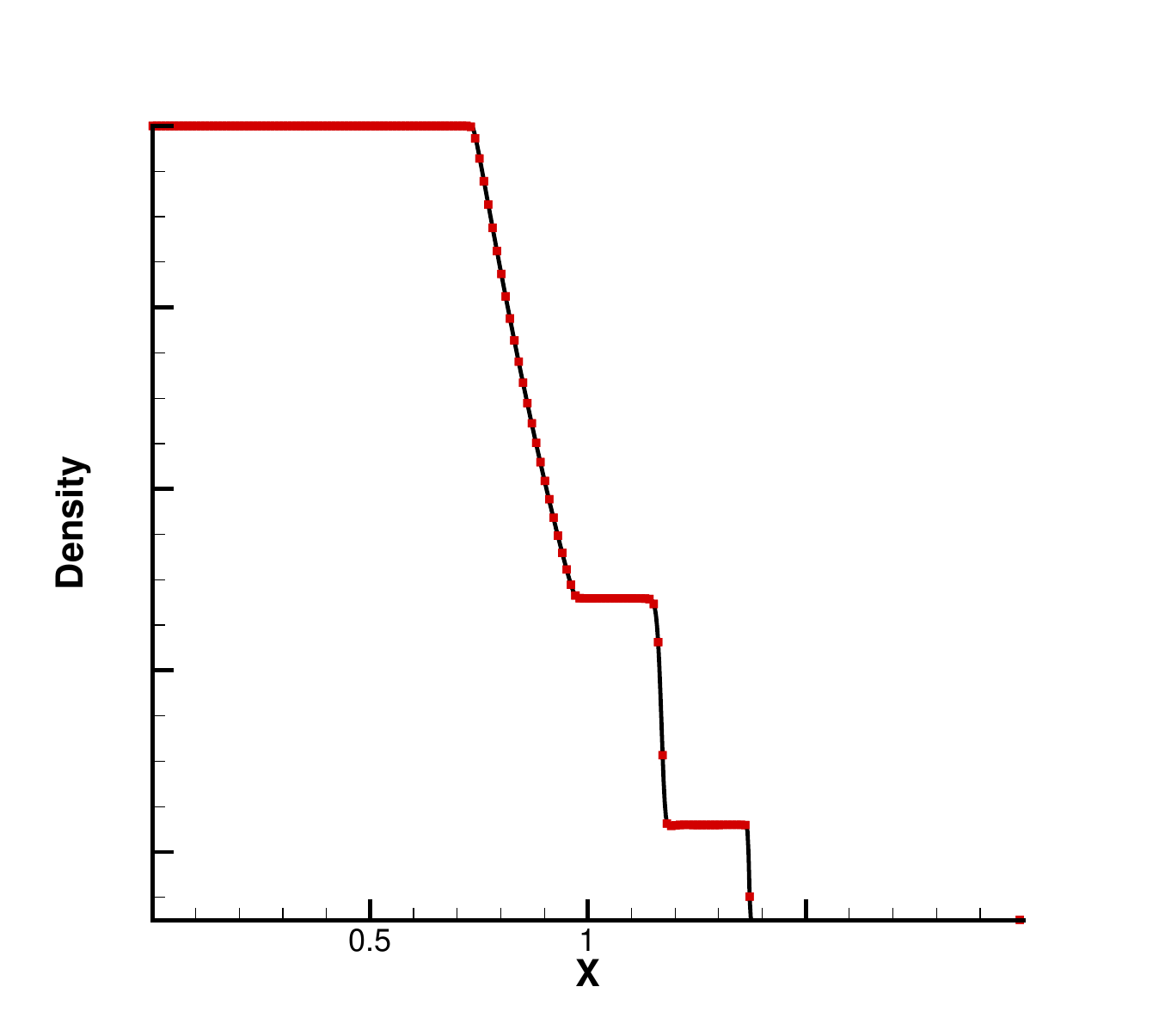}
        \caption{Density}
    \end{subfigure}
    \begin{subfigure}{0.32\textwidth}
        \centering
        \includegraphics[width=\linewidth]{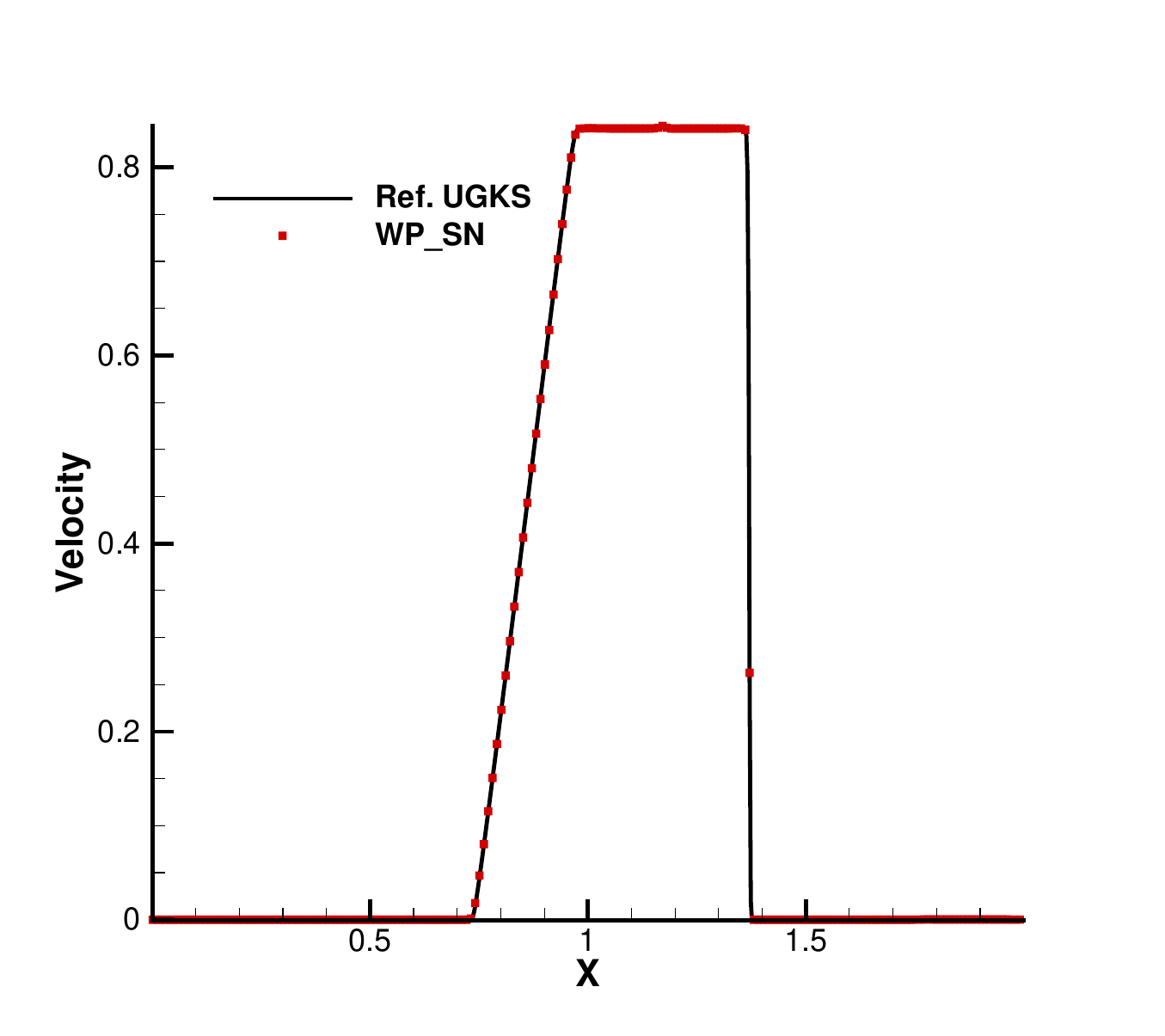}
        \caption{Velocity}
    \end{subfigure}
    \begin{subfigure}{0.32\textwidth}
        \centering
        \includegraphics[width=\linewidth]{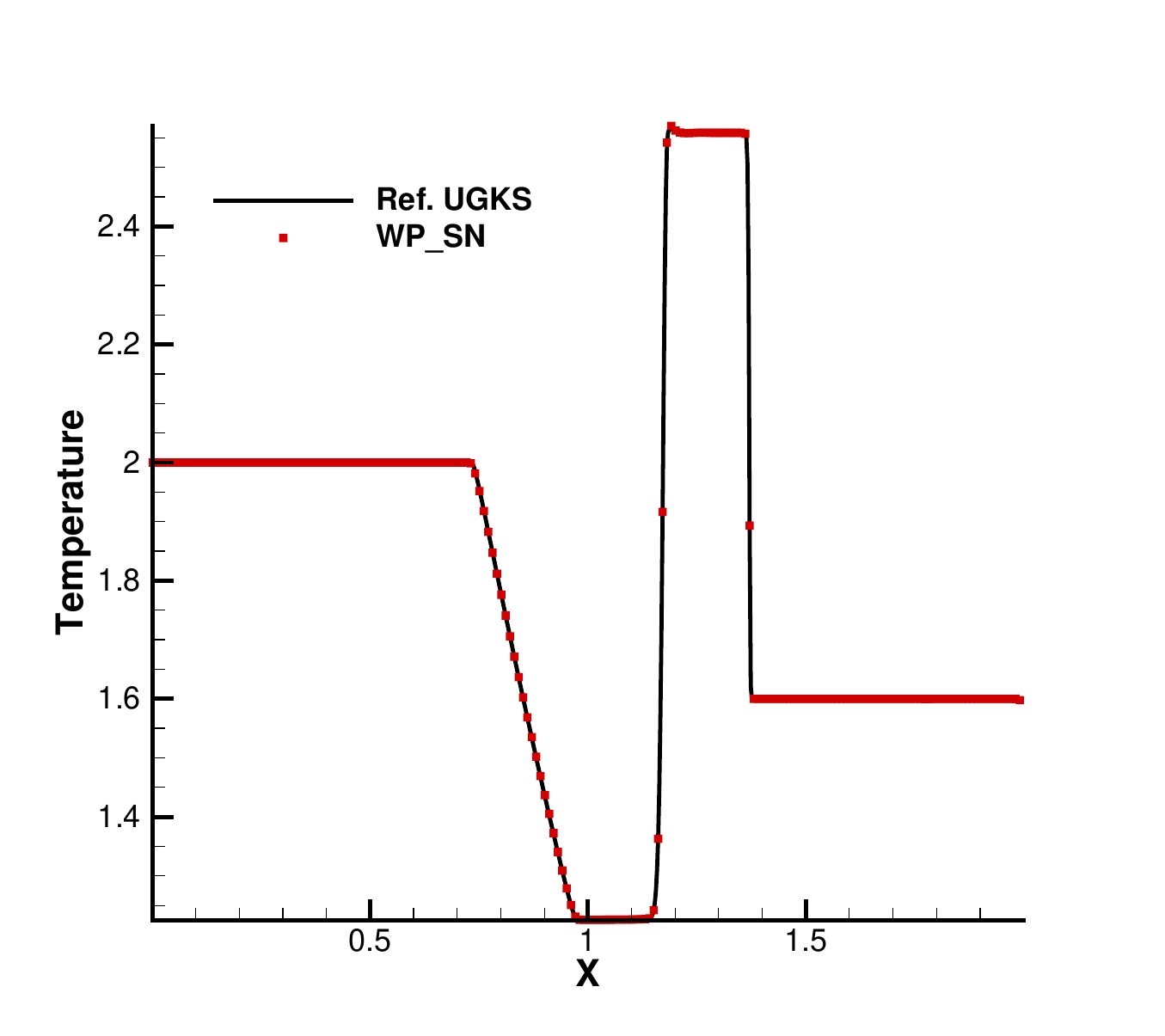}
        \caption{Temperature}
    \end{subfigure}
    \caption{Sod shock tube at $\mathrm{Kn}=10^{-4}$.}
    \label{fig:sod_kn1em4}
\end{figure}

\subsubsection{Couette flow}

The second one-dimensional benchmark is the planar Couette flow between two parallel diffuse-reflection walls. The lower wall is stationary and the upper wall moves tangentially, while the initial density and temperature are uniform. This test is sensitive to wall-induced rarefaction and velocity slip, and therefore provides a direct check of the kinetic wall treatment in the deterministic WPD-\(S_N\) formulation.

The physical domain is discretized by \(N_y=128\) cells in the wall-normal direction. A sequence of Knudsen numbers from \(\mathrm{Kn}=1\) down to \(10^{-4}\), including intermediate rarefaction levels, is considered. For rarefied and transition regimes, the reference solution is computed by the full \(S_N\) solver; for near-continuum cases, a continuum GKS reference is used. The WPD-\(S_N\) calculation uses the same mesh, boundary conditions, and gas model, with the local wave-particle horizon again specified by \(\mathrm{CFL}_l=1\).

Figure~\ref{fig:couette} compares the velocity and density profiles. The WPD-\(S_N\) solution reproduces the reference velocity slip and the weak density variation across the channel. The agreement over the full range of Knudsen numbers indicates that the wave-particle decomposition preserves the near-wall kinetic contribution while recovering the continuum behavior as \(\mathrm{Kn}\) decreases.

\begin{figure}[!htbp]
    \centering
    \begin{subfigure}{0.49\textwidth}
        \centering
        \includegraphics[width=\linewidth]{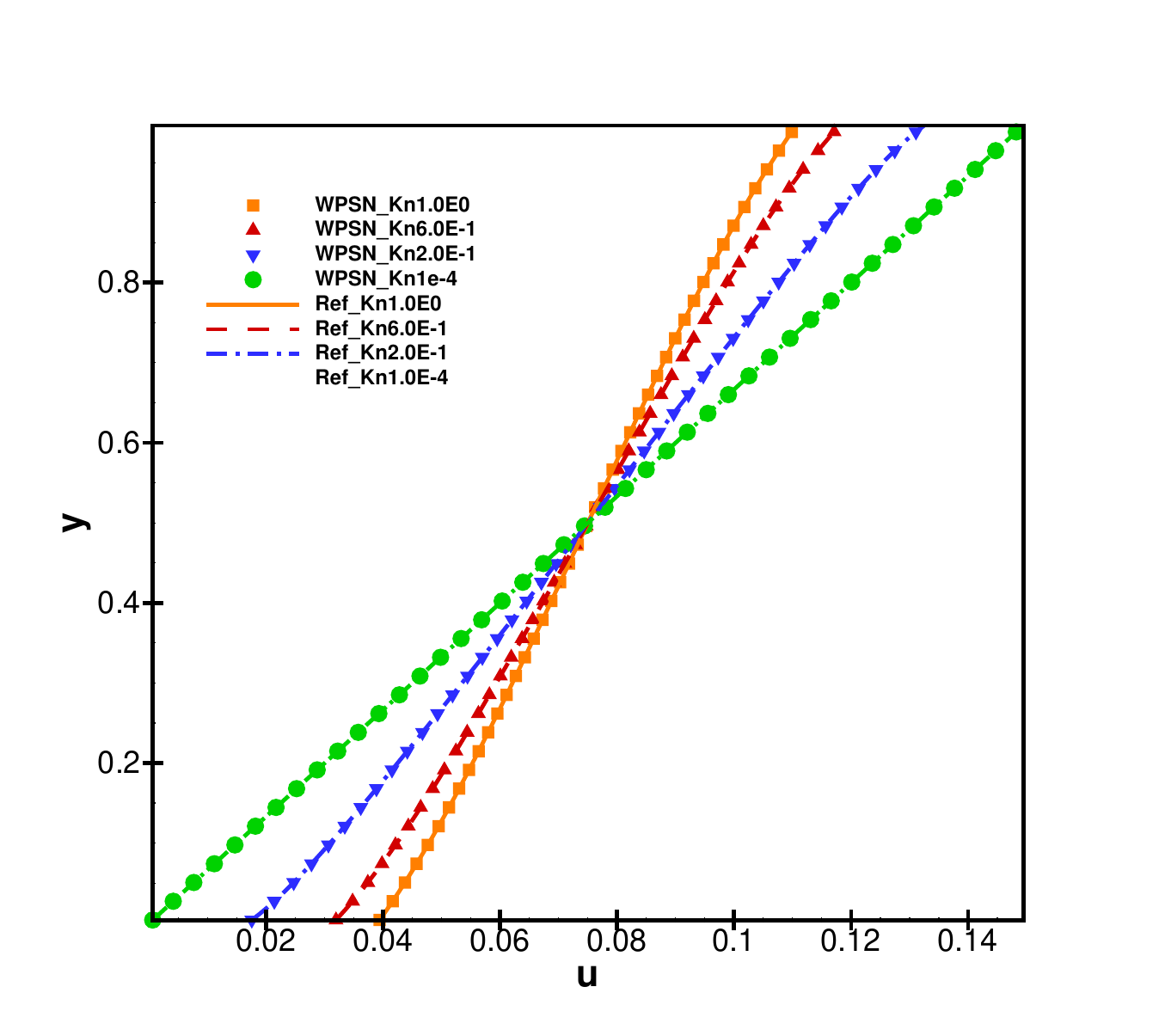}
        \caption{Velocity}
    \end{subfigure}
    \begin{subfigure}{0.49\textwidth}
        \centering
        \includegraphics[width=\linewidth]{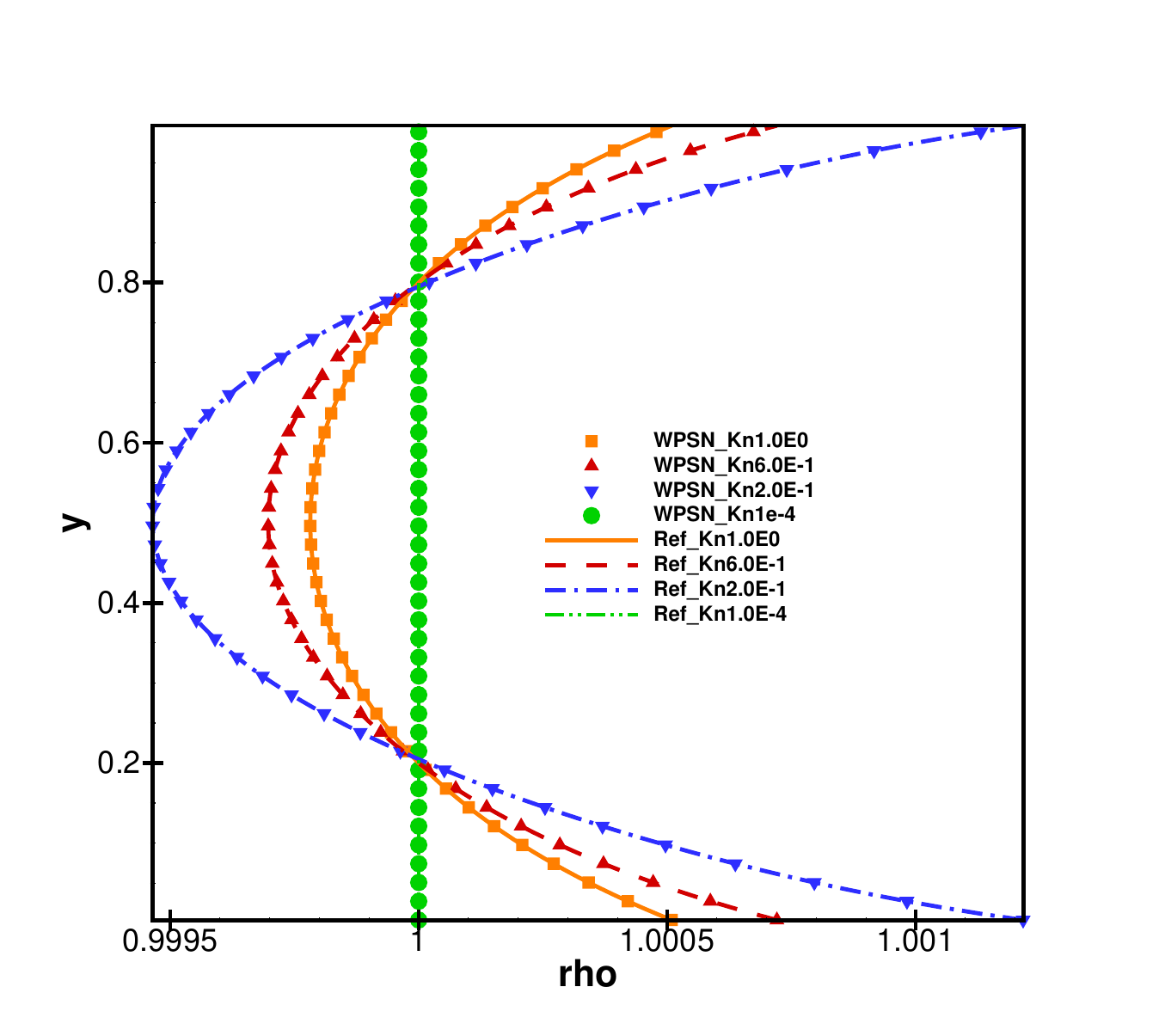}
        \caption{Density}
    \end{subfigure}
    \caption{Couette-flow profiles obtained by WPD-\(S_N\) and the reference solutions.}
    \label{fig:couette}
\end{figure}
\FloatBarrier

\subsubsection{Normal shock structure}

The third one-dimensional benchmark is the steady normal shock structure. This problem provides a more stringent kinetic test than the Sod tube because the solution contains a stationary non-equilibrium layer whose thickness and high-order moments are controlled by the collision model. We compute monatomic normal shocks at \(\mathrm{Kn}=1\) for upstream Mach numbers \(M=3\), \(8\), and \(10\). The computational domain length is \(40\), discretized by \(N_x=1000\) cells. The velocity space uses \(257\) uniform points; the velocity interval is \([-10,10]\) for \(M=3\), and \([-15,20]\) for \(M=8\) and \(10\). The calculations are advanced to \(t=200\) with \(\mathrm{CFL}=0.9\), and the shock profiles are shifted by the midpoint of the density transition for comparison.

The UGKS solution is used as the kinetic reference. The WPD-\(S_N\) calculation uses the same grid, velocity space, boundary states, and gas model, with \(\mathrm{CFL}_l=1\). Figures~\ref{fig:shock_ma8} and \ref{fig:shock_ma10} compare the normalized density and stress profiles for representative strong shocks. The density transition, shock thickness, and stress peak are captured by WPD-\(S_N\) with good agreement relative to UGKS. Figure~\ref{fig:shock_distribution} further compares the velocity distribution functions through the shock layer, showing that the deterministic particle component retains the non-equilibrium shape of the distribution rather than merely reproducing the macroscopic moments.

\begin{figure}[!htbp]
    \centering
    \begin{subfigure}{0.49\textwidth}
        \centering
        \includegraphics[width=\linewidth]{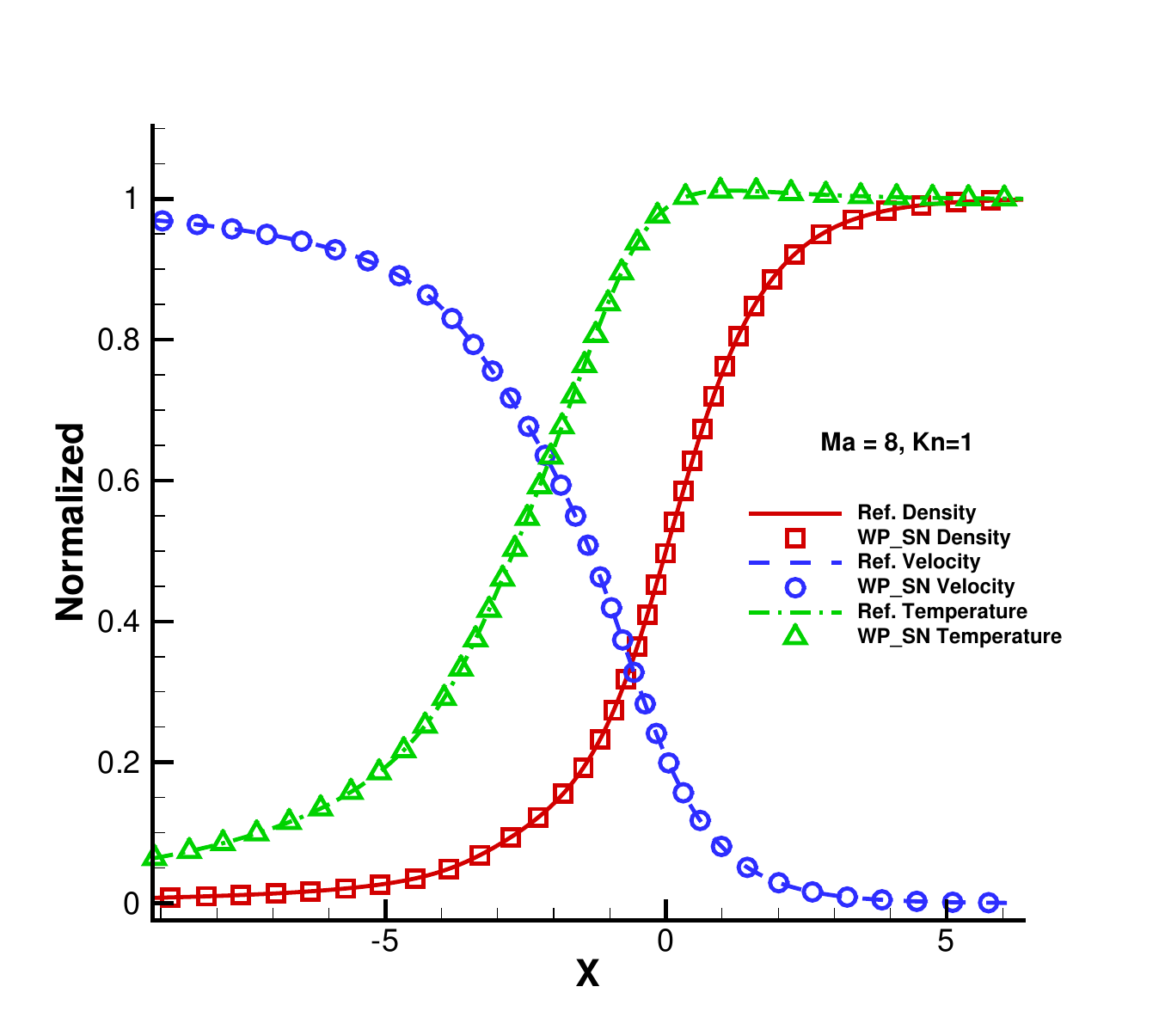}
        \caption{$M=8$, density}
    \end{subfigure}
    \begin{subfigure}{0.49\textwidth}
        \centering
        \includegraphics[width=\linewidth]{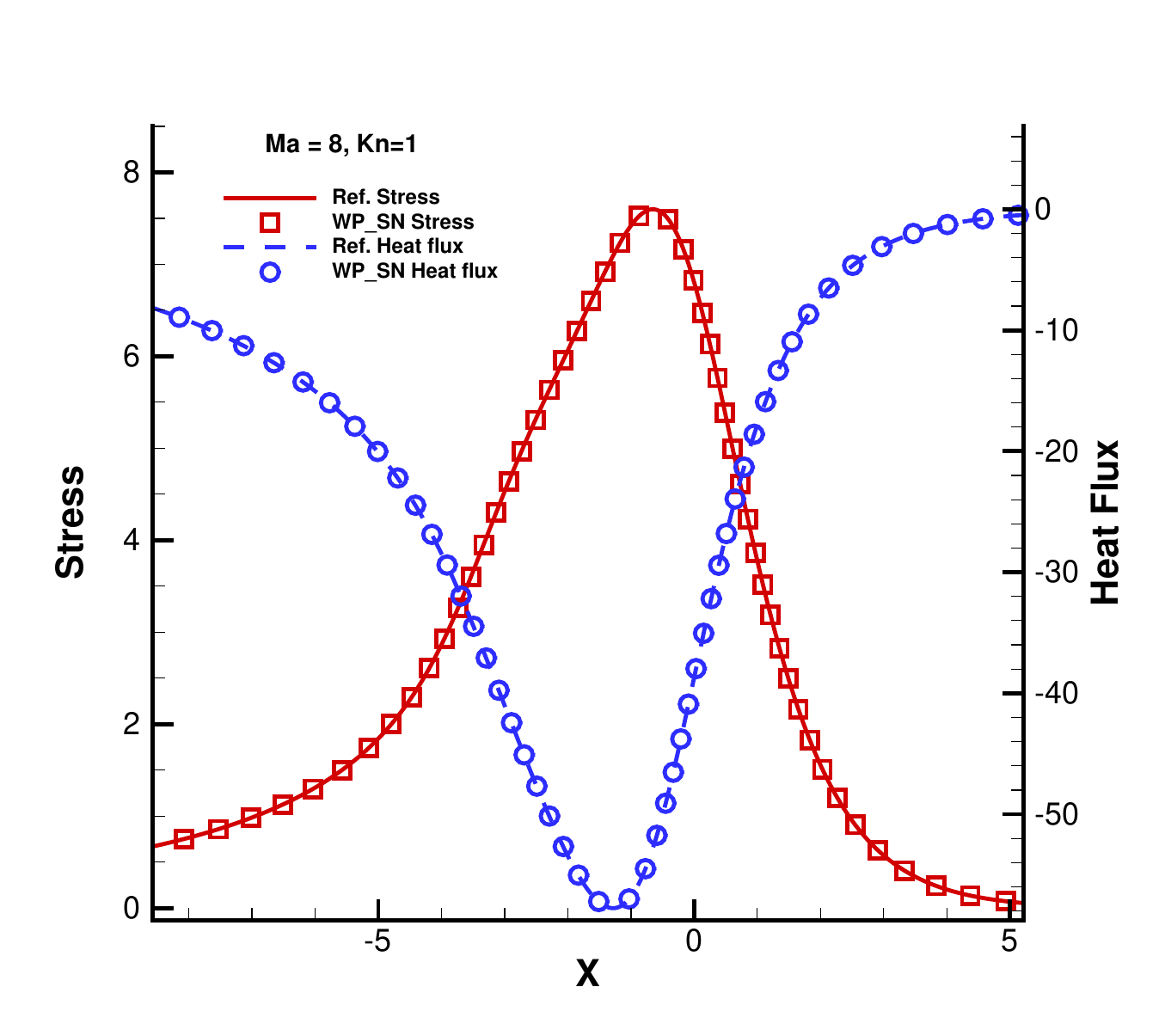}
        \caption{$M=8$, stress}
    \end{subfigure}
    \caption{Normal shock structure at \(M=8\) and \(\mathrm{Kn}=1\).}
    \label{fig:shock_ma8}
\end{figure}

\begin{figure}[!htbp]
    \centering
    \begin{subfigure}{0.49\textwidth}
        \centering
        \includegraphics[width=\linewidth]{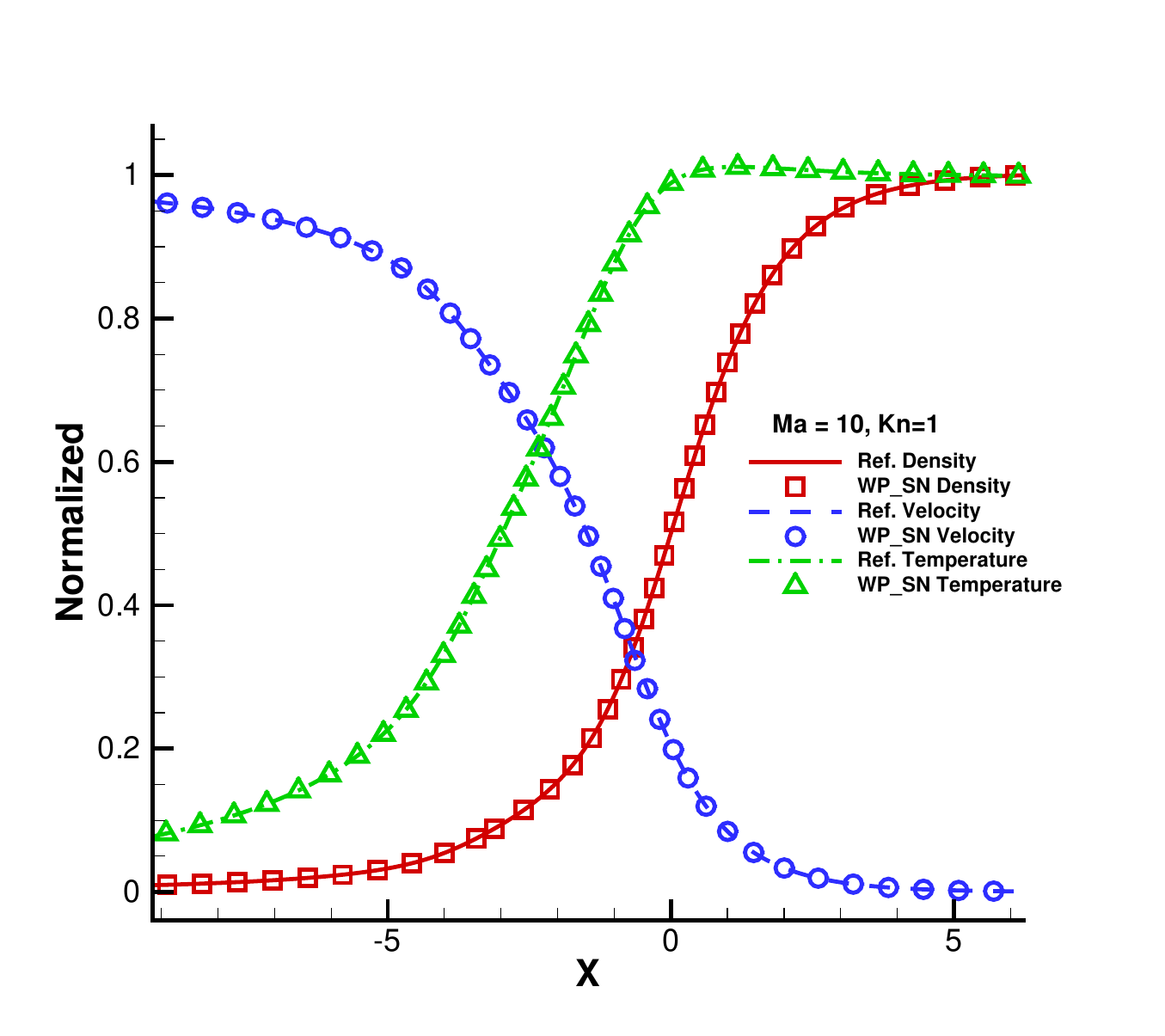}
        \caption{$M=10$, density}
    \end{subfigure}
    \begin{subfigure}{0.49\textwidth}
        \centering
        \includegraphics[width=\linewidth]{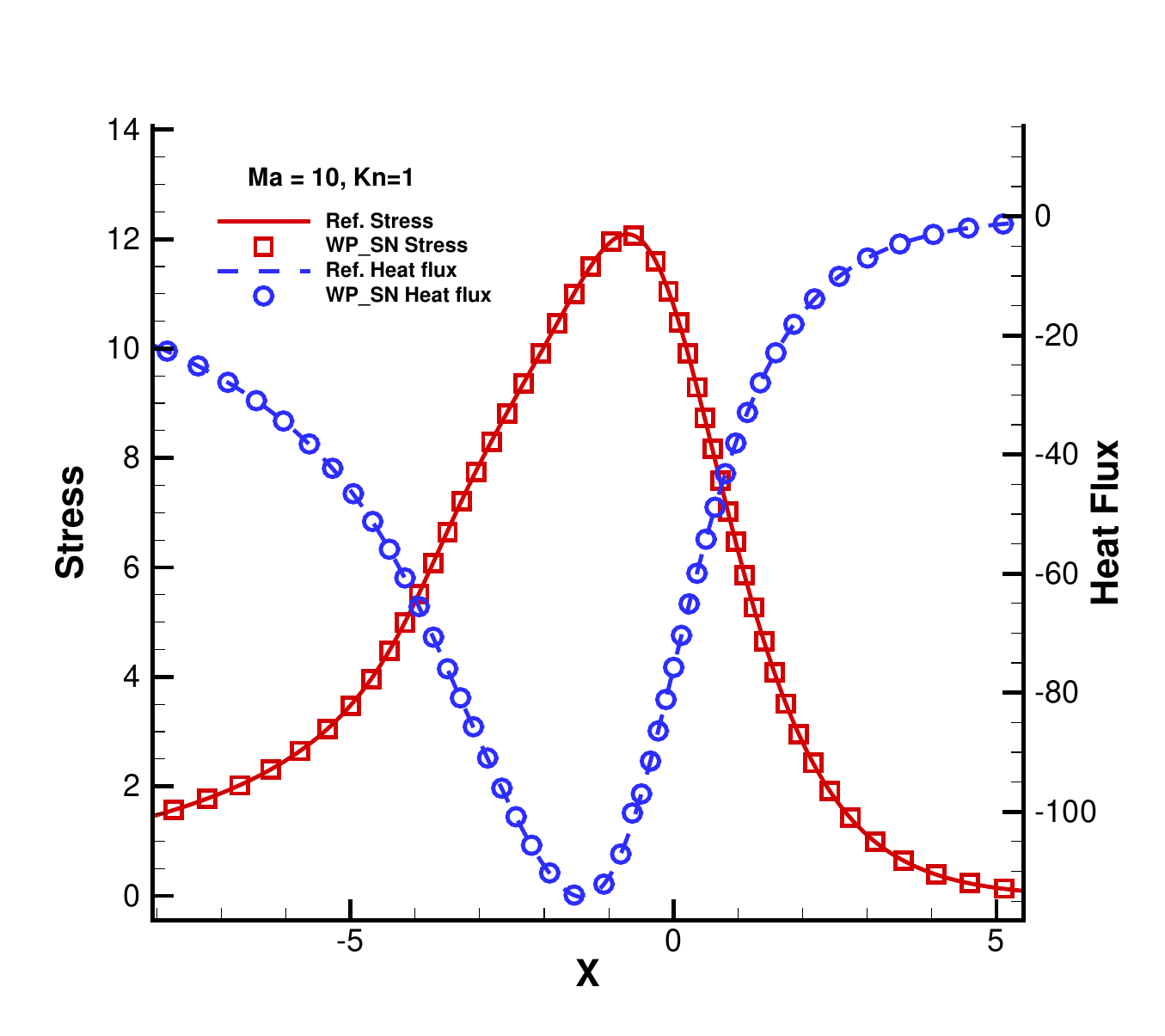}
        \caption{$M=10$, stress}
    \end{subfigure}
    \caption{Normal shock structure at \(M=10\) and \(\mathrm{Kn}=1\).}
    \label{fig:shock_ma10}
\end{figure}

\begin{figure}[!htbp]
    \centering
    \begin{subfigure}{0.49\textwidth}
        \centering
        \includegraphics[width=\linewidth]{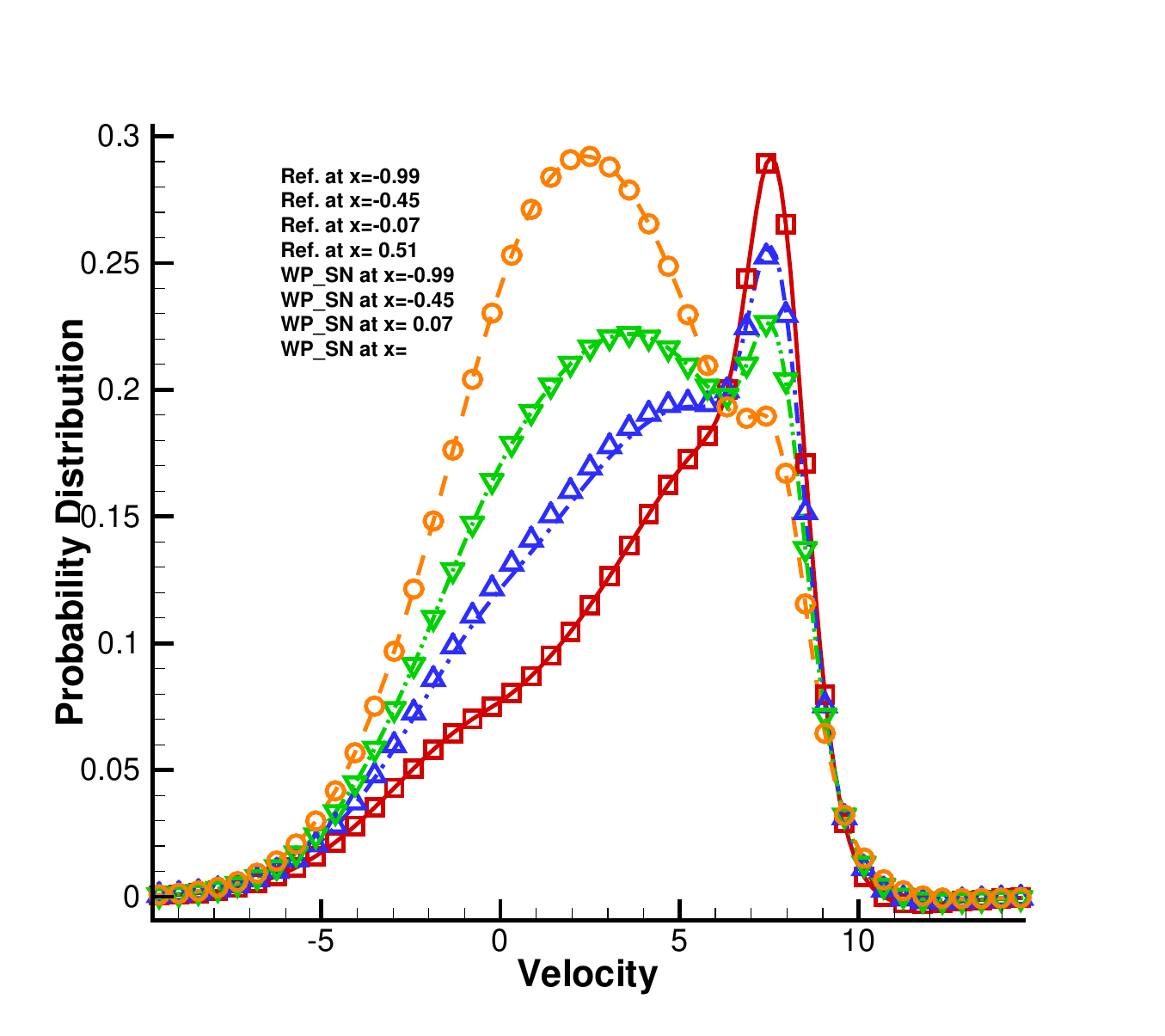}
        \caption{\(M=8\)}
    \end{subfigure}
    \begin{subfigure}{0.49\textwidth}
        \centering
        \includegraphics[width=\linewidth]{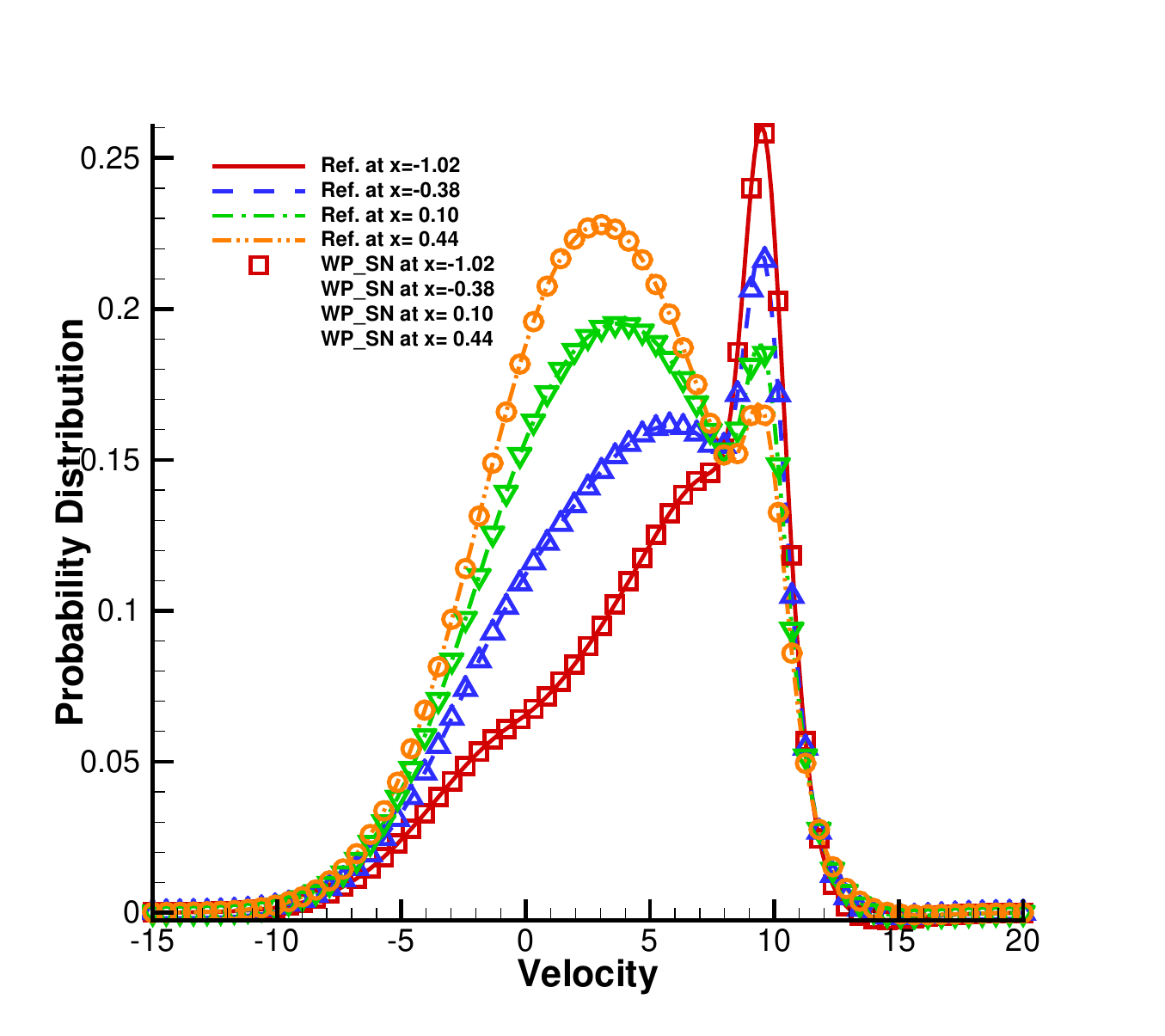}
        \caption{\(M=10\)}
    \end{subfigure}
    \caption{Velocity distribution functions across the normal shock layer.}
    \label{fig:shock_distribution}
\end{figure}

\subsection{Two-Dimensional Lid-Driven Cavity}

The lid-driven cavity extends the assessment to a two-dimensional wall-bounded flow. Two regimes are considered. The first is a rarefied cavity, where density and temperature fields are visibly affected by kinetic wall interaction. The second is the continuum incompressible cavity, where the WPD formulation recovers the standard viscous vortex structure and agrees with the benchmark data of Ghia et al.

For the WPD splitting, the local particle fraction is measured by
\[
\beta_P=\exp(-\mathcal{T}/\tau),\qquad \beta_W=1-\beta_P,
\]
where \(\mathcal{T}\) is the wave-particle horizon and \(\tau\) is the local relaxation time. With \(\mathrm{CFL}_l=1\), the reference horizon on a uniform square mesh is
\[
\mathcal{T} = \frac{\mathrm{CFL}_l\,\Delta x\Delta y}{\Delta y\,|\xi|_{\max}+\Delta x\,|\eta|_{\max}},
\]
with \(|\xi|_{\max}=|\eta|_{\max}=3\). This is Eq.~\eqref{eq:algorithm_local_horizon} specialized to a uniform square cell, using the maximal molecular speed for the reference estimate. Table~\ref{tab:cavity_wp_fraction} lists the resulting particle and wave fractions. The rarefied cases are particle dominated, whereas the \(Re=1000\) and \(Re=3200\) continuum cases are wave dominated and use mesh spacing and time step selected by the macroscopic CFL condition.

\begin{table}[!htbp]
\centering
\caption{Reference-state WPD fractions for the lid-driven cavity cases with \(\mathrm{CFL}_l=1\). Here \(\beta_P=\exp(-\mathcal{T}/\tau_0)\) is the particle fraction and \(\beta_W=1-\beta_P\) is the wave fraction.}
\label{tab:cavity_wp_fraction}
\begin{tabular}{ccccc}
\toprule
Case & Mesh & \(\tau_0\) & \(\beta_P\) & \(\beta_W\) \\
\midrule
\(\mathrm{Kn}=10\) & \(120\times120\) & \(1.1078\times10^{1}\) & \(9.9987\times10^{-1}\) & \(1.2537\times10^{-4}\) \\
\(\mathrm{Kn}=1\) & \(120\times120\) & \(1.1078\) & \(9.9875\times10^{-1}\) & \(1.2530\times10^{-3}\) \\
\(\mathrm{Kn}=0.075\) & \(120\times120\) & \(8.3084\times10^{-2}\) & \(9.8342\times10^{-1}\) & \(1.6578\times10^{-2}\) \\
\(Re=1000\) & \(200\times200\) & \(3.0000\times10^{-4}\) & \(6.2177\times10^{-2}\) & \(9.3782\times10^{-1}\) \\
\(Re=3200\) & \(600\times600\) & \(9.3750\times10^{-5}\) & \(5.1666\times10^{-2}\) & \(9.4833\times10^{-1}\) \\
\bottomrule
\end{tabular}
\end{table}

The discrete total mass is computed from the cell-centered density field as
\[
M_h=\sum_{i,j}\rho_{ij}\Delta x\Delta y .
\]
Since the initial density is uniform and the cavity area is unity, the reference mass is \(M_0=1\). Table~\ref{tab:cavity_mass_conservation} shows that the WPD-\(S_N\) calculations preserve the total mass to round-off accuracy.

\begin{table}[!htbp]
\centering
\caption{Discrete mass conservation in the WPD-\(S_N\) lid-driven cavity calculations.}
\label{tab:cavity_mass_conservation}
\begin{tabular}{ccc}
\toprule
Case & \(M_h\) & \((M_h-M_0)/M_0\) \\
\midrule
\(\mathrm{Kn}=10\) & \(1.0000000000\) & \(5.551\times10^{-15}\) \\
\(\mathrm{Kn}=1\) & \(1.0000000000\) & \(-3.331\times10^{-16}\) \\
\(\mathrm{Kn}=0.075\) & \(1.0000000000\) & \(2.354\times10^{-14}\) \\
\(Re=1000\) & \(1.0000000000\) & \(-6.661\times10^{-16}\) \\
\(Re=3200\) & \(1.0000000000\) & \(6.661\times10^{-16}\) \\
\bottomrule
\end{tabular}
\end{table}

\subsubsection{Rarefied cavity}

For the rarefied cavity calculations, the square cavity is discretized by a \(120\times120\) physical mesh. The discrete velocity space uses \(120\times120\) velocity points. Diffuse wall reflection is imposed on all solid boundaries, and the top lid moves tangentially. The WPD-\(S_N\) solution is computed with \(\mathrm{CFL}_l=1\). The representative cases shown here are \(\mathrm{Kn}=1\) and \(\mathrm{Kn}=0.075\), corresponding to a highly rarefied regime and a transition-regime case.

Figures~\ref{fig:cavity_kn1} and \ref{fig:cavity_kn0075} show density and temperature contours. At \(\mathrm{Kn}=1\), the wall-induced kinetic effect produces broad non-uniform density and temperature fields. As the Knudsen number is reduced to \(0.075\), the flow becomes more localized near the moving lid and the cavity core becomes closer to the continuum behavior. The centerline velocity profiles in Fig.~\ref{fig:cavity_rarefied_profiles} provide a quantitative comparison of the velocity structure and show that the deterministic particle component captures the rarefaction-induced slip and non-equilibrium profile variation.

\begin{figure}[!htbp]
    \centering
    \begin{subfigure}{0.48\textwidth}
        \centering
        \includegraphics[width=\linewidth]{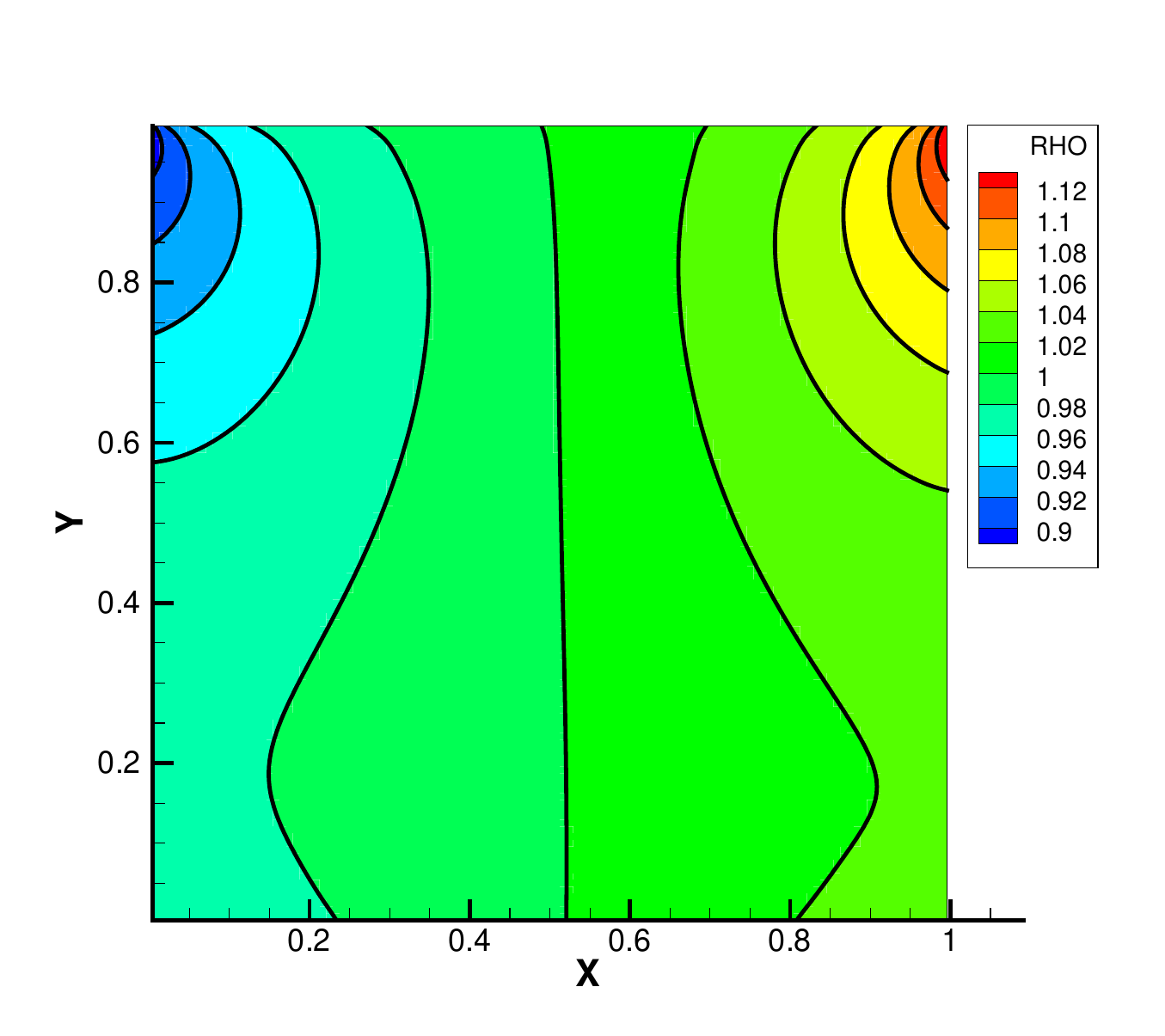}
        \caption{Density}
    \end{subfigure}
    \begin{subfigure}{0.48\textwidth}
        \centering
        \includegraphics[width=\linewidth]{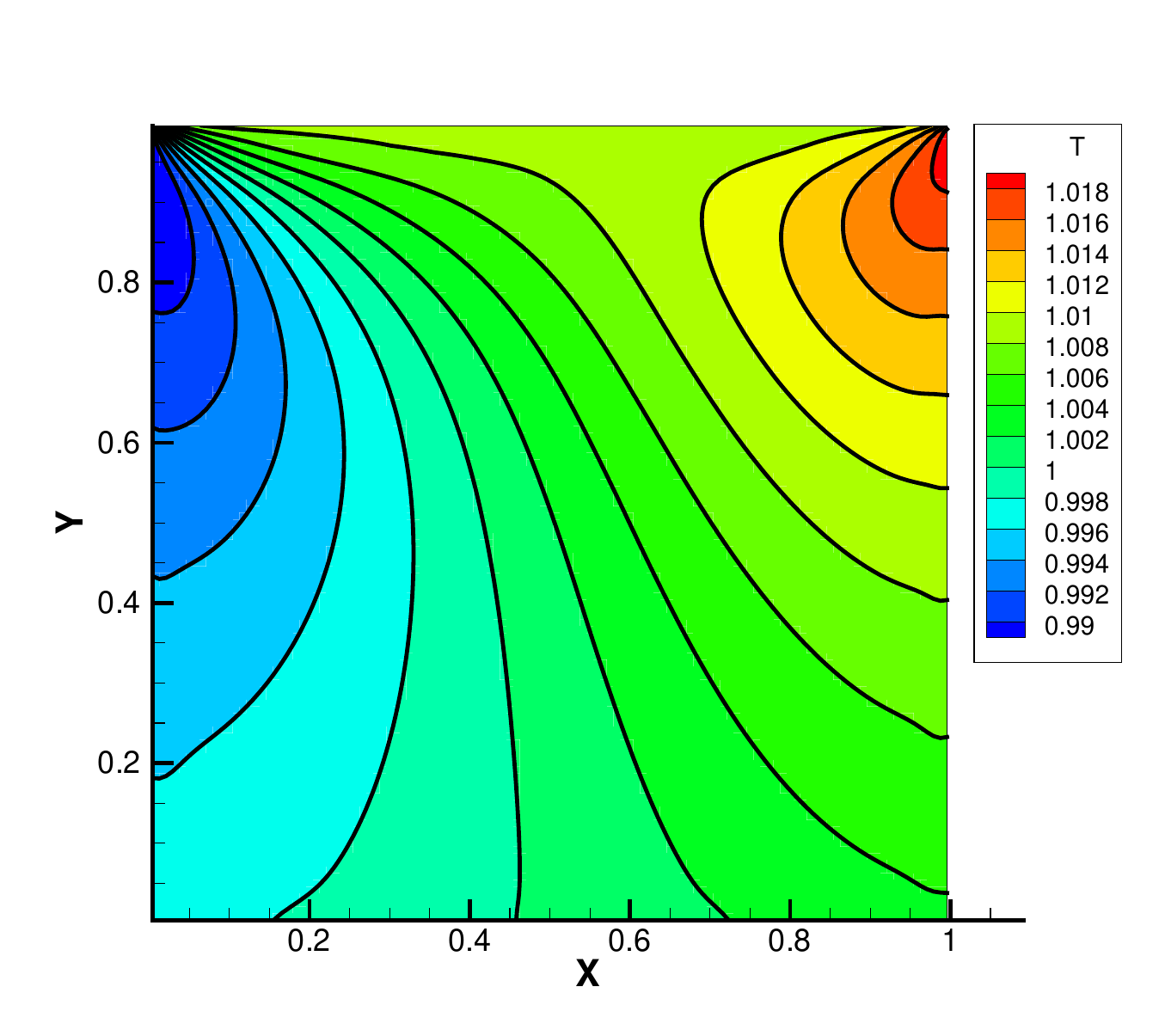}
        \caption{Temperature}
    \end{subfigure}
    \caption{Rarefied lid-driven cavity flow at $\mathrm{Kn}=1$.}
    \label{fig:cavity_kn1}
\end{figure}
\FloatBarrier

\begin{figure}[!htbp]
    \centering
    \begin{subfigure}{0.48\textwidth}
        \centering
        \includegraphics[width=\linewidth]{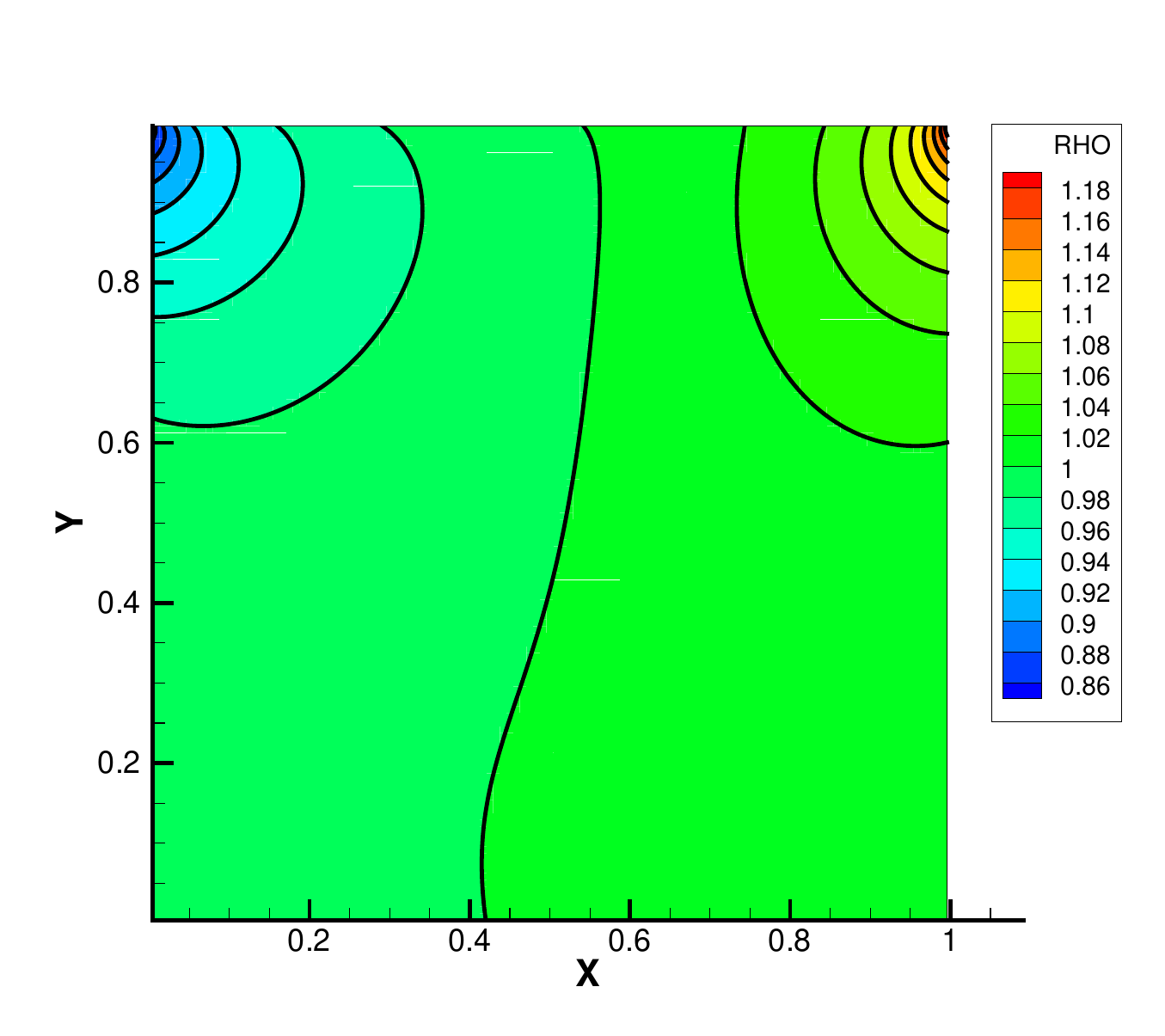}
        \caption{Density}
    \end{subfigure}
    \begin{subfigure}{0.48\textwidth}
        \centering
        \includegraphics[width=\linewidth]{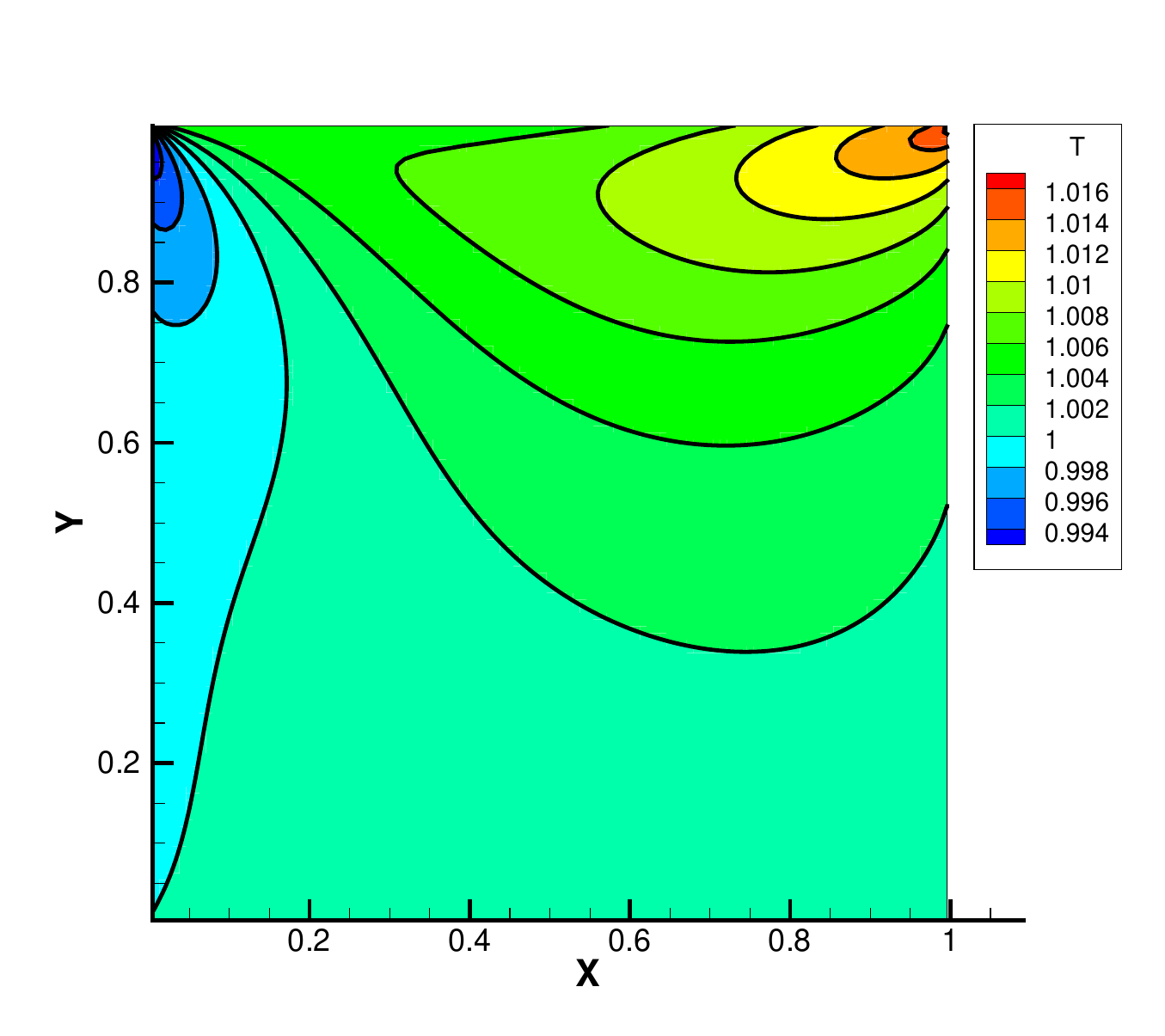}
        \caption{Temperature}
    \end{subfigure}
    \caption{Rarefied lid-driven cavity flow at $\mathrm{Kn}=0.075$.}
    \label{fig:cavity_kn0075}
\end{figure}

\begin{figure}[!htbp]
    \centering
    \begin{subfigure}{0.49\textwidth}
        \centering
        \includegraphics[width=\linewidth]{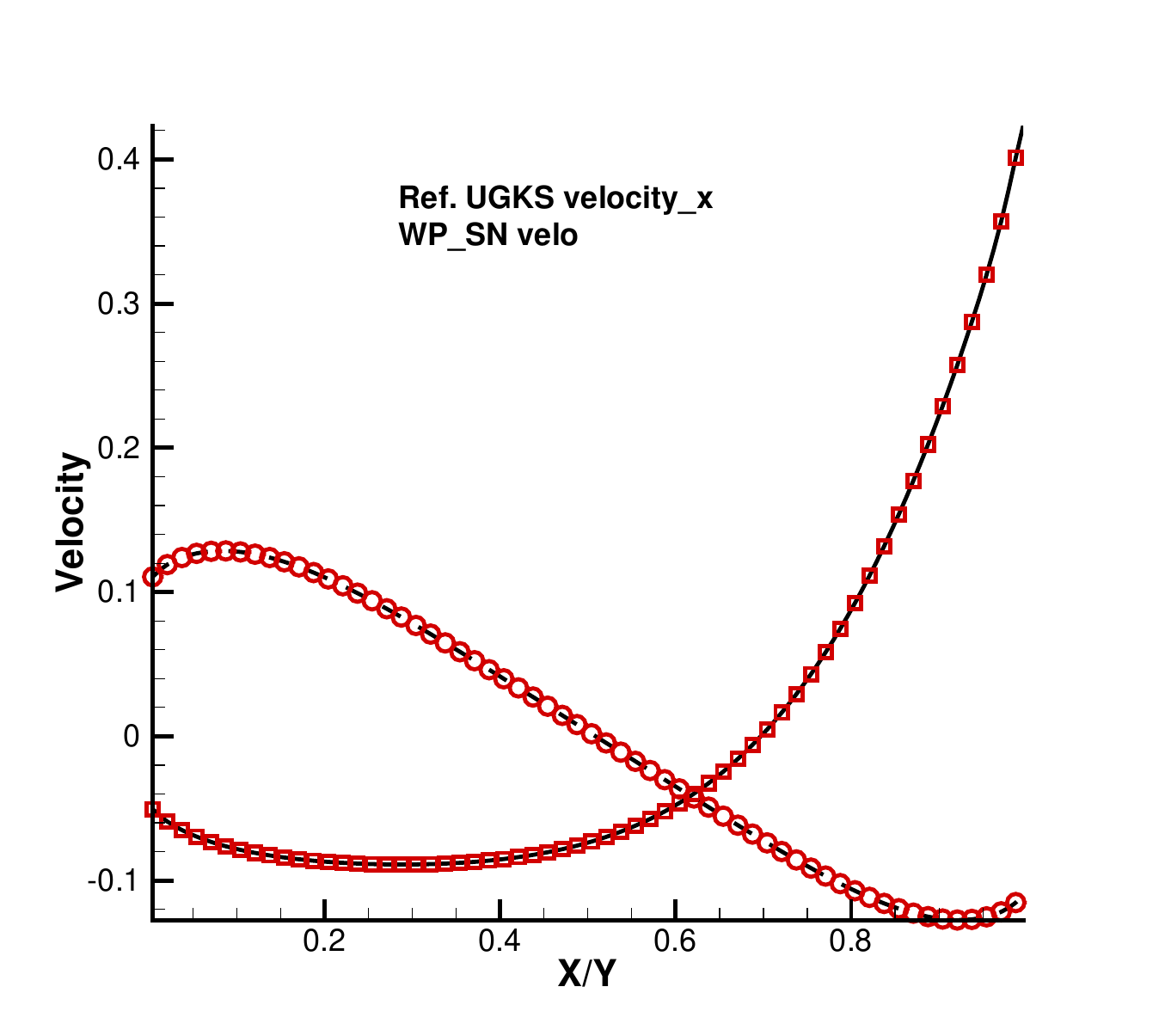}
        \caption{$\mathrm{Kn}=1$}
    \end{subfigure}
    \begin{subfigure}{0.49\textwidth}
        \centering
        \includegraphics[width=\linewidth]{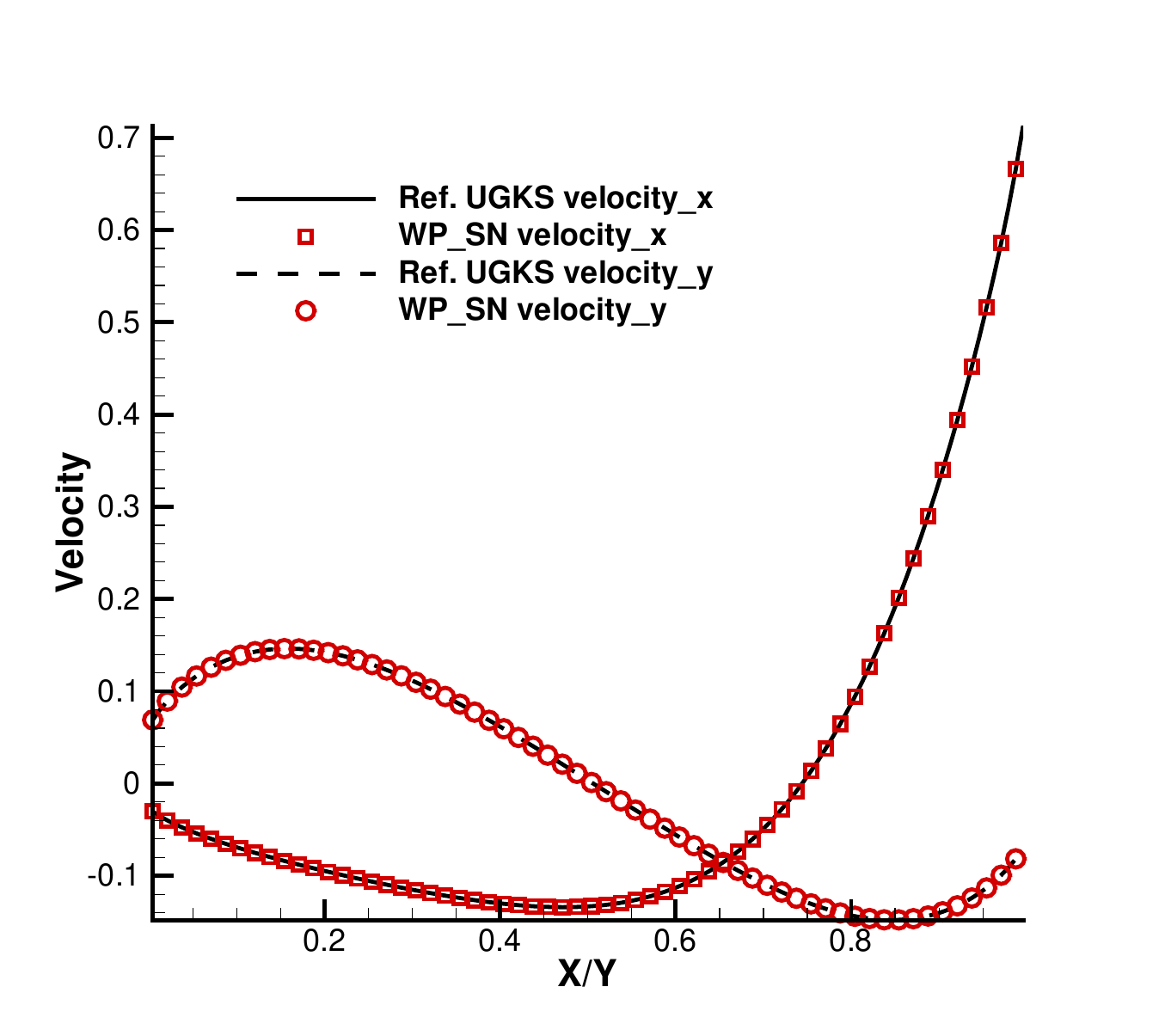}
        \caption{$\mathrm{Kn}=0.075$}
    \end{subfigure}
    \caption{Centerline velocity profiles for rarefied cavity flows.}
    \label{fig:cavity_rarefied_profiles}
\end{figure}

\subsubsection{Continuum cavity}

The continuum limit is assessed by the classical lid-driven cavity at \(Re=1000\) and \(Re=3200\). In these cases the wave component dominates and the WPD scheme reduces to the continuum GKS behavior. The \(Re=1000\) case is computed on a \(200\times200\) mesh, while the \(Re=3200\) case uses a finer \(600\times600\) mesh. The velocity profiles are compared with the benchmark data of Ghia et al. along the vertical and horizontal centerlines.

Figures~\ref{fig:cavity_re1000} and \ref{fig:cavity_re3200} show the streamlines overlaid on the velocity-magnitude contours and the corresponding centerline velocity comparisons. The primary vortex and the secondary corner structures are recovered, and the centerline profiles agree well with the benchmark data. This agreement supports the Navier--Stokes limiting behavior of the WPD formulation when the kinetic particle contribution becomes negligible.

\begin{figure}[!htbp]
    \centering
    \begin{subfigure}{0.47\textwidth}
        \centering
        \includegraphics[width=\linewidth]{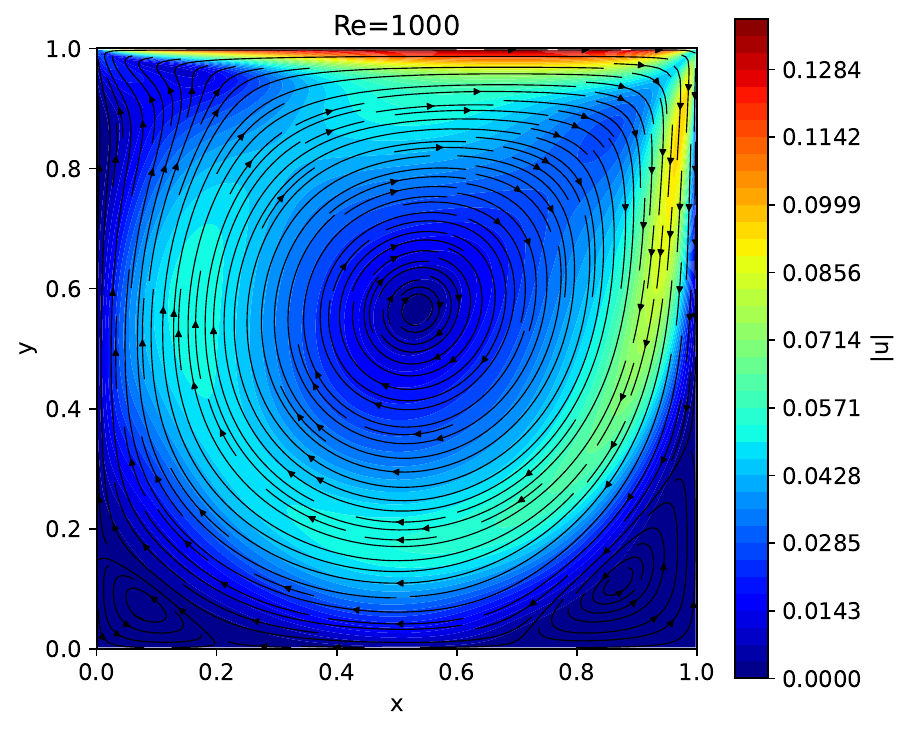}
        \caption{Streamlines and velocity magnitude}
    \end{subfigure}
    \begin{subfigure}{0.50\textwidth}
        \centering
        \includegraphics[width=\linewidth]{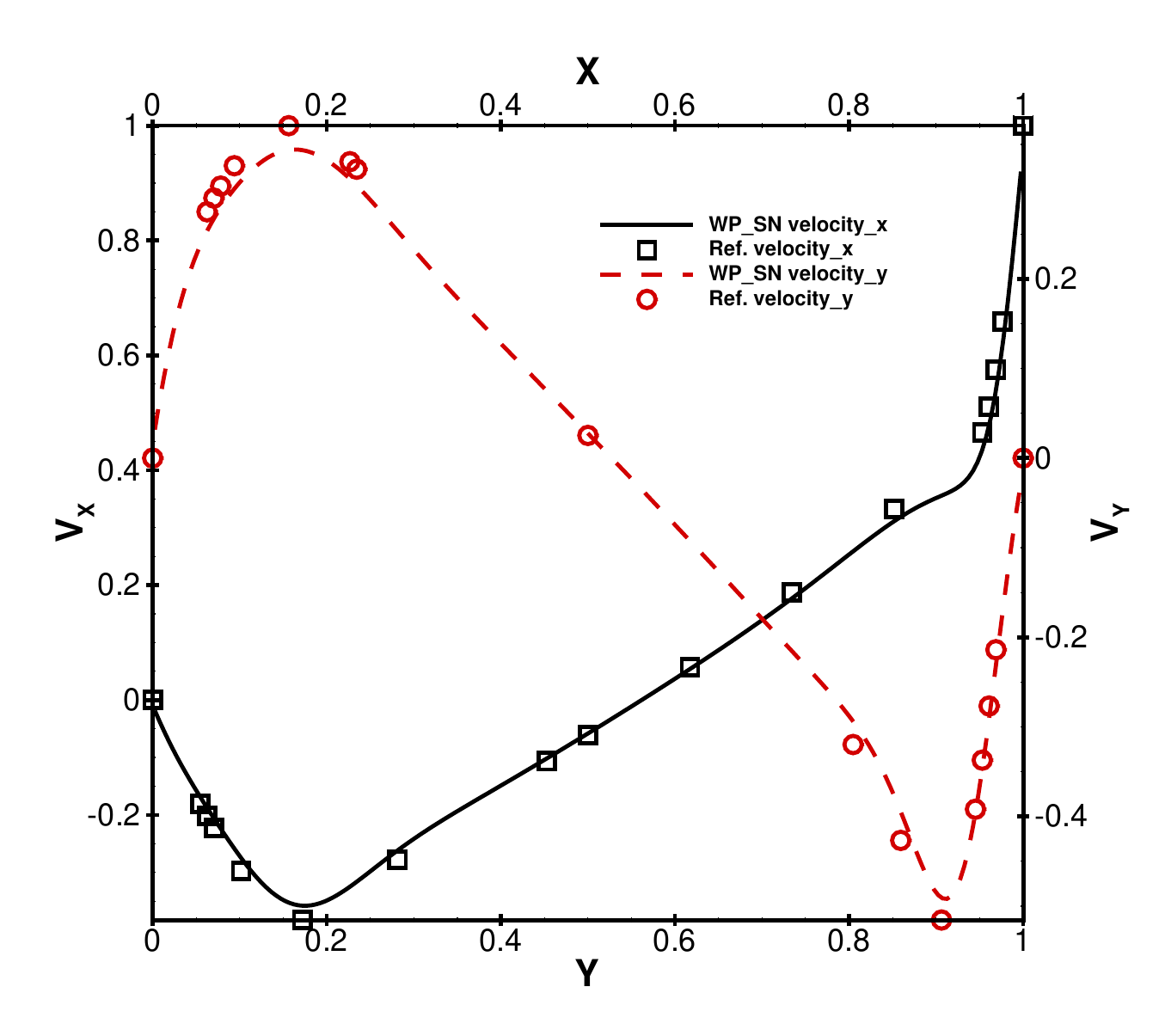}
        \caption{Centerline velocities}
    \end{subfigure}
    \caption{Continuum cavity flow at $Re=1000$ compared with the Ghia et al. benchmark.}
    \label{fig:cavity_re1000}
\end{figure}

\begin{figure}[!htbp]
    \centering
    \begin{subfigure}{0.47\textwidth}
        \centering
        \includegraphics[width=\linewidth]{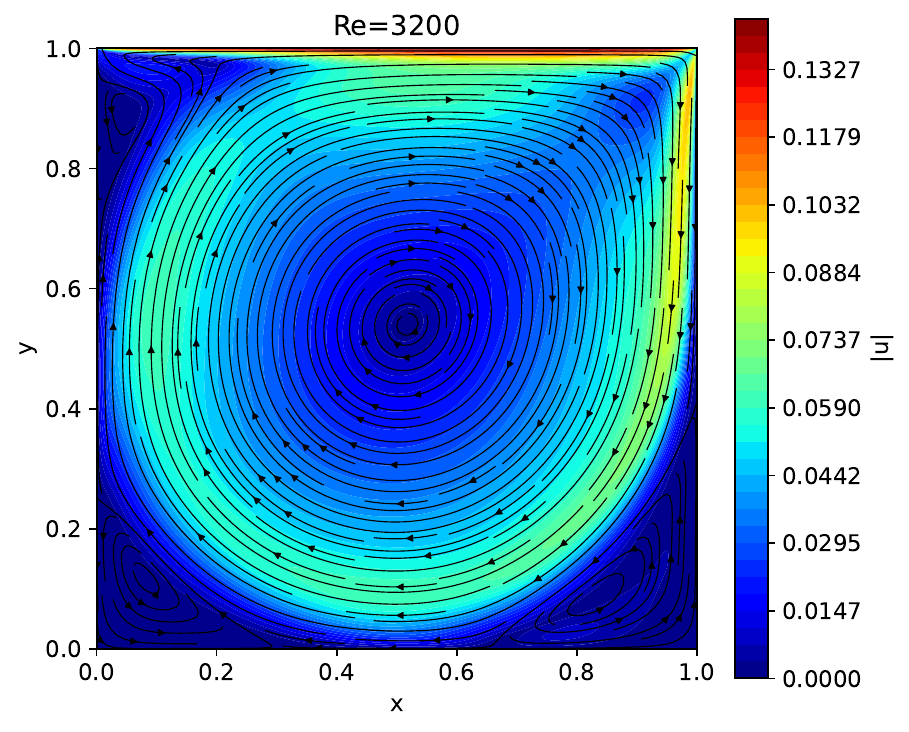}
        \caption{Streamlines and velocity magnitude}
    \end{subfigure}
    \begin{subfigure}{0.50\textwidth}
        \centering
        \includegraphics[width=\linewidth]{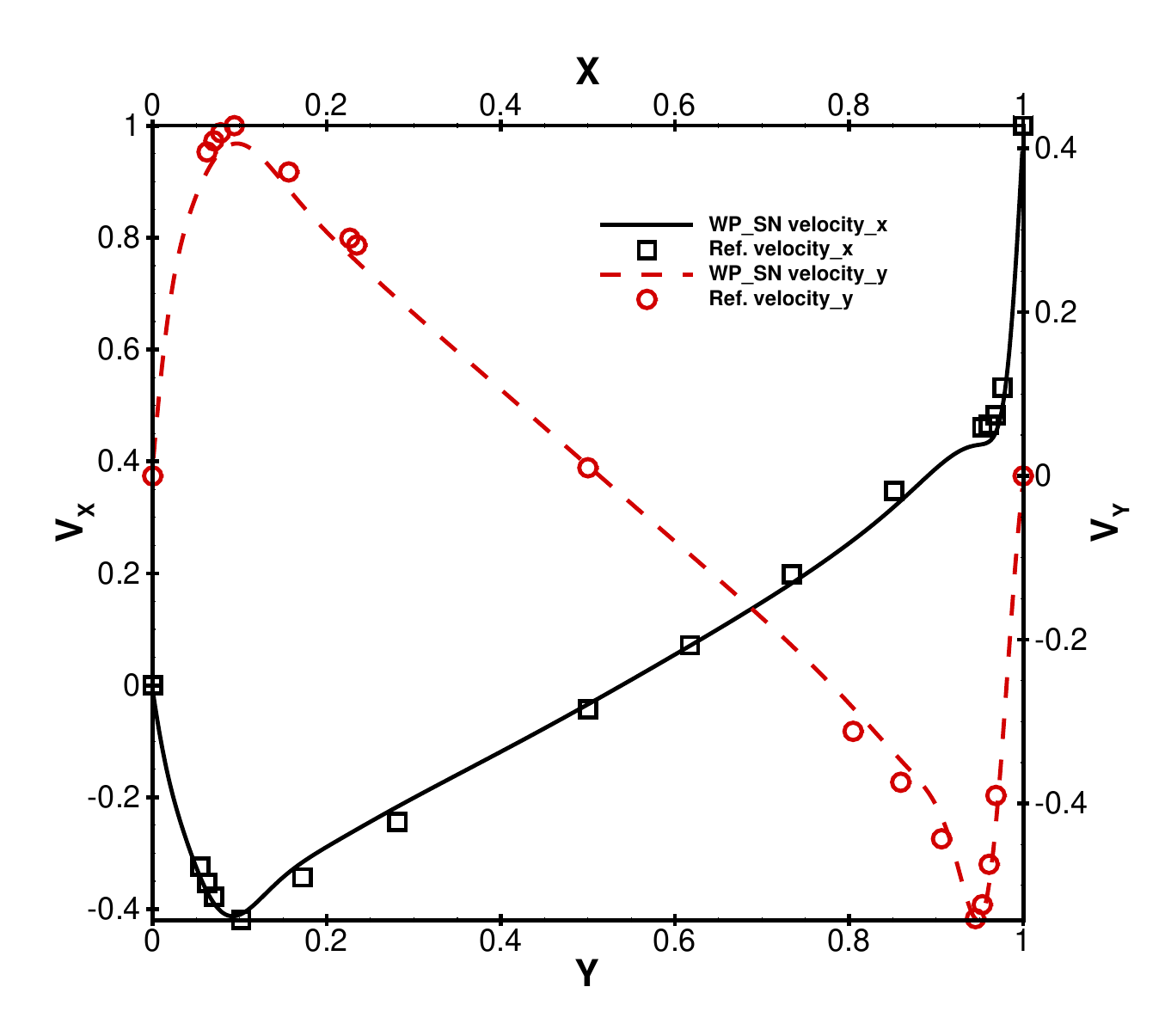}
        \caption{Centerline velocities}
    \end{subfigure}
    \caption{Continuum cavity flow at $Re=3200$ compared with the Ghia et al. benchmark.}
    \label{fig:cavity_re3200}
\end{figure}
\FloatBarrier

\subsection{Hypersonic Flow Around a Cylinder}

The hypersonic cylinder case examines the WPD framework in a non-equilibrium external flow with a curved diffuse-reflection wall. The free-stream Mach number is \(M_\infty=8\). The cylinder radius is \(R_c=1\), and the far-field radius is \(15R_c\). A body-fitted O-grid is used with \(100\) cells in the circumferential direction and \(64\) cells in the radial direction; the radial mesh is smoothly stretched with a factor of \(2.5\). The time step is determined by \(\mathrm{CFL}=0.2\), and all cases are advanced for \(25000\) steps. For the deterministic calculations, the velocity space uses \(64\times64\) points with \(u\in[-10,20]\) and \(v\in[-15,15]\). For the stochastic calculations, \(100\) reference particles per cell are used, and the reported MC fields are averaged after the transient stage.

The results compare WPD-\(S_N\) and WPD-MC for representative Knudsen numbers \(\mathrm{Kn}=1\), \(10^{-2}\), and \(10^{-3}\). The local wave-particle horizon is specified by \(\mathrm{CFL}_l=1\). The field contours assess the bow shock and wake structure, while the stagnation-line and wall-surface profiles provide more quantitative comparisons near the strongest compression and wall-interaction regions.

To quantify the consistency of the decomposition, Table~\ref{tab:cylinder_consistency_error} reports relative \(L_2\) differences on the final two-dimensional fields. The deterministic comparison uses a reference UGKS solution~\cite{huang2012ugks} for WPD-\(S_N\), while the stochastic comparison uses the averaged pure MC solution as the reference for WPD-MC. For a flow variable \(q\), the reported error is
\[
E_2(q)=
\left(
\frac{\sum_i (q_i^{\mathrm{WPD}}-q_i^{\mathrm{ref}})^2 V_i}
{\sum_i (q_i^{\mathrm{ref}})^2 V_i}
\right)^{1/2}.
\]
The deterministic \(S_N\) and WPD-\(S_N\) calculations were advanced to a steady state: the residual norm of the conservative variables decreases from \(O(1)\) to \(O(10^{-6})\) over the \(25000\) steps. The MC comparisons should instead be interpreted statistically, since the stochastic solver is assessed through time-averaged particle moments rather than residual decay.

\begin{table}[!htbp]
\centering
\caption{Relative \(L_2\) consistency errors for the cylinder calculations. WPD-\(S_N\) is compared with the reference UGKS solution, and WPD-MC is compared with the averaged pure MC reference on the same mesh.}
\label{tab:cylinder_consistency_error}
\small
\begin{tabular}{llccccc}
\toprule
\(\mathrm{Kn}\) & Comparison & \(\rho\) & \(u\) & \(v\) & \(p\) & \(T\) \\
\midrule
\(1\) & WPD-\(S_N\)/UGKS & \(1.22\times10^{-3}\) & \(1.18\times10^{-3}\) & \(5.95\times10^{-3}\) & \(2.95\times10^{-3}\) & \(4.89\times10^{-3}\) \\
\(1\) & WPD-MC/pure MC & \(1.19\times10^{-3}\) & \(8.69\times10^{-5}\) & \(5.85\times10^{-4}\) & \(1.72\times10^{-4}\) & \(2.21\times10^{-4}\) \\
\(10^{-2}\) & WPD-\(S_N\)/UGKS & \(1.09\times10^{-2}\) & \(5.14\times10^{-3}\) & \(1.52\times10^{-2}\) & \(1.17\times10^{-2}\) & \(2.16\times10^{-2}\) \\
\(10^{-2}\) & WPD-MC/pure MC & \(1.47\times10^{-2}\) & \(6.14\times10^{-3}\) & \(2.83\times10^{-2}\) & \(2.10\times10^{-2}\) & \(2.06\times10^{-2}\) \\
\(10^{-3}\) & WPD-\(S_N\)/UGKS & \(2.69\times10^{-3}\) & \(9.91\times10^{-4}\) & \(4.20\times10^{-3}\) & \(3.00\times10^{-3}\) & \(3.92\times10^{-3}\) \\
\(10^{-3}\) & WPD-MC/pure MC & \(1.30\times10^{-2}\) & \(9.53\times10^{-3}\) & \(3.77\times10^{-2}\) & \(2.77\times10^{-2}\) & \(4.13\times10^{-2}\) \\
\bottomrule
\end{tabular}
\end{table}

\subsubsection{Cylinder flow at \texorpdfstring{\(\mathrm{Kn}=1\)}{Kn=1}}

Figures~\ref{fig:cylinder_sn_kn1_pt}--\ref{fig:cylinder_mc_R_den_p_kn1} show the pressure, temperature, velocity, stagnation-line, and wall-surface results at \(\mathrm{Kn}=1\). The deterministic line and surface-pressure profiles compare WPD-\(S_N\) directly with the reference UGKS solution. Both WPD-\(S_N\) and WPD-MC capture the detached bow shock and the broad rarefied wake. The deterministic and stochastic implementations produce the same main flow structure, while WPD-MC naturally contains sampling noise in the low-density wake and other particle-dominated regions. Figure~\ref{fig:cylinder_pnumber_a_kn1} reports the equivalent particle-number field, identifying where the kinetic component is required.

\begin{figure}[!htbp]
    \centering
    \begin{subfigure}{0.49\textwidth}
        \centering
        \includegraphics[width=\linewidth]{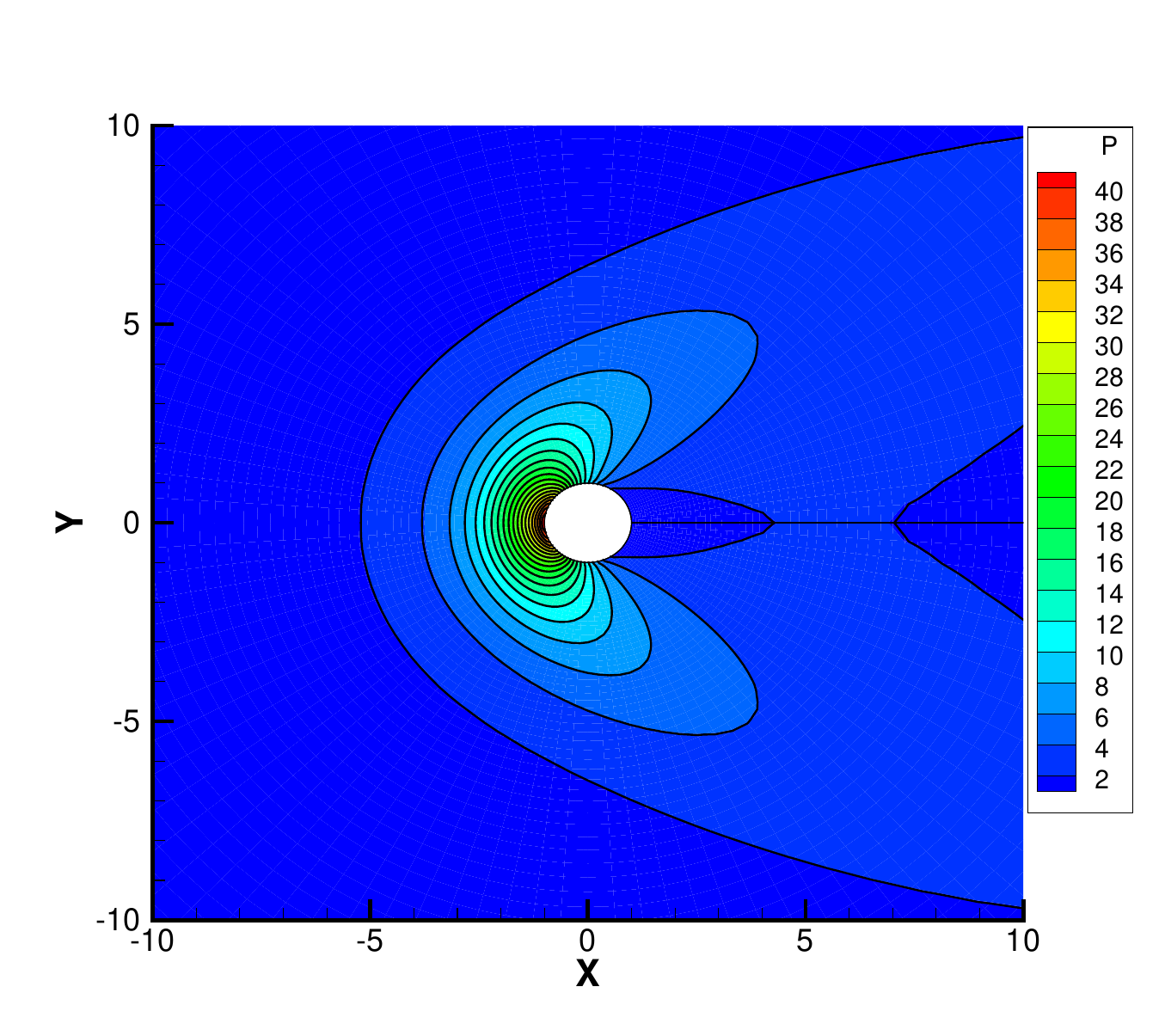}
        \caption{WPD-\(S_N\), $p$}
    \end{subfigure}
    \begin{subfigure}{0.49\textwidth}
        \centering
        \includegraphics[width=\linewidth]{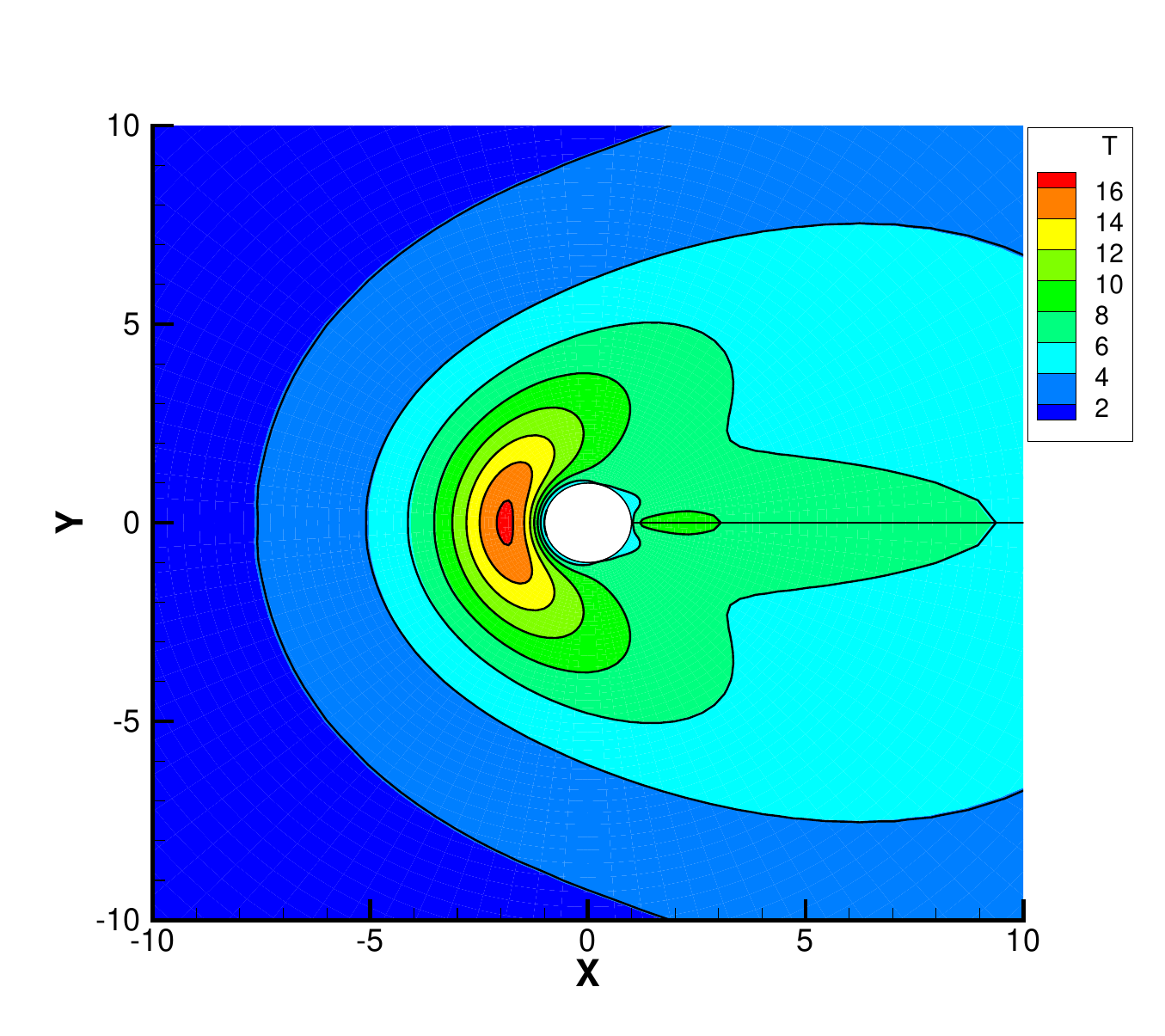}
        \caption{WPD-\(S_N\), $T$}
    \end{subfigure}
    \caption{Cylinder pressure and temperature contours at $\mathrm{Kn}=1$.}
    \label{fig:cylinder_sn_kn1_pt}
\end{figure}

\begin{figure}[!htbp]
    \centering
    \begin{subfigure}{0.49\textwidth}
        \centering
        \includegraphics[width=\linewidth]{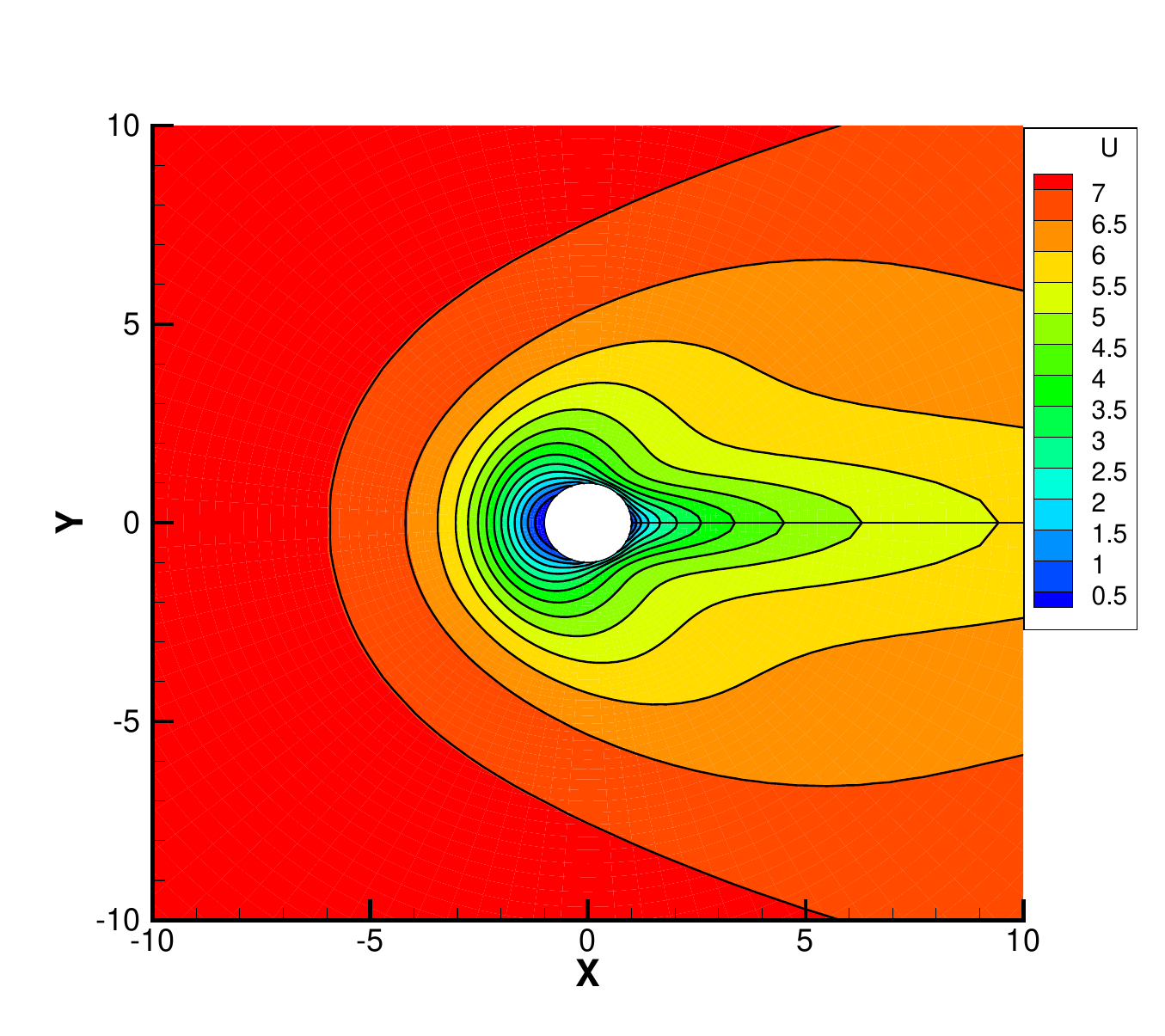}
        \caption{WPD-\(S_N\), $u$}
    \end{subfigure}
    \begin{subfigure}{0.49\textwidth}
        \centering
        \includegraphics[width=\linewidth]{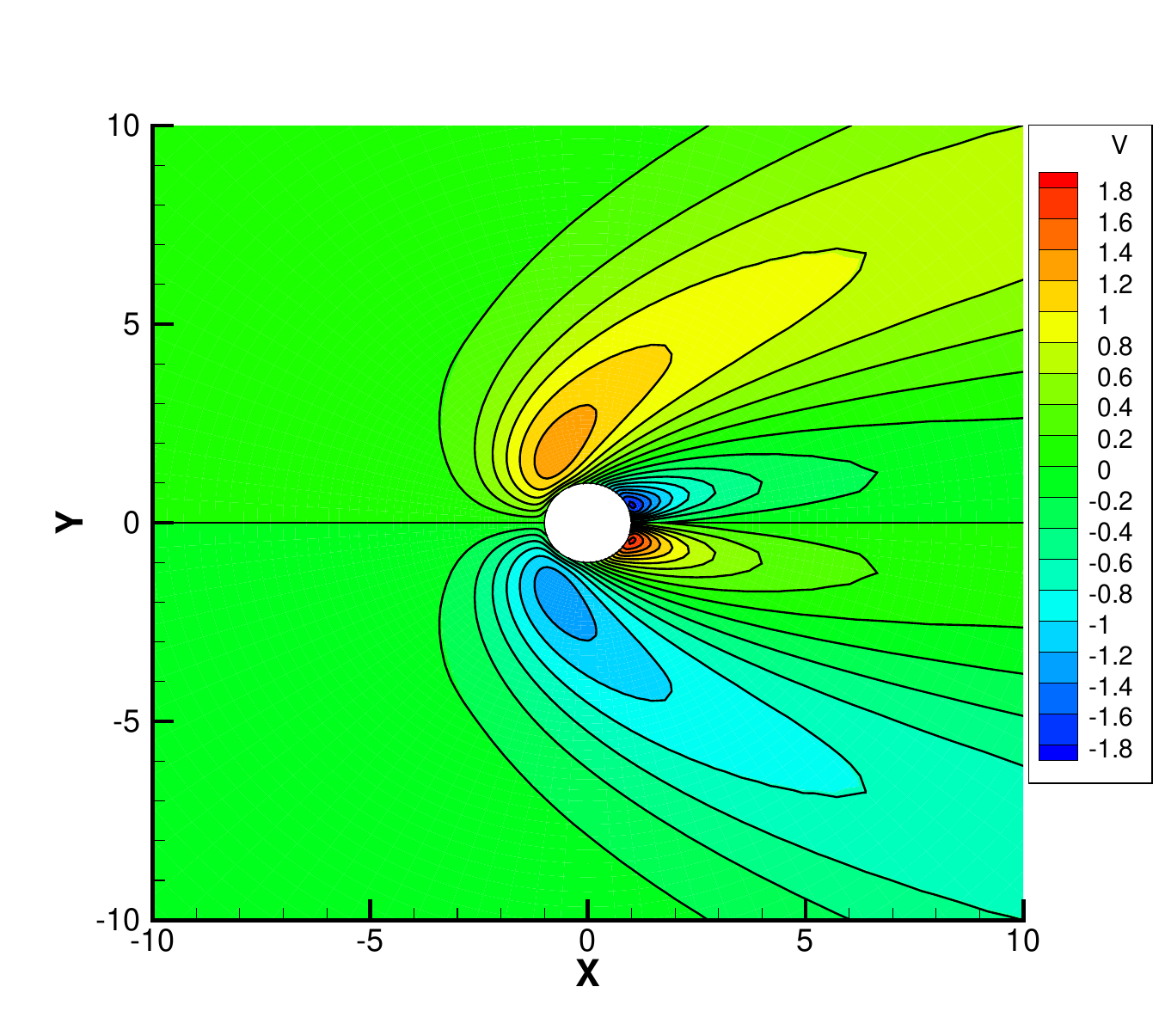}
        \caption{WPD-\(S_N\), $v$}
    \end{subfigure}
    \caption{Cylinder velocity contours at $\mathrm{Kn}=1$.}
    \label{fig:cylinder_sn_kn1_uv}
\end{figure}

\begin{figure}[!htbp]
    \centering
    \begin{subfigure}{0.49\textwidth}
        \centering
        \includegraphics[width=\linewidth]{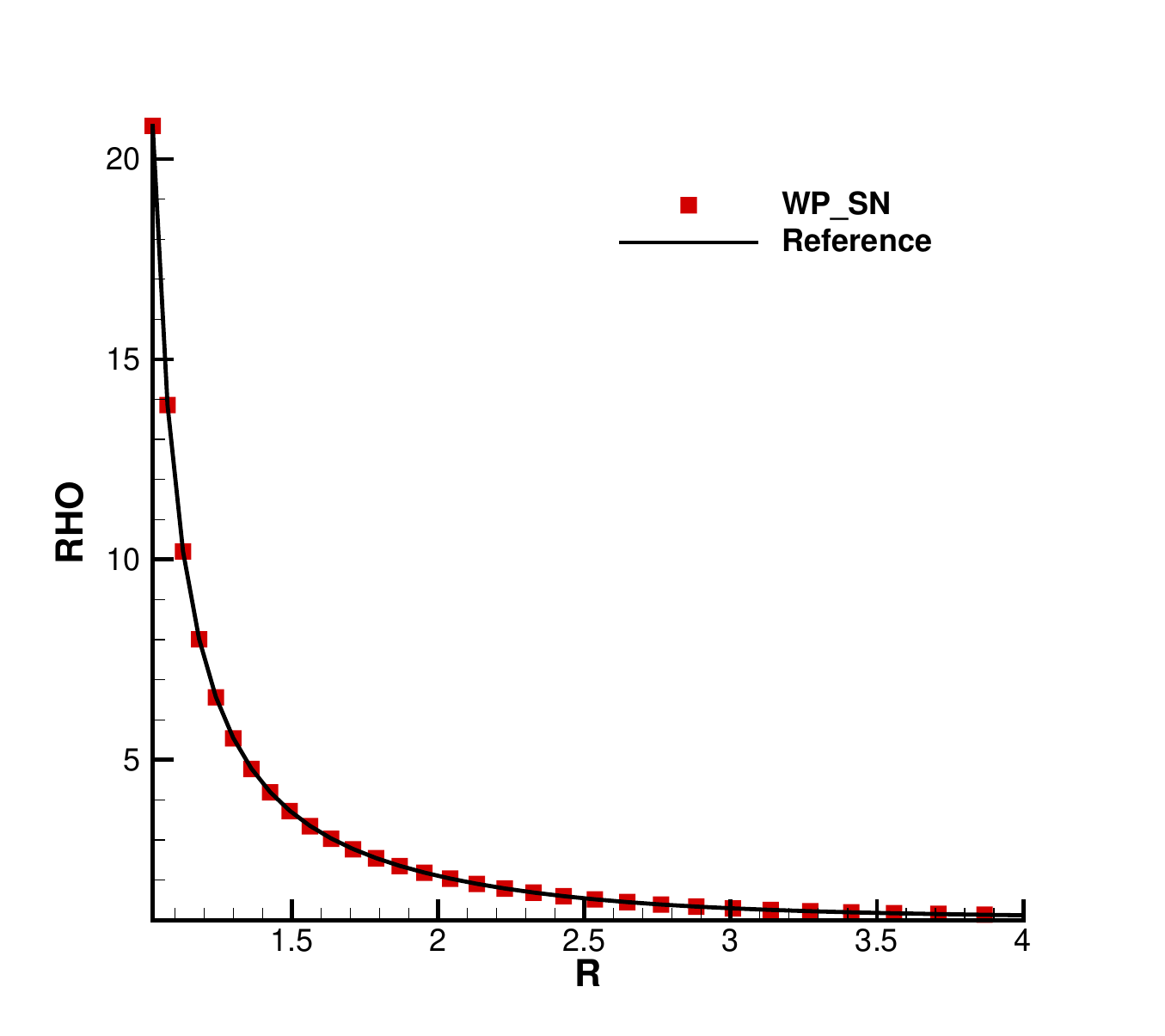}
        \caption{WPD-\(S_N\), $\rho$}
    \end{subfigure}
    \begin{subfigure}{0.49\textwidth}
        \centering
        \includegraphics[width=\linewidth]{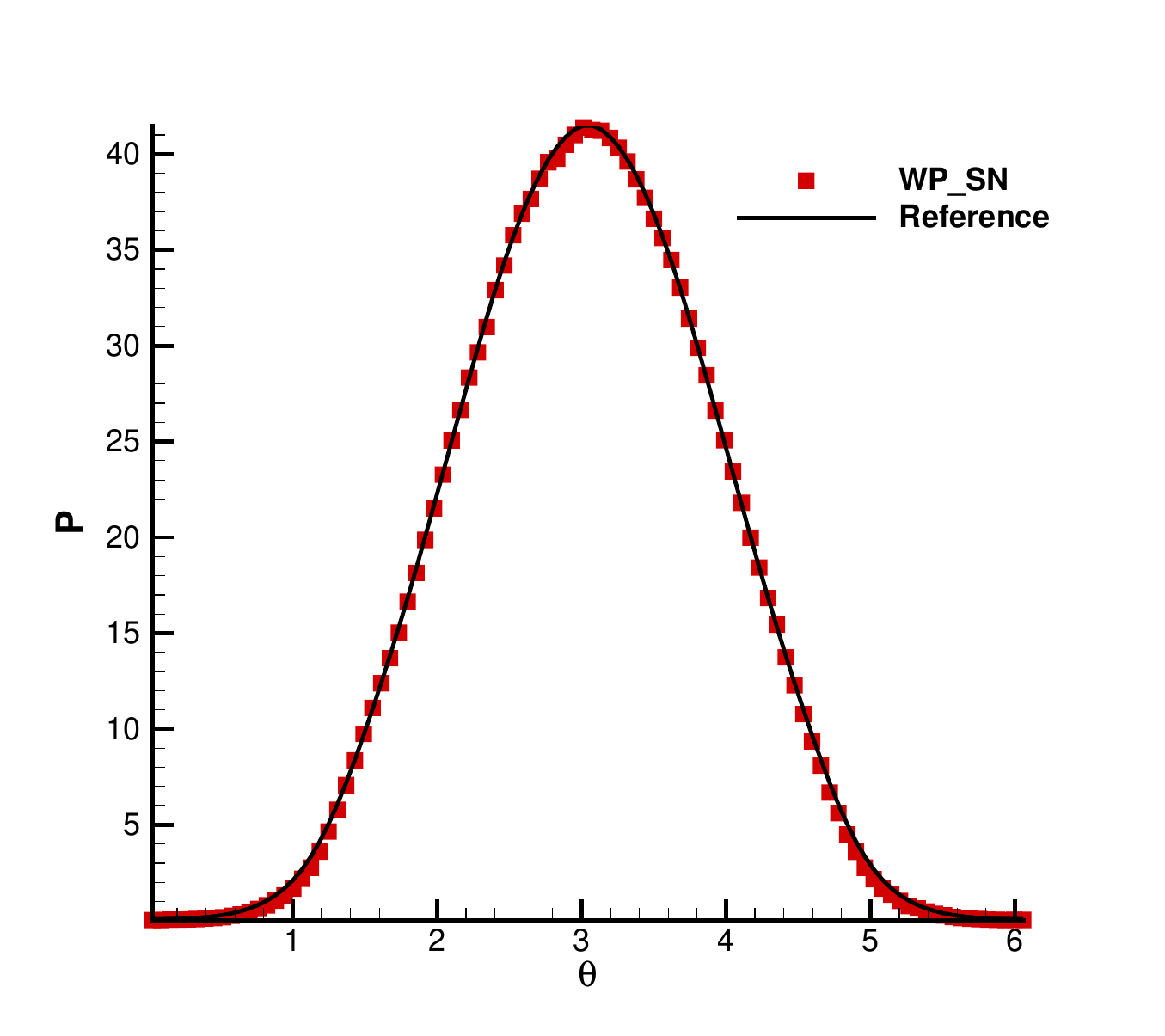}
        \caption{WPD-\(S_N\), $p$}
    \end{subfigure}
    \caption{Cylinder stagnation-line density and wall pressure at $\mathrm{Kn}=1$, comparing WPD-\(S_N\) with the reference UGKS solution.}
    \label{fig:cylinder_sn_R_den_p_kn1}
\end{figure}

\begin{figure}[!htbp]
    \centering
    \begin{subfigure}{0.49\textwidth}
        \centering
        \includegraphics[width=\linewidth]{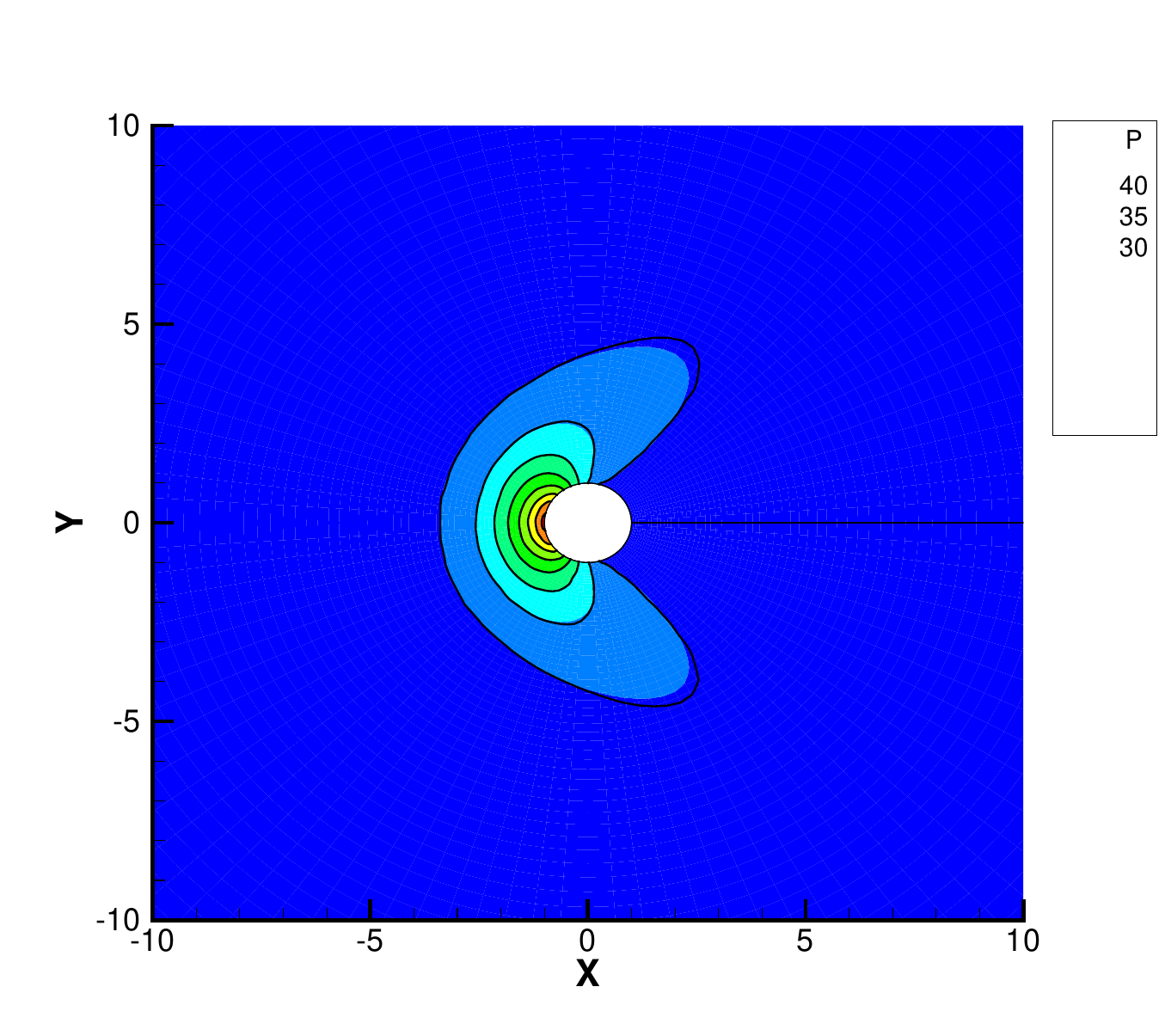}
        \caption{WPD-MC, $p$}
    \end{subfigure}
    \begin{subfigure}{0.49\textwidth}
        \centering
        \includegraphics[width=\linewidth]{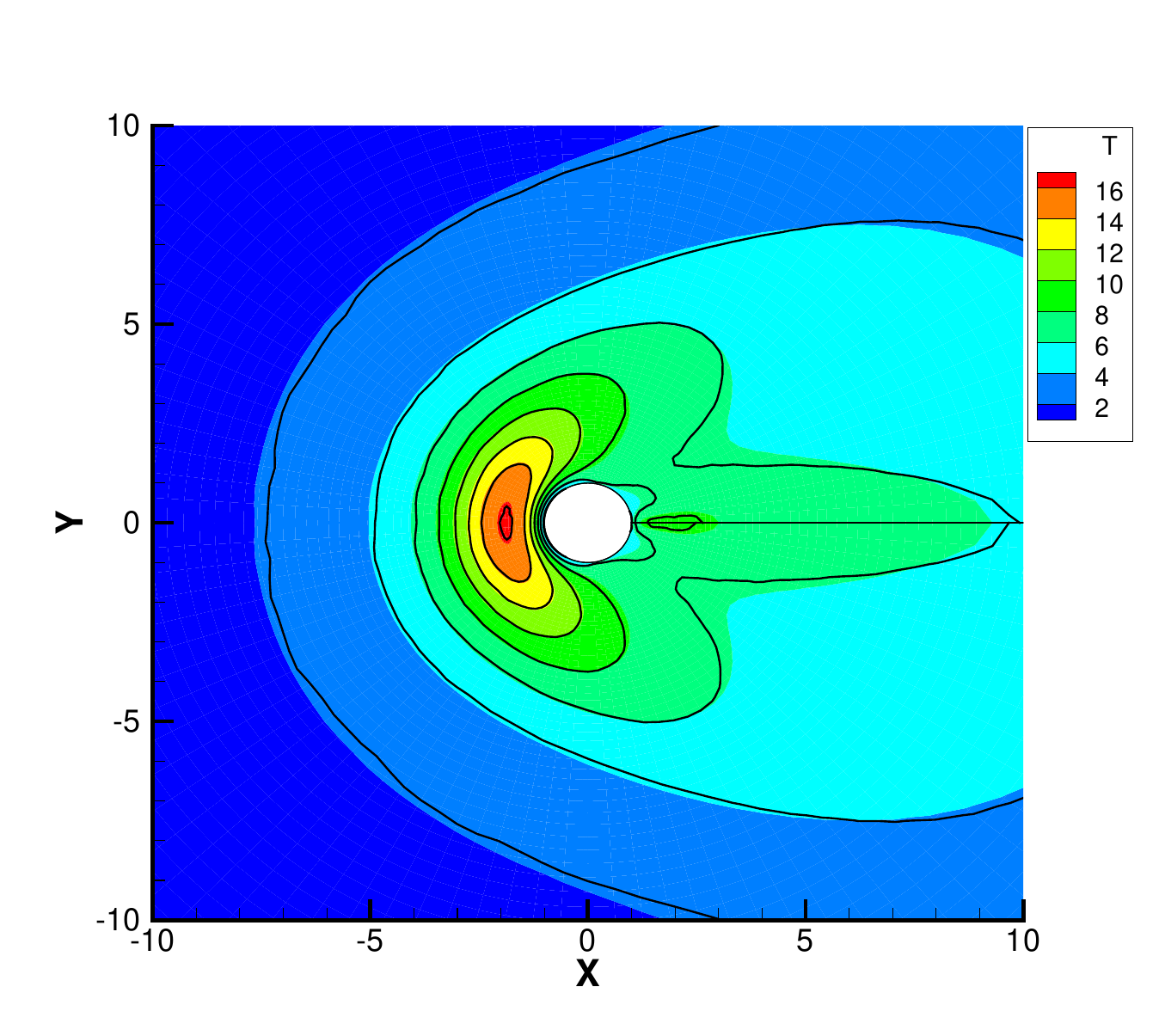}
        \caption{WPD-MC, $T$}
    \end{subfigure}
    \caption{Cylinder pressure and temperature contours at $\mathrm{Kn}=1$.}
    \label{fig:cylinder_mc_kn1_pt}
\end{figure}
\FloatBarrier

\begin{figure}[!htbp]
    \centering
    \begin{subfigure}{0.49\textwidth}
        \centering
        \includegraphics[width=\linewidth]{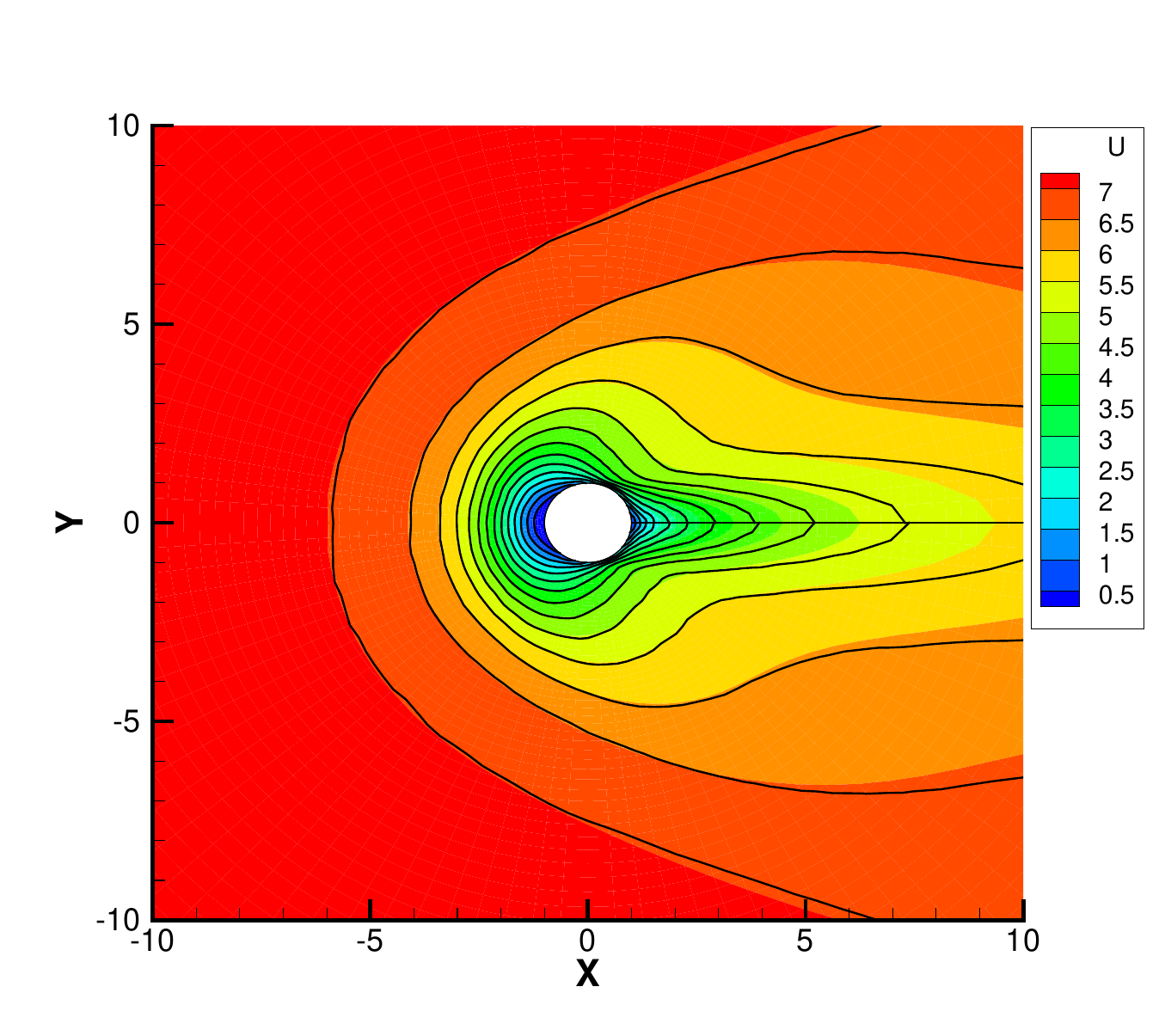}
        \caption{WPD-MC, $u$}
    \end{subfigure}
    \begin{subfigure}{0.49\textwidth}
        \centering
        \includegraphics[width=\linewidth]{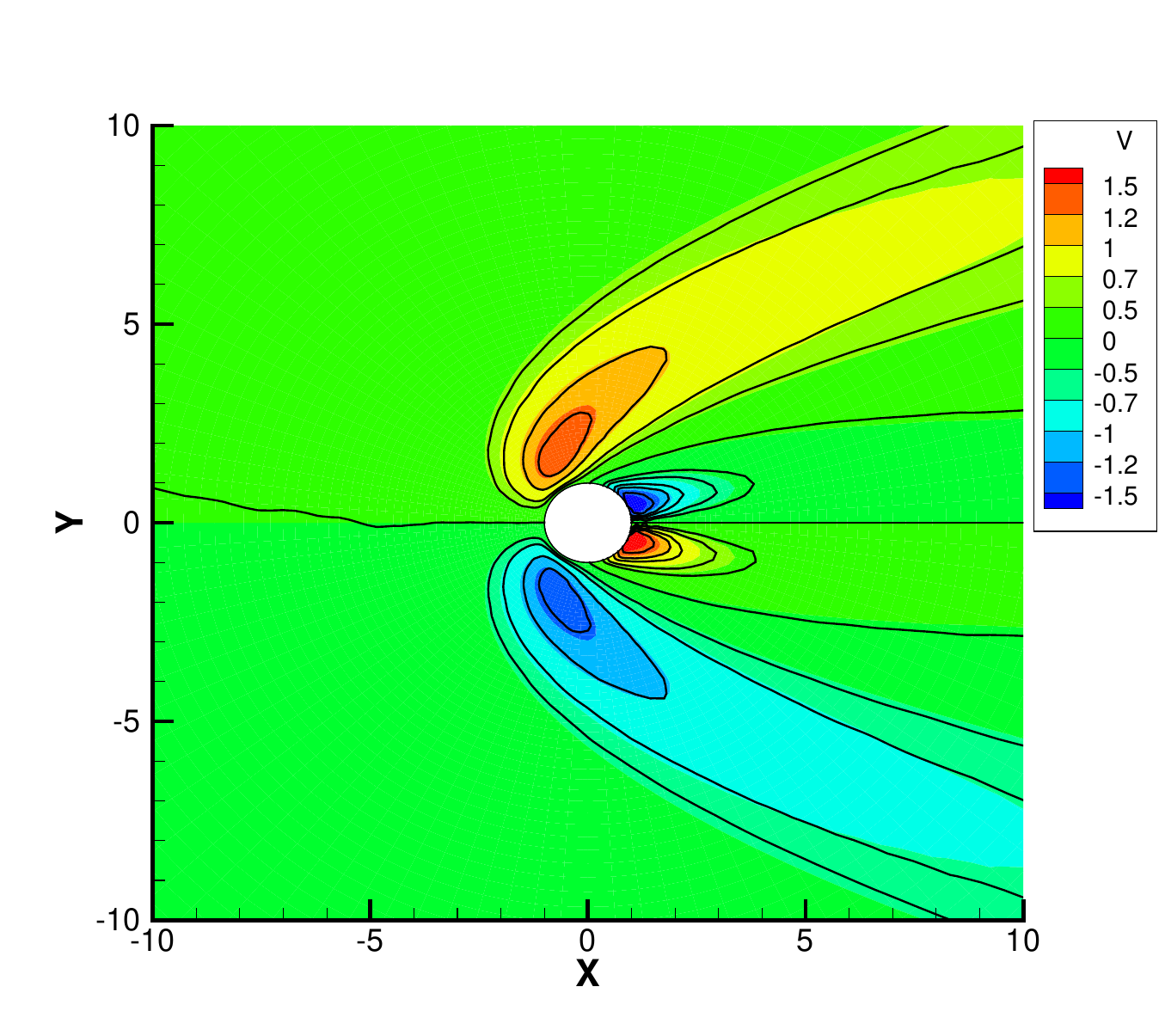}
        \caption{WPD-MC, $v$}
    \end{subfigure}
    \caption{Cylinder velocity contours at $\mathrm{Kn}=1$.}
    \label{fig:cylinder_mc_kn1_uv}
\end{figure}

\begin{figure}[!htbp]
    \centering
    \begin{subfigure}{0.49\textwidth}
        \centering
        \includegraphics[width=\linewidth]{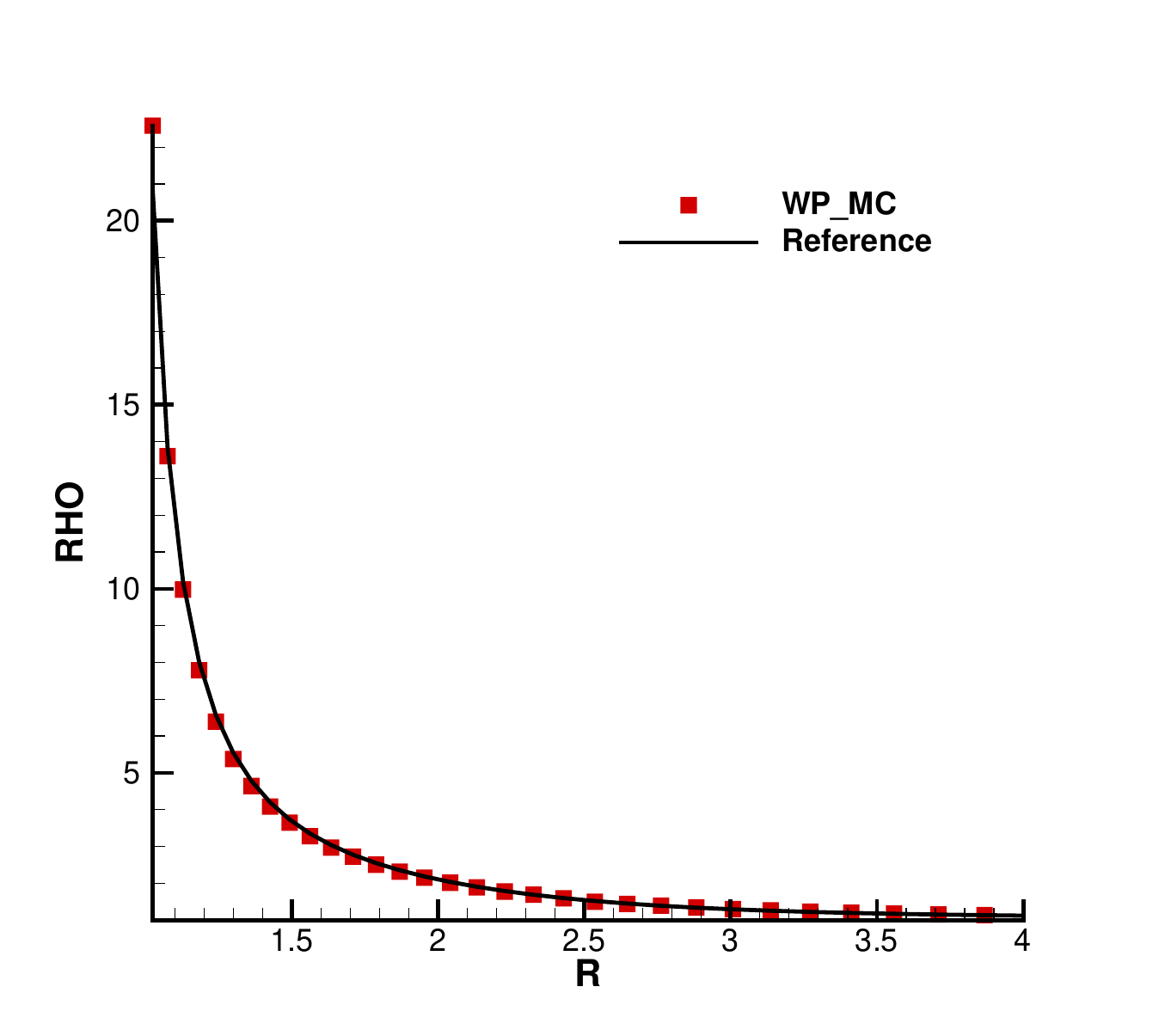}
        \caption{WPD-MC (local), $\rho$}
    \end{subfigure}
    \begin{subfigure}{0.49\textwidth}
        \centering
        \includegraphics[width=\linewidth]{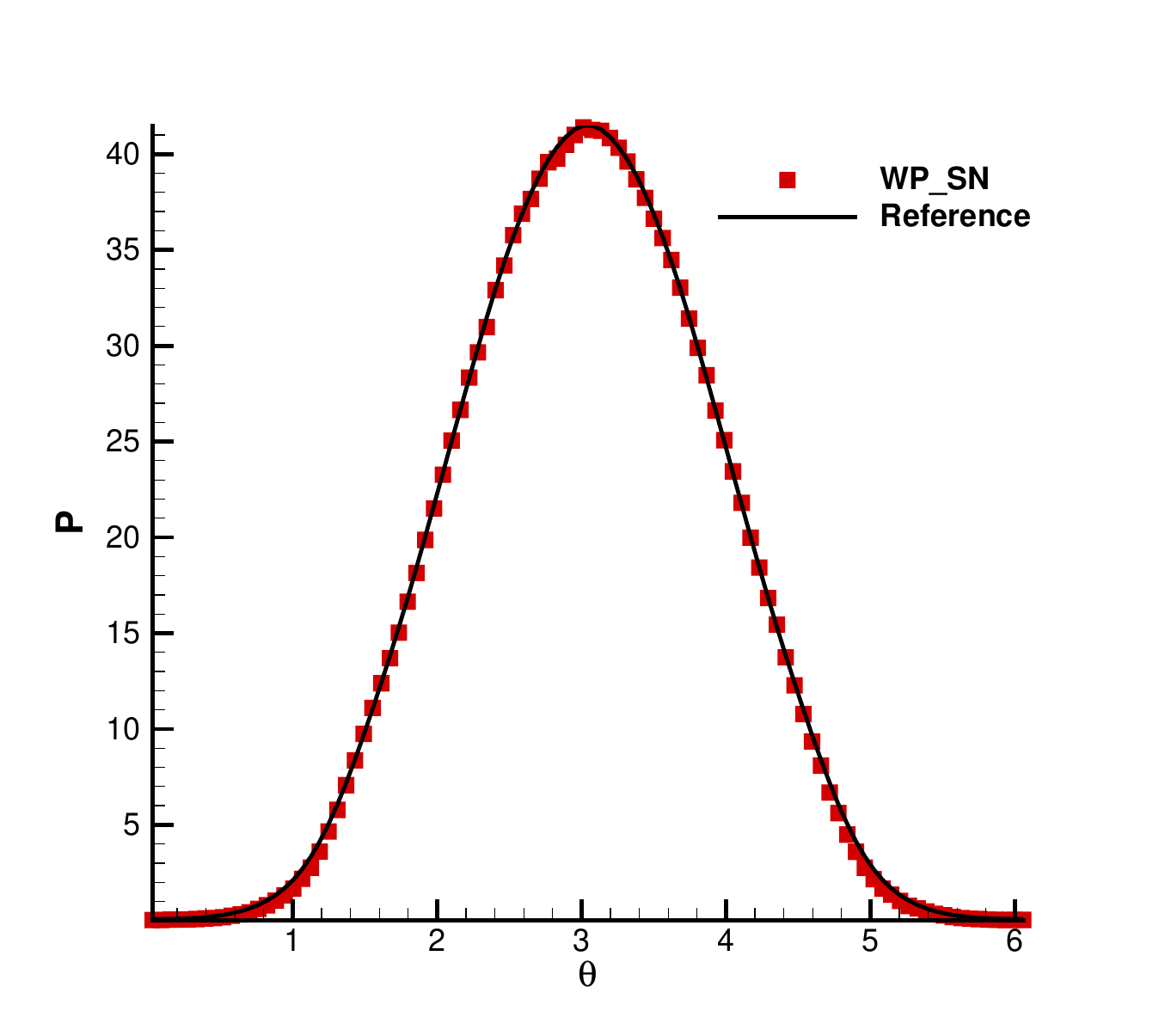}
        \caption{WPD-MC (local), $p$}
    \end{subfigure}
    \caption{Cylinder stagnation-line density and pressure at $\mathrm{Kn}=1$.}
    \label{fig:cylinder_mc_R_den_p_kn1}
\end{figure}

\begin{figure}[!htbp]
    \centering
    \begin{subfigure}{0.49\textwidth}
        \centering
        \includegraphics[width=\linewidth]{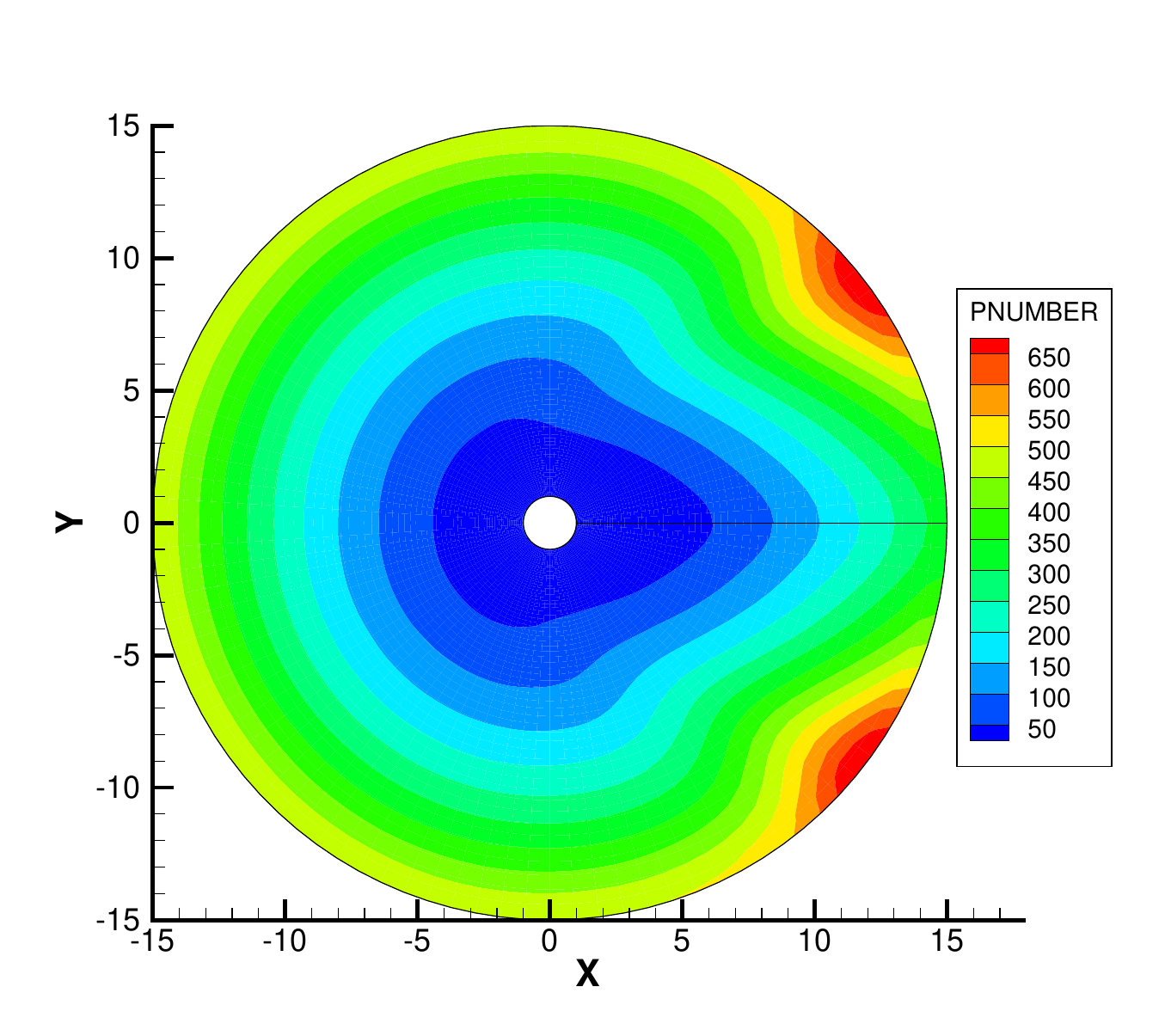}
        \caption{UGKWP, $\mathrm{Kn}=1$}
    \end{subfigure}
    \begin{subfigure}{0.49\textwidth}
        \centering
        \includegraphics[width=\linewidth]{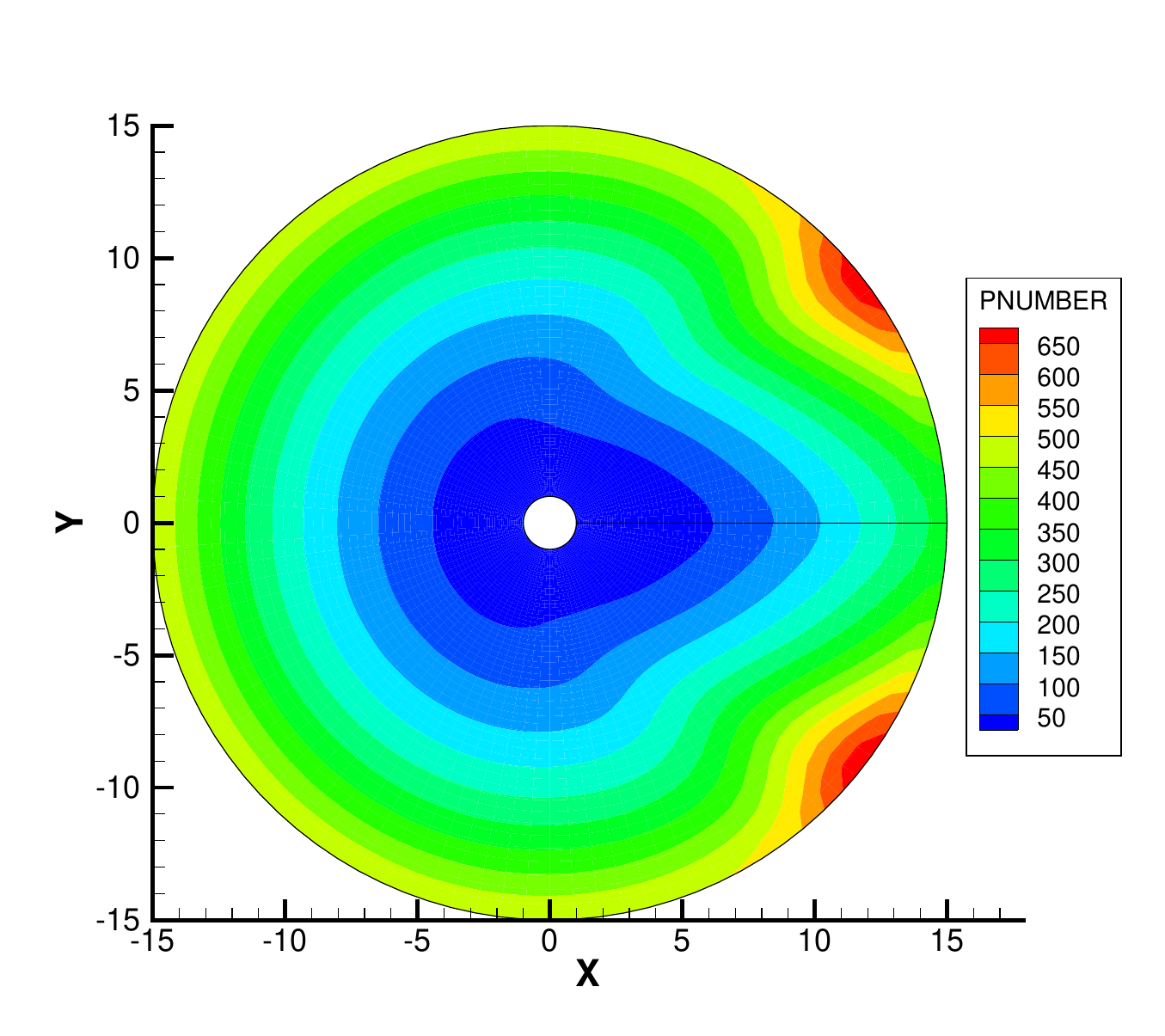}
        \caption{WPD-MC (local), $\mathrm{Kn}=1$}
    \end{subfigure}
    \caption{Equivalent particle-number fields at $\mathrm{Kn}=1$.}
    \label{fig:cylinder_pnumber_a_kn1}
\end{figure}

\subsubsection{Cylinder flow at \texorpdfstring{\(\mathrm{Kn}=10^{-2}\)}{Kn=1e-2}}

At \(\mathrm{Kn}=10^{-2}\), the flow is closer to the continuum regime but still contains a non-equilibrium shock layer and wall interaction. Figures~\ref{fig:cylinder_sn_kn1em2_pt}--\ref{fig:cylinder_mc_R_den_p_kn1em2} show that WPD-\(S_N\) and WPD-MC remain consistent in the shock stand-off distance, pressure rise, and surface pressure distribution. In particular, the deterministic wall-pressure and stagnation-line distributions are compared with the reference UGKS solution. Compared with \(\mathrm{Kn}=1\), the kinetic region is more localized, and the equivalent particle-number field in Fig.~\ref{fig:cylinder_pnumber_a_kn_1m2} correspondingly decreases over most of the domain.

\begin{figure}[!htbp]
    \centering
    \begin{subfigure}{0.49\textwidth}
        \centering
        \includegraphics[width=\linewidth]{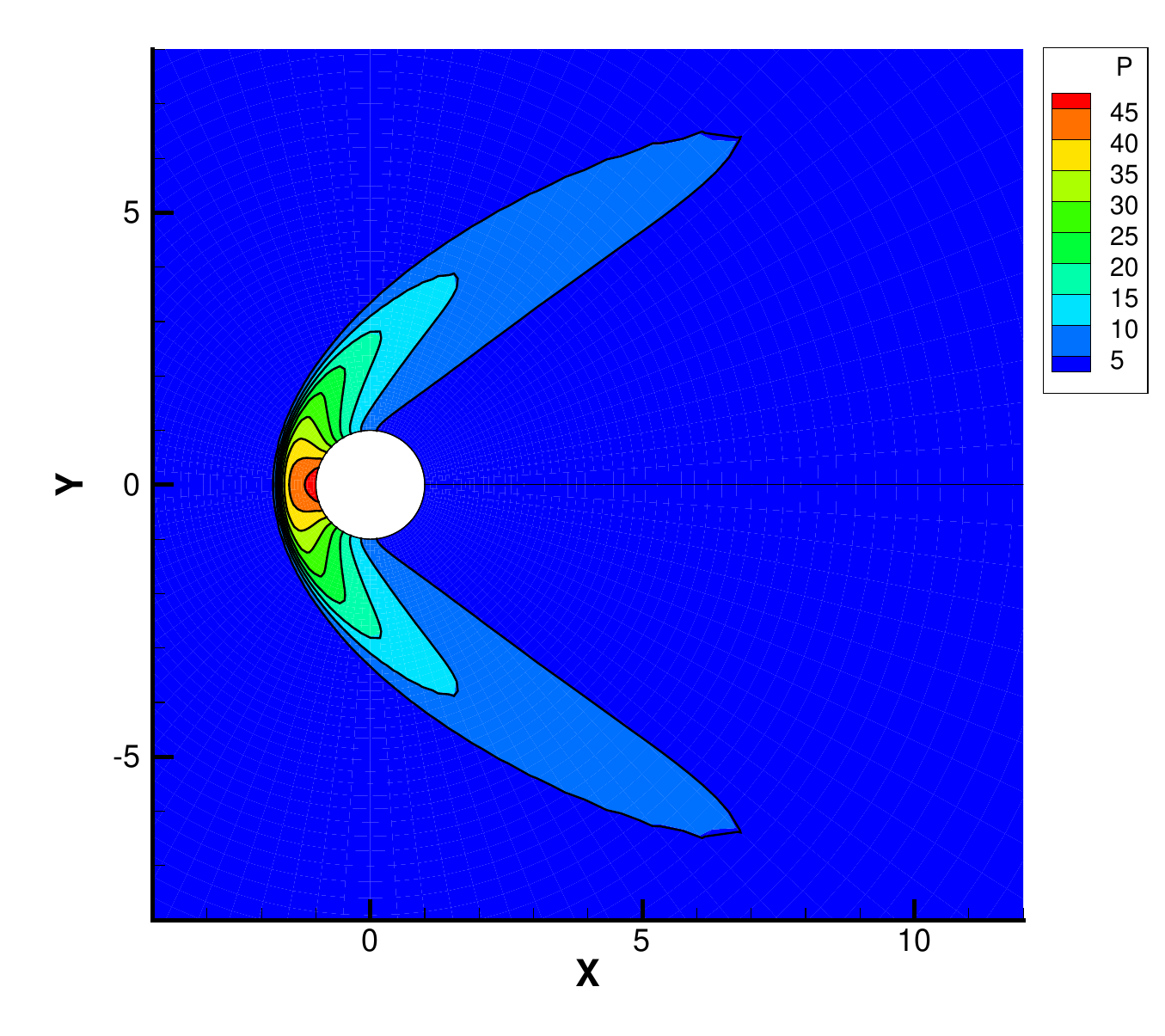}
        \caption{WPD-\(S_N\), $p$}
    \end{subfigure}
    \begin{subfigure}{0.49\textwidth}
        \centering
        \includegraphics[width=\linewidth]{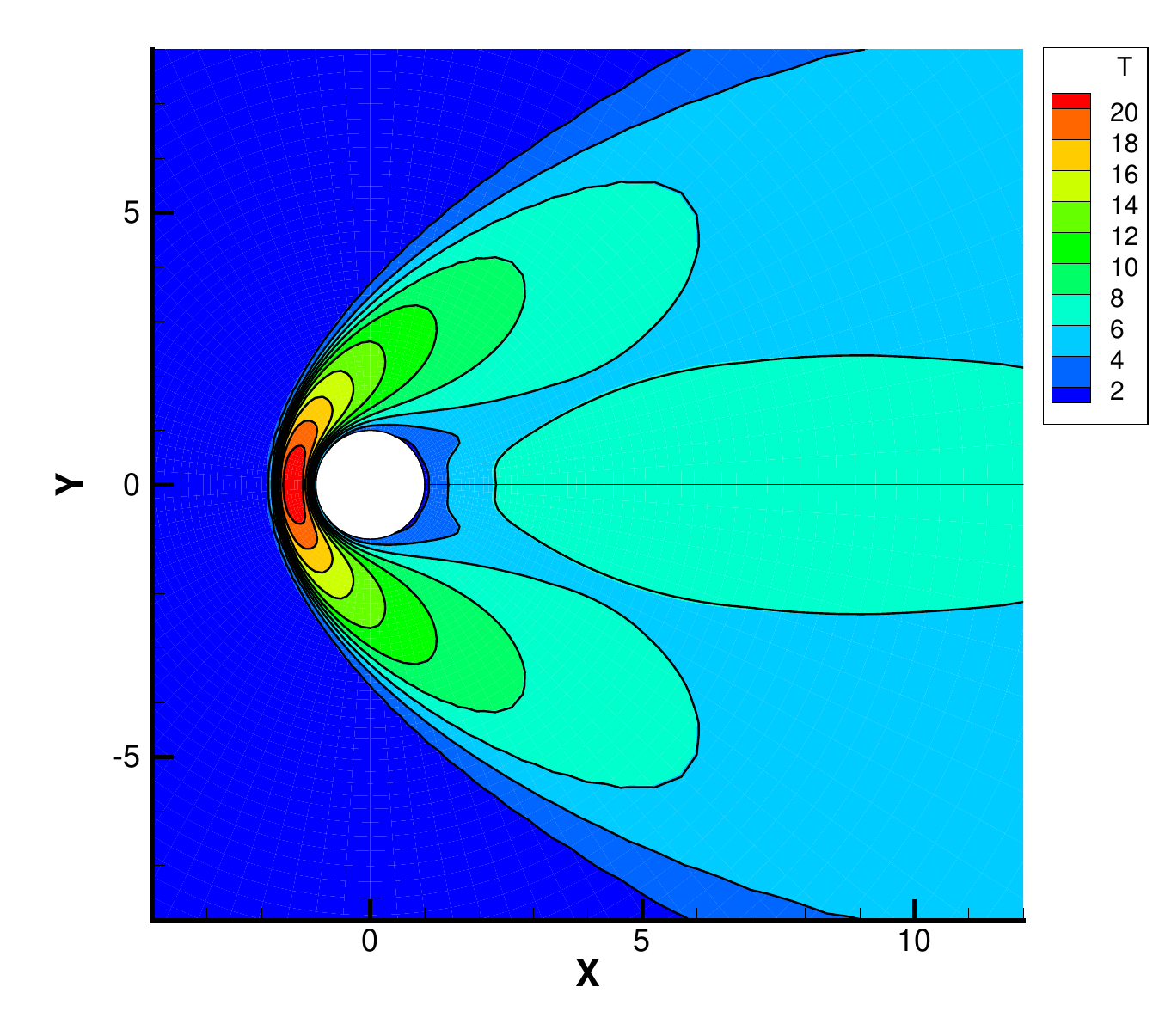}
        \caption{WPD-\(S_N\), $T$}
    \end{subfigure}
    \caption{Cylinder pressure and temperature contours at $\mathrm{Kn}=10^{-2}$.}
    \label{fig:cylinder_sn_kn1em2_pt}
\end{figure}
\FloatBarrier

\begin{figure}[!htbp]
    \centering
    \begin{subfigure}{0.49\textwidth}
        \centering
        \includegraphics[width=\linewidth]{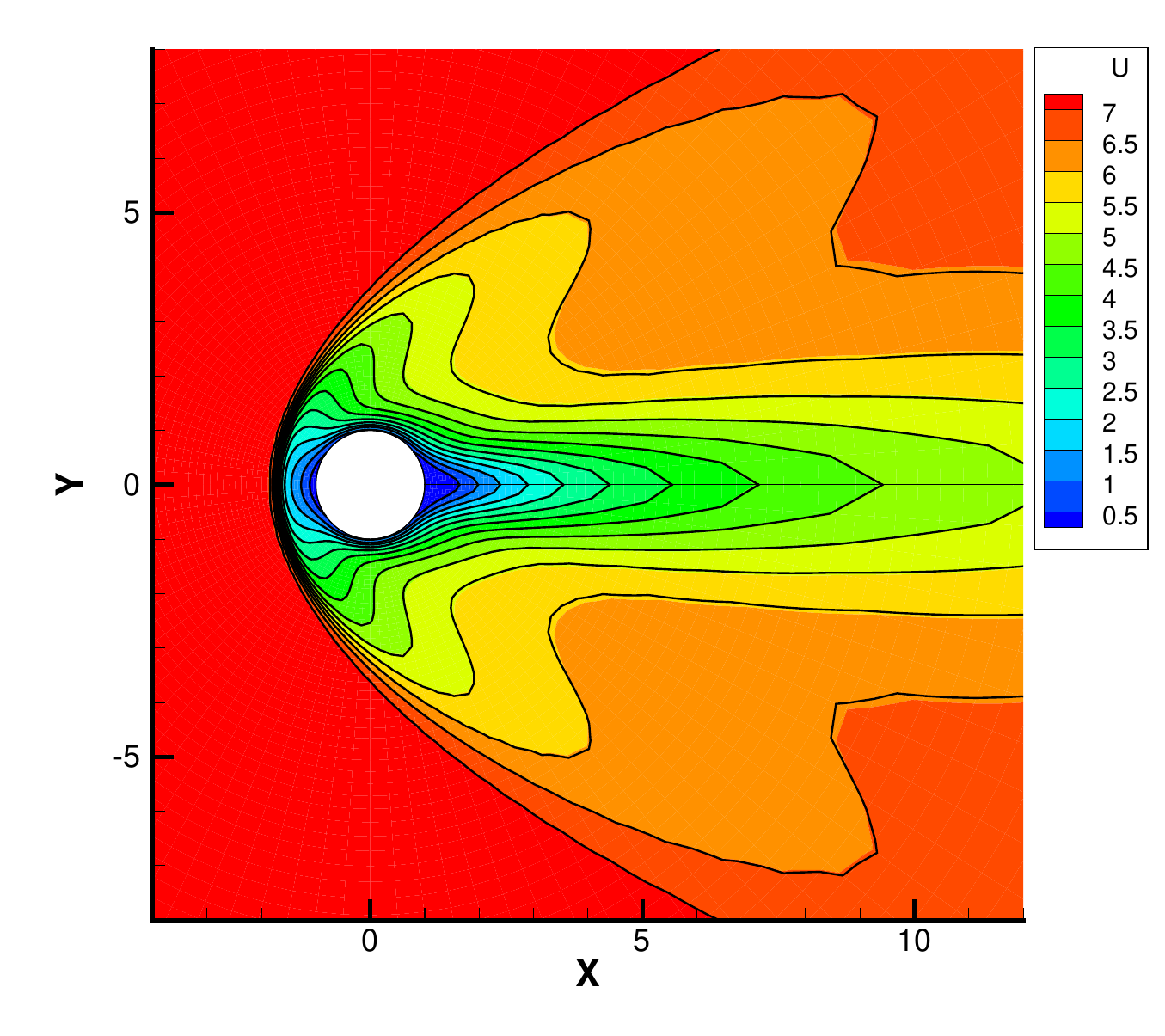}
        \caption{WPD-\(S_N\), $u$}
    \end{subfigure}
    \begin{subfigure}{0.49\textwidth}
        \centering
        \includegraphics[width=\linewidth]{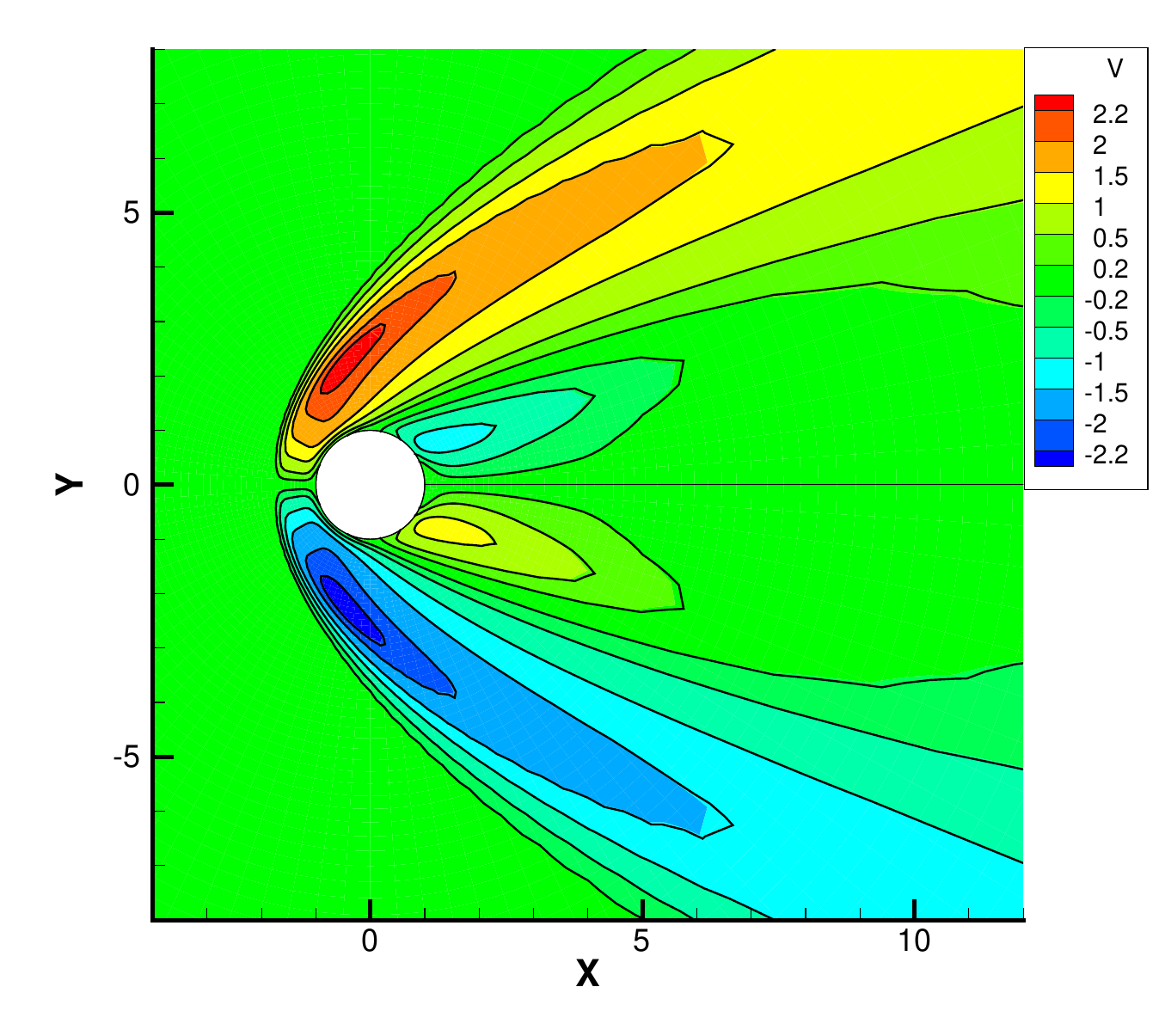}
        \caption{WPD-\(S_N\), $v$}
    \end{subfigure}
    \caption{Cylinder velocity contours at $\mathrm{Kn}=10^{-2}$.}
    \label{fig:cylinder_sn_kn1em2_uv}
\end{figure}

\begin{figure}[!htbp]
    \centering
    \begin{subfigure}{0.49\textwidth}
        \centering
        \includegraphics[width=\linewidth]{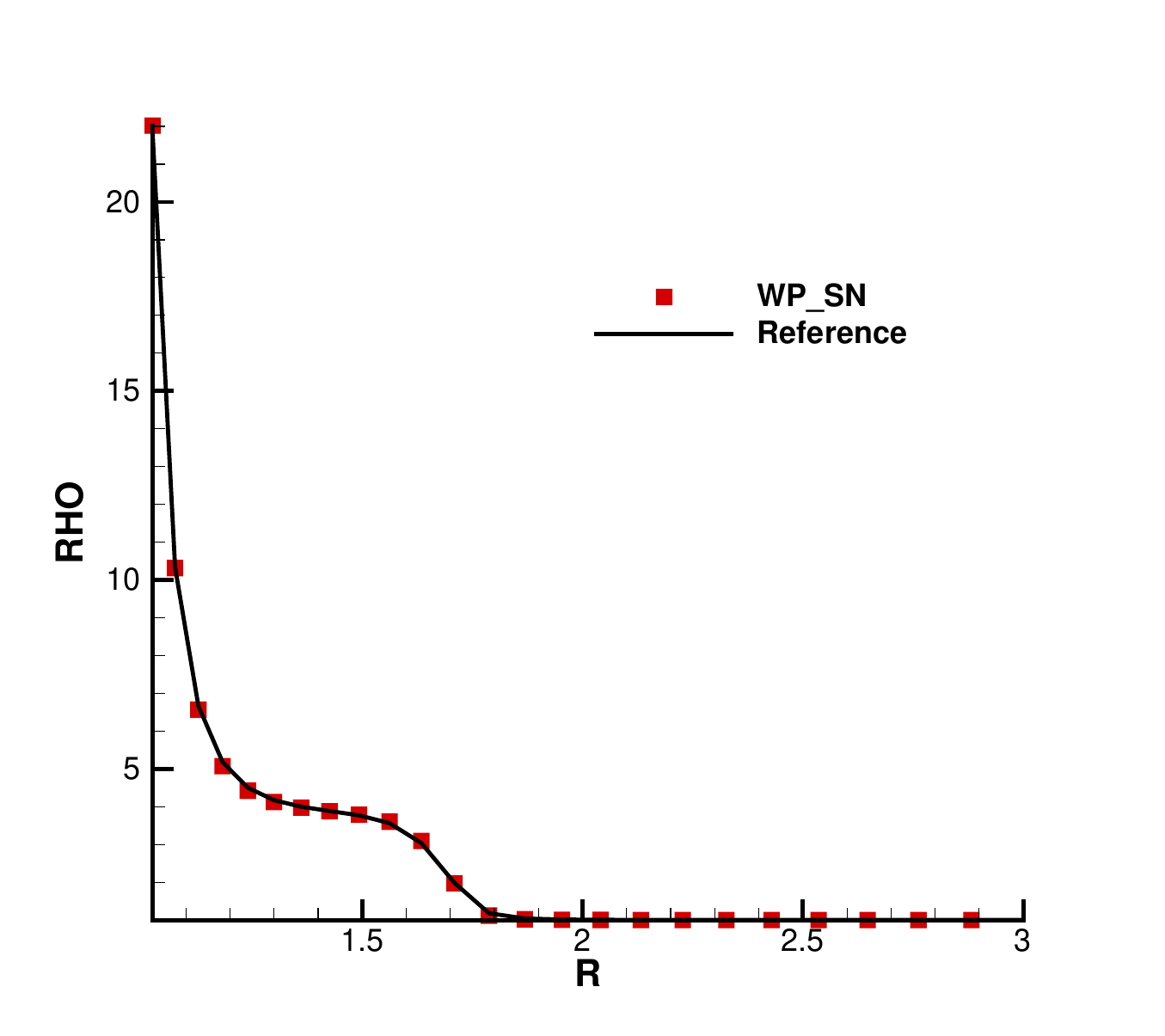}
        \caption{WPD-\(S_N\), $\rho$}
    \end{subfigure}
    \begin{subfigure}{0.49\textwidth}
        \centering
        \includegraphics[width=\linewidth]{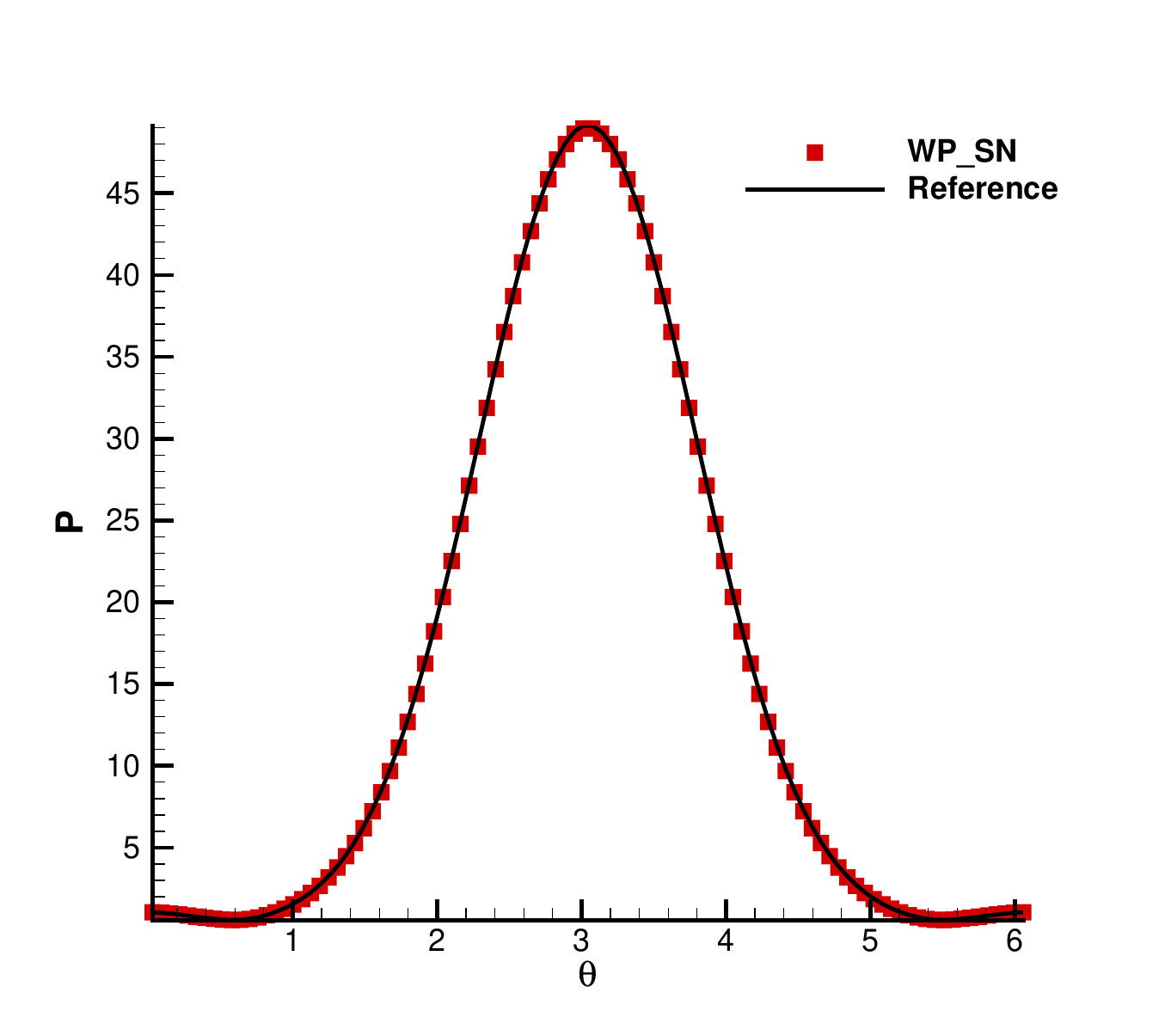}
        \caption{WPD-\(S_N\), $p$}
    \end{subfigure}
    \caption{Cylinder stagnation-line density and wall pressure at $\mathrm{Kn}=10^{-2}$, comparing WPD-\(S_N\) with the reference UGKS solution.}
    \label{fig:cylinder_sn_R_den_p_kn1em2}
\end{figure}

\begin{figure}[!htbp]
    \centering
    \begin{subfigure}{0.49\textwidth}
        \centering
        \includegraphics[width=\linewidth]{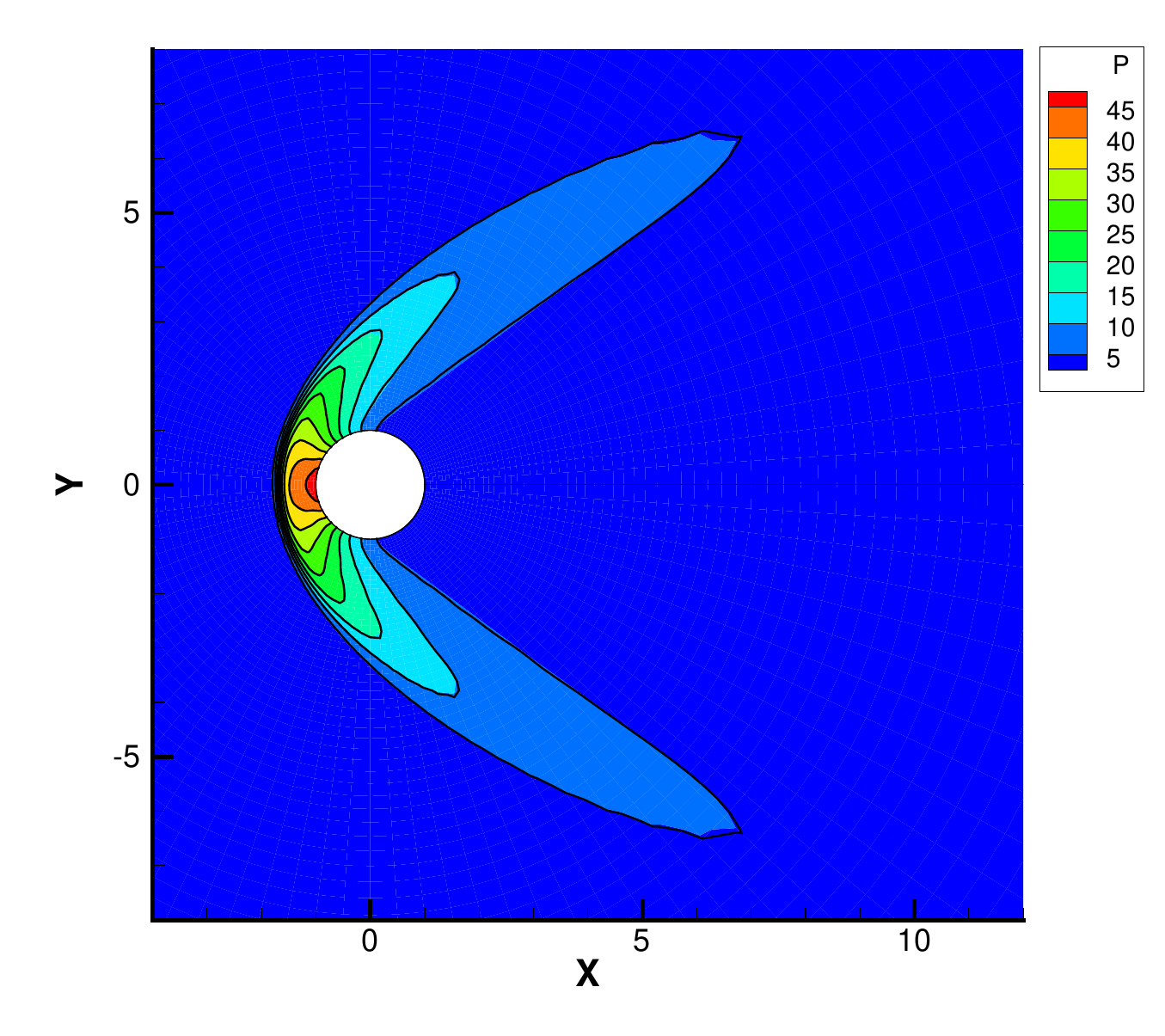}
        \caption{WPD-MC, $p$}
    \end{subfigure}
    \begin{subfigure}{0.49\textwidth}
        \centering
        \includegraphics[width=\linewidth]{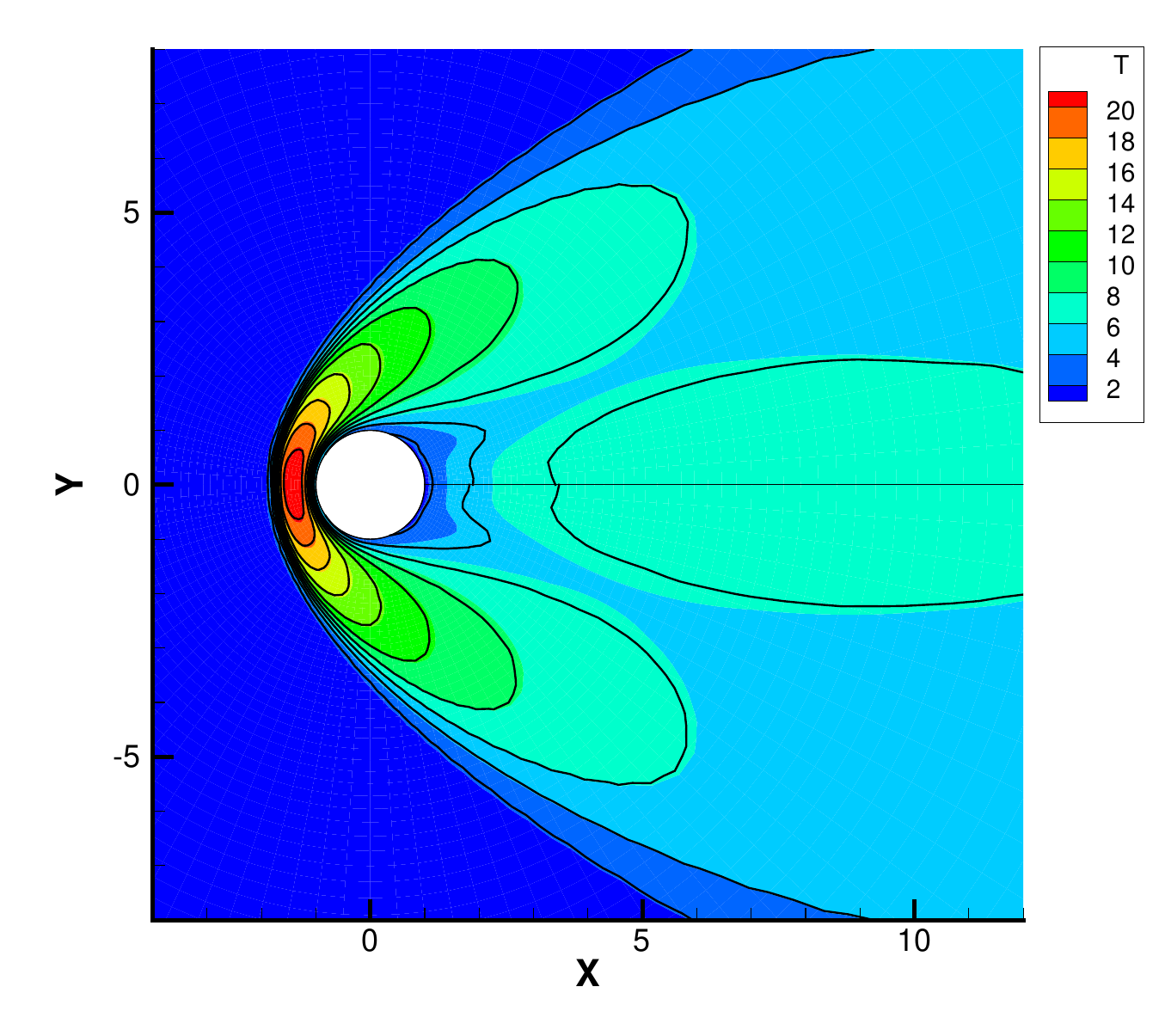}
        \caption{WPD-MC, $T$}
    \end{subfigure}
    \caption{Cylinder pressure and temperature contours at $\mathrm{Kn}=10^{-2}$.}
    \label{fig:cylinder_mc_kn1em2_pt}
\end{figure}

\begin{figure}[!htbp]
    \centering
    \begin{subfigure}{0.49\textwidth}
        \centering
        \includegraphics[width=\linewidth]{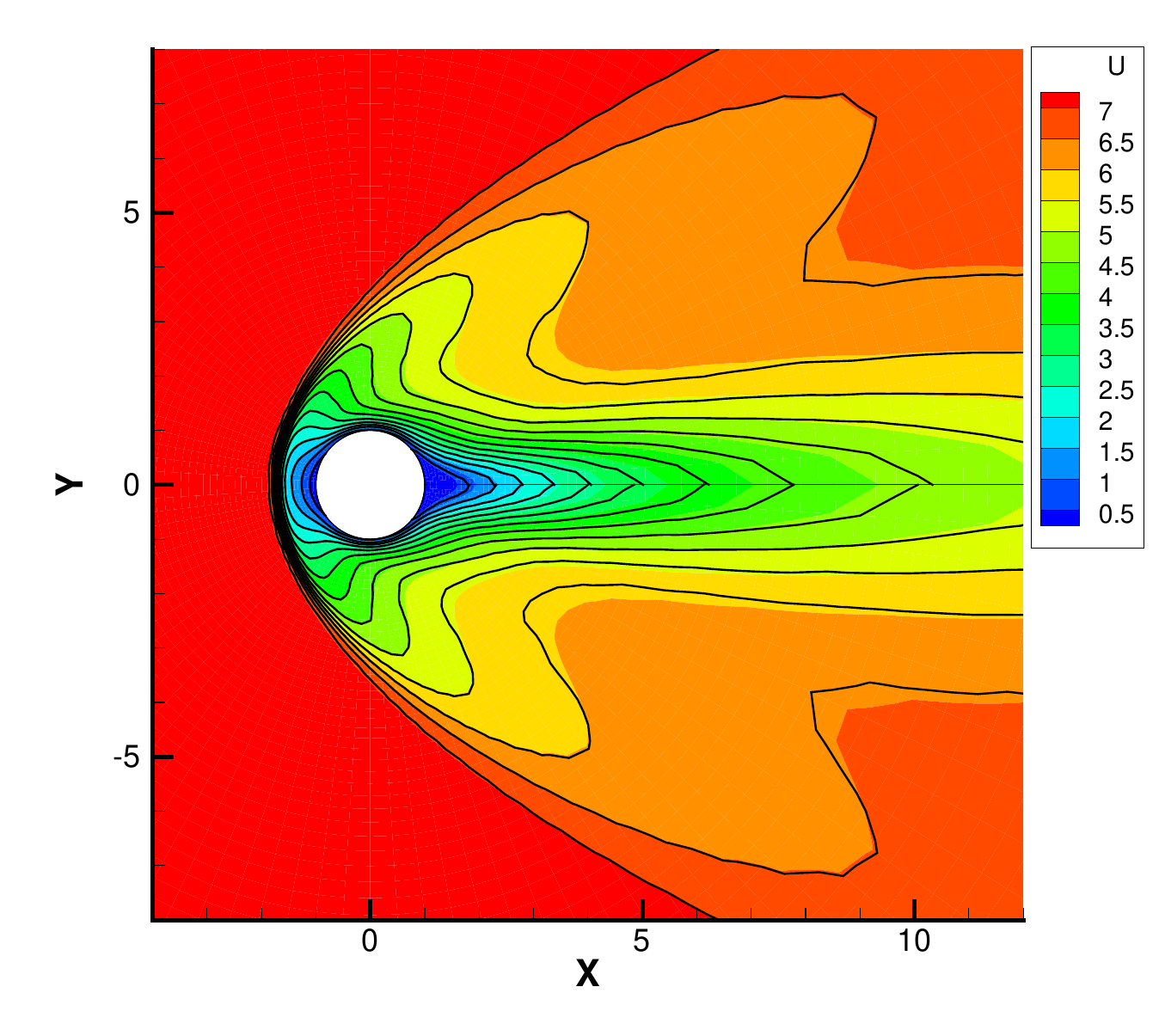}
        \caption{WPD-MC, $u$}
    \end{subfigure}
    \begin{subfigure}{0.49\textwidth}
        \centering
        \includegraphics[width=\linewidth]{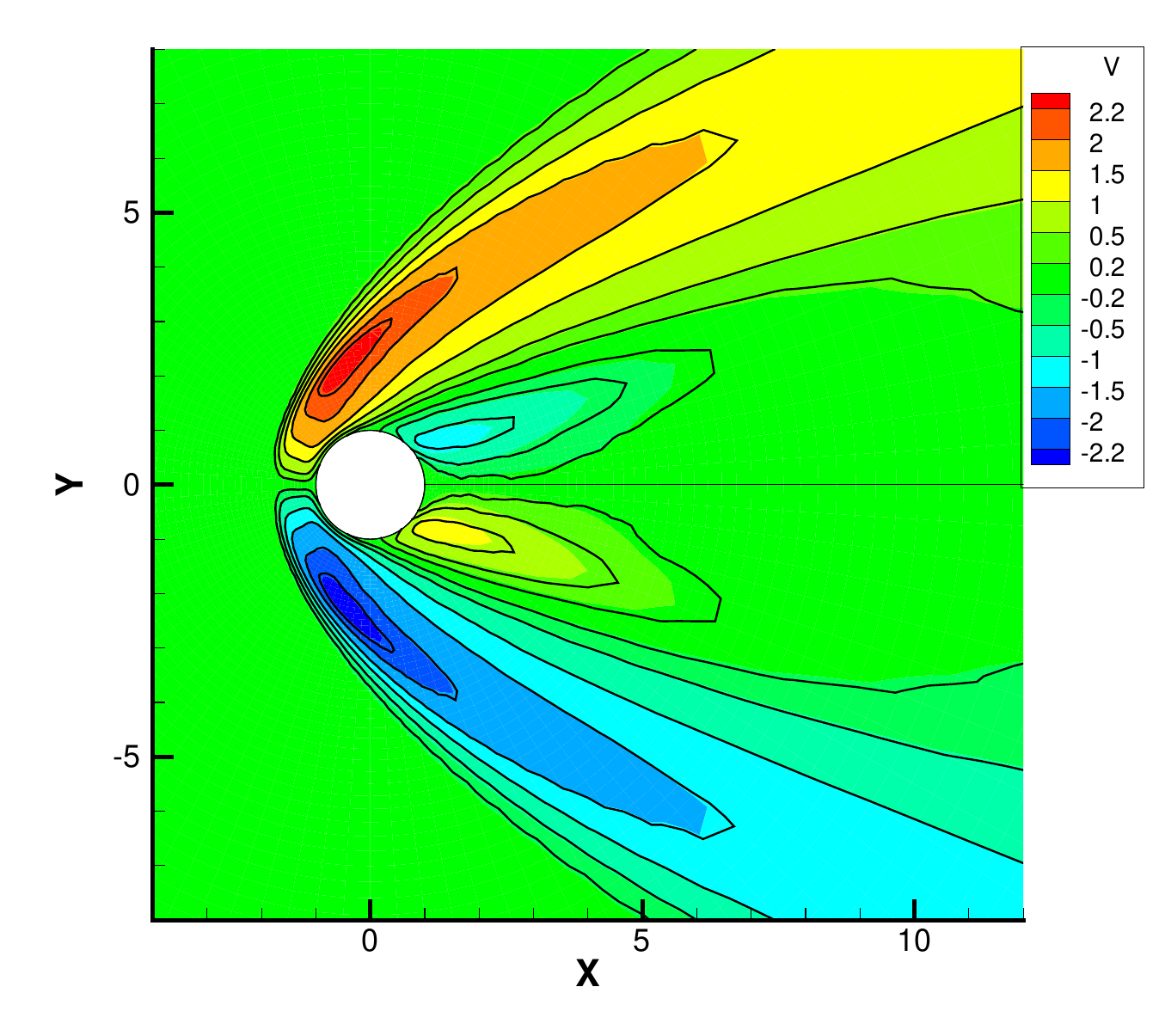}
        \caption{WPD-MC, $v$}
    \end{subfigure}
    \caption{Cylinder velocity contours at $\mathrm{Kn}=10^{-2}$.}
    \label{fig:cylinder_mc_kn1em2_uv}
\end{figure}
\FloatBarrier

\begin{figure}[!htbp]
    \centering
    \begin{subfigure}{0.49\textwidth}
        \centering
        \includegraphics[width=\linewidth]{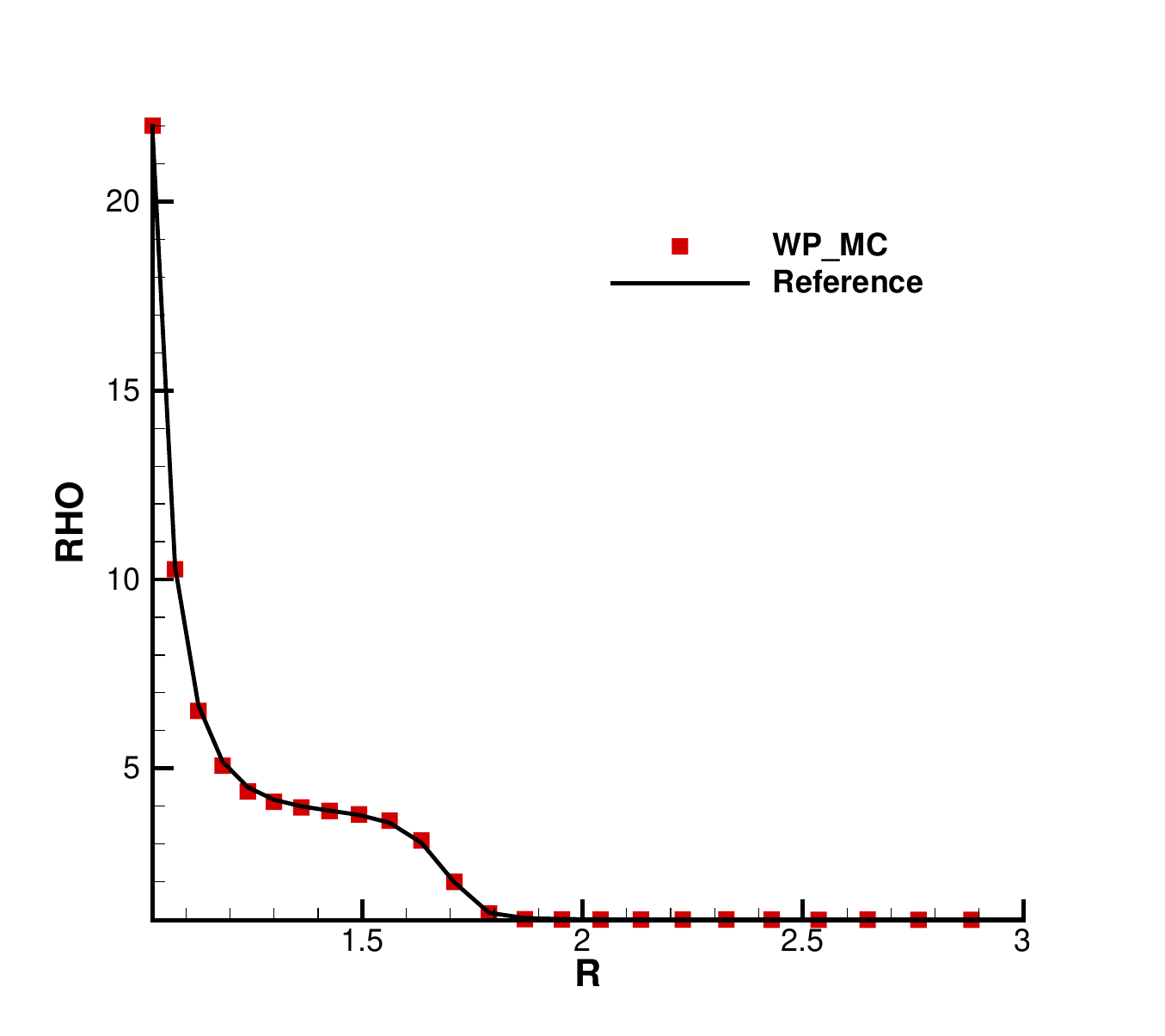}
        \caption{WPD-MC (local), $\rho$}
    \end{subfigure}
    \begin{subfigure}{0.49\textwidth}
        \centering
        \includegraphics[width=\linewidth]{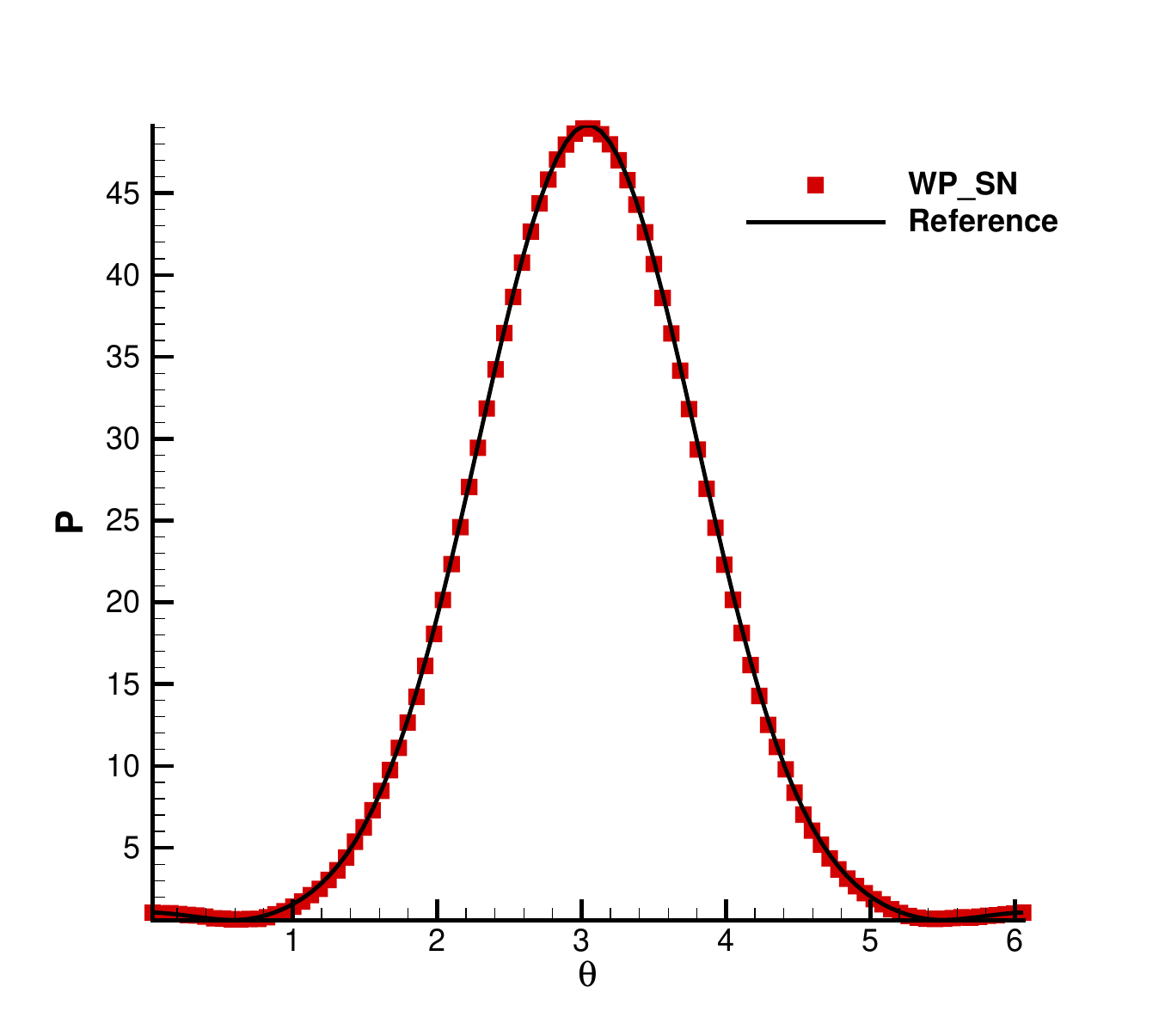}
        \caption{WPD-MC (local), $p$}
    \end{subfigure}
    \caption{Cylinder stagnation-line density and pressure at $\mathrm{Kn}=10^{-2}$.}
    \label{fig:cylinder_mc_R_den_p_kn1em2}
\end{figure}

\begin{figure}[!htbp]
    \centering
    \begin{subfigure}{0.49\textwidth}
        \centering
        \includegraphics[width=\linewidth]{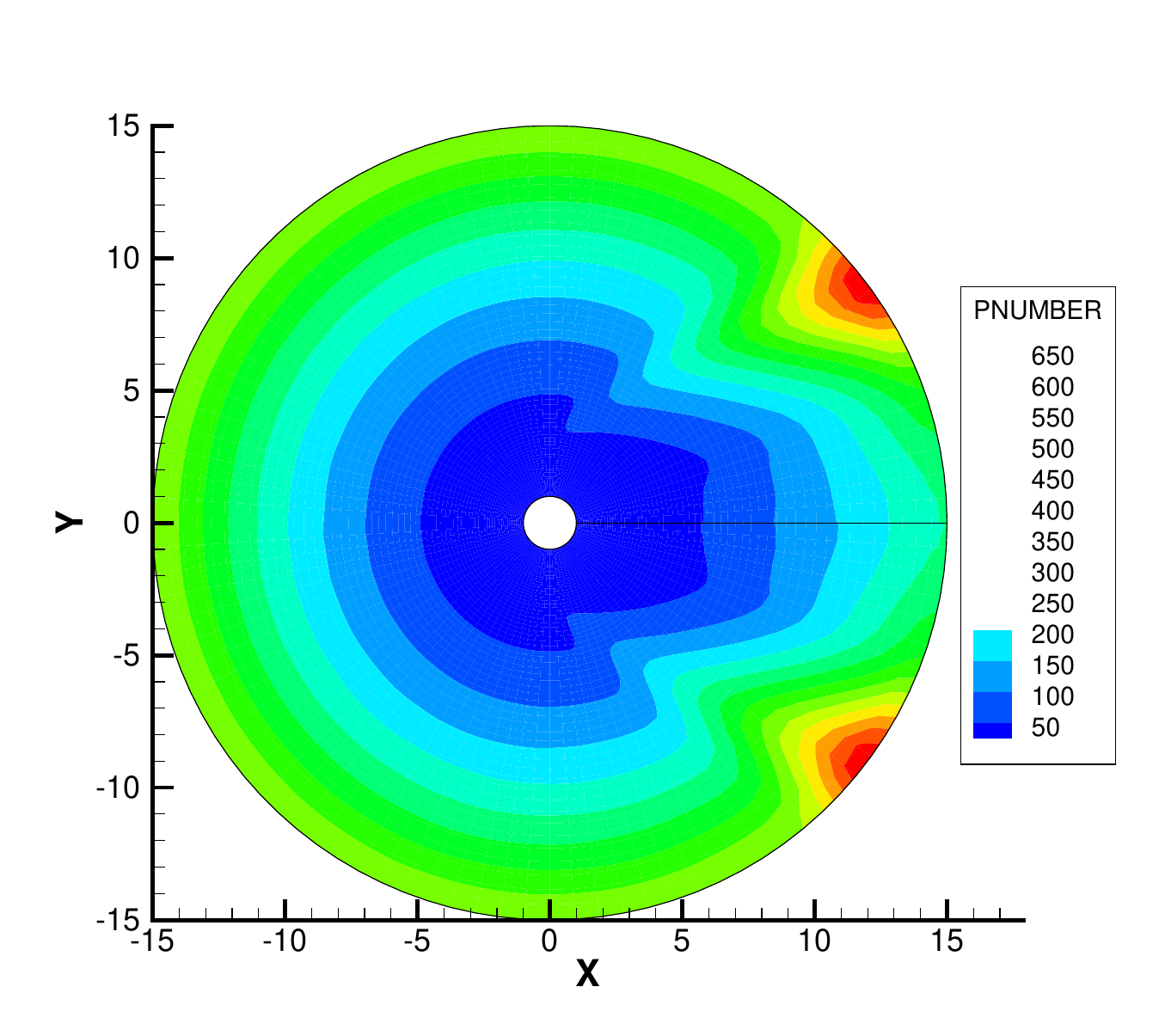}
        \caption{UGKWP, $\mathrm{Kn}=10^{-2}$}
    \end{subfigure}
    \begin{subfigure}{0.49\textwidth}
        \centering
        \includegraphics[width=\linewidth]{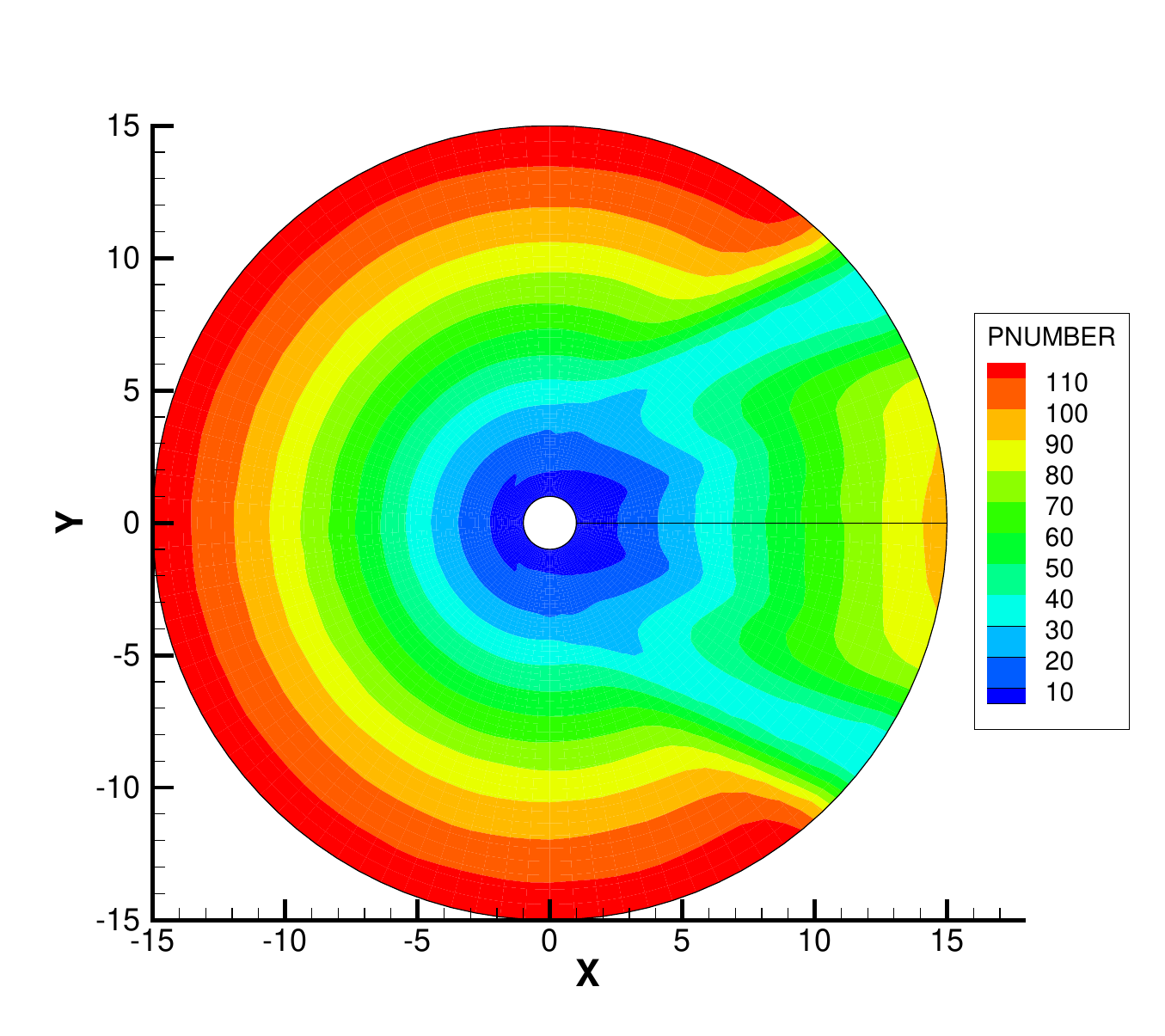}
        \caption{WPD-MC (local), $\mathrm{Kn}=10^{-2}$}
    \end{subfigure}
    \caption{Equivalent particle-number fields at $\mathrm{Kn}=10^{-2}$.}
    \label{fig:cylinder_pnumber_a_kn_1m2}
\end{figure}

\subsubsection{Cylinder flow at \texorpdfstring{\(\mathrm{Kn}=10^{-3}\)}{Kn=1e-3}}

The \(\mathrm{Kn}=10^{-3}\) case further approaches the continuum limit. Figures~\ref{fig:cylinder_sn_kn1em3_pt}--\ref{fig:cylinder_mc_R_den_p_kn1em3} show that the wave component dominates most of the domain, while the particle component is retained where strong gradients and non-equilibrium effects remain. The pressure, temperature, and velocity contours are smooth and consistent between WPD-\(S_N\) and WPD-MC, and the deterministic line and wall-pressure profiles remain close to the reference UGKS solution. The particle-number comparison in Fig.~\ref{fig:cylinder_pnumber_b_kn1m3} highlights the adaptive reduction of kinetic degrees of freedom as the flow becomes more continuum-like.

\begin{figure}[!htbp]
    \centering
    \begin{subfigure}{0.49\textwidth}
        \centering
        \includegraphics[width=\linewidth]{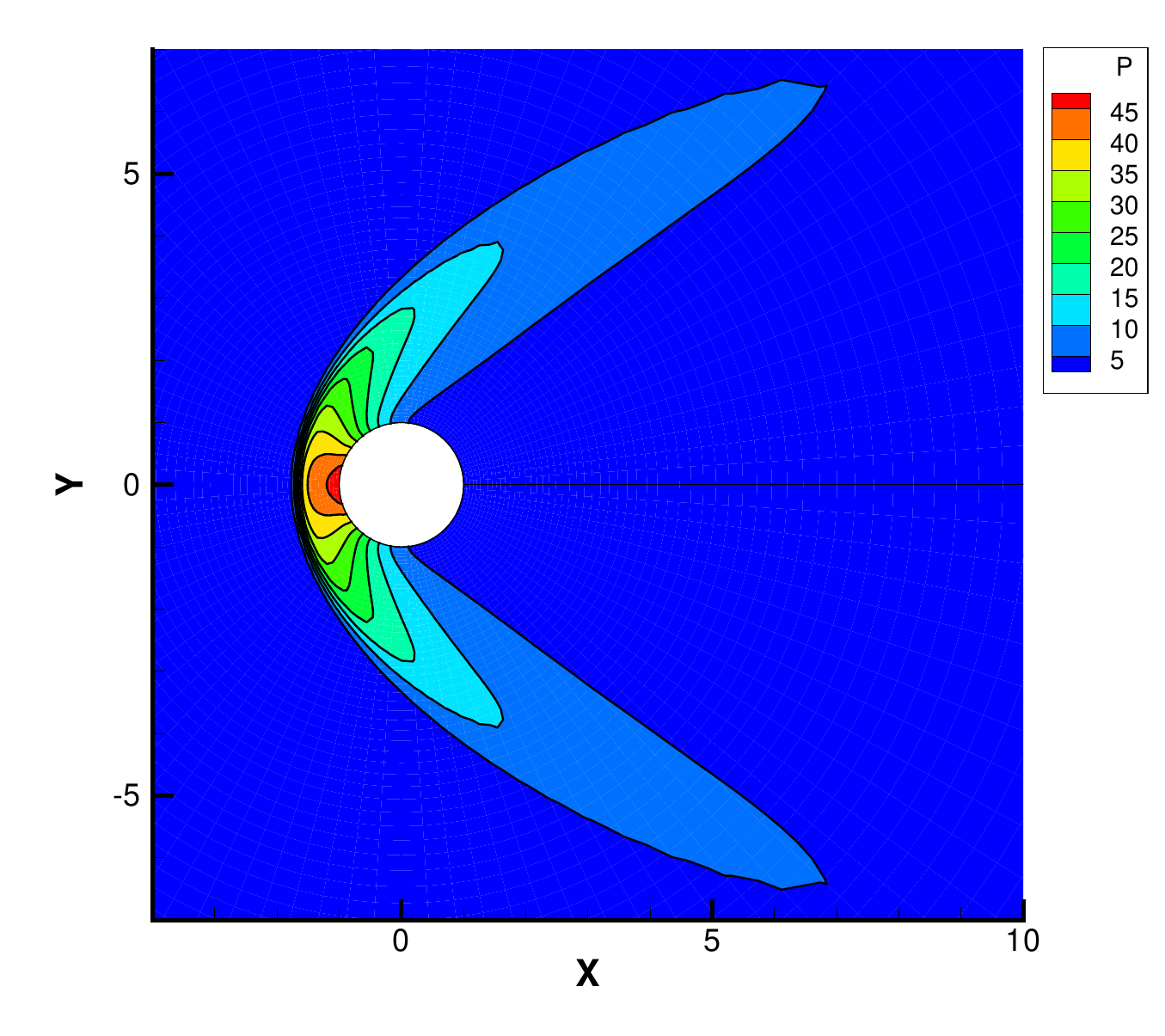}
        \caption{WPD-\(S_N\), $p$}
    \end{subfigure}
    \begin{subfigure}{0.49\textwidth}
        \centering
        \includegraphics[width=\linewidth]{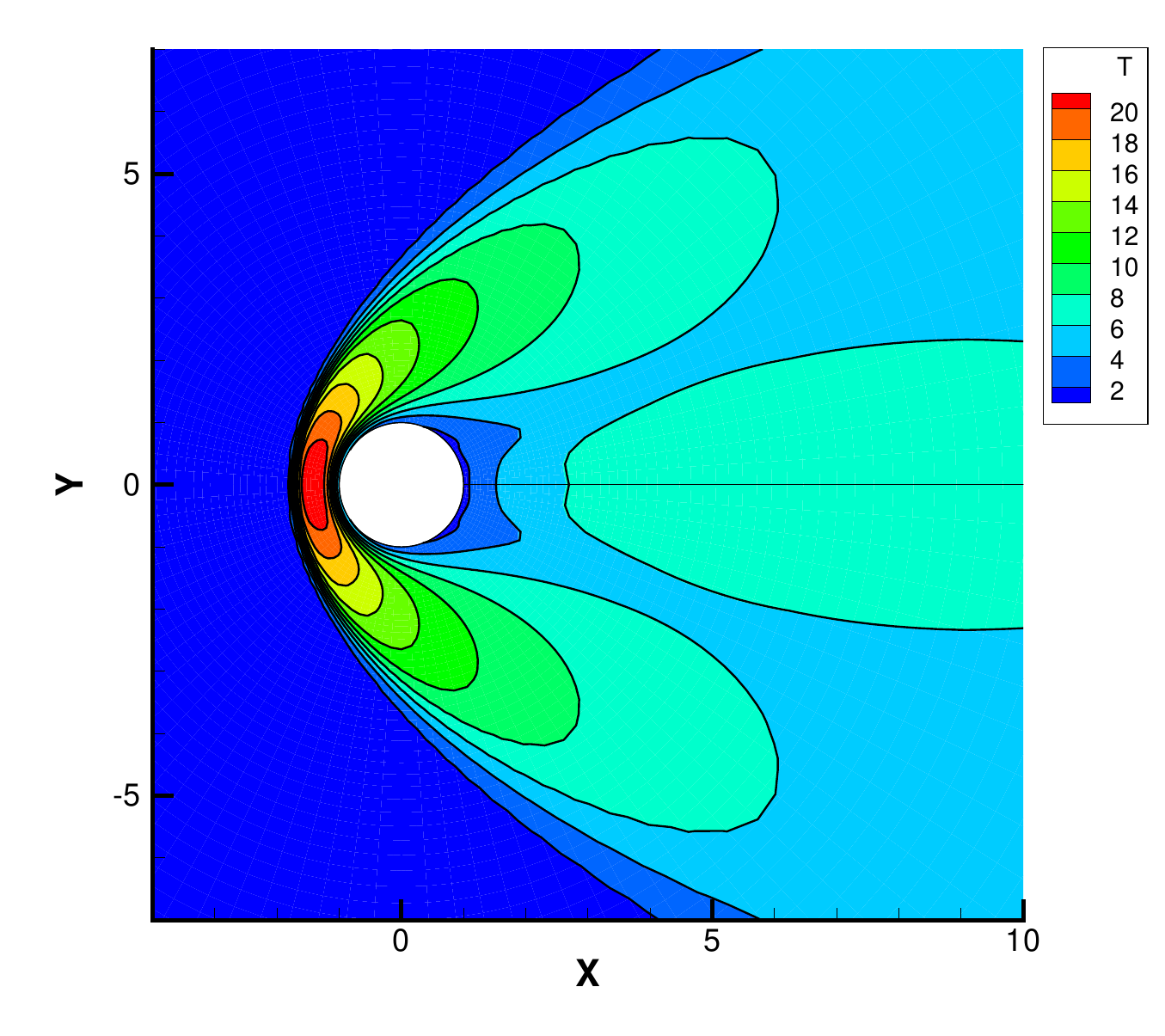}
        \caption{WPD-\(S_N\), $T$}
    \end{subfigure}
    \caption{Cylinder pressure and temperature contours at $\mathrm{Kn}=10^{-3}$.}
    \label{fig:cylinder_sn_kn1em3_pt}
\end{figure}

\begin{figure}[!htbp]
    \centering
    \begin{subfigure}{0.49\textwidth}
        \centering
        \includegraphics[width=\linewidth]{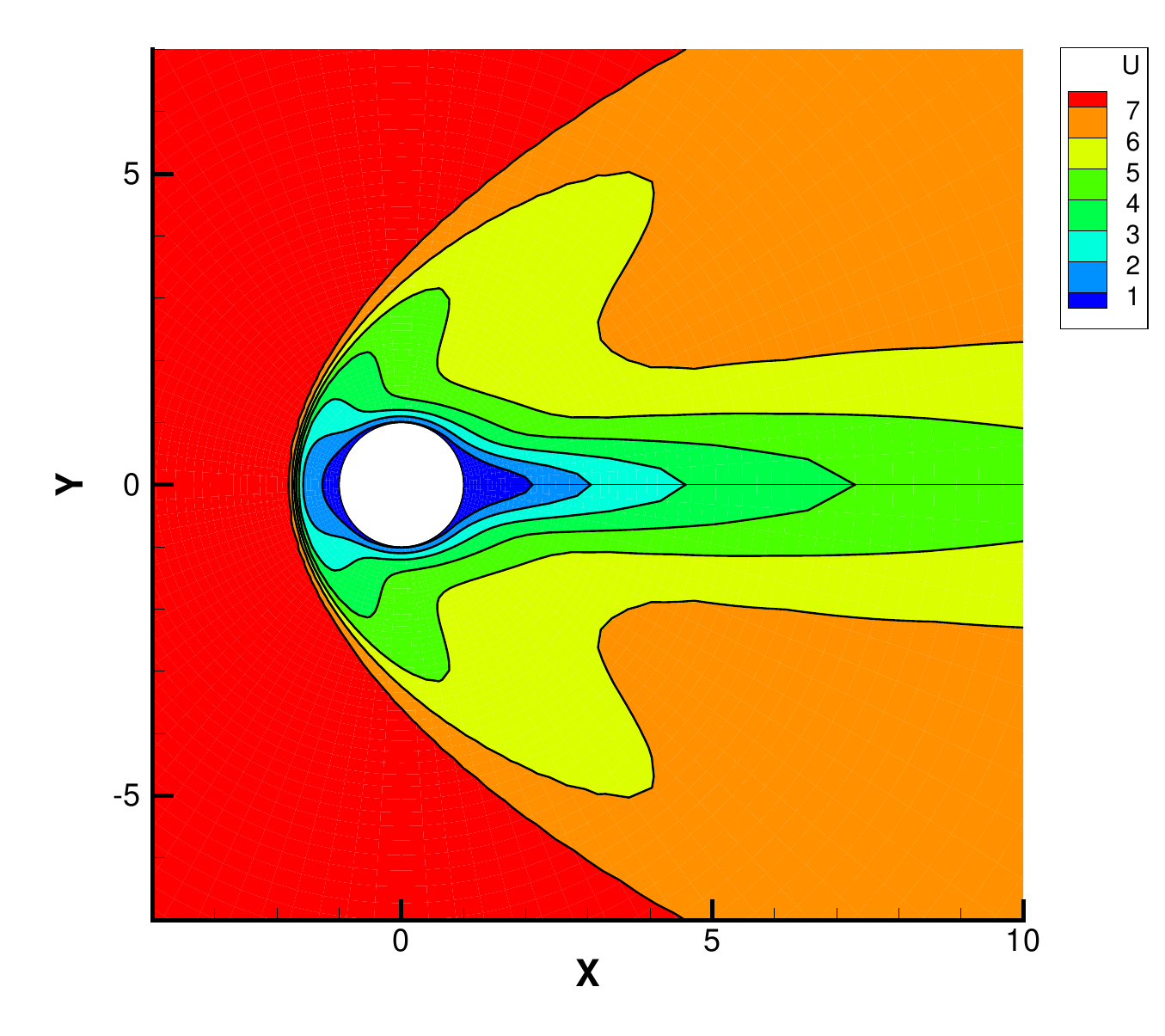}
        \caption{WPD-\(S_N\), $u$}
    \end{subfigure}
    \begin{subfigure}{0.49\textwidth}
        \centering
        \includegraphics[width=\linewidth]{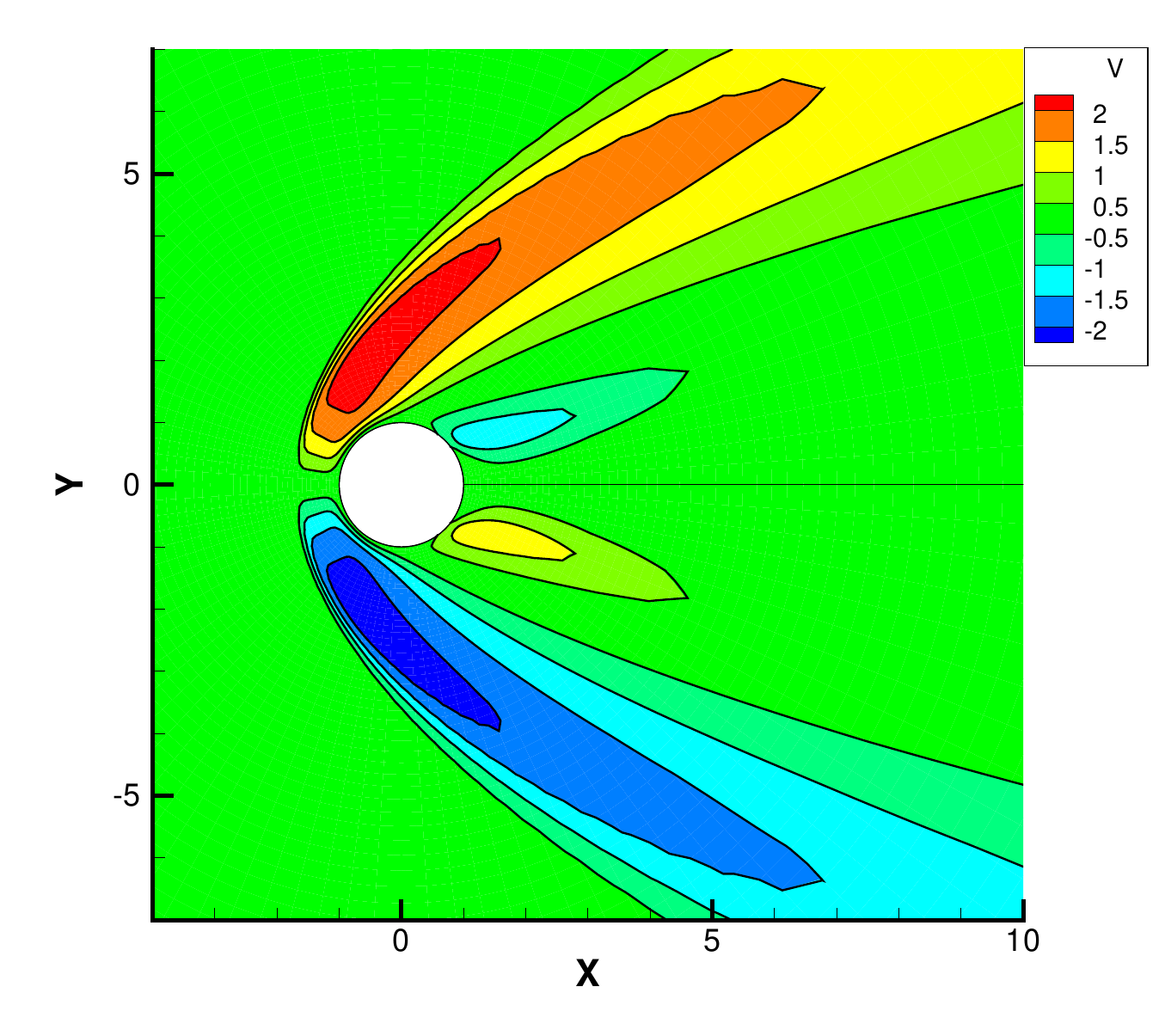}
        \caption{WPD-\(S_N\), $v$}
    \end{subfigure}
    \caption{Cylinder velocity contours at $\mathrm{Kn}=10^{-3}$.}
    \label{fig:cylinder_sn_kn1em3_uv}
\end{figure}
\FloatBarrier

\begin{figure}[!htbp]
    \centering
    \begin{subfigure}{0.49\textwidth}
        \centering
        \includegraphics[width=\linewidth]{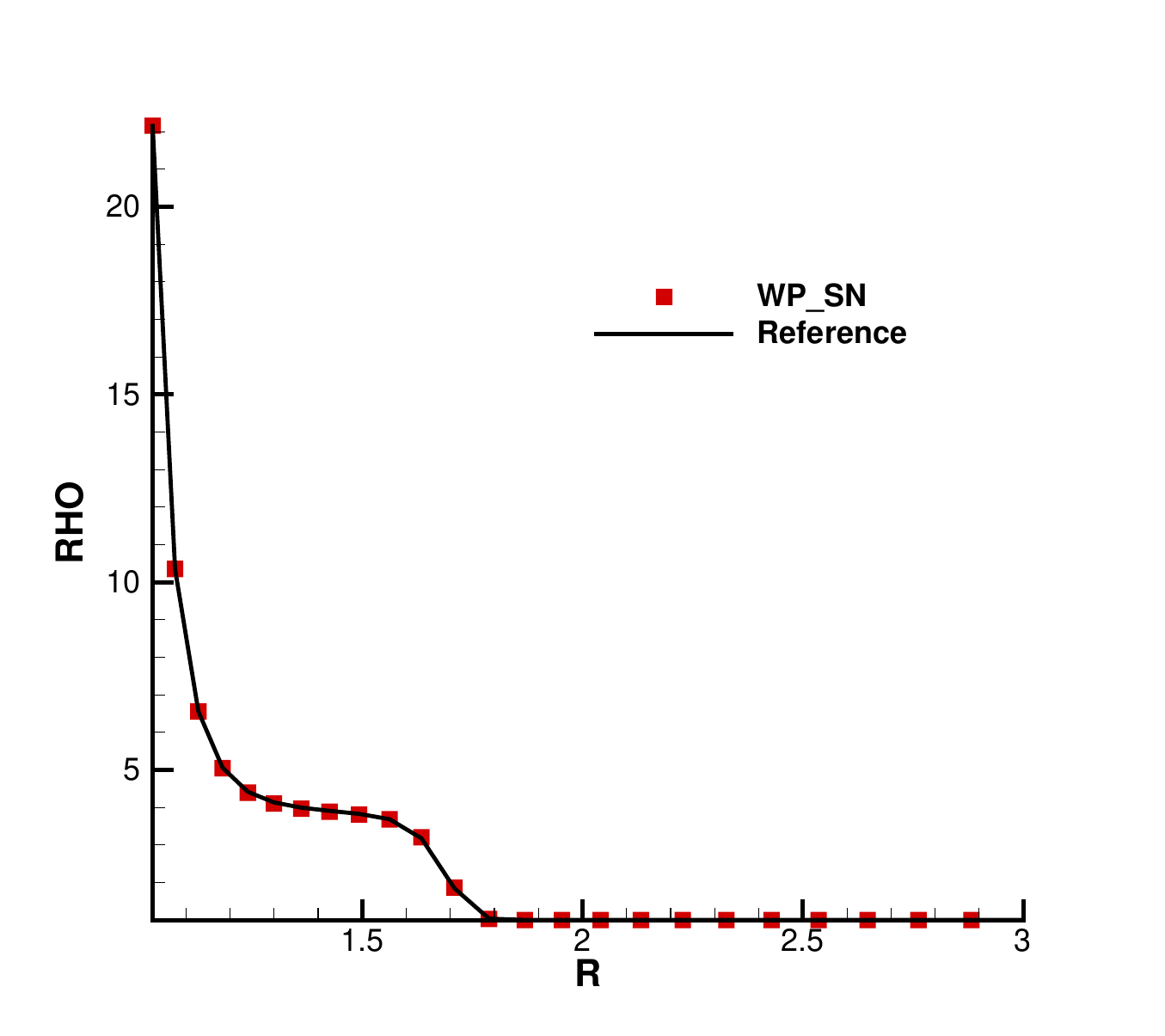}
        \caption{WPD-\(S_N\), $\rho$}
    \end{subfigure}
    \begin{subfigure}{0.49\textwidth}
        \centering
        \includegraphics[width=\linewidth]{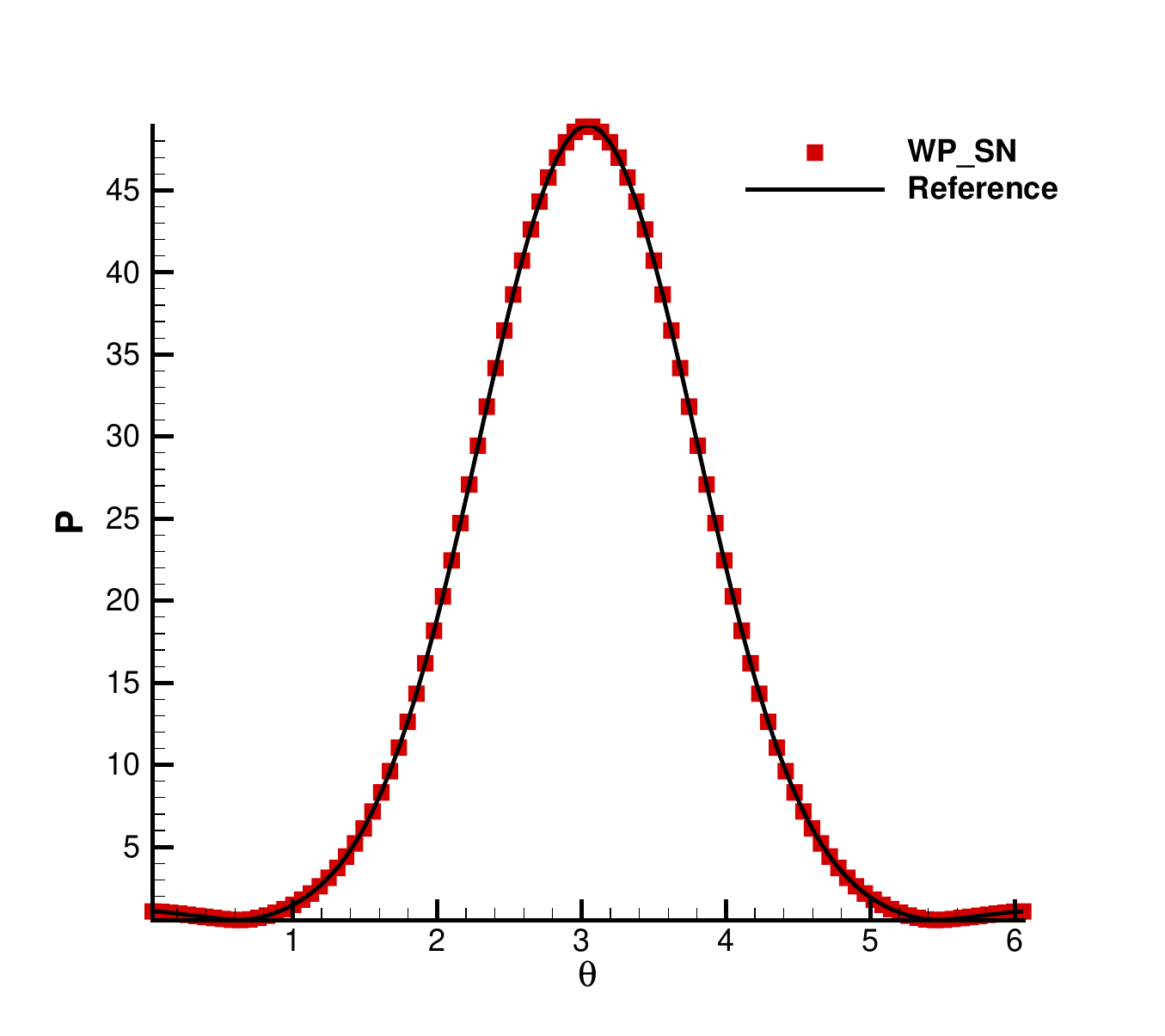}
        \caption{WPD-\(S_N\), $p$}
    \end{subfigure}
    \caption{Cylinder stagnation-line density and wall pressure at $\mathrm{Kn}=10^{-3}$, comparing WPD-\(S_N\) with the reference UGKS solution.}
    \label{fig:cylinder_sn_R_den_p_kn1em3}
\end{figure}

\begin{figure}[!htbp]
    \centering
    \begin{subfigure}{0.49\textwidth}
        \centering
        \includegraphics[width=\linewidth]{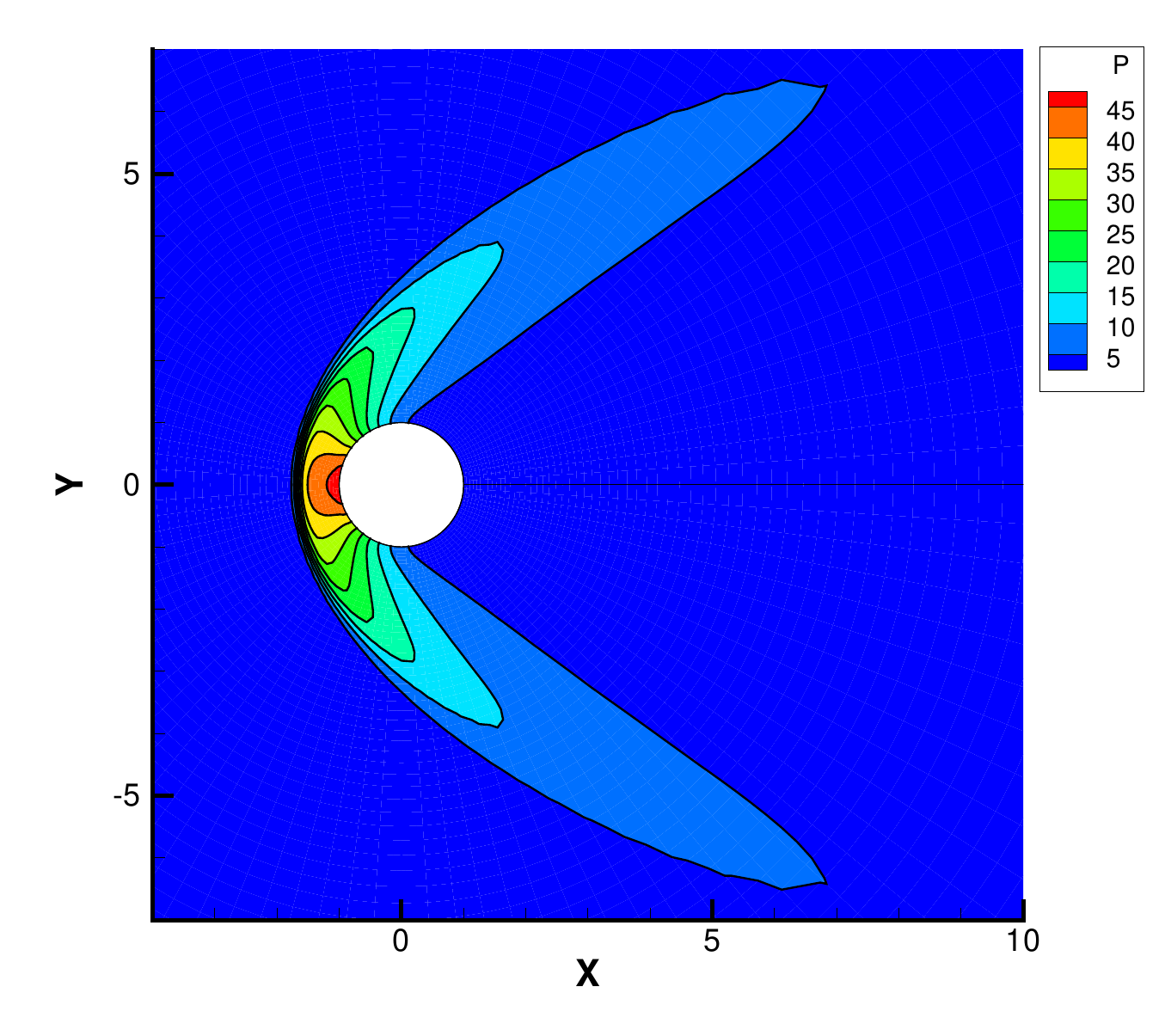}
        \caption{WPD-MC, $p$}
    \end{subfigure}
    \begin{subfigure}{0.49\textwidth}
        \centering
        \includegraphics[width=\linewidth]{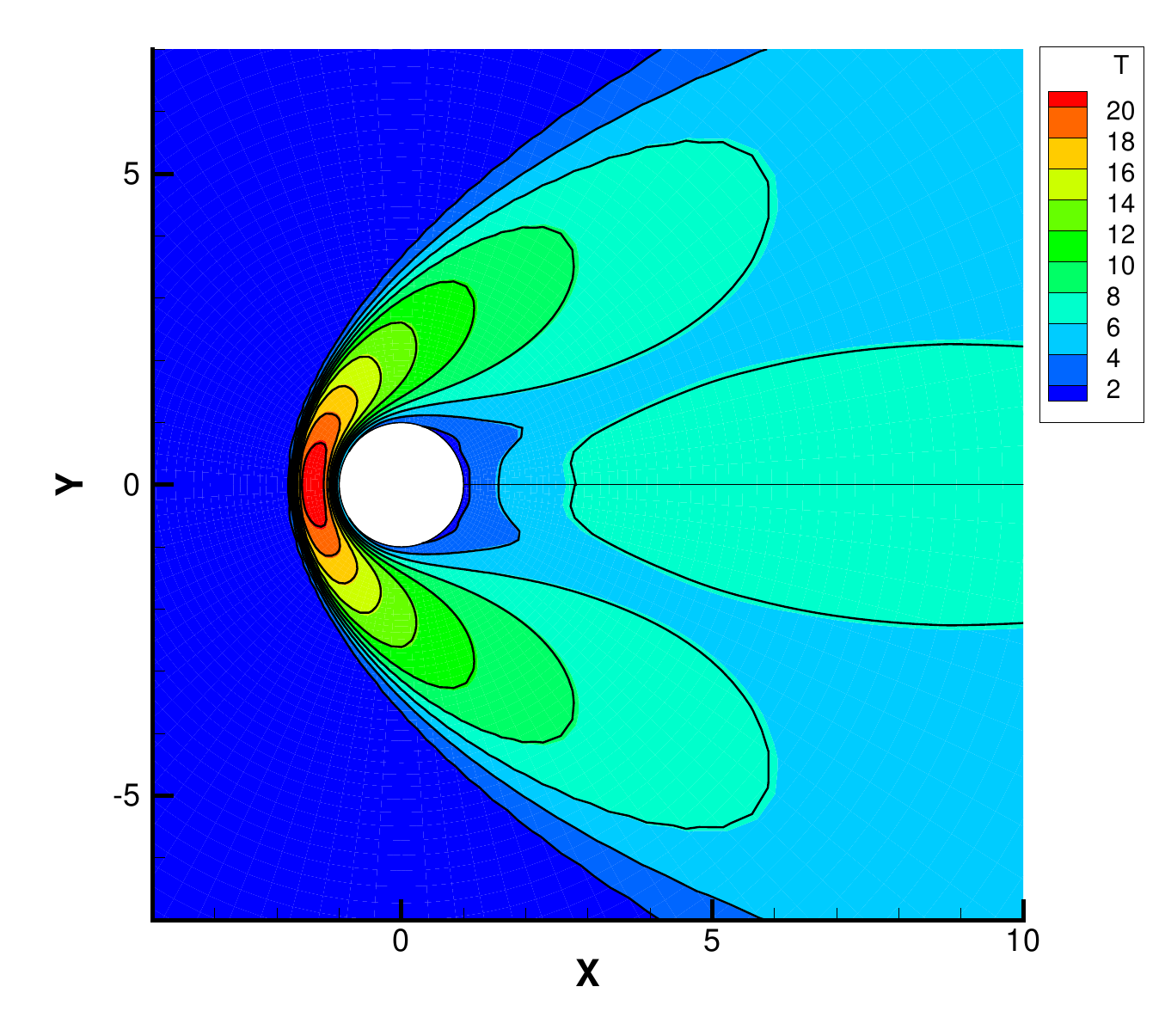}
        \caption{WPD-MC, $T$}
    \end{subfigure}
    \caption{Cylinder pressure and temperature contours at $\mathrm{Kn}=10^{-3}$.}
    \label{fig:cylinder_mc_kn1em3_pt}
\end{figure}

\begin{figure}[!htbp]
    \centering
    \begin{subfigure}{0.49\textwidth}
        \centering
        \includegraphics[width=\linewidth]{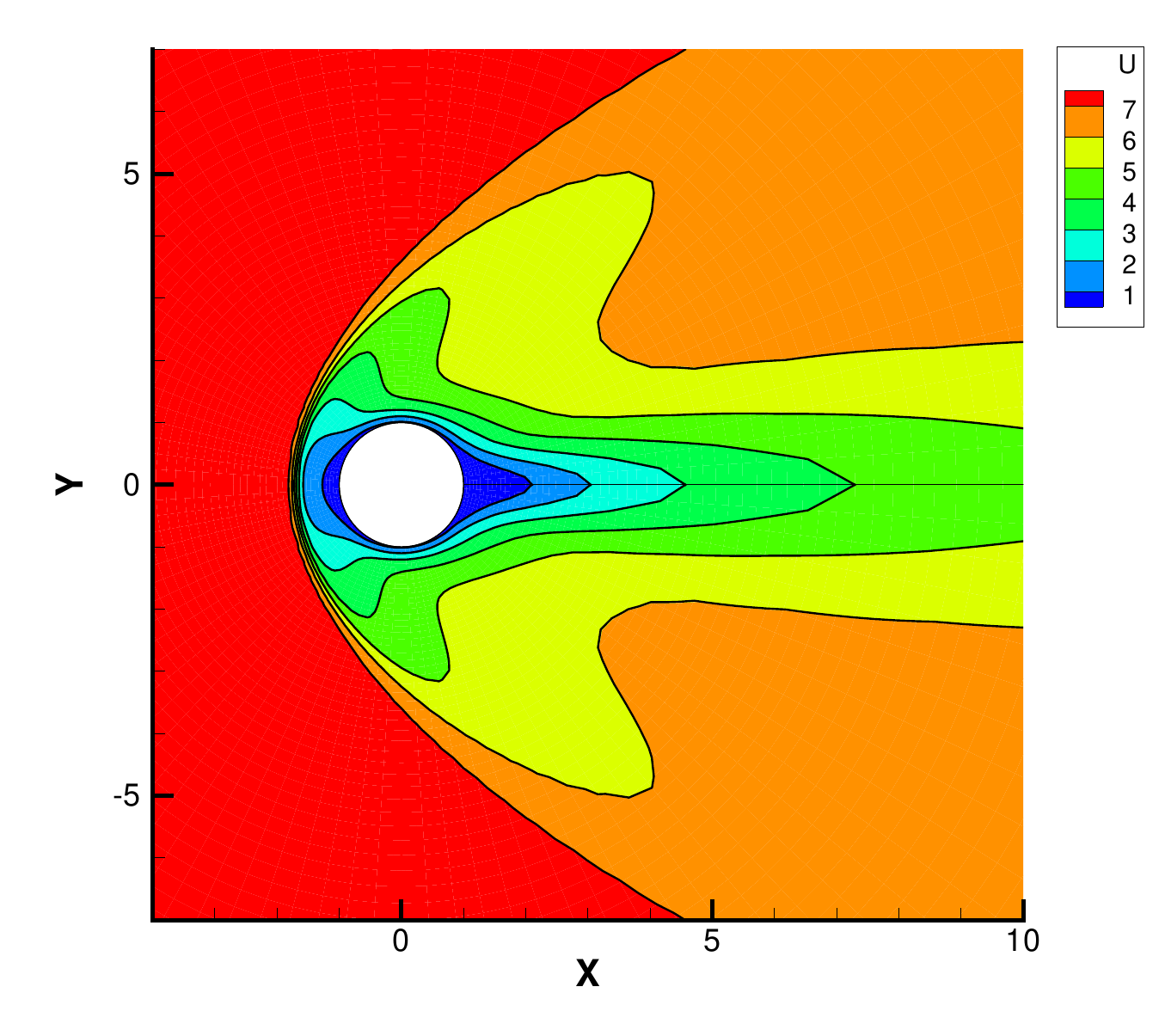}
        \caption{WPD-MC, $u$}
    \end{subfigure}
    \begin{subfigure}{0.49\textwidth}
        \centering
        \includegraphics[width=\linewidth]{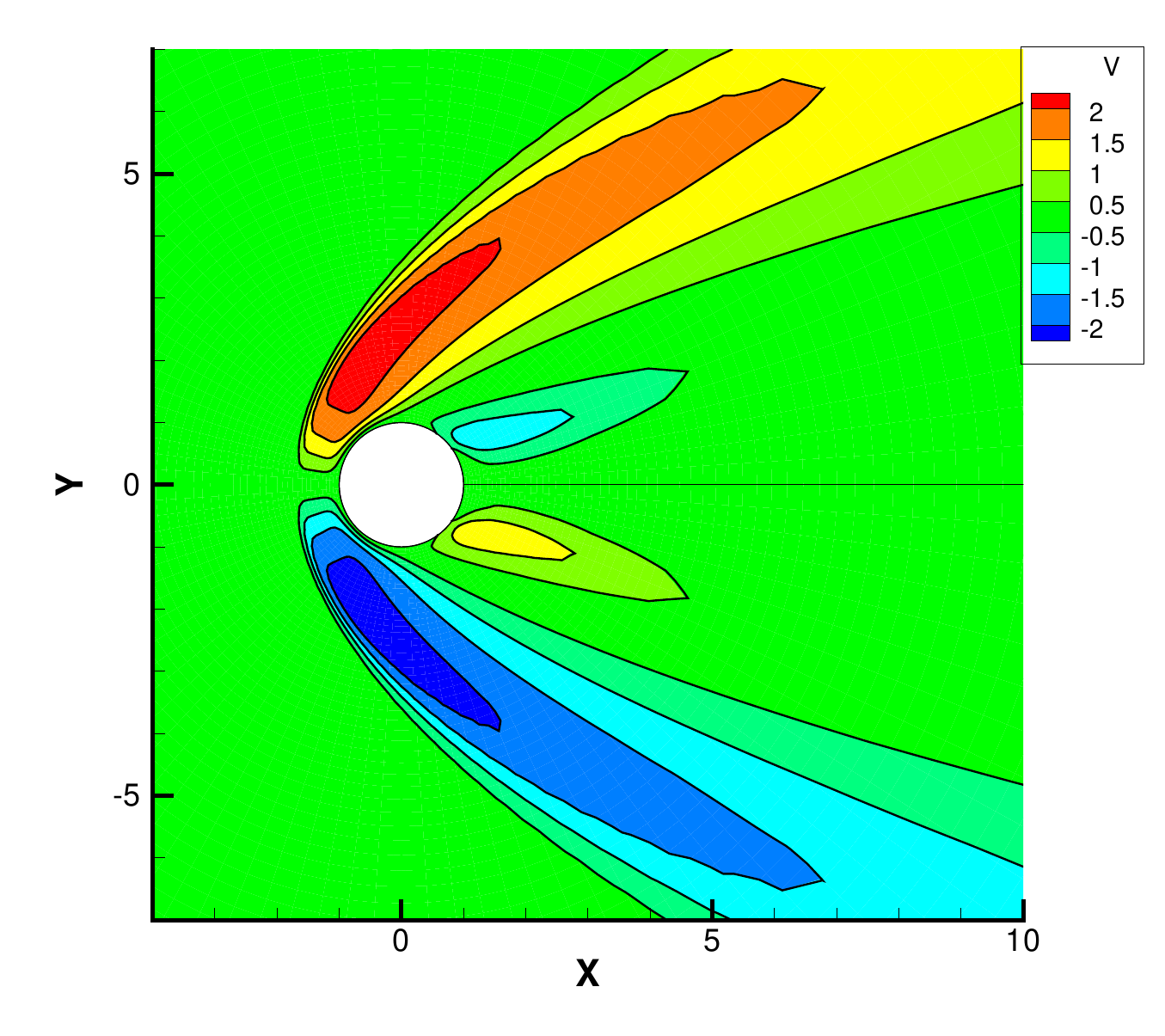}
        \caption{WPD-MC, $v$}
    \end{subfigure}
    \caption{Cylinder velocity contours at $\mathrm{Kn}=10^{-3}$.}
    \label{fig:cylinder_mc_kn1em3_uv}
\end{figure}

\begin{figure}[!htbp]
    \centering
    \begin{subfigure}{0.49\textwidth}
        \centering
        \includegraphics[width=\linewidth]{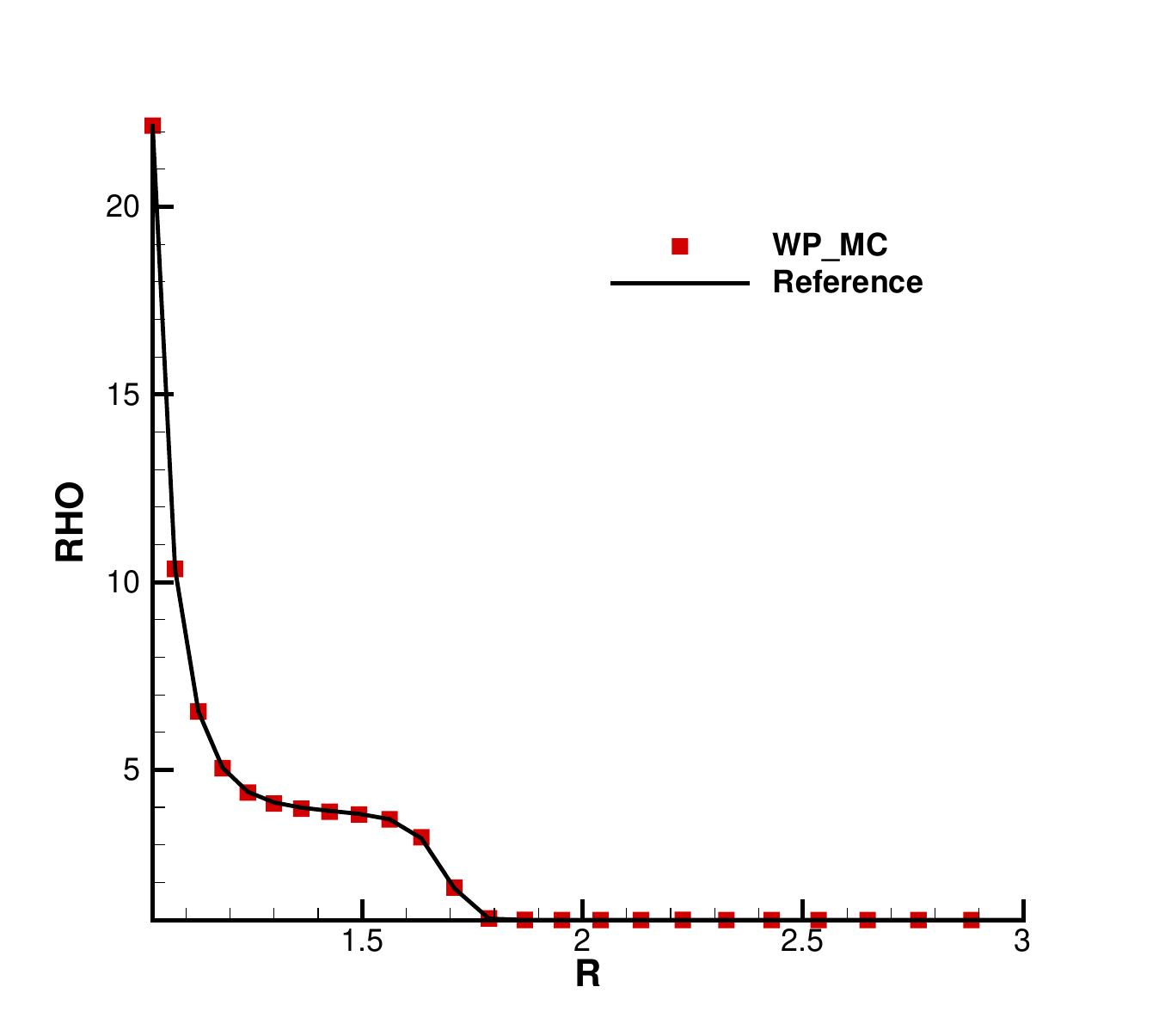}
        \caption{WPD-MC (local), $\rho$}
    \end{subfigure}
    \begin{subfigure}{0.49\textwidth}
        \centering
        \includegraphics[width=\linewidth]{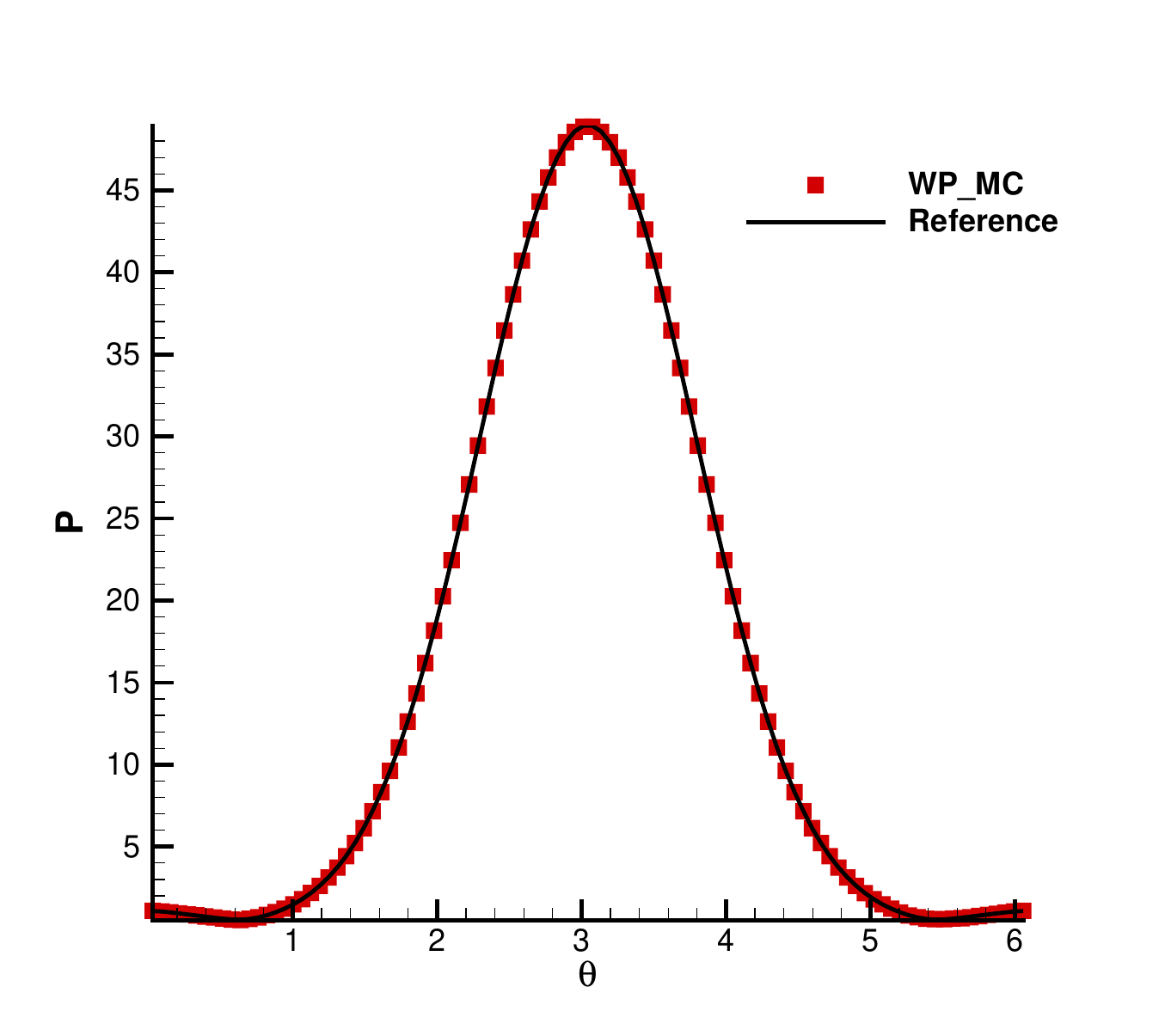}
        \caption{WPD-MC (local), $p$}
    \end{subfigure}
    \caption{Cylinder stagnation-line density and pressure at $\mathrm{Kn}=10^{-3}$.}
    \label{fig:cylinder_mc_R_den_p_kn1em3}
\end{figure}
\FloatBarrier

\begin{figure}[!htbp]
    \centering
    \begin{subfigure}{0.49\textwidth}
        \centering
        \includegraphics[width=\linewidth]{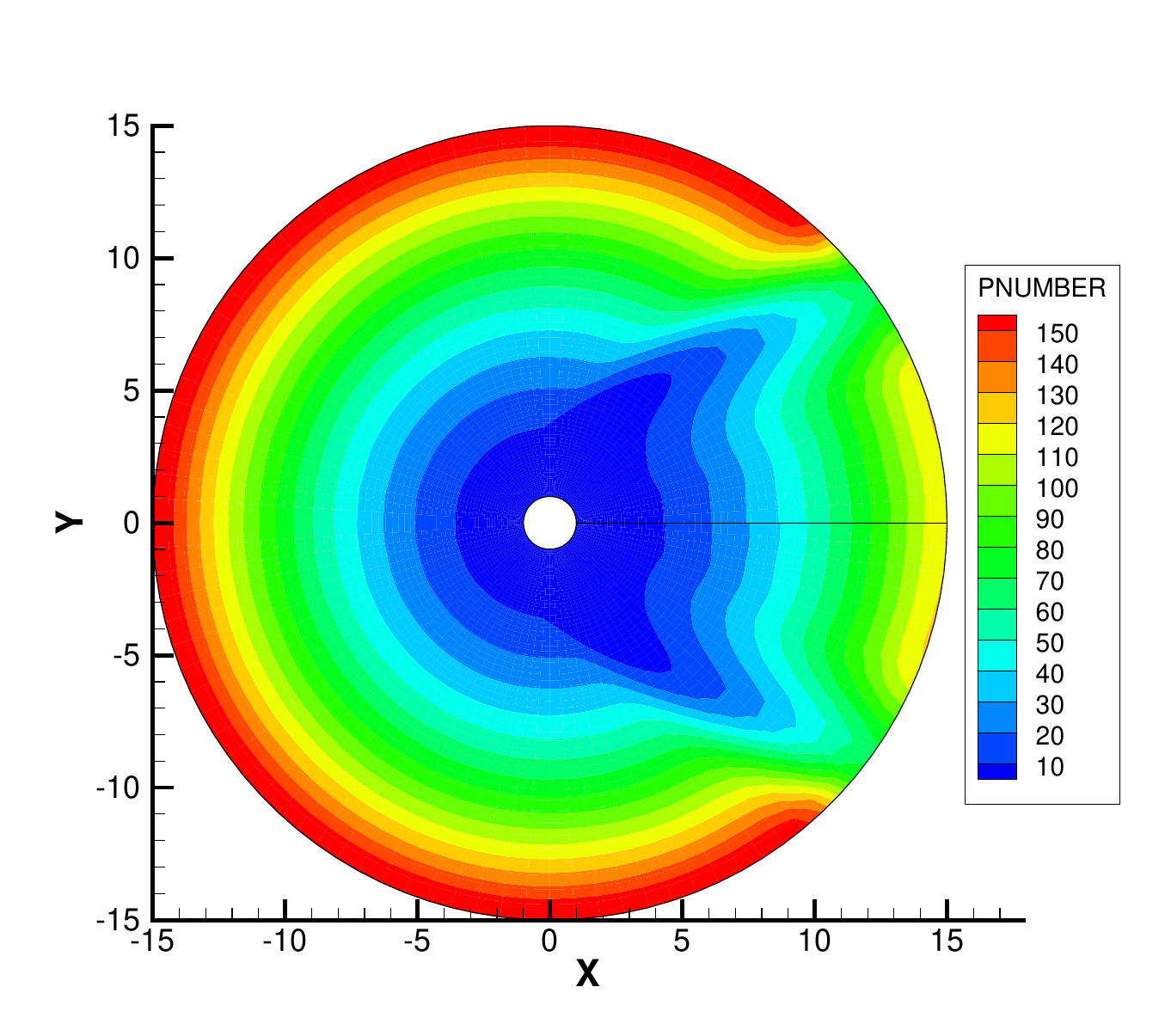}
        \caption{UGKWP, $\mathrm{Kn}=10^{-3}$}
    \end{subfigure}
    \begin{subfigure}{0.49\textwidth}
        \centering
        \includegraphics[width=\linewidth]{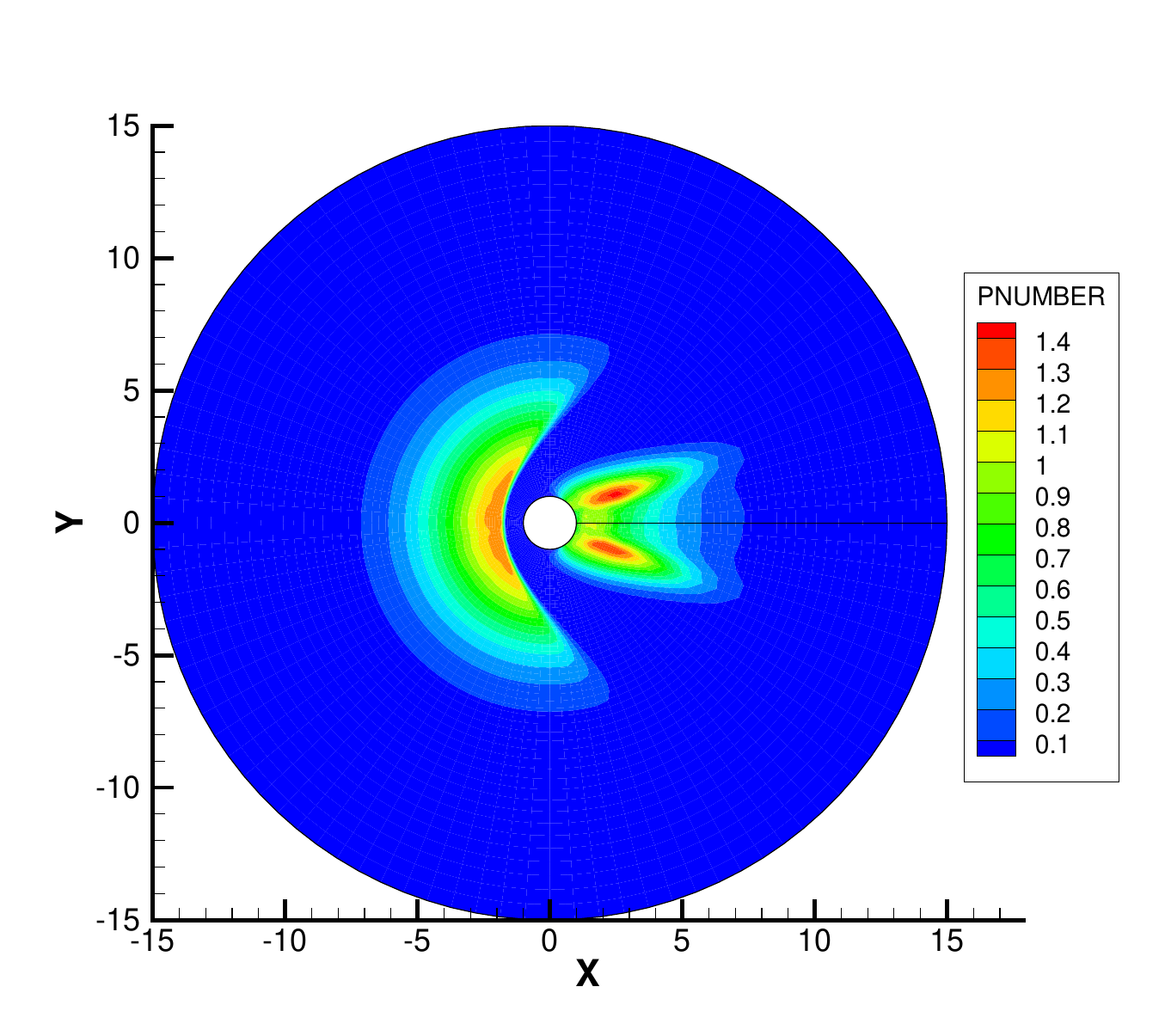}
        \caption{WPD-MC (local), $\mathrm{Kn}=10^{-3}$}
    \end{subfigure}
    \caption{Equivalent particle-number fields at $\mathrm{Kn}=10^{-3}$.}
    \label{fig:cylinder_pnumber_b_kn1m3}
\end{figure}

The computational saving in the cylinder calculation is measured by both the transported kinetic representation and the wall time. On the \(100\times64\) mesh, the deterministic \(S_N\) solver uses \(64\times64\) velocity points and two reduced distributions:
\[
N_{S_N}=100\times64\times64\times64\times2=5.24\times10^7 .
\]
For stochastic cases, the effective kinetic DoF is computed from the particle-phase mass and the fixed reference particle weight. It measures the active kinetic representation, not runtime directly, because Monte Carlo transport includes sampling, tracing, wall interaction, cell location, and memory overhead, while WPD-MC also retains the fixed wave/GKS update. Table~\ref{tab:cylinder_cost_estimate} therefore reports both effective DoF and measured wall time on one Intel Xeon Platinum 8358P node using \(64\) CPU cores.

\begin{table}[!htbp]
\centering
\caption{Effective kinetic degrees of freedom and measured wall time for the cylinder calculations on one Intel Xeon Platinum 8358P node using \(64\) CPU cores. The effective kinetic degrees of freedom measure the size of the explicitly transported kinetic representation; the wall time includes the fixed wave update and the additional overhead of particle tracing and sampling.}
\label{tab:cylinder_cost_estimate}
\small
\begin{tabular}{lccc}
\toprule
Method & Effective kinetic DoFs & DoF ratio & Measured wall time \\
\midrule
WPD-\(S_N\) & \(5.24\times10^7\) & \(1\) & \(40.71\) min \\
\midrule
\(\mathrm{Kn}=1\), UGKWP & \(6.97\times10^5\) & \(1.33\times10^{-2}\) & \(13.84\) min \\
\(\mathrm{Kn}=1\), WPD-MC (local) & \(6.90\times10^5\) & \(1.32\times10^{-2}\) & \(13.57\) min \\
\(\mathrm{Kn}=10^{-2}\), UGKWP & \(5.82\times10^5\) & \(1.11\times10^{-2}\) & \(12.49\) min \\
\(\mathrm{Kn}=10^{-2}\), WPD-MC (local) & \(2.28\times10^5\) & \(4.36\times10^{-3}\) & \(7.12\) min \\
\(\mathrm{Kn}=10^{-3}\), UGKWP & \(1.80\times10^5\) & \(3.43\times10^{-3}\) & \(4.43\) min \\
\(\mathrm{Kn}=10^{-3}\), WPD-MC (local) & \(2.03\times10^3\) & \(3.88\times10^{-5}\) & \(1.07\) min \\
\bottomrule
\end{tabular}
\end{table}

At \(\mathrm{Kn}=1\), WPD-MC (local) and UGKWP have nearly the same effective DoF and wall time because the particle fraction is close to unity. As the Knudsen number decreases, the local horizon reduces the kinetic representation in continuum-like regions. At \(\mathrm{Kn}=10^{-2}\), WPD-MC (local) reduces the effective DoF by about \(2.6\) and the wall time by about \(1.8\); at \(\mathrm{Kn}=10^{-3}\), the reductions are about \(88\) and \(4.1\), respectively. The smaller runtime gain reflects the fixed wave/GKS cost and Monte Carlo overhead.

\subsection{Three-Dimensional X38 Configuration}

The final case applies the same WPD formulation directly to a three-dimensional complex geometry. The X38 configuration is computed at a free-stream Mach number \(M_\infty=5\) with a nonzero angle of attack, \(\mathrm{AoA}=15^\circ\). The gas is monatomic with \(\gamma=5/3\), and diffuse wall reflection is used on the vehicle surface. The unstructured polyhedral volume mesh contains \(104723\) nodes and \(560593\) cells, with \(1138792\) faces.

The cases shown here correspond to \(\mathrm{Kn}=10^{-1}\), \(10^{-2}\), and \(10^{-4}\). This is a qualitative multidimensional demonstration rather than a quantitative aerodynamic validation. As \(\mathrm{Kn}\) decreases, the local-horizon coefficients shift more of the solution to the wave component, while particles remain active near rarefied, non-equilibrium, and wall-interaction regions. The figures show the resulting surface and near-body density, temperature, velocity, pressure, and Mach-number fields.

\begin{figure}[!htbp]
    \centering
    \begin{subfigure}{0.49\textwidth}
        \centering
        \includegraphics[width=\linewidth]{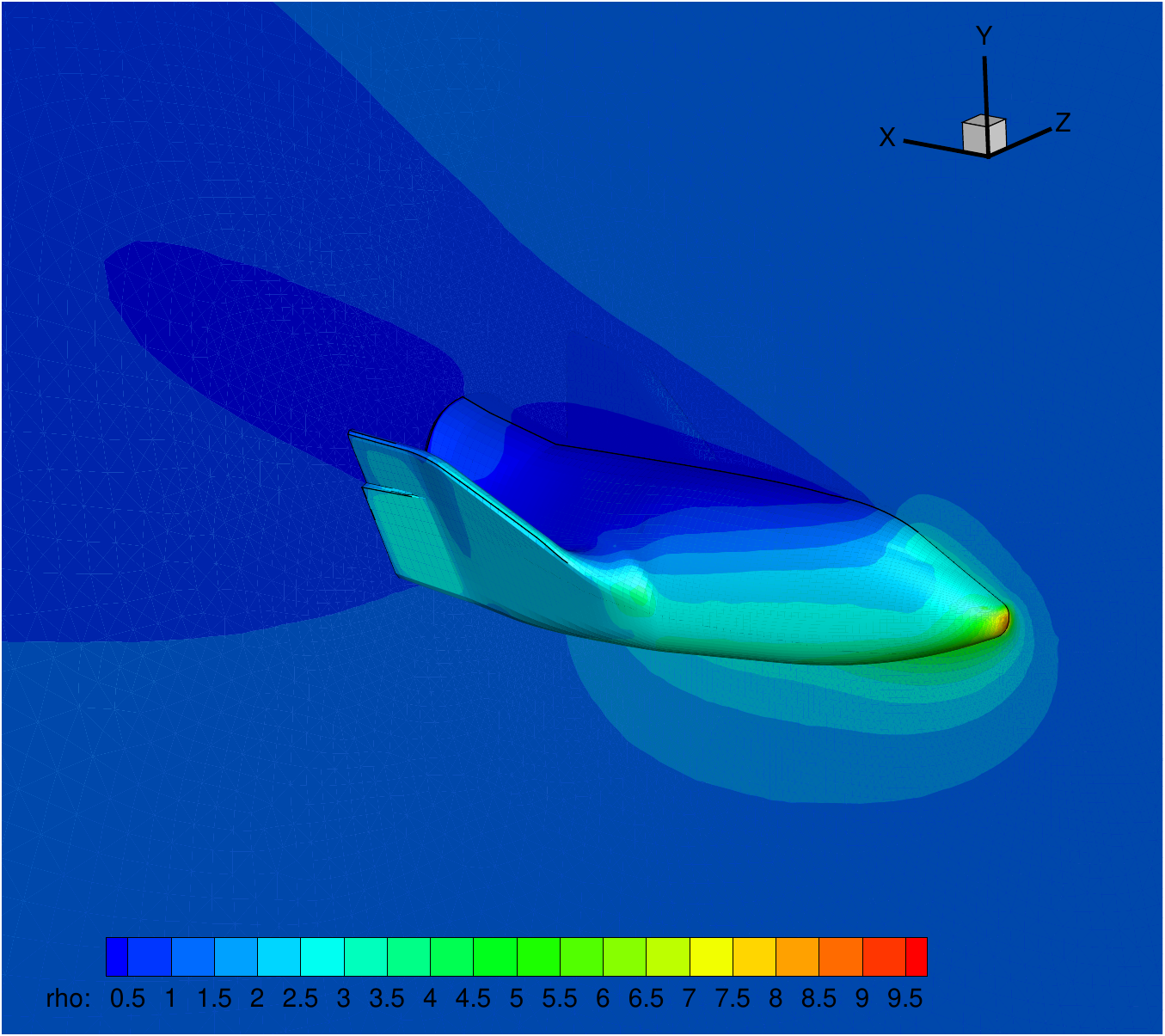}
        \caption{Density}
    \end{subfigure}
    \begin{subfigure}{0.49\textwidth}
        \centering
        \includegraphics[width=\linewidth]{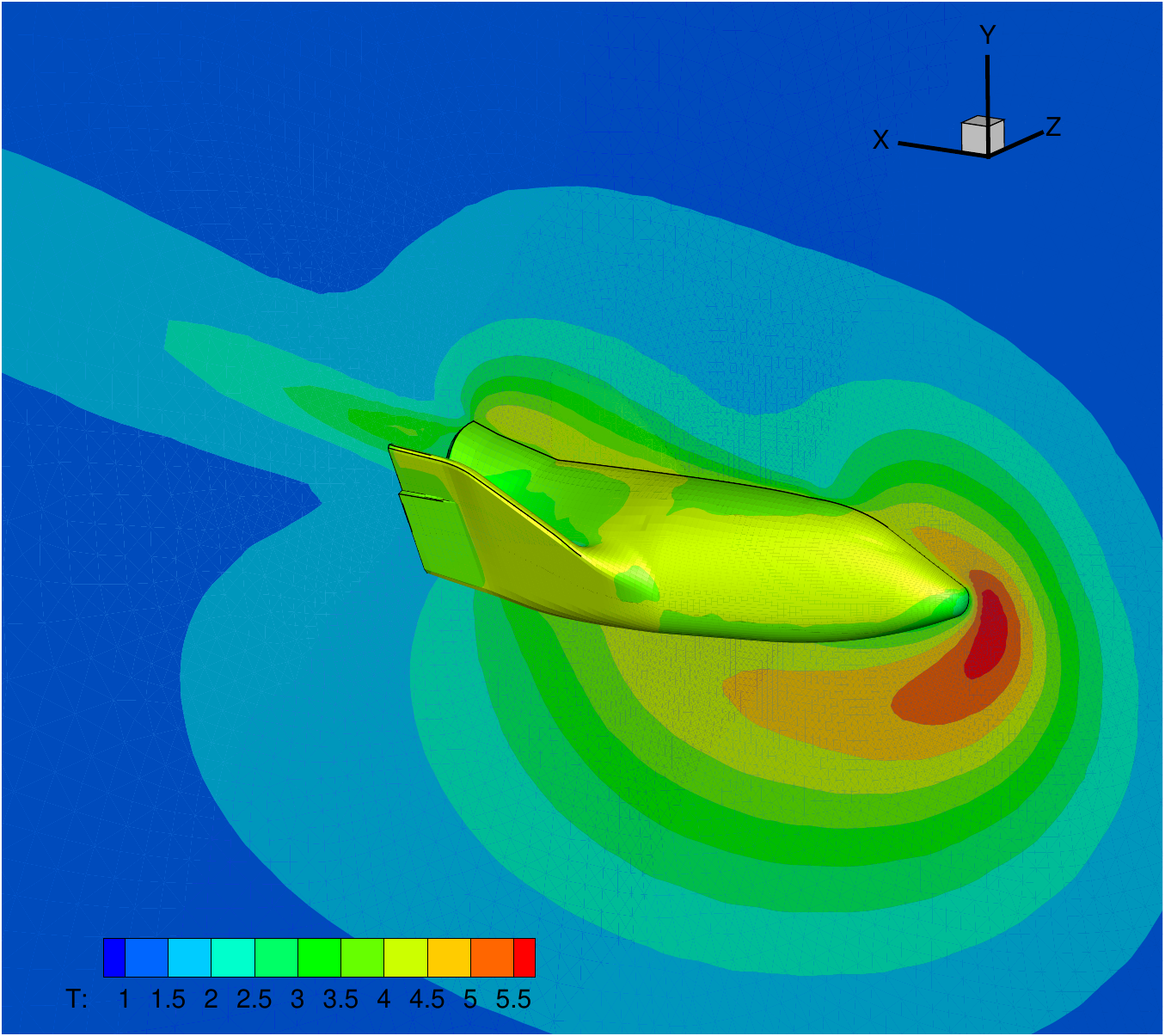}
        \caption{Temperature}
    \end{subfigure}
    \caption{Three-dimensional $x38$ case at $\mathrm{Kn}=10^{-1}$.}
    \label{fig:x38_kn1em1_rhot}
\end{figure}

\begin{figure}[!htbp]
    \centering
    \begin{subfigure}{0.49\textwidth}
        \centering
        \includegraphics[width=\linewidth]{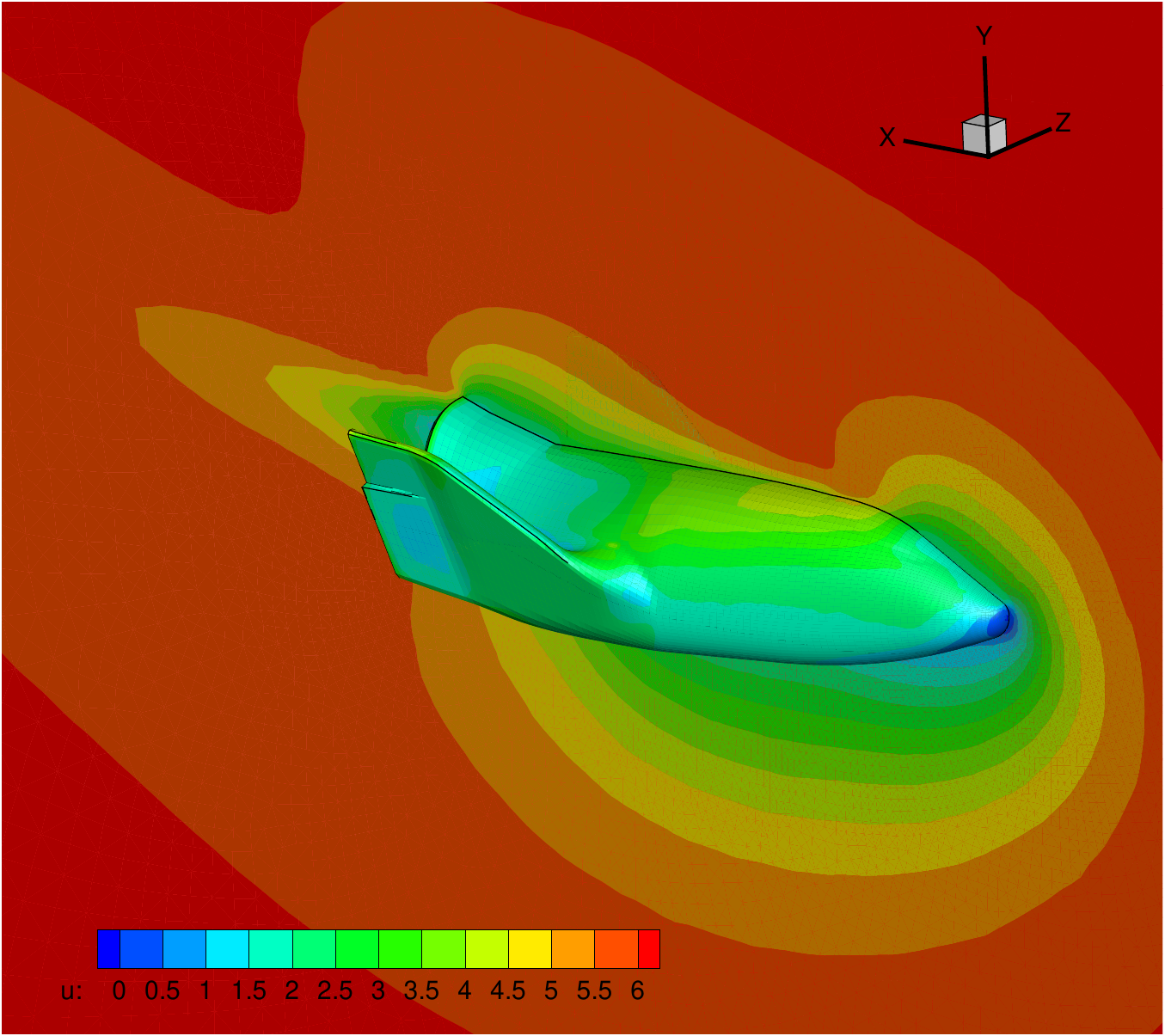}
        \caption{\(u\)}
    \end{subfigure}
    \begin{subfigure}{0.49\textwidth}
        \centering
        \includegraphics[width=\linewidth]{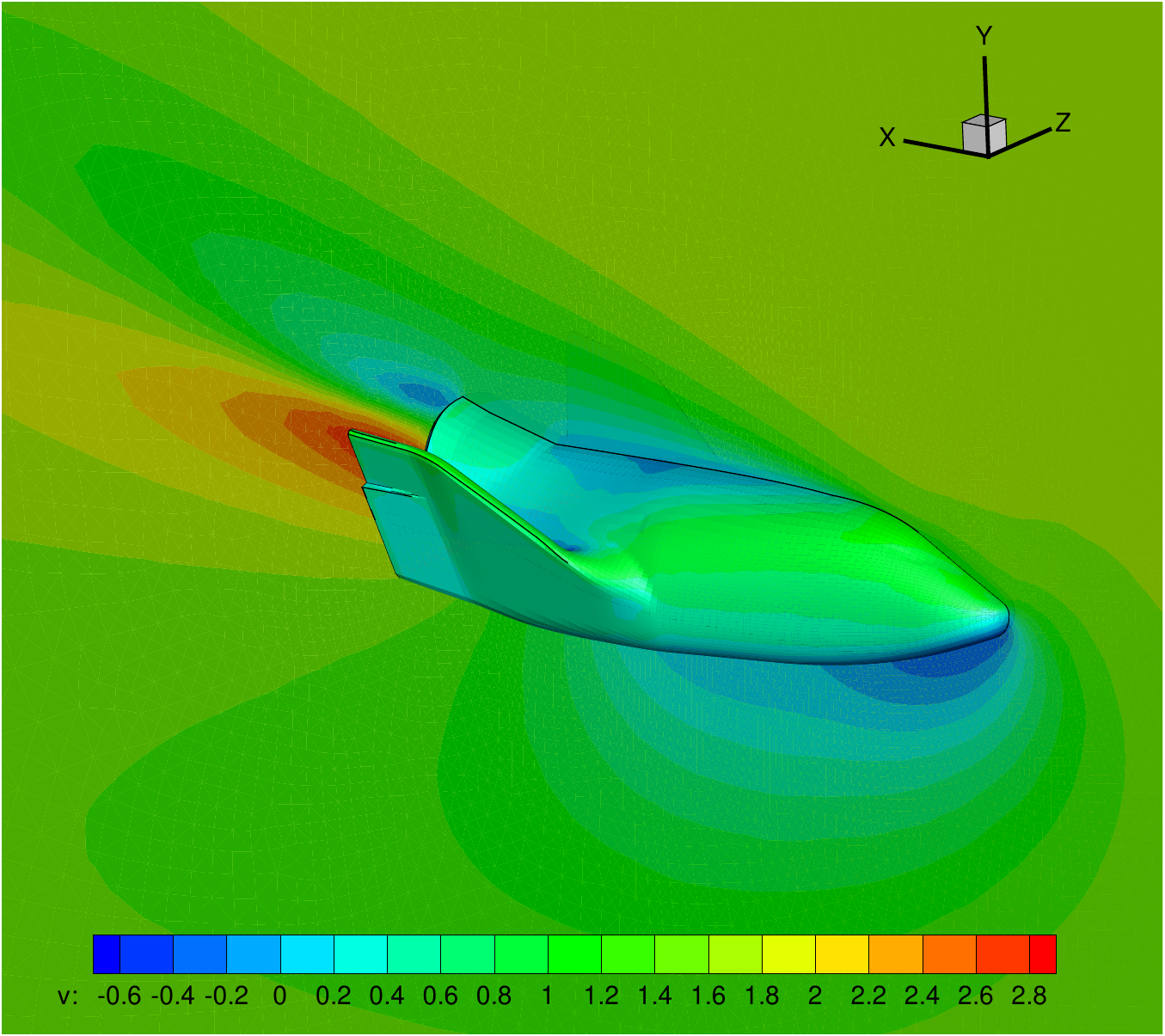}
        \caption{\(v\)}
    \end{subfigure}
    \caption{Three-dimensional $x38$ case at $\mathrm{Kn}=10^{-1}$.}
    \label{fig:x38_kn1em1_uv}
\end{figure}

\begin{figure}[!htbp]
    \centering
    \begin{subfigure}{0.49\textwidth}
        \centering
        \includegraphics[width=\linewidth]{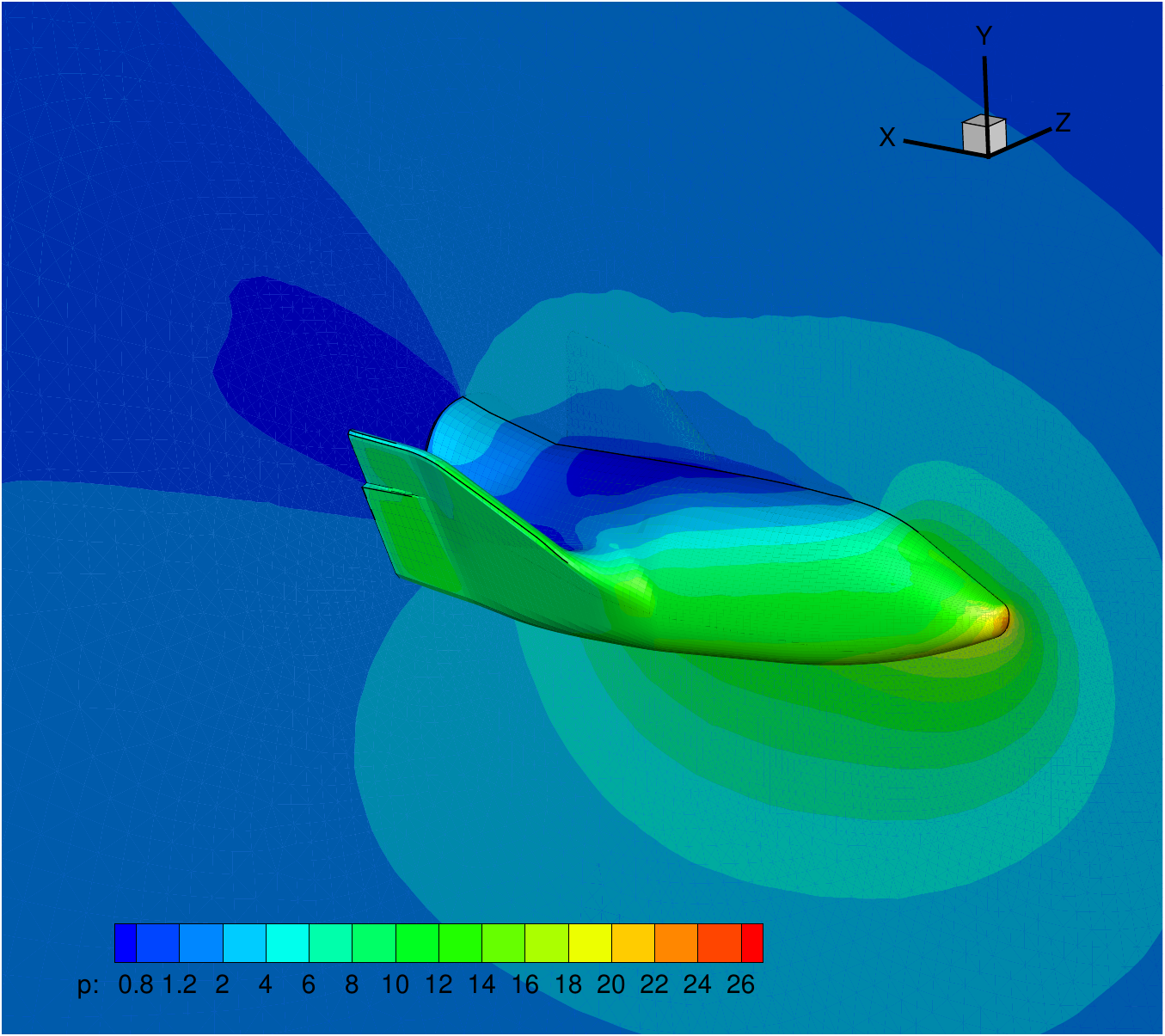}
        \caption{Pressure}
    \end{subfigure}
    \begin{subfigure}{0.49\textwidth}
        \centering
        \includegraphics[width=\linewidth]{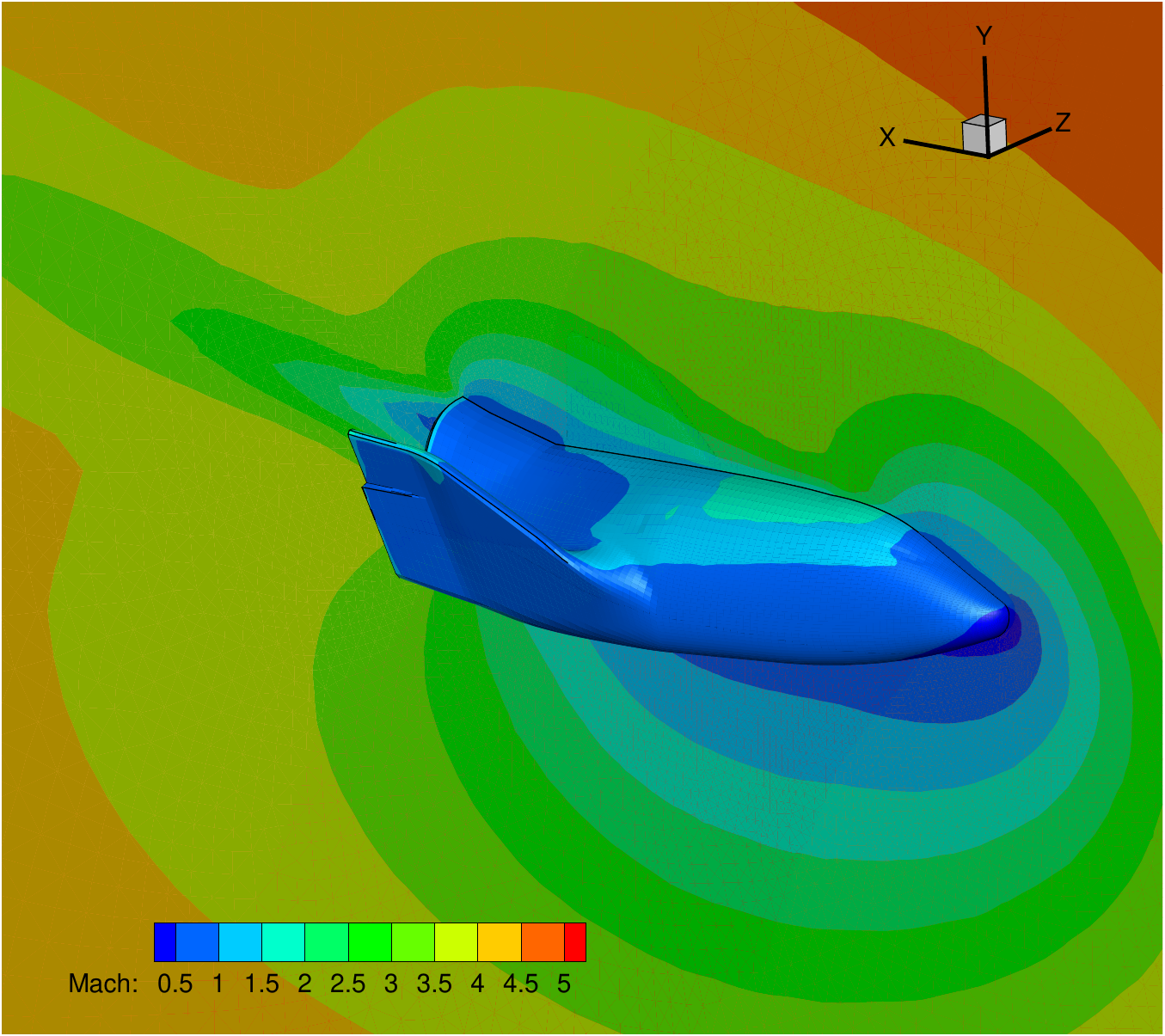}
        \caption{Mach number}
    \end{subfigure}
    \caption{Three-dimensional $x38$ case at $\mathrm{Kn}=10^{-1}$.}
    \label{fig:x38_kn1em1_pma}
\end{figure}

\begin{figure}[!htbp]
    \centering
    \begin{subfigure}{0.49\textwidth}
        \centering
        \includegraphics[width=\linewidth]{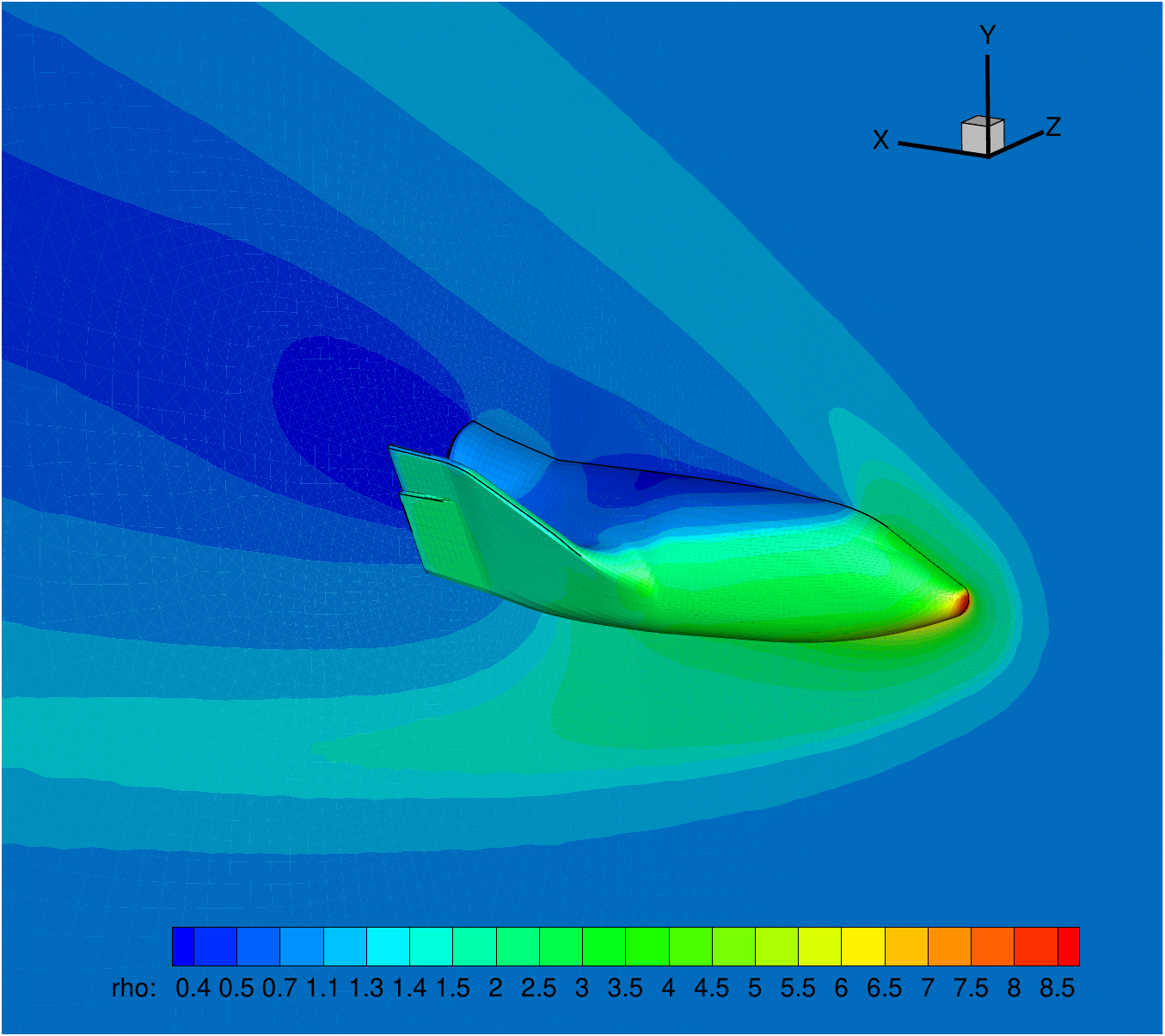}
        \caption{Density}
    \end{subfigure}
    \begin{subfigure}{0.49\textwidth}
        \centering
        \includegraphics[width=\linewidth]{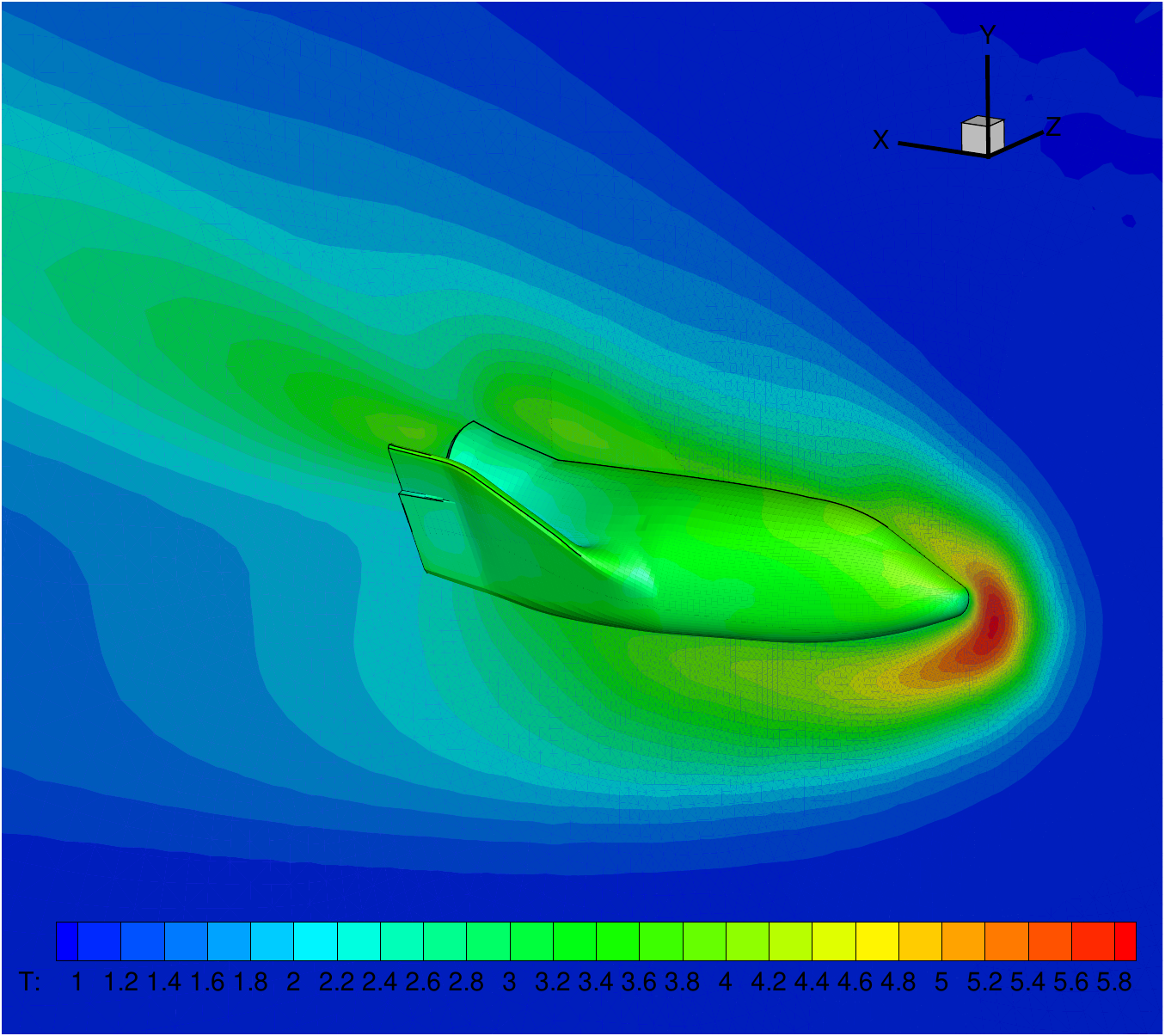}
        \caption{Temperature}
    \end{subfigure}
    \caption{Three-dimensional $x38$ case at $\mathrm{Kn}=10^{-2}$.}
    \label{fig:x38_kn1em2_rhot}
\end{figure}

\begin{figure}[!htbp]
    \centering
    \begin{subfigure}{0.49\textwidth}
        \centering
        \includegraphics[width=\linewidth]{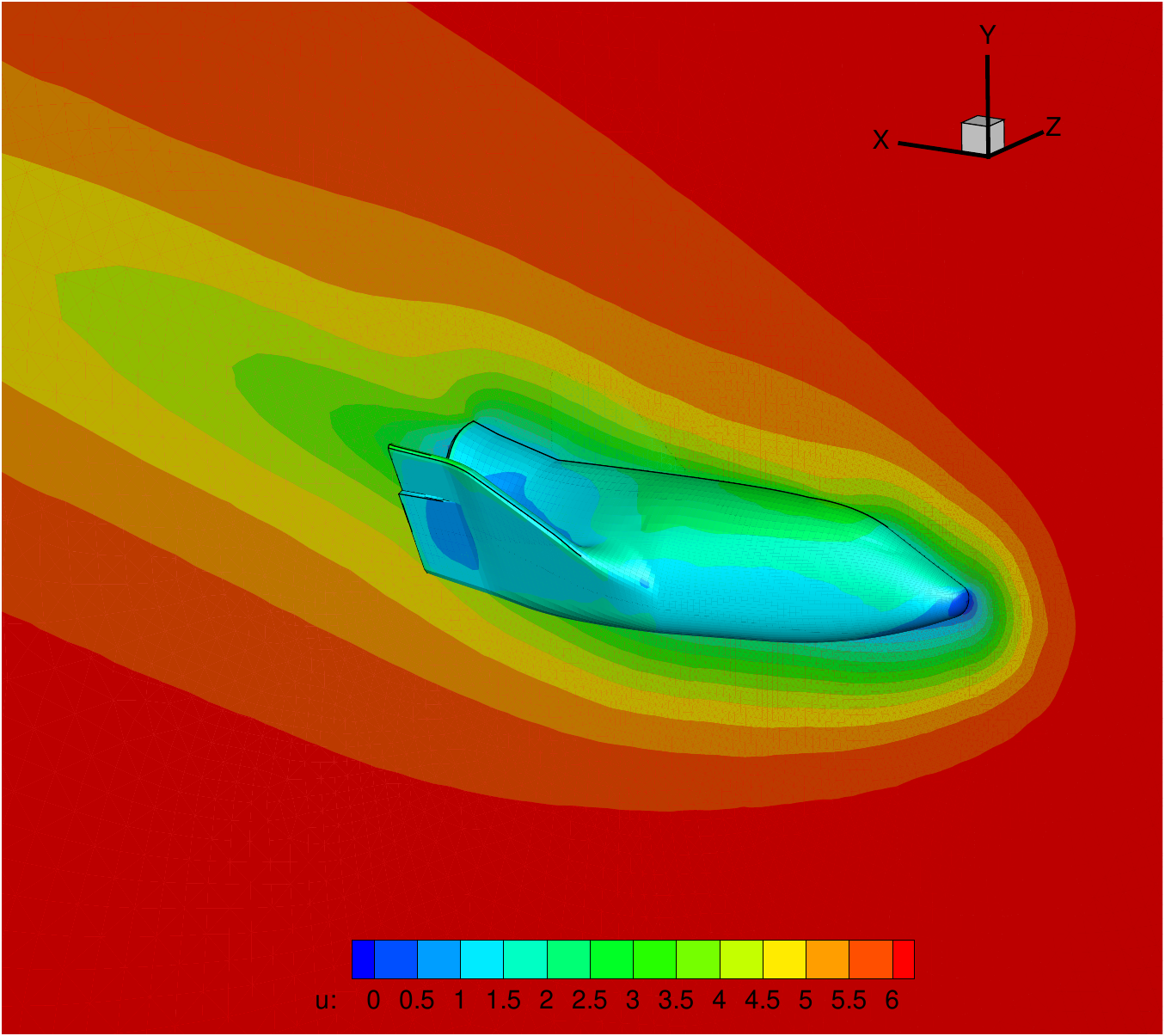}
        \caption{\(u\)}
    \end{subfigure}
    \begin{subfigure}{0.49\textwidth}
        \centering
        \includegraphics[width=\linewidth]{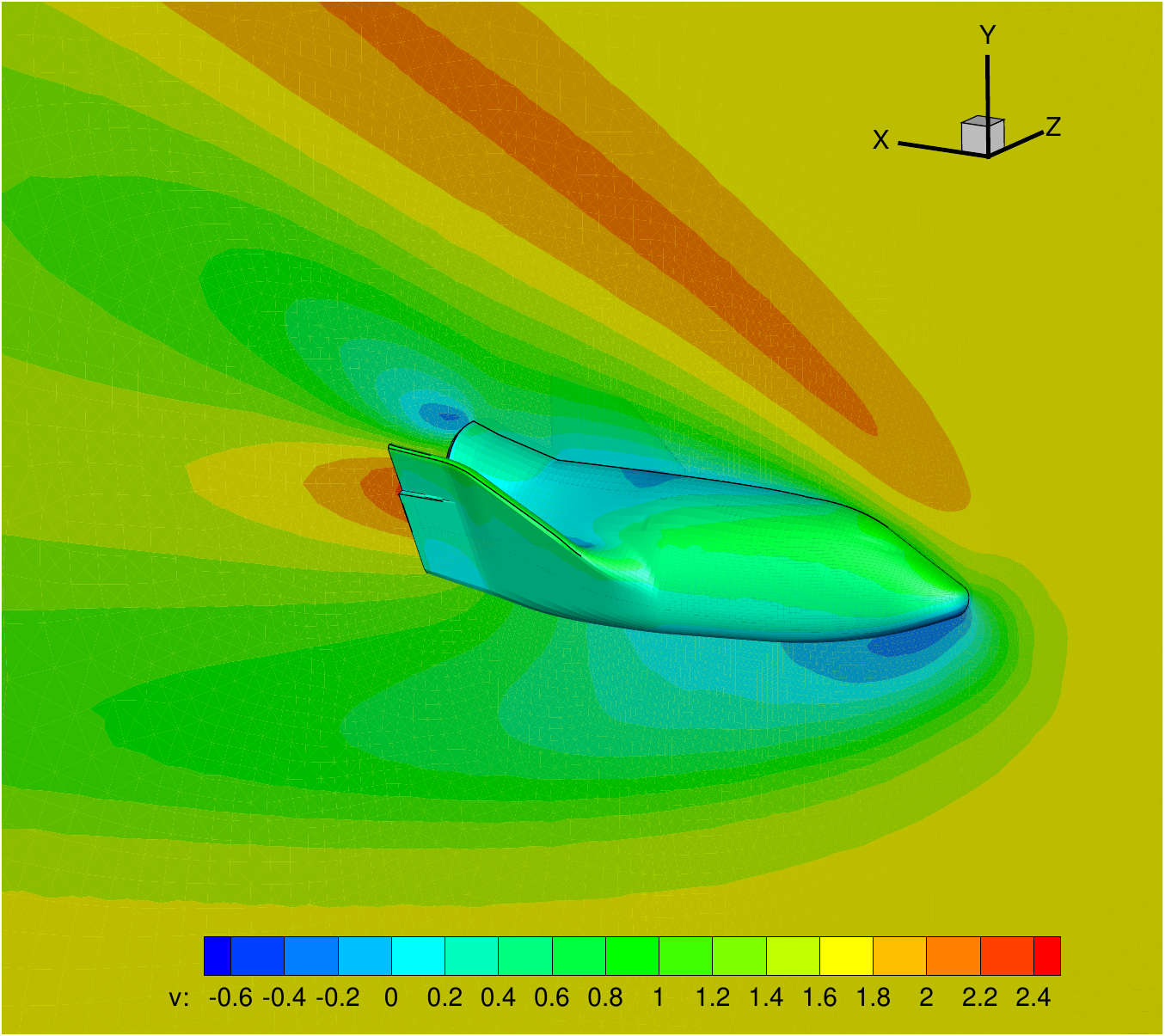}
        \caption{\(v\)}
    \end{subfigure}
    \caption{Three-dimensional $x38$ case at $\mathrm{Kn}=10^{-2}$.}
    \label{fig:x38_kn1em2_uv}
\end{figure}

\begin{figure}[!htbp]
    \centering
    \begin{subfigure}{0.49\textwidth}
        \centering
        \includegraphics[width=\linewidth]{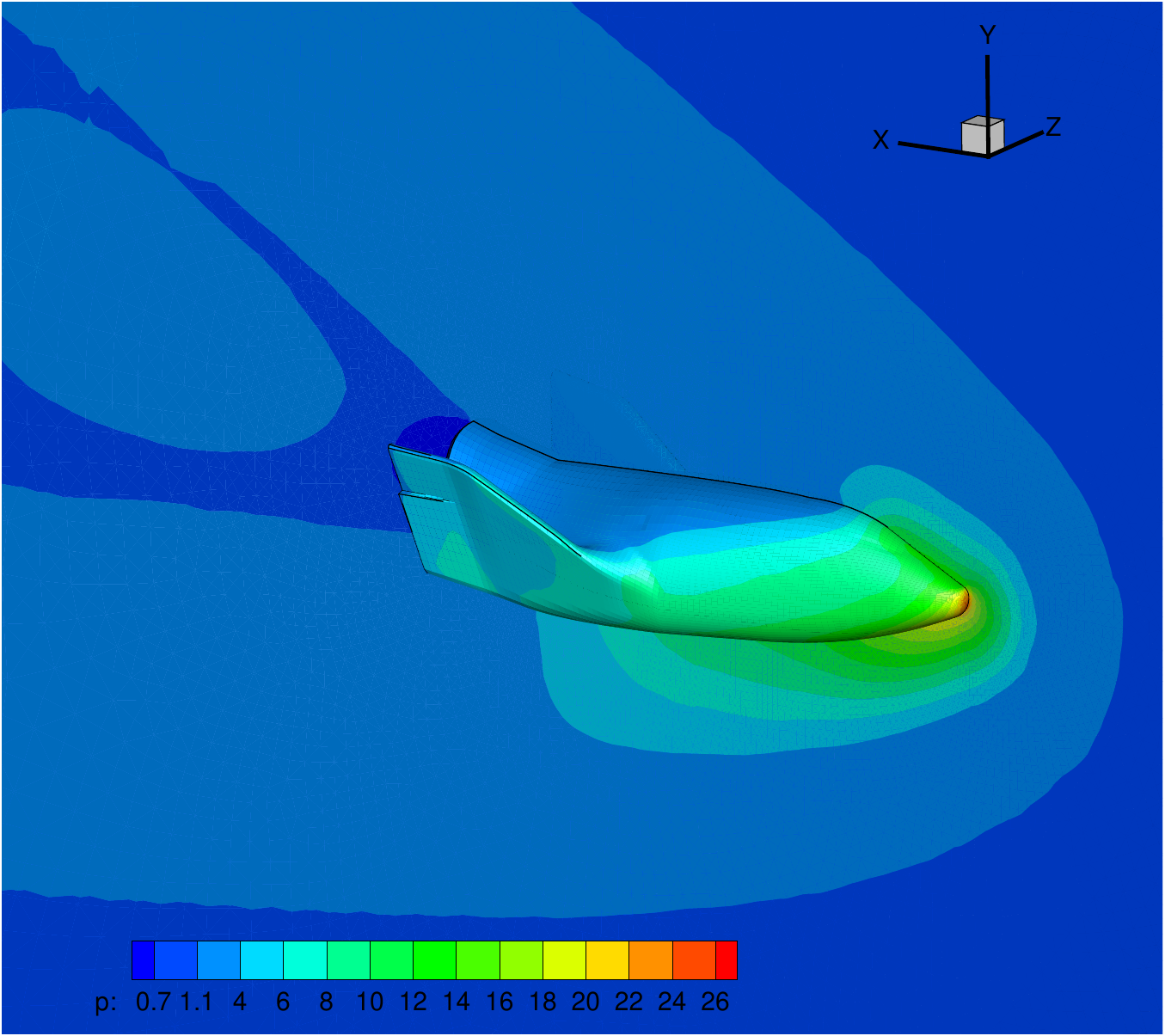}
        \caption{Pressure}
    \end{subfigure}
    \begin{subfigure}{0.49\textwidth}
        \centering
        \includegraphics[width=\linewidth]{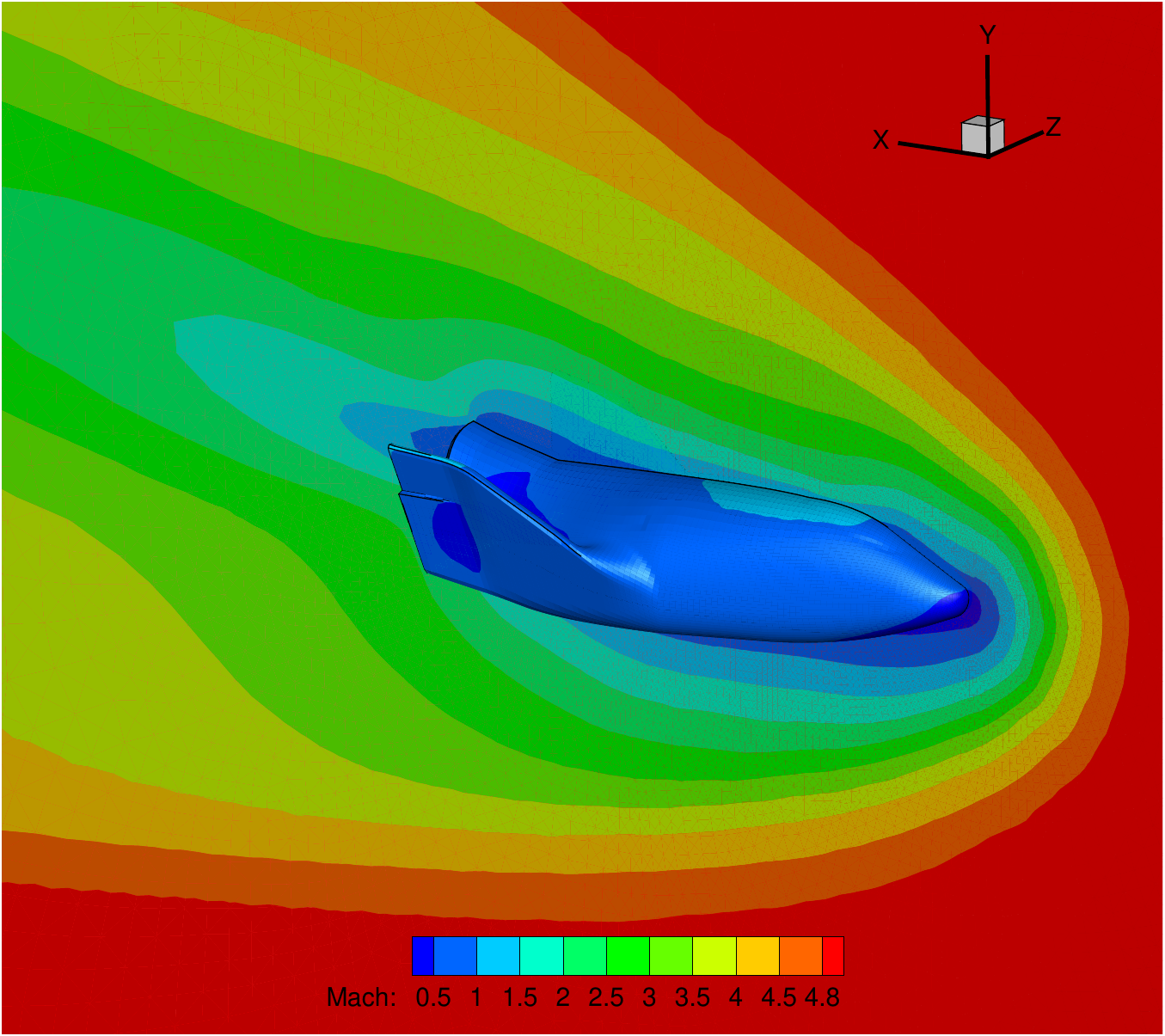}
        \caption{Mach number}
    \end{subfigure}
    \caption{Three-dimensional $x38$ case at $\mathrm{Kn}=10^{-2}$.}
    \label{fig:x38_kn1em2_pma}
\end{figure}

\begin{figure}[!htbp]
    \centering
    \begin{subfigure}{0.49\textwidth}
        \centering
        \includegraphics[width=\linewidth]{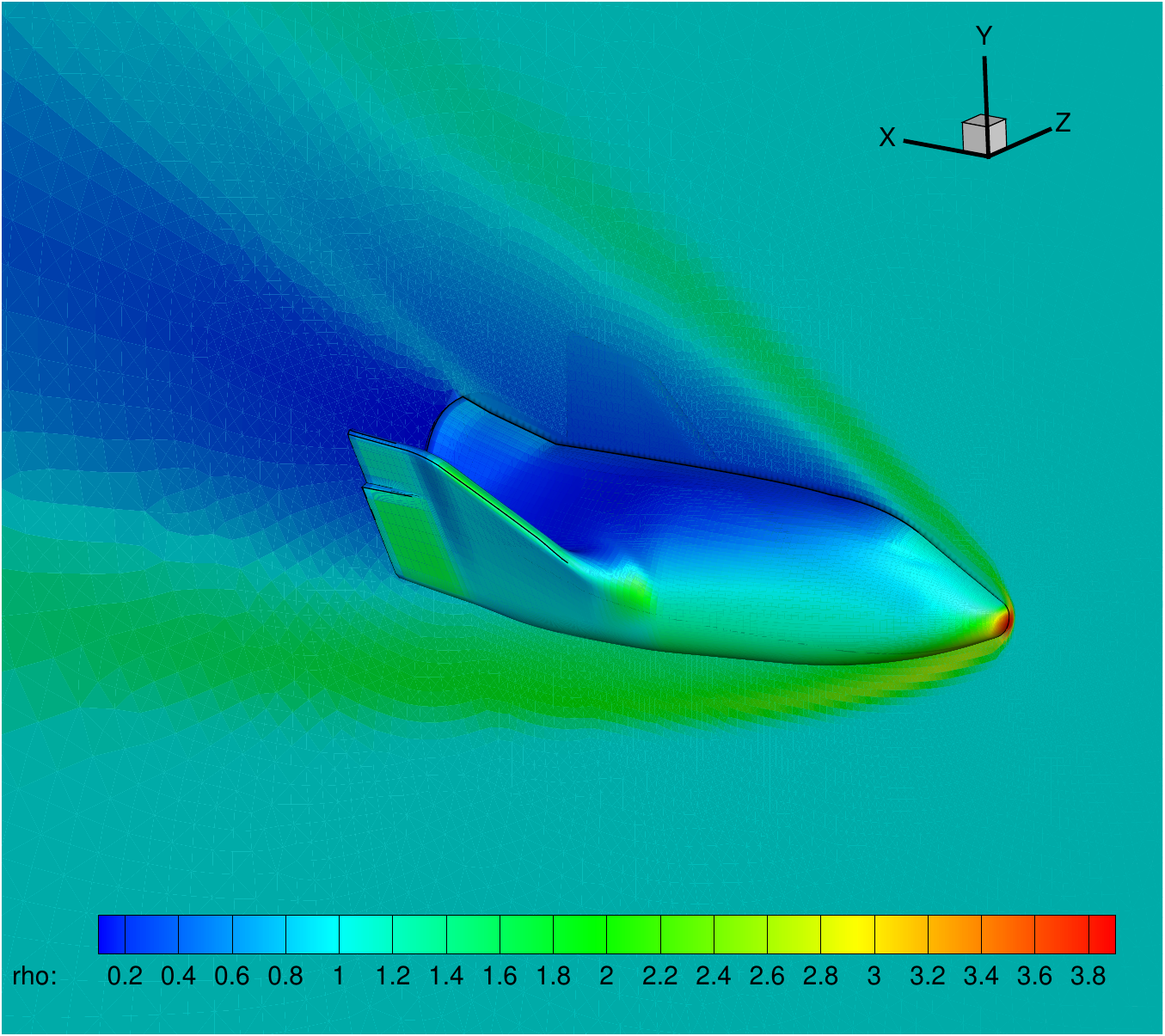}
        \caption{Density}
    \end{subfigure}
    \begin{subfigure}{0.49\textwidth}
        \centering
        \includegraphics[width=\linewidth]{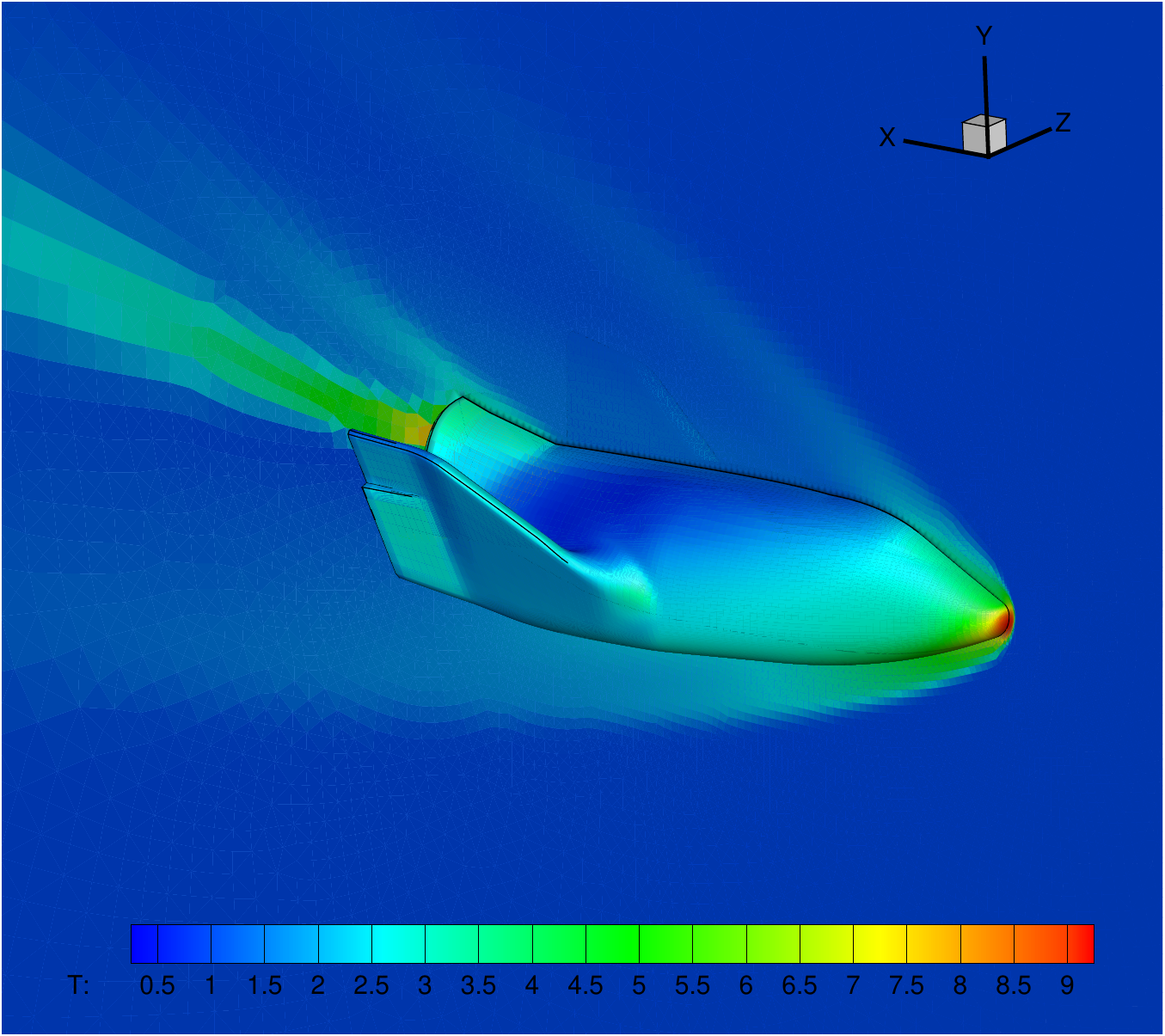}
        \caption{Temperature}
    \end{subfigure}
    \caption{Three-dimensional $x38$ case at $\mathrm{Kn}=10^{-4}$.}
    \label{fig:x38_kn1em4_rhot}
\end{figure}

\begin{figure}[!htbp]
    \centering
    \begin{subfigure}{0.49\textwidth}
        \centering
        \includegraphics[width=\linewidth]{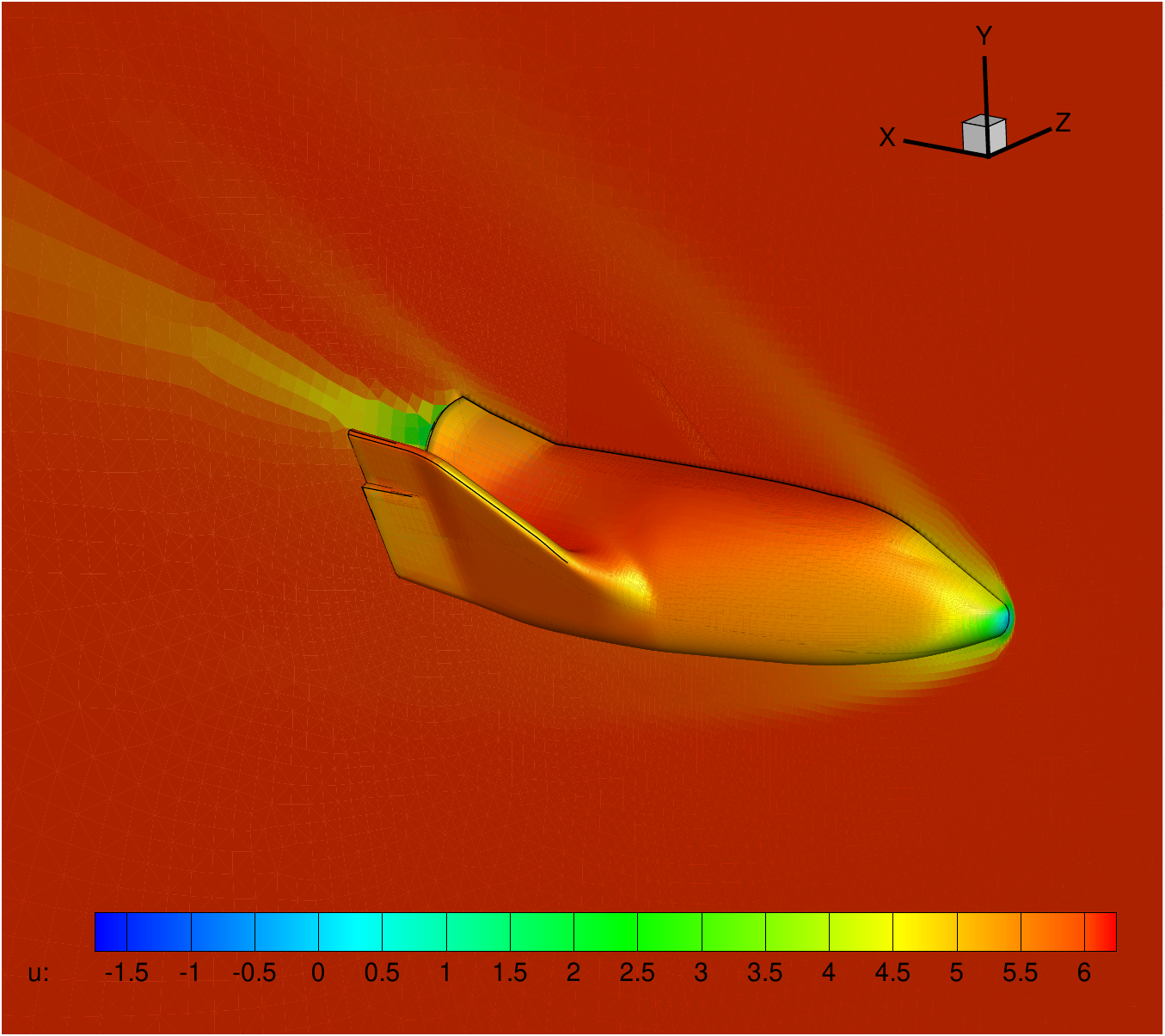}
        \caption{\(u\)}
    \end{subfigure}
    \begin{subfigure}{0.49\textwidth}
        \centering
        \includegraphics[width=\linewidth]{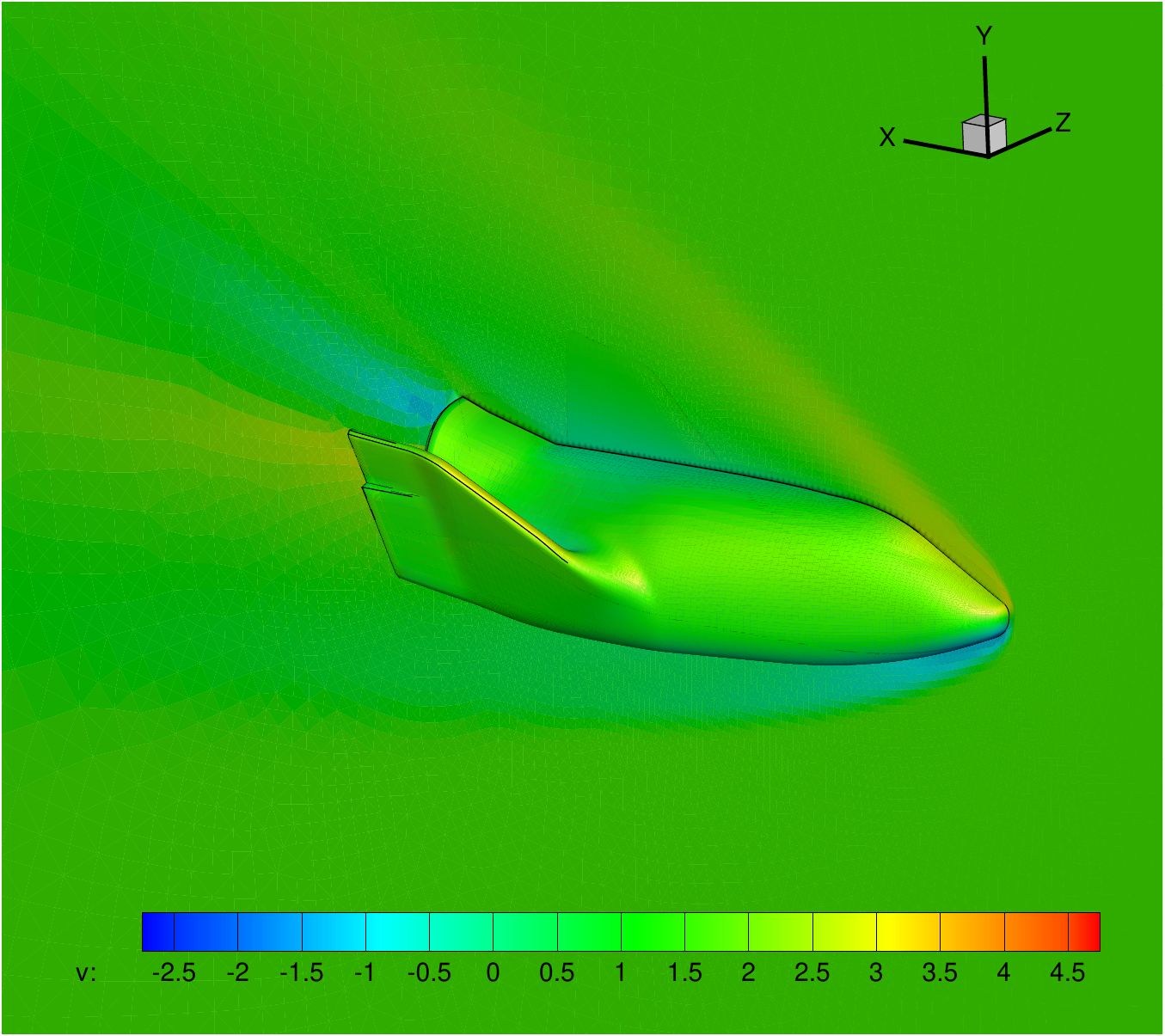}
        \caption{\(v\)}
    \end{subfigure}
    \caption{Three-dimensional $x38$ case at $\mathrm{Kn}=10^{-4}$.}
    \label{fig:x38_kn1em4_uv}
\end{figure}

\begin{figure}[!htbp]
    \centering
    \begin{subfigure}{0.49\textwidth}
        \centering
        \includegraphics[width=\linewidth]{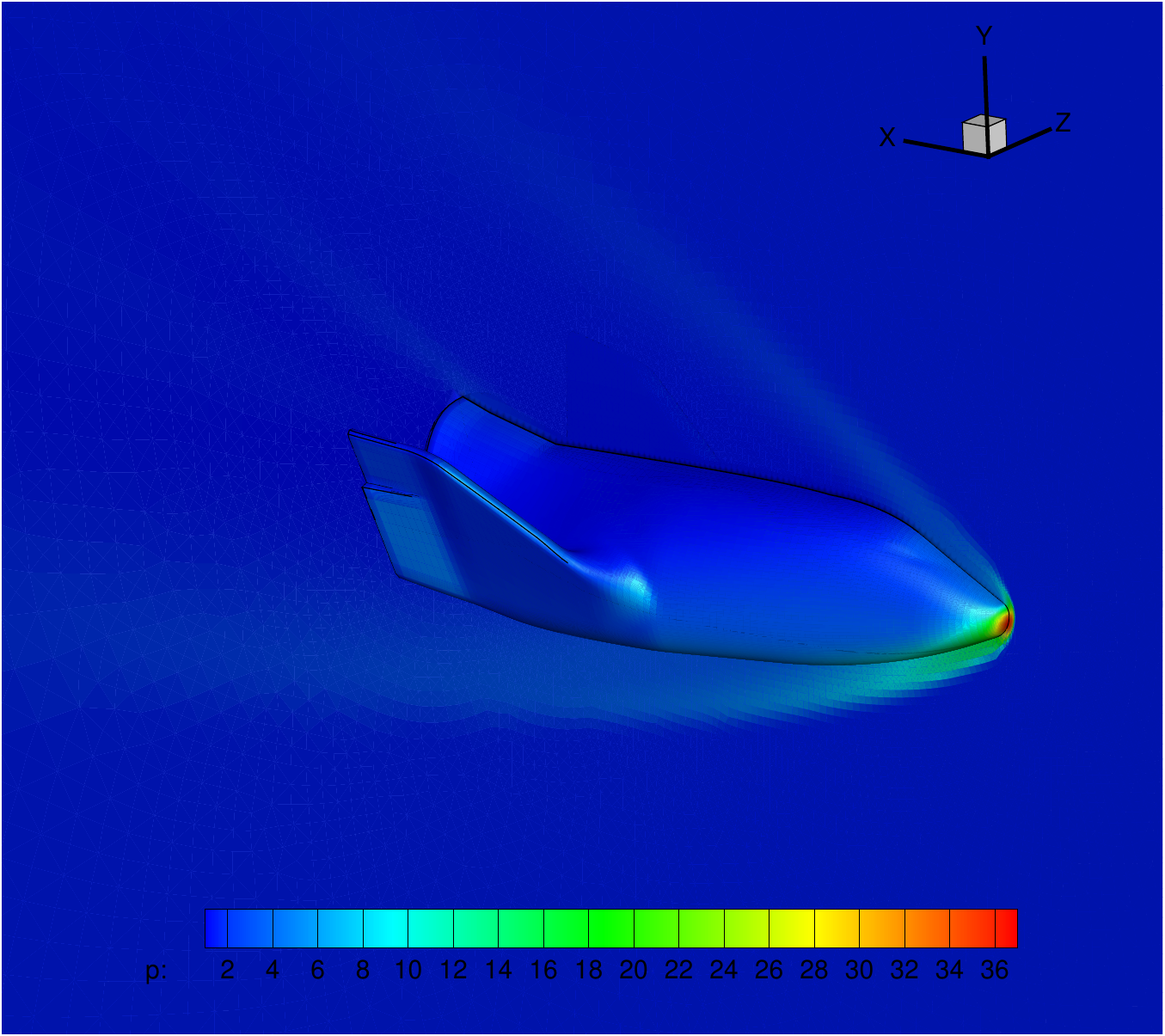}
        \caption{Pressure}
    \end{subfigure}
    \begin{subfigure}{0.49\textwidth}
        \centering
        \includegraphics[width=\linewidth]{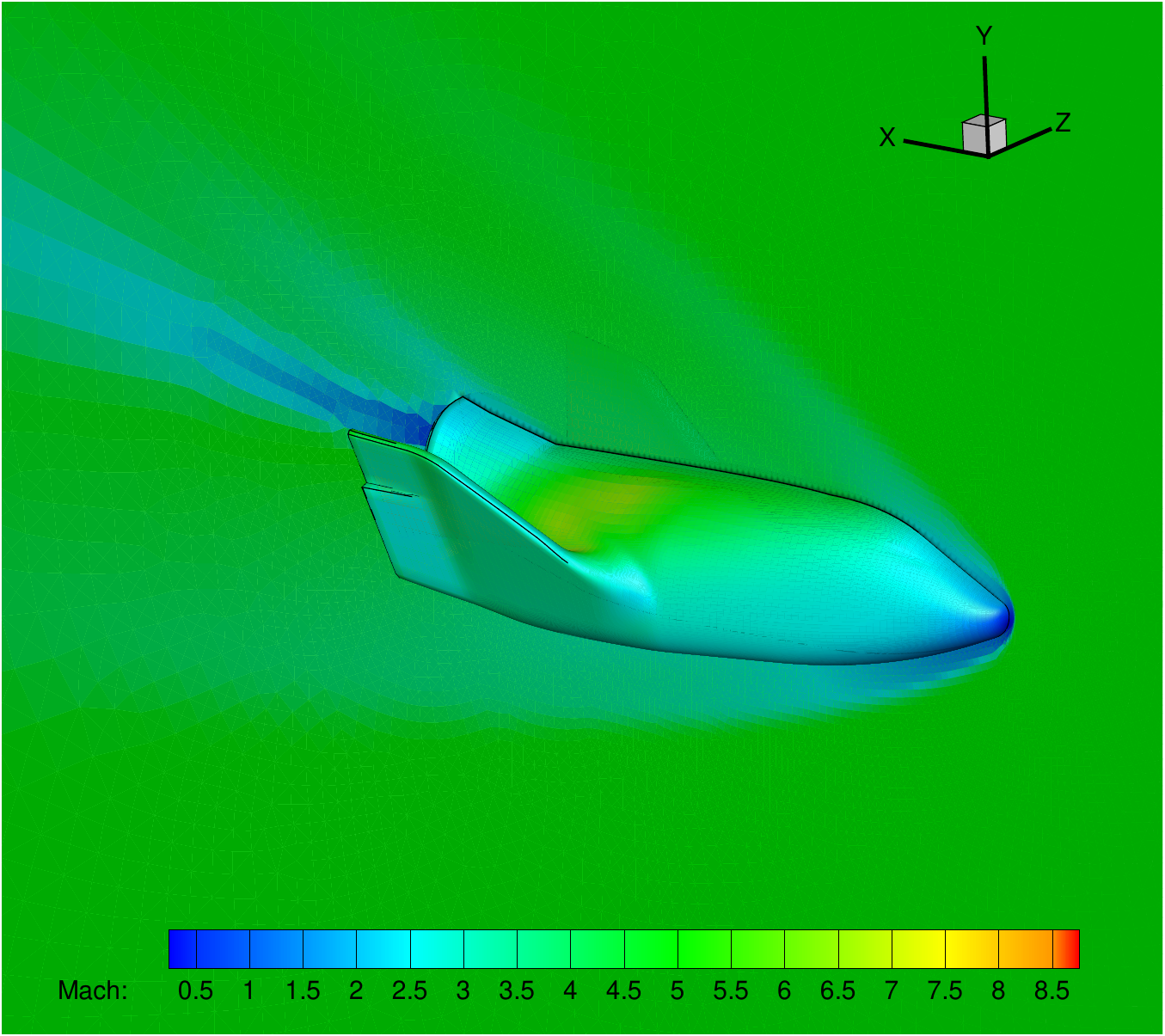}
        \caption{Mach number}
    \end{subfigure}
    \caption{Three-dimensional $x38$ case at $\mathrm{Kn}=10^{-4}$.}
    \label{fig:x38_kn1em4_pma}
\end{figure}

\section{Conclusion}\label{sec:conclusion}

In this work, we have formulated a wave-particle decomposition featuring a local kinetic horizon for kinetic relaxation equations. At the theoretical level, the characteristic integral solution effectively partitions the distribution function into an accumulated wave contribution and a collisionless particle contribution. The resulting unified wave-particle system comprises a wave equation, a particle equation, and a source-free total conservation law. Furthermore, the wave operator admits a Chapman--Enskog expansion, whose moments yield horizon-weighted Euler and Navier--Stokes fluxes, while the particle equation accurately captures the remaining kinetic transport.

At the algorithmic level, we have constructed a conservative macro-micro numerical method based on this coupled system. The total conservative variables are advanced within a finite-volume framework using the sum of the wave and particle fluxes. The wave flux is evaluated via a Navier--Stokes gas-kinetic scheme, whereas the particle flux is computed using either a deterministic discrete-ordinate discretization or a Monte Carlo particle method. Coupling these components through the target state, the local-horizon source, and conservative moment reconstruction ensures that the total macroscopic conservation law is strictly preserved at the discrete level.

This formulation distinguishes itself from the global time-step wave-particle splitting of the UGKWP method by defining the decomposition directly at the PDE level and by tying the particle fraction to a local kinetic horizon. Additionally, it departs from classical macro-micro decompositions: rather than subtracting an equilibrium target from the full distribution, the particle component is generated by the collisionless factor of the integral solution and is therefore interpreted as a fractional kinetic population rather than a signed residual.

Our analysis demonstrates the discrete asymptotic-preserving continuum limit, rarefied-regime consistency, and the appropriate scaling of active kinetic degrees of freedom. Comprehensive numerical experiments---encompassing smooth and discontinuous one-dimensional flows, wall-bounded cavity flows, hypersonic cylinder flows, and complex three-dimensional configurations---confirm the method's robustness. These tests indicate that the proposed formulation seamlessly retains kinetic transport in non-equilibrium regions while recovering a purely hydrodynamic discretization in continuum-dominated regimes.

Future work will focus on higher-order discretizations of the local-horizon source, adaptive velocity-space treatments for the deterministic particle solver, and large-scale parallel implementations tailored for hypersonic rarefied-continuum aerodynamics and other multiscale kinetic applications.

\section*{Acknowledgements}

The authors are partially supported by the National Key R\&D Program of China (2022YFA1004500). Chang Liu is partially supported by the National Natural Science Foundation of China (12102061, 12031001), the Beijing Natural Science Foundation (Z230003), and the Presidential Foundation of the China Academy of Engineering Physics (YZJJZQ2022017). Kun Xu is partially supported by the National Natural Science Foundation of China (92371107) and the Hong Kong Research Grants Council (16208324).

\end{document}